\documentclass[11pt,letterpaper]{article}

\usepackage{setspace}

\usepackage{bm,aliascnt,amsthm,amsmath,amsfonts,rotate}

\usepackage[]{datetime}
\usepackage{todonotes}
\usepackage{subcaption}
\usepackage{enumerate}
\usepackage{cite}
\usepackage{enumitem,xspace}   

\usepackage[letterpaper, top=.8in, bottom=.8in, left=1in, right=1in, includefoot]{geometry}

\usepackage[pdftex,
backref=page,
plainpages = false,
pdfpagelabels,
hyperfootnotes=false,
pdfpagemode=FullScreen,
bookmarks=true,
bookmarksopen = true,
bookmarksnumbered = true,
breaklinks = true,
hyperfigures,
pagebackref,
urlcolor = magenta,
urlcolor = black!30!green,
anchorcolor = green,
hyperindex = true,
colorlinks = true,
linkcolor = black!30!blue,
citecolor = black!30!green
]{hyperref}

\newcounter{func}
\newcommand{\newfun}[1]{f_{\refstepcounter{func}\label{#1}\thefunc}}
\newcommand{\funref}[1]{\hyperref[#1]{f_{\ref*{#1}}}} 

\newcounter{con}
\newcommand{\newcon}[1]{c_{\refstepcounter{con}\label{#1}\thecon}}
\newcommand{\conref}[1]{\hyperref[#1]{c_{\ref*{#1}}}} 

\usepackage[polutonikogreek,english,]{babel}
\usepackage[utf8]{inputenc}

\usepackage[T1]{fontenc} 
\usepackage{alphabeta}

\newtheorem{observation}{Observation}
\newtheorem{proposition}{Proposition}
\newtheorem{lemma}{Lemma}

\newtheorem{claim}{Claim}
\newtheorem{theorem}{Theorem}

\usepackage{color}
\definecolor{MidnightBlack}{rgb}{0.1,0.1,.34}
\definecolor{MidnightBlue}{rgb}{0.1,0.1,0.44}
\definecolor{Black}{rgb}{0,0, 0}
\definecolor{Blue}{rgb}{0, 0 ,1}
\definecolor{Red}{rgb}{1, 0 ,0}
\definecolor{White}{rgb}{1, 1, 1}
\definecolor{Grey}{rgb}{.6, .6, .6}
\definecolor{Mygreen}{rgb}{.0, .7, .0}
\definecolor{Yellow}{rgb}{.55,.55,0}
\definecolor{Mustard}{rgb}{1.0, 0.86, 0.35}
\definecolor{applegreen}{rgb}{0.55, 0.71, 0.0}
\definecolor{darkturquoise}{rgb}{0.0, 0.81, 0.82}
\definecolor{celestialblue}{rgb}{0.29, 0.59, 0.82}
\definecolor{green_yellow}{rgb}{0.68, 1.0, 0.18}
\definecolor{crimsonglory}{rgb}{0.75, 0.0, 0.2}
\definecolor{darkmagenta}{rgb}{0.30, 0.0, 0.30}
\definecolor{internationalorange}{rgb}{1.0, 0.31, 0.0}
\definecolor{darkorange}{rgb}{1.0, 0.55, 0.0}
\definecolor{ao}{rgb}{0.0, 0.5, 0.0}
\definecolor{awesome}{rgb}{1.0, 0.13, 0.32}
\definecolor{darkblue}{rgb}{0,0,0.545}
\definecolor{gold}{rgb}{1,0.843,0}

\newcommand{\blue}[1]{{\color{Blue}#1}}

\newcommand{\green}[1]{{\color{Mygreen}#1}}

\newcommand{\darkorange}[1]{{\color{darkorange}#1}}

\newcommand{\hh}{\end{document}}

\usepackage{float}
\usepackage{tikz,pgf}
\usetikzlibrary{backgrounds}
\usetikzlibrary{calc,3d}
\usetikzlibrary{intersections,decorations.pathmorphing,shapes,decorations.pathreplacing,fit,fadings}

\usetikzlibrary{arrows.meta}

\makeatletter

\pgfdeclarearrow{
  name = ipe _linear,
  defaults = {
    length = +1bp,
    width  = +.666bp,
    line width = +0pt 1,
  },
  setup code = {
    \pgfarrowssetbackend{0pt}
    \pgfarrowssetvisualbackend{
      \pgfarrowlength\advance\pgf@x by-.5\pgfarrowlinewidth}
    \pgfarrowssetlineend{\pgfarrowlength}
    \ifpgfarrowreversed
      \pgfarrowssetlineend{\pgfarrowlength\advance\pgf@x by-.5\pgfarrowlinewidth}
    \fi
    \pgfarrowssettipend{\pgfarrowlength}
    \pgfarrowshullpoint{\pgfarrowlength}{0pt}
    \pgfarrowsupperhullpoint{0pt}{.5\pgfarrowwidth}
    \pgfarrowssavethe\pgfarrowlinewidth
    \pgfarrowssavethe\pgfarrowlength
    \pgfarrowssavethe\pgfarrowwidth
  },
  drawing code = {
    \pgfsetdash{}{+0pt}
    \ifdim\pgfarrowlinewidth=\pgflinewidth\else\pgfsetlinewidth{+\pgfarrowlinewidth}\fi
    \pgfpathmoveto{\pgfqpoint{0pt}{.5\pgfarrowwidth}}
    \pgfpathlineto{\pgfqpoint{\pgfarrowlength}{0pt}}
    \pgfpathlineto{\pgfqpoint{0pt}{-.5\pgfarrowwidth}}
    \pgfusepathqstroke
  },
  parameters = {
    \the\pgfarrowlinewidth,%
    \the\pgfarrowlength,%
    \the\pgfarrowwidth,%
  },
}

\pgfdeclarearrow{
  name = ipe _pointed,
  defaults = {
    length = +1bp,
    width  = +.666bp,
    inset  = +.2bp,
    line width = +0pt 1,
  },
  setup code = {
    \pgfarrowssetbackend{0pt}
    \pgfarrowssetvisualbackend{\pgfarrowinset}
    \pgfarrowssetlineend{\pgfarrowinset}
    \ifpgfarrowreversed
      \pgfarrowssetlineend{\pgfarrowlength}
    \fi
    \pgfarrowssettipend{\pgfarrowlength}
    \pgfarrowshullpoint{\pgfarrowlength}{0pt}
    \pgfarrowsupperhullpoint{0pt}{.5\pgfarrowwidth}
    \pgfarrowshullpoint{\pgfarrowinset}{0pt}
    \pgfarrowssavethe\pgfarrowinset
    \pgfarrowssavethe\pgfarrowlinewidth
    \pgfarrowssavethe\pgfarrowlength
    \pgfarrowssavethe\pgfarrowwidth
  },
  drawing code = {
    \pgfsetdash{}{+0pt}
    \ifdim\pgfarrowlinewidth=\pgflinewidth\else\pgfsetlinewidth{+\pgfarrowlinewidth}\fi
    \pgfpathmoveto{\pgfqpoint{\pgfarrowlength}{0pt}}
    \pgfpathlineto{\pgfqpoint{0pt}{.5\pgfarrowwidth}}
    \pgfpathlineto{\pgfqpoint{\pgfarrowinset}{0pt}}
    \pgfpathlineto{\pgfqpoint{0pt}{-.5\pgfarrowwidth}}
    \pgfpathclose
    \ifpgfarrowopen
      \pgfusepathqstroke
    \else
      \ifdim\pgfarrowlinewidth>0pt\pgfusepathqfillstroke\else\pgfusepathqfill\fi
    \fi
  },
  parameters = {
    \the\pgfarrowlinewidth,%
    \the\pgfarrowlength,%
    \the\pgfarrowwidth,%
    \the\pgfarrowinset,%
    \ifpgfarrowopen o\fi%
  },
}

\newdimen\ipeminipagewidth

\tikzstyle{ipe import} = [
  x=1bp, y=1bp,
  ipe node stretch/.store in=\ipenodestretch,
  ipe stretch normal/.style={ipe node stretch=1},
  ipe stretch normal,
  ipe node/.style={
    anchor=base west, inner sep=0, outer sep=0, scale=\ipenodestretch
  },
  ipe mark scale/.store in=\ipemarkscale,
  ipe mark tiny/.style={ipe mark scale=1.5},
  ipe mark small/.style={ipe mark scale=2},
  ipe mark normal/.style={ipe mark scale=3},
  ipe mark large/.style={ipe mark scale=5},
  ipe mark normal, 
  ipe circle/.pic={
    \draw[line width=0.2*\ipemarkscale]
      (0,0) circle[radius=0.5*\ipemarkscale];
    \coordinate () at (0,0);
  },
  ipe disk/.pic={
    \fill (0,0) circle[radius=0.6*\ipemarkscale];
    \coordinate () at (0,0);
  },
  ipe fdisk/.pic={
    \filldraw[line width=0.2*\ipemarkscale]
      (0,0) circle[radius=0.5*\ipemarkscale];
    \coordinate () at (0,0);
  },
  ipe box/.pic={
    \draw[line width=0.2*\ipemarkscale, line join=miter]
      (-.5*\ipemarkscale,-.5*\ipemarkscale) rectangle
      ( .5*\ipemarkscale, .5*\ipemarkscale);
    \coordinate () at (0,0);
  },
  ipe square/.pic={
    \fill
      (-.6*\ipemarkscale,-.6*\ipemarkscale) rectangle
      ( .6*\ipemarkscale, .6*\ipemarkscale);
    \coordinate () at (0,0);
  },
  ipe fsquare/.pic={
    \filldraw[line width=0.2*\ipemarkscale, line join=miter]
      (-.5*\ipemarkscale,-.5*\ipemarkscale) rectangle
      ( .5*\ipemarkscale, .5*\ipemarkscale);
    \coordinate () at (0,0);
  },
  ipe cross/.pic={
    \draw[line width=0.2*\ipemarkscale, line cap=butt]
      (-.5*\ipemarkscale,-.5*\ipemarkscale) --
      ( .5*\ipemarkscale, .5*\ipemarkscale)
      (-.5*\ipemarkscale, .5*\ipemarkscale) --
      ( .5*\ipemarkscale,-.5*\ipemarkscale);
    \coordinate () at (0,0);
  },
  /pgf/arrow keys/.cd,
  ipe arrow normal/.style={scale=1},
  ipe arrow tiny/.style={scale=.4},
  ipe arrow small/.style={scale=.7},
  ipe arrow large/.style={scale=1.4},
  ipe arrow normal,
  /tikz/.cd,
  ipe arrows/.style={
    ipe normal/.tip={
      ipe _pointed[length=1bp, width=.666bp, inset=0bp,
                   quick, ipe arrow normal]},
    ipe pointed/.tip={
      ipe _pointed[length=1bp, width=.666bp, inset=0.2bp,
                   quick, ipe arrow normal]},
    ipe linear/.tip={
      ipe _linear[length = 1bp, width=.666bp,
                  ipe arrow normal, quick]},
    ipe fnormal/.tip={ipe normal[fill=white]},
    ipe fpointed/.tip={ipe pointed[fill=white]},
    ipe double/.tip={ipe normal[] ipe normal},
    ipe fdouble/.tip={ipe fnormal[] ipe fnormal},
    ipe arc/.tip={ipe normal},
    ipe farc/.tip={ipe fnormal},
    ipe ptarc/.tip={ipe pointed},
    ipe fptarc/.tip={ipe fpointed},
  },
  ipe arrows, 
]

\tikzset{
  rgb color/.code args={#1=#2}{%
    \definecolor{tempcolor-#1}{rgb}{#2}%
    \tikzset{#1=tempcolor-#1}%
  },
}

\makeatother

\tikzstyle{ipe stylesheet} = [
  ipe import,
  even odd rule,
  line join=round,
  line cap=butt,
  ipe pen normal/.style={line width=0.4},
  ipe pen heavier/.style={line width=0.8},
  ipe pen fat/.style={line width=1.2},
  ipe pen ultrafat/.style={line width=2},
  ipe pen normal,
  ipe mark normal/.style={ipe mark scale=3},
  ipe mark large/.style={ipe mark scale=5},
  ipe mark small/.style={ipe mark scale=2},
  ipe mark tiny/.style={ipe mark scale=1.1},
  ipe mark normal,
  /pgf/arrow keys/.cd,
  ipe arrow normal/.style={scale=7},
  ipe arrow large/.style={scale=10},
  ipe arrow small/.style={scale=5},
  ipe arrow tiny/.style={scale=3},
  ipe arrow normal,
  /tikz/.cd,
  ipe arrows, 
  <->/.tip = ipe normal,
  ipe dash normal/.style={dash pattern=},
  ipe dash dotted/.style={dash pattern=on 1bp off 3bp},
  ipe dash dashed/.style={dash pattern=on 4bp off 4bp},
  ipe dash dash dotted/.style={dash pattern=on 4bp off 2bp on 1bp off 2bp},
  ipe dash dash dot dotted/.style={dash pattern=on 4bp off 2bp on 1bp off 2bp on 1bp off 2bp},
  ipe dash normal,
  ipe node/.append style={font=\normalsize},
  ipe stretch normal/.style={ipe node stretch=1},
  ipe stretch normal,
  ipe opacity 10/.style={opacity=0.1},
  ipe opacity 30/.style={opacity=0.3},
  ipe opacity 50/.style={opacity=0.5},
  ipe opacity 75/.style={opacity=0.75},
  ipe opacity opaque/.style={opacity=1},
  ipe opacity opaque,
]
\definecolor{red}{rgb}{1,0,0}
\definecolor{blue}{rgb}{0,0,1}
\definecolor{green}{rgb}{0,1,0}
\definecolor{yellow}{rgb}{1,1,0}
\definecolor{orange}{rgb}{1,0.647,0}
\definecolor{gold}{rgb}{1,0.843,0}
\definecolor{purple}{rgb}{0.627,0.125,0.941}
\definecolor{gray}{rgb}{0.745,0.745,0.745}
\definecolor{brown}{rgb}{0.647,0.165,0.165}
\definecolor{navy}{rgb}{0,0,0.502}
\definecolor{pink}{rgb}{1,0.753,0.796}
\definecolor{seagreen}{rgb}{0.18,0.545,0.341}
\definecolor{turquoise}{rgb}{0.251,0.878,0.816}
\definecolor{violet}{rgb}{0.933,0.51,0.933}
\definecolor{darkblue}{rgb}{0,0,0.545}
\definecolor{darkcyan}{rgb}{0,0.545,0.545}
\definecolor{darkgray}{rgb}{0.663,0.663,0.663}
\definecolor{darkgreen}{rgb}{0,0.392,0}
\definecolor{darkmagenta}{rgb}{0.545,0,0.545}
\definecolor{darkorange}{rgb}{1,0.549,0}
\definecolor{darkred}{rgb}{0.545,0,0}
\definecolor{lightblue}{rgb}{0.678,0.847,0.902}
\definecolor{lightcyan}{rgb}{0.878,1,1}
\definecolor{lightgray}{rgb}{0.827,0.827,0.827}
\definecolor{lightgreen}{rgb}{0.565,0.933,0.565}
\definecolor{lightyellow}{rgb}{1,1,0.878}
\definecolor{black}{rgb}{0,0,0}
\definecolor{white}{rgb}{1,1,1}

\tikzset{red node/.style={draw=red, circle, fill = red, minimum size = 4pt, inner sep = 0pt}}
\tikzset{yellow node/.style={draw=yellow, circle, fill = yellow, minimum size = 4pt, inner sep = 0pt}}
\tikzset{blue node/.style={draw=celestialblue, circle, fill =celestialblue, minimum size = 4pt, inner sep = 0pt}}
\tikzset{triangle/.style = { regular polygon, regular polygon sides=3, rotate=180}}
\tikzset{small red/.style={draw=red, triangle, fill = red, minimum size = 2pt, inner sep = 0pt}}

\tikzset{black node/.style={draw, circle, fill = black, minimum size = 3pt, inner sep = 0pt}}
\tikzset{small black node/.style={draw, circle, fill = black, minimum size = 3pt, inner sep = 0pt}}
\tikzset{model node/.style={draw=celestialblue, circle, fill = celestialblue, minimum size = 5pt, inner sep = 0pt}}
\tikzset{model node small/.style={draw=celestialblue, circle, fill = celestialblue, minimum size = 3pt, inner sep = 0pt}}
\tikzset{rep node/.style={draw=red, circle, fill = red, minimum size = 3pt, inner sep = 0pt}}
\tikzset{track node 1/.style={draw, circle, fill = black, minimum size = 2pt, inner sep = 0pt}}
\tikzset{track node 2/.style={draw=black!30!white, circle, fill = black!30!white, minimum size = 2pt, inner sep = 0pt}}
\tikzset{track node 3/.style={draw=black!10!white, circle, fill = black!10!white, minimum size = 2pt, inner sep = 0pt}}
\tikzfading[name=fade out, inner color=transparent!0, outer
color=transparent!100]
\tikzfading[name=fade in, inner color=transparent!100, outer
color=transparent!0]
\tikzfading[name=middle, top color=transparent!80, bottom
color=transparent!80, middle color=transparent!0]
\tikzset{terminal/.style={circle, draw=black, fill=black!50,
                        inner sep=0pt, minimum width=4pt}}
\tikzset{root terminal/.style={circle,thick, draw=black, fill=orange!90!yellow,
                        inner sep=1pt, minimum width=6pt}}
\tikzset{simple/.style={circle, draw=black, fill=black,
                        inner sep=0pt, minimum width=4pt}}
\tikzset{treenode/.style={circle, draw=black, fill=black,
                        inner sep=0pt, minimum width=4pt}}
\tikzset{
diagonal fill/.style 2 args={fill=#2, path picture={
\fill[#1, sharp corners] (path picture bounding box.south west) -|
                         (path picture bounding box.north east) -- cycle;}},
reversed diagonal fill/.style 2 args={fill=#2, path picture={
\fill[#1, sharp corners] (path picture bounding box.north west) |- 
                         (path picture bounding box.south east) -- cycle;}}
}

\tikzstyle{ipe stylesheet} = [
  ipe import,
  even odd rule,
  line join=round,
  line cap=butt,
  ipe pen normal/.style={line width=0.4},
  ipe pen heavier/.style={line width=0.8},
  ipe pen fat/.style={line width=1.2},
  ipe pen ultrafat/.style={line width=2},
  ipe pen normal,
  ipe mark normal/.style={ipe mark scale=3},
  ipe mark large/.style={ipe mark scale=5},
  ipe mark small/.style={ipe mark scale=2},
  ipe mark tiny/.style={ipe mark scale=1.1},
  ipe mark normal,
  /pgf/arrow keys/.cd,
  ipe arrow normal/.style={scale=7},
  ipe arrow large/.style={scale=10},
  ipe arrow small/.style={scale=5},
  ipe arrow tiny/.style={scale=3},
  ipe arrow normal,
  /tikz/.cd,
  ipe arrows, 
  <->/.tip = ipe normal,
  ipe dash normal/.style={dash pattern=},
  ipe dash dotted/.style={dash pattern=on 1bp off 3bp},
  ipe dash dashed/.style={dash pattern=on 4bp off 4bp},
  ipe dash dash dotted/.style={dash pattern=on 4bp off 2bp on 1bp off 2bp},
  ipe dash dash dot dotted/.style={dash pattern=on 4bp off 2bp on 1bp off 2bp on 1bp off 2bp},
  ipe dash normal,
  ipe node/.append style={font=\normalsize},
  ipe stretch normal/.style={ipe node stretch=1},
  ipe stretch normal,
  ipe opacity 10/.style={opacity=0.1},
  ipe opacity 30/.style={opacity=0.3},
  ipe opacity 50/.style={opacity=0.5},
  ipe opacity 75/.style={opacity=0.75},
  ipe opacity opaque/.style={opacity=1},
  ipe opacity opaque,
]

\newcommand{\Mod}{{\rm Mod}}
\newcommand{\MSOL}{\mbox{\sf MSOL}}
\newcommand{\FOL}{{\sf FOL}}
\newcommand{\DP}{{\sf dp}}

\newcommand{\FOLDP}{\FOL[\tau{\normalfont+}\DP]}

\newcommand{\labels}[1]{\label{#1}}
\newcommand{\hw}{{\sf hw}}

\newcommand{\bd}{{\sf bd}}
\newcommand{\yes}{{\sf yes}}
\newcommand{\no}{{\sf no}}
\newcommand{\tw}{{\sf tw}}
\newcommand{\ann}{{\sf ann}}
\newcommand{\inter}{{\sf int}}
\newcommand{\remove}[1]{}

\newcommand{\bigmid}{\;\big|\;}
\newcommand{\cupall}{\pmb{\pmb{\bigcup}}}

\newcommand{\injec}{\mu}
\newcommand{\prem}{\preceq_{\sf m}}

\newcommand{\NP}{{\sf NP}\xspace}
\newcommand{\FPT}{{\sf FPT}\xspace}

\newcommand{\frR}{{\mathfrak{R}}}

\newcommand{\excl}{{\sf Excl}}
\newcommand{\imprint}{{\sf imp}}
\newcommand{\Models}{{\sf Pairings}}

\newcommand{\injection}{{\bf inj}}
\newcommand{\repact}{\mathcal{L}}

\RequirePackage{stmaryrd}
\usepackage{textcomp}
\DeclareUnicodeCharacter{2286}{\subseteq}
\DeclareUnicodeCharacter{2192}{\ifmmode\to\else\textrightarrow\fi}
\DeclareUnicodeCharacter{2203}{\ensuremath\exists}
\DeclareUnicodeCharacter{183}{\cdot}
\DeclareUnicodeCharacter{2200}{\forall}
\DeclareUnicodeCharacter{2264}{\leq}
\DeclareUnicodeCharacter{2265}{\geq}
\DeclareUnicodeCharacter{8614}{\mathbin{\mapsto}}
\DeclareUnicodeCharacter{8656}{\Leftarrow}
\DeclareUnicodeCharacter{8657}{\Uparrow}
\DeclareUnicodeCharacter{8658}{\Rightarrow}
\DeclareUnicodeCharacter{8659}{\Downarrow}
\DeclareUnicodeCharacter{8669}{\rightsquigarrow}
\newcommand{\eqdef}{\stackrel{{\scriptsize\rm def}}{=}}
\DeclareUnicodeCharacter{8797}{\eqdef}
\DeclareUnicodeCharacter{8870}{\vdash}
\DeclareUnicodeCharacter{8873}{\Vdash}
\DeclareUnicodeCharacter{22A7}{\models}
\DeclareUnicodeCharacter{9121}{\lceil}
\DeclareUnicodeCharacter{9123}{\lfloor}
\DeclareUnicodeCharacter{9124}{\rceil}
\DeclareUnicodeCharacter{2208}{\in}
\DeclareUnicodeCharacter{9126}{\rfloor}
\DeclareUnicodeCharacter{9655}{\triangleright}
\DeclareUnicodeCharacter{9665}{\triangleleft}
\DeclareUnicodeCharacter{9671}{\diamond}
\DeclareUnicodeCharacter{9675}{\circ}
\DeclareUnicodeCharacter{10178}{\bot}
\DeclareUnicodeCharacter{10214}{} 
\DeclareUnicodeCharacter{10215}{} 
\DeclareUnicodeCharacter{10229}{\longleftarrow}
\DeclareUnicodeCharacter{10230}{\longrightarrow}
\DeclareUnicodeCharacter{10231}{\longleftrightarrow}
\DeclareUnicodeCharacter{10232}{\Longleftarrow}
\DeclareUnicodeCharacter{10233}{\Longrightarrow}
\DeclareUnicodeCharacter{10234}{\Longleftrightarrow}
\DeclareUnicodeCharacter{10236}{\longmapsto}
\DeclareUnicodeCharacter{10238}{\Longmapsto} 
\DeclareUnicodeCharacter{10503}{\Mapsto}    
\DeclareUnicodeCharacter{10971}{\mathrel{\not\hspace{-0.2em}\cap}}
\DeclareUnicodeCharacter{65294}{\ldotp}
\DeclareUnicodeCharacter{65372}{\mid}

\DeclareFontEncoding{LS1}{}{}
\DeclareFontSubstitution{LS1}{stix}{m}{n}
\DeclareSymbolFont{symbolsstix}{LS1}{stixscr}{m}{n}
\SetSymbolFont{symbolsstix}{bold}{LS1}{stixscr}{b}{n}
\DeclareMathSymbol{\mathvisiblespace}{0}{symbolsstix}{"B6}
\newcommand{\mathspace}{\mathvisiblespace}

\usepackage{titling}
\thanksmarkseries{arabic}
\newcommand*\samethanks[1][\value{footnote}]{\footnotemark[#1]}

\begin{document}

\title{Model-Checking for First-Order Logic with Disjoint Paths Predicates in Proper Minor-Closed Graph Classes\thanks{An extended abstract of this paper appeared in the \emph{Proceedings of the 34th Annual ACM-SIAM Symposium on Discrete Algorithms (SODA 2023).}}}

\author{\bigskip
Petr A. Golovach\thanks{Department of Informatics, University of Bergen, Norway. Supported by the Research Council of Norway  via the project BWCA (314528).\
Email:  \texttt{petr.golovach@uib.no}.}\and Giannos Stamoulis\thanks{LIRMM, Univ Montpellier, CNRS, Montpellier, France. {Supported}  by the ANR projects DEMOGRAPH (ANR-16-CE40-0028), ESIGMA (ANR-17-CE23-0010), and the French-German Collaboration ANR/DFG Project UTMA (ANR-20-CE92-0027). The third author was also supported by the  \emph{Ministère de l'Europe et des Affaires étrangères}
(MEAE) and the \emph{Ministère de l'Enseignement supérieur et de la
Recherche} (MESR), via the Franco-Norwegian project PHC Aurora projet N° 51260WL (2024).\  Emails:
\texttt{giannos.stamoulis@lirmm.fr}, \texttt{sedthilk@thilikos.info}.}
\and
Dimitrios  M. Thilikos\samethanks[3]}
\date{}

\pagenumbering{Alph}
\maketitle

\begin{abstract}
\noindent The \emph{disjoint paths logic}, \textsf{FOL+DP}, is an extension of First-Order Logic (\textsf{FOL}) with the extra atomic predicate $\textsf{dp}_k(x_1,y_1,\ldots,x_k,y_k),$ expressing the existence of internally vertex-disjoint paths between $x_i$ and $y_i,$ for $i\in\{1,\ldots, k\}$. This logic can express a wide variety of problems that escape the expressibility potential of \textsf{FOL}. We prove that for every proper minor-closed graph class, model-checking for \textsf{FOL+DP} can be done in quadratic time. We also introduce an extension of \textsf{FOL+DP}, namely the \emph{scattered disjoint paths logic}, \textsf{FOL+SDP}, where we further consider the atomic predicate $s\textsf{-sdp}_k(x_1,y_1,\ldots,x_k,y_k),$ demanding that the disjoint paths are within distance bigger than some fixed value $s$. Using the same technique we prove that model-checking for \textsf{FOL+SDP} can be done in quadratic time on classes of graphs with bounded Euler genus.
\end{abstract}
\thispagestyle{empty}


\newpage\thispagestyle{empty}

\tableofcontents\thispagestyle{empty}

\newpage
\pagenumbering{arabic}
\newpage

\setcounter{page}{1}
\section{Introduction}

Logic plays a fundamental role in algorithmic research. It provides a universal language for 
formally describing computational problems and is important 
for the investigation of their computational complexity.  
In many cases, the accumulation of knowledge on algorithm design revealed that several algorithmic techniques have conceptual similarities that result from some common logical description of the problems where they apply. These similarities become effective when the inputs of the corresponding problems have certain structural characteristics. In some cases, this empirical evidence has been materialized in the so called Algorithmic Meta-Theorems (AMTs), a term introduced by Martin Grohe in \cite{Grohe07logi}. Such theorems typically provide two types of conditions, a logical one and a combinatorial one, such that every problem that is expressible by the logical condition can be solved efficiently when its inputs are restricted by the combinatorial condition. The importance of AMTs resides to the fact that they 
are able to unify wide families of computational problems (and also the algorithmic solutions for them) under a single 
model-theoretic/combinatorial framework (see~\cite{Grohe07logi,Kreutzer09algo,GroheK11meth}).

\subsection{AMTs for \MSOL\ and \FOL}

Probably, the most prototypical AMT is  known as {\sl Courcelle's Theorem} proved in~\cite{Courcelle90them} (see also~\cite{BoriePT92auto,ArnborgLS91easy} and~\cite{DreierR21appro}),
asserting that every problem on graphs that is expressible by a sentence $\varphi$ in Monadic Second Order Logic (\MSOL)
can be solved in time\footnote{Let ${\bf t}=(x_{1},\ldots,x_{l})\in \mathbb{N}^l$ and $f,g: \mathbb{N}
\rightarrow \mathbb{N}.$
We adopt the notation $f(n)=\mathcal{O}_{\textbf{t}}(g(n))$ in order to denote that there exists a computable
function $\ell:\mathbb{N}^{l} \rightarrow \mathbb{N}$ such that  $f(n)=\mathcal{O}(\ell({\bf t})\cdot g(n)).$}
$\mathcal{O}_{|\varphi|,k}(n)$, when restricted to graphs of treewith at most $k$.
Clearly, when $\varphi$ and $k$ are fixed, this readily implies a linear-time algorithm. However, 
we prefer to display the dependencies on $\varphi$ and $k$, under the $\mathcal{O}_{|\varphi|,k}$ notation, so as to make clear that this theorem
provides a linear-time parameterized\footnote{More generally, we  say that a computational problem, parameterized by $k$, is {\sf FPT} (Fixed Parameter Tractable)
when it admits an algorithm running in  $\mathcal{O}_{k}(n^{O(1)})$ time. We  assume that the reader is familiar 
with the basic concepts of parameterized algorithms and parameterized complexity classes -- see \cite{cygan2015parameterized,FlumG06para,Niedermeier06invi}.} algorithm, for every problem expressible in \MSOL, when it is parameterized by treewidth. 

Clearly, the logical/combinatorial compromise of Courcelle's theorem is not the only possible one. In fact,
each AMT constitutes a different compromise between the logical and the combinatorial condition
and a considerable amount of research in the theory of algorithms has been  dedicated to the 
conception of alternative such compromises.  Also, for particular logics,  research has been dedicated to the identification of their combinatorial horizon, i.e., the most general combinatorial conditions that can 
accompany them in an AMT. For instance, for \MSOL,  the meta-algorithmic horizon is, under certain assumptions,  delimited by the graph classes of bounded treewidth (see~\cite{GanianHLORS14lower,KreutzerT10lower,BojanczykP16defi}).

\paragraph{Meta-algorithmics of \FOL.}
Given that the meta-algorithmic limits of \MSOL\ are fairly well understood, research on AMTs has been largely oriented to the meta-algorithmics of  First-Order Logic (\FOL).
The two most powerful results in this direction concern two different types of combinatorial conditions. The first was given by Grohe, Kreutzer, and Siebertz in \cite{GroheKS17dec} and is the graph class property of being {\sl nowhere dense}. The second was given by  Bonnet,   Kim,  Thomassé, and  Watrigant in~\cite{BonnetKTW22twinI}
and is the graph class property of having {\sl bounded twin-width}. We should stress here that 
the notion of being nowhere dense originates from a long line of research on graph sparsity, initiated by Nešetřil and Ossona de Mendez in \cite{NesetrilM05theg} (see also \cite{DvorakKT13,KreutzerD09param,NesetrilM11aonno,NesetrilM12spar,KreutzerRS19polyn,PilipczukST18onthe}). On the other side, twin-width is a recently introduced graph parameter, defined in terms of sequences of vertex identifications (see \cite{BonnetKTW22twinI,BonnetGMSTT22twinIV,BonnetKRT22twinVI,BonnetCKKLT22twinVIII,BonnetGTT22twinVII,BonnetGKTW21twinII,BonnetG0TW21twinIII,Bonnet0RTW21twinkernels,BonnetNMST21twinperm,BergeBD21deci,GajarskyPT21stab,GanianPSSS22weig,BonnetKW22redu,PilipczukS22grap,JacobP22bound,DreierGJMR22twin,PilipczukSZ22comp,KratschNS22ontr,SchidlerS22asat,PetterssonS22boun,HlinenyP22twin,BalabanH21twin} for a sample of the vibrant current research on the algorithmic and combinatorial properties of twin-width).
Seminal results on the above two combinatorial conditions indicate that, under certain assumptions, they approach the combinatorial horizon of \FOL\ (see \cite{DvorakKT13} for nowhere density and \cite{BonnetGMSTT22twinIV} for  bounded twin-width). Research on the meta-algorithmics of \FOL\ is nowadays quite active and has moved to several directions such 
as the study of \FOL-interpretability \cite{Bonnet22model,PilipczukOS22transd,NesetrilMS22stru,NesetrilMPRS21rankw,NesetrilRMS20linea,GajarskyKNMPST20first} or the enhancement of \FOL\ with  counting/numerical predicates~\cite{KuskeS17first,KuskeS18gaifm,DreierR21appro,GroheS18fisrt} (see also \cite{HeuvelKPQRS17model,Grange21succe,EickmeyerEH17succi,GroheS00local} for other extensions).

\paragraph{AMTs between \FOL\ and \MSOL.}
A challenging direction is the introduction of new logics whose expressive 
power is between \FOL\ and \MSOL\  that can lead to AMTs under combinatorial conditions 
that are less general than those applicable for \FOL\ and more general than those  applicable for \MSOL.
Two approaches that have been initiated in this direction are the following. The first direction is the introduction of
compound logics that can express problems whose description combines both \FOL\ and \MSOL\ queries. 
Such a compound logic has been  recently introduced in \cite{FominGSST21acomp} and yielded AMTs that  are applicable on a wide family of graph
modification problems. The second direction is the extension of \FOL\ with additional 
predicates that are not expressible in \FOL. An important step in this direction was the introduction of 
the separator logic \FOL{\sf +conn}. This extension of \FOL\ was introduced independently 
by  Schirrmacher, Siebertz, and Vigny in \cite{SchirrmacherSV22first} and 
by   Bojańczyk in \cite{Bojanczyk21separ} (under the name {\sl separator logic}), who considered, for every $k\geq 1$, the general predicate ${\sf conn}_{k}({\sf x},{\sf y},{\sf z}_1,...,{\sf z}_k)$, that evaluates to true on a graph $G$ if (the valuations of) ${\sf x}$ and ${\sf y}$ are joined in $G$ by a path  that avoids (the valuations) of the variables $\{{\sf z}_1,...,{\sf z}_k\}$.
According to the recent meta-algorithmic results of Pilipczuk, Schirrmacher, Siebertz, Toruńczyk,  and Vigny in \cite{PilipczukSSTV22algo}, 
every problem on graphs that is expressible by some formula $\varphi\in$ \FOL{\sf +conn} can be solved in time $\mathcal{O}_{|\varphi|,r}(n^3),$ where $r={\sf hj}(G)$ is the Hajós number\footnote{The \emph{Hajós number} is the maximum $h$ for which $G$ contains a subivision of $K_{h}$ as a subgraph.} of $G$. Notice that this AMT implies the existence of parameterized algorithms on graph classes with bounded Hajós number for  problems
whose definition uses connectivity queries and therefore are not expressible in \FOL. The most indicative 
example of a (meta-) problem displaying the expressibility power of  \FOL{\sf +conn}  is  {\sc Elimination Distance to $\varphi$}, 
asking whether the elimination distance 
of a graph $G$ from some model of $\varphi\in\FOL$ is at most $k$.
This problem is {\sf W[2]}-hard (when parameterized by $k$) even for simple instantiations of $\varphi$ \cite{FominGT21param}, however, it admits a time $\mathcal{O}_{|\varphi|,k,r}(n^3)$ algorithm when restricted to   graph classes with Hajós number at most $r$, because of the results in \cite{PilipczukSSTV22algo}.
For more examples of the expressibility power of  \FOL{\sf +conn}, see \cite{SchirrmacherSV22first}.
Also, in \cite{PilipczukSSTV22algo} it was proved that the tractability horizon of \FOL{\sf +conn} is, under certain assumptions, delimited by   graph classes of bounded Hajós number.

\subsection{AMTs for \FOL{\sf +}{\sf DP} and extensions}

As a next step towards a more expressive logic, Schirrmacher, Siebertz, and Vigny~\cite{SchirrmacherSV22first}
defined the, more expressive, \emph{disjoint-paths logic} \FOL{\sf +}{\sf DP}\ by adding, for every $k$,  the atomic predicates {\sf dp}$_{k}({\sf x}_1,{\sf y}_1,\ldots,{\sf x}_k,{\sf y}_k)$ that evaluate to true if there are internally vertex-disjoint paths between (the valuations of) ${\sf x}_i$ and ${\sf y}_i$, for all $i\in\{1,\ldots,k\}$.
In the same paper they defined\footnote{In the definition of \FOL{\sf +}{\sf DP} we insist that the paths are {\sl disjoint} while in~\cite{SchirrmacherSV22first} paths are required to be {\sl internally disjoint} (certainly,  these two variants define equivalent formulas). We insist on the ``complete disjointness''  as this permits us to see  \FOL{\sf +}{\sf DP} as a special case of  the more general \FOL{\sf +}{\sf SDP}  that we introduce in this paper.}   the logical hierarchy  
 \FOL{\sf +}\DP$_k$ that uses predicates for at most $k$ disjoint paths and proved that these fragments of \FOL{\sf +}{\sf DP}\ define a {\sl strict} descriptive complexity hierarchy.

  The challenging open question is whether an AMT exists for this 
logic, under some suitable combinatorial restriction that is more general than the one of having bounded treewidth (that is where Courcelle's theorem is applicable).

\paragraph{Our results.}
Given a graph $G$, we define the \emph{Hadwiger number} of $G$, denoted by $\mathsf{hw}(G)$, as the maximum $r$ for which $G$ contains $K_{r}$ as a minor.\footnote{Given two graphs $G$ and $H$, $H$ is a \emph{minor} of $G$ if $G$ contains a contraction of $H$ as a subgraph.}
Our main result is the following AMT.

\begin{theorem}
\label{@definitionen}
Every problem on graphs that is expressible by some formula $\varphi$ in \FOL{\sf +}{\sf DP} can be solved by an algorithm running in time $\mathcal{O}_{|\varphi|,r}(n^2)$, where $r$ is the Hadwiger number of $G$.
\end{theorem}

Some indicative (meta) problems whose standard parameterizations (i.e., those defined by the parameter $k$ in their inputs) are automatically classified in \FPT
because of \autoref{@definitionen} are
{\sc Minor Containment},
{\sc Topological Minor Containment},
{\sc Cyclability},
{\sc Unordered Linkability},
{\sc Ordered Linkability},
{\sc $\mathcal{F}$-Minor-Deletion},
{\sc $\mathcal{F}$-Topological Minor-Deletion},
{\sc $\mathcal{F}$-Contraction Deletion} (for bounded genus graphs),
{\sc Annotated $\mathcal{F}$-$\preceq$-Deletion},
{\sc Subset $\mathcal{F}$-$\preceq$-Deletion},
{\sc $\varphi$-Deletion},
{\sc $\varphi$-Amalgamation},
{\sc $\mathcal{L}$-$\varphi$-Replacement},
{\sc $\varphi$-Elimination distance},
{\sc $\varphi$-Reconfiguration}.
For the definitions, the complexity, and the \FOL{\sf +}{\sf DP}-expressibility 
of all these problems we refer the reader to \autoref{sec_problemswesolve}.
\smallskip

Our next step towards a more expressive logic, is to extend \FOL{\sf +}{\sf DP}\ by considering, for every $k\geq 1$, 
a more general form of predicate {\sf $s$-dp}$_{k}({\sf x}_1,{\sf y}_1,\ldots,{\sf x}_k,{\sf y}_k)$ where we now demand that the disjoint paths in question are pairwise $s$-scattered, i.e., there are no two vertices of two distinct
paths that are within distance at most  $s$.  We call the new logic  \FOL{\sf +}{\sf SDP}.
As {\sf $0$-dp}$_{k}({\sf x}_1,{\sf y}_1,\ldots,{\sf x}_k,{\sf y}_k)=\text{\sf dp}_{k}({\sf x}_1,{\sf y}_1,\ldots,{\sf x}_k,{\sf y}_k)$, we readily have that 
 \FOL{\sf +}{\sf SDP} is an extension of  \FOL{\sf +}{\sf DP}.
Our second result is the following  AMT.
 
 \begin{theorem}
\label{@geographical}
Every problem on graphs that is expressible by some formula $\varphi$ in \FOL{\sf +}{\sf SDP} can be solved by an algorithm running in time $\mathcal{O}_{|\varphi|,k}(n^2)$, where $k$ is the Euler genus of $G$.
\end{theorem}
 
Some indicative problems whose standard parameterizations are automatically classified in \FPT
because of \autoref{@geographical} are
 {\sc Induced Minor},
 {\sc Induced Topological Minor Containment},
 {\sc Contraction Containment},
 {\sc Induced Unordered Linkability},
 {\sc Induced Ordered Linkability},
 {\sc $\mathcal{F}$-Induced Minor Deletion},
 {\sc $\mathcal{F}$-Induced Topological Minor Deletion},
{\sc $\varphi$-Deletion},
{\sc $\varphi$-Amalgamation},
{\sc $\mathcal{L}$-$\varphi$-Replacement},
{\sc $\varphi$-Elimination distance},
{\sc $\varphi$-Reconfi\-guration}.
For the definitions, the complexity, and the \FOL{\sf +}{\sf SDP}-expressibility 
of all these problems we refer the reader to \autoref{sec_problemswesolve}.

%
%
%

\subsection{The irrelevant vertex technique}

In Volume XIII of their Graphs Minors series, Robertson and Seymour introduced the celebrated \emph{irrelevant vertex technique}  in order design a time $\mathcal{O}_{k}(n^3)$ algorithm for 
the following problem \cite{RobertsonS95GMXIII}.
\smallskip

\fbox{
\begin{minipage}{15cm}
\noindent{\sc Disjoint Paths}\\
\noindent{\sl Input}: a graph $G$ and pairs $(s_{1},t_{1}),\ldots,(s_{k},t_{k})$ of vertices of $G$.\\
\noindent{\sl Question}: Are there pairwise vertex-disjoint paths between $s_i$ and $t_i$, for $i\in\{1,\ldots,k\}$?
\end{minipage}
}
\medskip

Notice that the description of the above problem does not fit in \FOL. It demands the existence of $k$ pairwise disjoint sets of vertices each inducing a {\sl connected} graph containing the terminals $s_{i}$ and $t_{i}$. While connectivity is expressible in \MSOL\ it 
is known that 
it cannot be expressed in \FOL\ (see e.g., \cite{EbbinghausF95finit,Libkin04eleme}).\medskip

The general idea behind  the irrelevant vertex technique is  that if some part of the input graph is ``sufficiently insulated'' from the rest of the graph, then it may be {\sl irrelevant} in the sense that the solution can be ``reconfigured'' away from it.
When this idea applies, then this irrelevant part can be safely deleted and produce an equivalent, and simpler, instance of the problem.
 
The original application of the irrelevant vertex technique for 
the \textsc{Disjoint Paths} problem have had two phases: 

\begin{itemize}
\item {\bf 1st phase}:   when the input graph $G$ contains a big (as a function of $k$) clique minor

\item {\bf 2nd phase}:   when $G$ minor-excludes a clique, that is $G$ has small Hadwiger number.
\end{itemize}

\autoref{@definitionen} and \autoref{@geographical} deal with the applicability of the 2nd phase.
We next give a brief outline of how this phase was applied for  the \textsc{Disjoint Paths} problem
and, in particular, in the ``non-trivial'' situation where $G$ has ``big'' 
treewidth.
To deal with this situation, Robertson and Seymour proved in \cite{RobertsonS95GMXIII}
the so called {\sl Flat Wall Theorem}, asserting that if a graph has small Hadwiger number and big treewidth,  then after the removal from $G$ of ``few'' vertices, called {\sl apex vertices}, the resulting graph contains a big wall that is ``flat''. Intuitively, by the term ``flat wall'' we refer to a wall $W$ whose perimeter $P$ contains a separator $S$ of $G$ where the inner part of the wall
is inside one of the connected components $C$ of $G\setminus V(S)$ and where 
no two disjoint $(s_{i},t_{i})$-paths, $i\in\{1,2\}$, exist in the graph $G[V(P)\cup V(C)]$, called the \emph{compass} of $W$,  
where  $s_{1},s_{2},t_{1},$ and $t_{2}$ are vertices of $S$, appearing in this ordering in $P$.
This flatness property implies that every set of homocentric cycles around the central part of the wall can act as a system of separators ``insulating'' the two central vertices of $W$   from the part of the graph that lies outside the compass of the wall.
With this structural result at hand, Robertson and Seymour 
proved that every set of $k$ disjoint paths that may certify 
a \yes-instance of the \textsc{Disjoint Paths} problem can be rerouted 
away from its central vertices, that is it they be declared irrelevant 
and be safely discarded from $G$. The proof of this rerouting 
argument is quite technical and was given  in Volumes XXI and XXII of the Graph Minors series (this result is now known as the Unique Linkage Theorem -- see~\cite{KawarabayashiW2010asho,AdlerKKLST17irrel,Mazoit13asin,GolovachST20hitti,GolovachST22comb} for later proofs and improvements).
Moreover, Kawarabayashi, Kobayashi, and Reed proved in \cite{KawarabayashiKR12thedis} that the flat wall $W$ (and therefore its central vertices as well)  can be found in linear time.
Given now that a simpler equivalent instance of the \textsc{Disjoint Paths} problem
is found, we may repeat the above procedure a linear number of times until  the treewidth is ``small'' so that the problem can be solved by a dynamic programming algorithm (which exists because of Courcelle's theorem).
This, taken into account the improvement of \cite{KawarabayashiKR12thedis}, takes a total of $\mathcal{O}_{k}(n^2)$ time.

\paragraph{The potential of irrelevant vertex technique.}
To adapt the above arguments for other problems has been 
a challenging enterprise for graph algorithm designers during the last 20 years.
For a indicative (while not exhaustive) list of papers that made use of this technique, see~\cite{FominGT19modif,JansenK021verte,DLindermayrSV20elimi,KawarabayashiMR08asim,MarxS07obta,KobayashiK09algo,GroheKMW11find,KawarabayashiK10impr,Kawarabayashi07half,KawarabayashiR09hadw,KawarabayashiLR10reco,GolovachKPT12indu,FominLST12line,TakehiroKPT09para,Thilikos12grap,KawarabayashiKM10link,KaminskiN12find,KaminskiT12cont,FominLRS11sube,HeggernesHLP13obta,KawarabayashiR10oddc,AdlerGK08comp,CattellDDFL00onco,Kawarabayashi09plan,FominLP0Z20hittin,BasteST20acomp,SauST20anftp,SauST21amor,SauST21kapiI,SauST21kapiII,FominGSST21acomp,FominGST20analgo,GolovachST20hitti}. 
Typically, for each problem, the challenge is to give an algorithm that is  
able to detect, in polynomial time, some vertex that can be declared irrelevant 
and then prove that this vertex is  indeed irrelevant in the sense that  
discarding it from the input graph creates an equivalent instance. 
In some cases, apart from declaring a vertex irrelevant, 
``annotated versions'' of problems have been considered 
and vertices may also be declared {\sl annotation-irrelevant} in the sense that 
they can be safely excluded from the set of annotated vertices.
This extended concept of irrelevancy was used in \cite{GolovachKMT17thep} for the {\sc Cyclability} problem  and 
in~\cite{GolovachST20hitti}
for the  \textsc{$\mathcal{F}$-Topological Minor-Deletion} problem (see also the 
meta-algorithmic results in \cite{FominGSST21acomp,FominGST20analgo,FominGT19modif}).
In our proof of  \autoref{@definitionen} and  \autoref{@geographical} 
we largely make use of the annotation technology. In fact we consider an annotated set for the 
variables quantified in the \FOL{\sf +}{\sf DP} formula.

Most of the problems that are amenable to the application of the irrelevant vertex technique have a common denominator: they are not \FOL-expressible and they deal with graph classes with unbounded treewidth, that go beyond the combinatorial applicability of \MSOL. Thus, they escape the logical/combinatorial
conditions of the known meta-algorithmic technology of \FOL\ and \MSOL.

The proof of \autoref{@definitionen} and  \autoref{@geographical} is abstracting the  irrelevant vertex technology of all the aforementioned problems into two AMTs. \autoref{@definitionen}
(resp. \autoref{@geographical}) essentially indicates that, for graphs of bounded Hadwiger number (Euler genus),  
the descriptive potential of the \textsc{Disjoint (Induced) Paths} problem can be ``embedded᾽᾽ inside  \FOL\ in the form of the predicates {\sf dp}$_{k}({\sf x}_1,{\sf y}_1,\ldots,{\sf x}_k,{\sf y}_k)$ ({\sf $s$-dp}$_{k}({\sf x}_1,{\sf y}_1,\ldots,{\sf x}_k,{\sf y}_k)$).

\paragraph{What to do with a clique.}
Clearly \autoref{@definitionen} and  \autoref{@geographical}  concern the 2nd phase of the irrelevant vertex technique  where the Hadwiger number is bounded. At this point we wish to mention that the applicability of the 1st phase, concerning the question ``{\sl what to do with a clique\,}'', may  vary depending on the problem in question. We distinguish three main categories of parameterized problems. 
\begin{itemize}

\item[A.] The first category contains standard parameterizations  of  problems,  such as {\sc Minor Containment}, {\sc $\mathcal{F}$-Minor Amalgamation}, {\sc $\mathcal{F}$-Minor Deletion}, or {\sc $\mathcal{F}$-Minor Local Replacement},  where \autoref{@definitionen}       
applies and, moreover, 
big enough Hadwiger number immediately certifies a \yes- or a \no-instance. Given that checking whether 
a clique $K_{r}$ is a minor of a graph can be done in time $\mathcal{O}_{r}(n^2)$ \cite{KawarabayashiKR12thedis},  for such  problems,
just {\sf FO+DP}-expressibility is enough for implying that a problem is \FPT, even in general graphs.

\item[B.]  The second category of parameterized problems contains those where the 1st phase is non-applicable in general, in the sense that they are already intractable. For instance,  the standard parameterization of 
  {\sc Cyclability} or {\sc Subset Linkability}, where \autoref{@definitionen} applies, is {\sf co-W[1]}-hard  and problems, where \autoref{@geographical} applies,  
  such as {\sc Induced Disjoint Paths}, {\sc Contraction Containment}, and  {\sc Induced Minor Containment}   are {\sf NP}-complete even for fixed values of the parameter (see \cite{KawarabayashiK08thei}, \cite{BrouwerV87contra,LevinPW08,LevinPW08a}, and \cite{LevequeLMT09dete,FellowsKMP95thec} respectively).

\item[C.]
The third category concerns standard parameterizations of problems, such as 
 {\sc Disjoint Paths} \cite{RobertsonS95GMXIII}, {\sc Topological Minor Containment} \cite{GroheKMW11find} , {\sc $\mathcal{F}$-Topological Minor Deletion} \cite{FominLP0Z20hittin}, {\sc $\mathcal{F}$-TM Elimination Distance} \cite{AgrawalKLPRSZ22dele}, and {\sc $\mathcal{F}$-TM-Treewidth} \cite{AgrawalKLPRSZ22dele},  where extra algorithmic machinery is employed 
for dealing with the 1st phase (typically related with the recursive understanding technique \cite{ChitnisCHPP16desig,GroheKMW11find,CaiCC06,LokshtanovRSZ18redu}).  
To our knowledge, there is no  general treatment 
 of the 1st phase, further than the results of \cite{RobertsonS95GMXIII,GroheKMW11find,FominLP0Z20hittin,AgrawalKLPRSZ22dele}.
 To investigate the meta-algorithmic conditions that may unify them, even for some combinatorial condition that is more general that having bounded Hadwiger number,  is an interesting open challenge.
  \end{itemize}
  
  For more on the classification of the problems that are treated by  \autoref{@definitionen} and  \autoref{@geographical} according to the above three categories, see \autoref{sec_problemswesolve}.

\paragraph{Organization of the paper.}
In~\autoref{sec_overview}, we provide an overview of our proof.
In~\autoref{sec_preliminaries}, we provide some basic definitions that will be used throughout the paper.
In~\autoref{sec_alternative}, we present a way to translate model-checking to (recursive) folio containment.
Then, in~\autoref{sec_reroutinginannuli}, we give some additional definitions and results for dealing with collections of paths inside (partially planar) graphs.
In~\autoref{sec_annotated}, we present a trick to transform a disjoint paths query to one that can deal with the presence of some \emph{apex} vertices that can ``spoil'' flatness and in~\autoref{sec_signaturesexchengability} we present the combinatorial result that supports the correctness of the main subroutine of the algorithm of~\autoref{@definitionen} and~\autoref{@geographical}.
Next, in~\autoref{sec_proofoftheorem}, we present the proof of~\autoref{@definitionen} and in~\autoref{sec_logicscattered} we present the proof of~\autoref{@geographical}.
We conclude the paper with~\autoref{sec_conclusion}.
In~\autoref{sec_problemswesolve} we present a list of problems expressible in
\FOL{\sf +}{\sf DP} and in  \FOL{\sf +}{\sf SDP} and in~\autoref{sec_flatwalls} we present the flat wall framework that we use in this paper, which was introduced in~\cite{SauST21amor}. \autoref{@ressemblerois} contains missing complexity proofs of problems mentioned in~\autoref{sec_conclusion} and~\autoref{sec_problemswesolve}.

\section{Overview of the proof}

\label{sec_overview}
 In this section we summarize the main ideas involved in the proof of~\autoref{@definitionen} and \autoref{@geographical}.
We describe our approach for graphs.
However, our results are proven for {\sl colored graphs} (i.e., graphs equipped with a sequence of subsets of their vertex set).

\subsection{General scheme of the algorithm}

For our algorithms we follow the typical motif of the irrelevant vertex technique:
if the treewidth of the input graph $G$ is upper-bounded by an appropriately chosen function, depending only on the sentence $\varphi\in\FOL{\sf +DP}$ and $\hw(G)$,
then because of Courcelle's Theorem~\cite{Courcelle90them,Courcelle92,Courcelle97} we can check whether $G$ satisfies $\varphi$ in linear time, using the fact that $\FOL{\sf +DP}$ is a fragment of $\MSOL$ (see~\autoref{subsec_dplogic}).
Otherwise, if the treewidth is ``large enough'', we identify an irrelevant vertex in linear time, that is a vertex whose removal does not affect satisfiability of $\varphi$.
This highly non-trivial procedure of finding an irrelevant vertex is our main goal and, in what follows, we describe in an intuitive level the main ingredients of this approach.

\paragraph{Irrelevant vertices for model-checking.}
Applying the irrelevant vertex technique in a model-checking setting demands
the restriction of the part of the graph that is used to interpret variables and predicates of the sentence.
To show that a vertex is irrelevant,
one has to prove that whether the given sentence is satisfied or not does not depend on the presence of this
vertex inside the graph.
The ``building blocks'' of our sentences are either first-order variables (quantified by universal or existential quantifiers),
or interpretations of relation symbols of the given vocabulary
(in our case, these are edges and colors in the vertices),
or disjoint-path predicates between some of the already quantified first-order variables.
Towards building an irrelevant vertex argument, an essential step is to
{\sl reduce the scope} of the quantification of the variables, while preserving the satisfiability status of the sentence.
The fact that a vertex can be discarded from the scope of a quantifier,
permits to declare it {\sl annotation-irrelevant} with respect to this quantifier.
Moreover, even if some vertex is annotation-irrelevant with respect to all quantifiers,
this vertex could still be important for the existence (or not) of disjoint paths between vertices (that are picked inside the annotation).
Such a vertex can be removed from the graph, therefore be declared {\sl problem-irrelevant}, only if we can guarantee
that disjoint paths can be safely rerouted away from it.
The pursue of such a problem-irrelevant vertex is executed inside a ``big enough'' bidimensional area of vertices that are annotation-irrelevant for {\sl all} quantifiers.
While this does not deviate from the ``typical'' rerouting arguments of the irrelevant vertex technique,
dealing with the restriction of the annotation is the most demanding part.
In the rest of this section, we aim to demonstrate a way to tackle this problem.

\subsection{Translating model-checking to (recursive) folio containment}
Our main tool in order to create equivalent (annotated) instances
is to express model-checking in graph-theoretic terms.
Throughout this section, we assume that all sentences are given in prenex-normal form,
i.e., $\varphi = Q_1 {\sf x}_1\ldots Q_r {\sf x}_r \psi({\sf x}_1,\ldots, {\sf x}_r)$,
where $Q_1,\ldots,Q_r\in\{\forall,\exists\}$,
${\sf x}_1,\ldots,{\sf x}_r$ are first-order variables,
and $\psi({\sf x}_1,\ldots, {\sf x}_r)$ is a quantifier-free formula
with ${\sf x}_1,\ldots, {\sf x}_r$ as free variables.

\paragraph{Assigning annotated graphs to rooted trees.}
Our first step (\autoref{sec_alternative}) is to interpret the satisfaction of a sentence $\varphi$ from a graph in terms of
the existence of a subtree inside a tree where the graph is embedded, such that the bifurcations of this subtree
correspond to the quantifiers of $\varphi$ and the vertices collected in each root-to-leaf path evaluate to true the 
quantifer-free ``tail'' of $\varphi$.
These trees, known as \emph{game trees}, appear with different names in the literature, like \emph{evaluation trees}~\cite{GajarskyGK20diff} or \emph{morphism trees}~\cite{BonnetKTW22twinI} (see also~\cite{gradel2003automata}).
These trees follow the recursive structure of a set $\mathsf{sig}^r(G,R_1,\ldots,R_r)$, where $G$ is a graph and $R_1,\ldots,R_r\subseteq V(G)$.
This set, called \emph{signature} of $(G,R_1,\ldots,R_r)$, is defined as the recursive collection of the different types of tuples $v_1,\ldots,v_i$ $i\leq r$ of vertices of $V(G)$, where each $v_j$ belongs to $R_j$ for every $j\leq i$ (see the definition in~\autoref{subsec_treeembed}).
We can construct a rooted tree $(T,t_0)$ following the recursive structure (of depth $r$) of this set and this naturally gives a mapping of each node of the tree to the vertex of $G$ of a particular type (when considering the tuple of all ancestors of it).
We call this mapping an \emph{assignement} of $(G,R_1,\ldots,R_r)$ to $(T,t_0)$.
Intuitively, every root-to-leaf path is mapped to a tuple of (possibly repeating) vertices $v_1,\ldots,v_r$,
where $v_i\in R_i$, for every $i\in\{1,\ldots,r\}$ (see~\autoref{fig_overview_assignment}).
Also, we define a \emph{$\varphi$-spanning} subtree of a rooted tree
to be the (sub)tree $T'$ with the same root that is obtained following the quantifiers of $\varphi$, i.e., if $Q_i=\exists$, then every node of $T'$ of depth $i$ has only one child in $T'$, while if $Q_i=\forall$, every node of $T'$ of depth $i$ bifurcates (in $T'$) to all its children in $T$ (see~\autoref{fig_overview_assignment} for an example).
In this setting, an (annotated) formula $\varphi = Q_1 {\sf x}_1\in {\sf R}_1\ldots Q_r {\sf x}_r\in{\sf R}_r\ \psi({\sf x}_1,\ldots, {\sf x}_r)$ is satisfied from a tuple $(G,R_1,\ldots,R_r)$ iff there is an assignment of $(G,R_1,\ldots,R_r)$ to a rooted tree and a $\varphi$-spanning subtree of this tree such that for every root-to-leaf path of the  $\varphi$-spanning subtree, the formula $\psi$ is satisfied when interpreting its free variables as the vertices collected in this path (\autoref{obs_treeformula}).

\paragraph{Patterns and pattern-coloring.}
A crucial observation that is central to our approach is the following:
for every quantifier-free formula $\psi({\sf x}_1,\ldots, {\sf x}_r)\in\FOL{\sf +DP}$, given a graph $G$ and a tuple of vertices $(v_1,\ldots,v_r)$ of $G$, the question whether $\psi$ is satisfied when ${\sf x}_1,\ldots, {\sf x}_r$ are interpreted as $v_1,\ldots,v_r$ boils down to checking whether edges and/or disjoint paths between particular pairs from $v_1,\ldots,v_r$ (indicated by the atomic formulas of $\psi$) are present in $G$.
For this reason,
we define the \emph{pattern} of a {\sl boundaried} graph $(G,v_1,\ldots,v_r)$ to be an encoding of the edges and the disjoint paths between all pairs in $v_1,\ldots,v_r$ in $G$.
Note that, knowing the pattern of a boundaried graph, we can determine which quantifier-free formulas in $\FOL{\sf +DP}$
are evaluated to true in this graph and which do not.
Therefore, when considering a graph assigned in a rooted tree, we can ``color'' each leaf of the tree by the pattern of the corresponding boundaried graph $(G,v_1,\ldots,v_r)$, where $v_1,\ldots,v_r$ are the vertices collected in the corresponding root-to-leaf path.
This, we call the \emph{pattern-coloring} of the leafs.
Also, we introduce encoding in terms of patterns for formulas, i.e., we define the pattern of every clause of a quantifier-free formula, that encodes the presence (or not) of predicates or their nagation.
This allows to restate \autoref{obs_treeformula} and formulate model-chacking as the search of a $\varphi$-spanning tree with particular colors in its leafs, inside a leaf-colored tree (to which the annotated graph is assigned) where colors are given by the pattern-coloring (\autoref{obs_grleafs}). See~\autoref{fig_overview_assignment}.

\begin{figure}[ht]
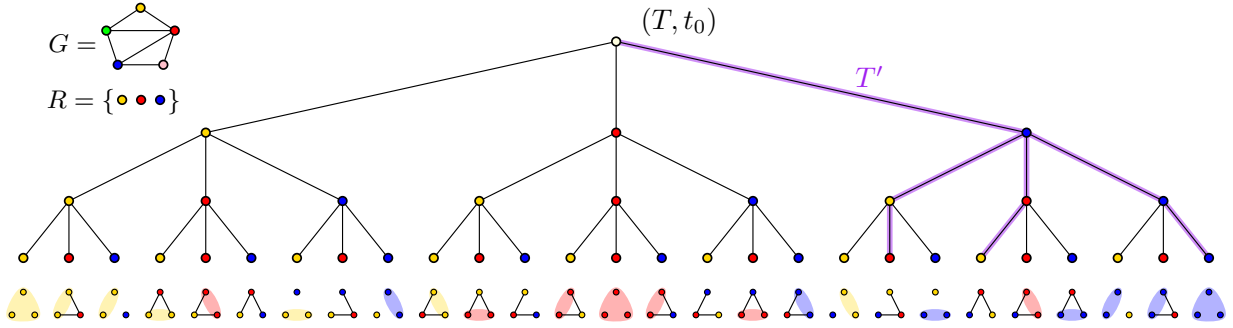

\centering
\resizebox{\textwidth}{!}{

}\caption{An example of an assignment of an annotated graph $(G,R,R,R)$ (depicted on the upper left part of the figure)
to a rooted tree $(T,t_0)$ (in the center of the figure).
We use colors for each vertex of $G$ in order to distinguish them, without using indices and we use colors in the nodes of $T$ to show where each vertex of $G$ is mapped to. 
The small graphs below each leaf of $T$ encode the pattern-coloring, i.e., are the (boundaried) graphs induced by the vertices collected in each root-to-leaf path of $T$ and the linear orderings of their vertex sets (from left to right) follow the ordering of the corresponding path.
Vertices that are picked twice in the same path are drawn inside a same-colored bag (respecting the ordering).
The subtree $T'$ is a $\varphi$-spanning subtree of $(T,t_0)$ for the \FOL-sentence
$\varphi = \exists {\sf x}_1\in {\sf R}_1\ \forall {\sf x}_2\in{\sf R}_2\ \exists {\sf x}_3\in{\sf R}_3\  ({\sf x}_1 = {\sf x}_2 \vee {\sf E}({\sf x}_2,{\sf x}_3))$,
that certifies that $(G,R,R,R)\models \varphi$.}
\label{fig_overview_assignment}
\end{figure}

\paragraph{Same (recursive) patterns imply satisfaction of the same formulas.}
Our intuition is that for every graph $G$, the information encoded by the patterns, organised in a recursive
tree structure as described above, is sufficient to evaluate {\sl every} sentence in $\FOL{\sf +DP}$ of a given number of quantifiers.
In~\autoref{subsec_equiv}, we formalize this intuition by defining equivalence of leaf-colored trees
and proving that two annotated graphs $(G,R_1,\ldots,R_r),(G',R_1',\ldots,R_r')$ that have the same recursive patterns for depth $r$ (i.e., $\mathsf{sig}^r(G,R_1,\ldots,R_r)=\mathsf{sig}^r(G',R_1',\ldots,R_r')$) satisfy the same sentences in $\FOL{\sf +DP}$
(\autoref{lemma_reducing}).
The proof of~\autoref{lemma_reducing} is an induction on the numbers of quantifiers of a sentence.

We use~\autoref{lemma_reducing} in the course of the proof of~\autoref{thm_main} in the following way:
As long as we can find sets $R_1',\ldots,R_r'\subseteq V(G)$ such that $R_i'\subseteq R_i$ and $(G,R_1,\ldots,R_r)$ and $(G,R_1',\ldots,R_r')$ have the same signature,
we can safely reduce the former instance to the latter and report progress.
Finding the sets $R_1',\ldots,R_r'$ is not straightforward and in the next subsection we explain how to deal with this situation.
Moreover, in our reduction, we will find a bidimensional annotation-irrelevant area and return a problem-irrelevant vertex $v$, i.e., a vertex $v$ such that
$(G,R_1,\ldots,R_r)$ and $(G\setminus v,R_1',\ldots,R_r')$ 
are also equivalent.
Finally, we reduce to an instance $(G',R_1',\ldots,R_r')$ with the same $r$-signature as the input graph (i.e., $\mathsf{sig}^r(G,V(G),\ldots,V(G))=\mathsf{sig}^r(G',R_1',\ldots,R_r')$,
where the treewidth of $G'$ depends only on $r$ and the minor we exclude.
Then, following the tree-like assignment  given by the signature, we can evaluate any  $\FOL{\sf +DP}$ sentence of quantifier rank at most $r$.
This latter is done by applying Courcelle's theorem.

\subsection{Combinatorial trick to compute folios}
\label{subsec_overviewrep}

Our main difficulty is to compute {\sl recursive} folios of the given annotated graph and obtain a different annotation that gives the same recursive patterns.
To tackle this, we need some further tools to handle disjoint paths.
From a high-level point of view, our approach considers expressing ``partial'' patterns inside a bounded treewidth part of the given graph.
Using Courcelle's Theorem, we can compute partial patterns and the boundary (tuples of) vertices that give rise to these patterns.

\paragraph{Flat walls with bounded treewidth compasses.}
In the introduction,
we already sketched a definition of flat walls,
originating in the work of Robertson and Seymour~\cite{RobertsonS95GMXIII}.
An alternative intuition for a flat wall is to see it
as a structure made up of (not necessarily planar) pieces, called \emph{flaps}, that are glued together with boundaries of size at most three,
in a way that follows the bidimensional structure of a wall.
In this article, we use the framework recently introduced in~\cite{SauST21amor} that provides a more accurate view on some previously defined notions on flat walls, particularly in~\cite{KawarabayashiTW18anew} (see~\autoref{sec_flatwalls} for formal definitions).
In the course of our main algorithm, we use a variant of the Flat Wall Theorem
proved in~\cite{KawarabayashiTW18anew,SauST21amor}
(see~\autoref{prop_flatwallbdtw}) that provides a guarantee that the {\sl compass} of the flat wall has bounded treewidth.
This permits us, because of Courcelle's Theorem, to answer in linear time \MSOL-queries inside the compass.

In fact, for our arguments, instead of working with a flat wall, we work with a \emph{flat railed annulus} (contained inside the flat wall).
This is a structure similar to a flat wall, whose ``underlying'' structure is not a wall but a sequence of nested cycles transversed by some disjoint paths, called {\sl rails}. See~\autoref{fig_flatannulus}.

\begin{figure}[ht]
\centering
\scalebox{0.5}{\begin{tikzpicture}

\begin{scope}[on background layer]
    
\draw[very thick,black,fill=black!10!white] (0,0) circle (7cm);
\draw[very thick,black,fill=red!25!white] (0,0) circle (5.8cm);
\draw[very thick,black,fill=cyan!30!white] (0,0) circle (4.6cm);
\draw[very thick,black,fill=blue!30!white] (0,0) circle (3.4cm);
\draw[very thick,black,fill=white] circle (2.2cm);

\foreach \x in {7,6.8,...,4.8} \draw (0,0) circle (\x cm);
\foreach \x in {4.6,4.4,...,2.2} \draw (0,0) circle (\x cm);
\foreach \x in {0,10,..., 350} \draw (\x:7) -- (\x:2.2);

\draw[red, line width=1.5pt] plot [smooth, tension=1.2] coordinates { (0,-1) (-160:1) (-120:1.6) (-140:3)  (-140:4)};
\draw[red, line width=1.5pt] plot [smooth, tension=1.2] coordinates {(-140:4)
(-145:5.5)  (-160:6.5) (200:8)};

\draw[red, line width=1.5pt] plot [smooth, tension=1.2] coordinates { (0.5,1) (-50:1) (-80:3) (-100:1.8) (-125:3)  (-130:4)};
\draw[red, line width=1.5pt] plot [smooth, tension=1.2] coordinates {(-130:4)
(-125:5.5)  (-110:6.5) (220:8.5)};

\draw[red, line width=1.5pt] plot [smooth, tension=1.2] coordinates { (160:8.3) (170:7) (160:5.5) (190:4.8)  (-160:4)};
\draw[red, line width=1.5pt] plot [smooth, tension=1.2] coordinates {(-160:4)
(190:3)  (160:2.5) (140:1.5)};

\draw[red, line width=1.5pt] plot [smooth, tension=1.2] coordinates { (-30:8.3) (-30:5.5) (-40:7.5) (-60:6) (-60:8)};

\node[terminal,black] (s1) at (0,-1) {};
\node[terminal,black] (t1) at (200:8) {};
\node[terminal,black] (s2) at (0.5,1) {};
\node[terminal,black] (t2) at (220:8.5) {};
\node[terminal,black] (s3) at (160:8.3) {};
\node[terminal,black] (t3) at (140:1.5) {};
\node[terminal,black] (s4) at (-30:8.3) {};
\node[terminal,black] (t4) at (-60:8) {};

\node[root terminal] (r1) at (-160:4) {};
\node[root terminal] (r2) at (-150:4) {};
\node[root terminal] (r3) at (-140:4) {};
\node[root terminal] (r4) at (-130:4) {};
\node[root terminal] (r5) at (-120:4) {};
    
\draw[ultra thick,yellow] (40:4.6) circle (0.2cm);
\node (o1) at ($(40:4.6)+(135:0.2)$) {};
\node (o2) at ($(40:4.6)+(-60:0.2)$) {};
\end{scope}

\begin{scope}[yshift=3cm,xshift=9cm,scale=0.8]

\draw[ultra thick,yellow] ($(0,2.8)+(135:2.6)$) -- (o1.center)
($(0,2.8)+(-90:2.6)$) -- (o2.center);

\draw (0,2.8) circle (2.6cm);
\clip (0,2.8) circle (2.6cm);

\begin{scope}

    \path (-1.2,2.7)  arc (40:77:5)
    node[terminal,pos=0.1] (P1) {}
    node[terminal,pos=0.3] (P0) {}
    (90:5)--(90:5);
    \draw [thick,black,path fading=north, fading angle =40] (-1.2,2.7)  arc (40:64:5)  (100:4)--(100:4);
    
    \path (-1.2,2.7) arc (40:-15:5)
    node[terminal,pos=0.2] (P2) {}
    node[terminal,pos=0.4] (P3) {}
    node[pos=1] (A) {}
    (-15:5)--(-15:5);
    \draw [thick,black,path fading=south] (-1.2,2.7) arc (40:5:5)    (7:5)--(7:5);
    
    \path (0.2,3)  arc (40:82:6.5)
    node[terminal,pos=0.1] (P6) {}
    node[terminal,pos=0.25] (P5) {}
    node[terminal,pos=0.45] (P4) {}
    (85:7)--(85:7);
    \draw [thick,black!50!white] (0.2,3)  arc (40:59:6.5)  (85:7)--(85:7);
    \draw [thick,black,path fading=west] (0.2,3)  arc (40:66:6.5)  (85:7)--(85:7);

    \path (0.2,3) arc (40:-5:6.5)
    node[terminal,pos=0.1] (P7) {}
    node[terminal,pos=0.3] (P8) {}
    node[terminal,pos=0.4] (P9) {}
    node[pos=1] (B) {}
    (-5:7)--(-5:7);
    \draw [thick,black,path fading=south] (0.2,3) arc (40:13:6.5)
    (15:7)--(15:7);
    
    \path (1.7,3)  arc (40:85:8.5)
    node[pos=1] (D) {}
    node[terminal,pos=0.4] (P10) {}
    node[terminal,pos=0.25] (P11) {}
    node[terminal,pos=0.05] (P12) {}
    (90:9)--(90:9);
    \draw [thick,black,path fading=west] (1.7,3)  arc (40:69:8.5)   (80:9)--(80:9);
    
    \path (1.7,3) arc  (40:5:8.5)
    node[terminal,pos=0.15] (P13) {}
    node[terminal,pos=0.35] (P14) {}
    node[terminal,pos=0.45] (P15) {}
    node[pos=1] (C) {}
    (5:9)--(5:9);
    \draw [thick,black,path fading=south] (1.7,3) arc  (40:25:8.5)   (25:9)--(25:9);
    
    \foreach \x in {0,...,3} \node[terminal,blue,circle] () at (P\x) {};
    \foreach \x in {4,...,9} \node[terminal,blue,circle] () at (P\x) {};
    \foreach \x in {10,...,15} \node[terminal,red,circle] () at (P\x) {};

    \node[] (A1) at ($(A)+(-45:1)$) {$A_{\bar{w}0}$};
    \node[] (A2) at ($(B)+(-45:1)$) {$A_{\bar{w}1}$};

    \node[terminal,blue,circle](Q9) at (0.4,.6) {};
    \node[terminal,blue,circle](Q8) at (0.8,1.3) {};
    \node[terminal,blue,circle](Q7) at (-0.1,1.6) {};
    \node[terminal,blue,circle](Q6) at (0.2,1.8) {};
    \node[terminal,blue,circle](Q5) at (-.5,2.4) {};
    \node[terminal,blue,circle](Q4) at (-0.3,2.8) {};        
    \node[terminal,blue,circle](Q3) at (-1,2.8) {};    
    \node[terminal,blue,circle](Q2) at (-0.8,3.4) {};
    \node[terminal,blue,circle](Q1) at (-1.2,3.8) {};
    
    \node[terminal,blue,circle](ex2) at (-1.9,4) {};

    \node[terminal,red,circle](B8) at (2,1.6) {};
    \node[terminal,blue,circle](B7) at (1.2,2.2) {};
    \node[terminal,blue,circle](B6) at (1.4,2.4) {};
    \node[terminal,red,circle](B5) at (0.5,3.3) {};
    \node[terminal,red,circle](B4) at (0.7,3.6) {};
    \node[terminal,red,circle](B3) at (0.3,3.7) {};    
    \node[terminal,blue,circle](B2) at (-.65,4.25) {};
    \node[terminal,blue,circle](B1) at (-.85,4.45) {};

    \node[terminal, red,circle] (C7) at (0.7,5) {};
    \node[terminal, red,circle] (C6) at (0.4,4.6) {};
    \node[terminal, red,circle] (C5) at (1.3,4.5) {};
    \node[terminal, red,circle] (C4) at (1.2,4.25) {};
    \node[terminal, red,circle] (C3) at (1.45,3.95) {};
    \node[terminal, red,circle] (C2) at (2.1,3.6) {};
    \node[terminal, red,circle] (C1) at (2,3) {};

    \node[terminal, blue,circle] (A5) at (-2,2.8) {};
    \node[terminal, blue,circle] (A4) at (-1.8,2.3) {};
    \node[terminal, blue,circle] (A3) at (-1.2,1.9) {};
    \node[terminal, blue,circle] (A2) at (-1.5,1.5) {};
    \node[terminal, blue,circle] (A1) at (-1,1) {};

\begin{scope}[on background layer]
\fill[white] (0,2.8) circle (2.6cm);

\clip (0,2.8) circle (2.6cm);

    \begin{scope}
        \fill[cyan] (P9.center) to [bend left= 20] ($(P9)+(-70:1)$) to [bend left= 20] (P9.center);
        
        \fill[cyan] (P3.center) to [bend right= 10] ($(P3)+(-90:1)$) to ($(P3)+(-25:2)$) to [bend right=10] (P3.center);
        
        \fill[cyan] (A1.center) to [bend left= 20] ($(A1)+(-80:2)$) to ($(A1)+(-100:1)$) to [bend right=10] (A1.center);
        
        \fill[cyan] (A1.center) to [bend left= 10] ($(A1)+(-170:1)$) to ($(A1)+(-120:1)$) to [bend right=10] (A1.center);
        
        \fill[cyan] (A4.center) to [bend left= 10] ($(A4)+(-160:1)$) to ($(A4)+(-120:1)$) to [bend right=10] (A4.center);
        
        \fill[cyan] (A5.center) to [bend right= 10] ($(A5)+(-170:1)$) to ($(A5)+(-140:1)$) to [bend right=10] (A5.center);
  
        \fill[cyan] (P0.center) to [bend right= 10] ($(P0)+(140:1)$) to [bend right= 20]
        (A5.center) to [bend right=30] (P0.center);
        
        \fill[cyan] (ex2.center) to [bend left= 20] ($(ex2)+(150:1)$) to [bend left=30] (ex2.center);
        
        \fill[cyan] (P4.center) to [bend left= 20] ($(P4)+(150:1)$) to [bend left=30] (P4.center);
        
        \fill[red] (P10.center) to [bend left= 10] ($(P10)+(170:2)$) to ($(P10)+(120:1)$) to [bend right=10] (P10.center);
    
        \fill[red] (P10.center) to [bend left= 20] ($(P10)+(85:1)$) to [bend left =20](C7.center) to [bend right=30] (P10.center);

        \fill[red] (C7.center) to [bend right= 20] ($(C7)+(45:1)$) to [bend right =20]($(C7)+(70:1)$) to [bend right=15] (C7.center);
        
        \fill[red] (C5.center) to [bend right= 20] ($(C5)+(45:1)$) to [bend right =20]($(C5)+(70:1)$) to [bend right=15] (C5.center);
        
        \fill[red] (C2.center) to [bend right= 20] ($(C2)+(45:1)$) to [bend right =20]($(C2)+(90:2)$) to [bend right=15] (C2.center);
        
        \fill[red] (C2.center) to [bend right= 20] ($(C2)+(-40:1)$) to [bend right =20]($(C2)+(-10:1)$) to [bend right=15] (C2.center);
        
        \fill[red] (P13.center) to [bend right= 20] ($(P13)+(-60:1)$) to [bend right =20]($(P13)+(30:1)$) to [bend right=15] (P13.center);
        
         \fill[red] (B8.center) to [bend right= 20] ($(B8)+(-100:1)$) to [bend right =20]($(B8)+(-10:1)$) to [bend right=15] (B8.center);

        \fill[white,path fading=fade in] (0,2.8) circle (3cm);
    \end{scope}

\fill[cyan!70!white,opacity=0.5] (P1.center) to [bend right= 20] (P2.center) to [bend right= 20] (Q4.center) to [bend right=20] (P1.center);

\fill[cyan!70!white,opacity=0.5] (P0.center) to [bend right= 30] (P4.center) to [bend right= 20] (P6.center) to [bend left=20] (P0.center);
   
\fill[cyan!70!white,opacity=0.5] (P5.center) to [bend right= 20] (P6.center) to [bend right=20] (P5.center);

\fill[cyan!70!white,opacity=0.5] (P1.center) to [bend left= 20] (A4.center) to [bend left= 20] (A5.center) to [bend left=20] (P1.center);

    \fill[cyan!70!white,opacity=0.5] (P1.center) to [bend right= 20] (P0.center) to [bend right=20] (P1.center);
    
    \fill[cyan!70!white,opacity=0.5] (P6.center) to [bend left= 20] (P7.center) to [bend left= 20] (Q4.center) to [bend left=40] (P6.center);
    
    \fill[cyan!70!white,opacity=0.5] (P2.center) to [bend left= 20] (P7.center) to [bend left= 20] (P2.center);

    \fill[cyan!70!white,opacity=0.5] (A4.center) to [bend left= 20] (P2.center) to [bend left= 20] (A1.center) to [bend left=40] (A4.center);
    
    \fill[cyan!70!white,opacity=0.5] (P3.center) to [bend left= 10] (P2.center) to [bend left= 20] (P8.center) to [bend right=20] (P3.center);
    
    \fill[cyan!70!white,opacity=0.5] (P3.center) to [bend right= 5] (P8.center) to [bend left= 40] (P9.center) to [bend left=20] (P3.center);
 
    \fill[cyan!70!white,opacity=0.5] (A1.center) to [bend right= 20] (P3.center) to [bend right=20] (A1.center);
    
   \fill[cyan!70!white,opacity=0.5] (P0) to [bend left= 20] (ex2) to [bend left =20](P4.center) to [bend left=15] (P0);
   
   \fill[cyan!70!white,opacity=0.5] (P4.center) to [bend left= 20] (P10.center) to [bend left =20](P5.center) to [bend left=15] (P4.center);
   
   \fill[red!50!white,opacity=0.5] (P10.center) to [bend right= 10] (C7.center) to [bend left =20](P11.center) to [bend left=15] (P10.center);
   
    \fill[red!50!white,opacity=0.5] (P5.center) to [bend right= 10] (P11.center) to [bend right=15] (P5.center);
   
   \fill[red!50!white,opacity=0.5] (P11.center) to [bend left= 20] (C5.center) to [bend left =20](C2.center) to [bend left=10] (P11.center);
   
   \fill[red!50!white,opacity=0.5] (P11.center) to [bend right= 20] (P6.center) to [bend right =20](P12.center) to [bend right=20] (P11.center);
   
   \fill[red!50!white,opacity=0.5] (P12.center) to [bend left= 20] (C2.center) to [bend left =20](P13.center) to [bend left=15] (P12.center);
   
   \fill[cyan!70!white,opacity=0.5] (P13.center) to [bend right= 20] (P7.center) to [bend right =20](P8.center) to [bend right=15] (P13.center);
   
   \fill[red!50!white,opacity=0.5] (P13.center) to [bend left= 20] (B8.center) to  [bend left=15] (P13.center);
   
   \fill[red!50!white,opacity=0.5] (P8.center) to [bend left= 20] (B8.center) to [bend left=15] (P8.center);
   
   \fill[red!50!white,opacity=0.5] (P7.center) to [bend left= 20] (P12.center) to [bend left=15] (P7.center);
\end{scope}

    \draw[black,path fading= west] (ex2) -- ($(ex2)+(145:0.4)$);

    \draw[black] (P4)--(Q1) (P6) -- (Q4) (P7)--(Q4) (P7)--(P2) (P8)--(Q6) (P8)--(Q7) (P8)--(Q8) (P9)--(Q8);
    \draw[ultra thick,black] (P2) -- (Q3) (Q3) -- (Q4) (Q4) -- (P6);
    
    \draw[black] (P0) -- (Q2) (Q2)-- (P6);
    \draw[black] (P3) -- (P2) (P2)-- (Q6) (Q6)-- (P8);

    \draw[black]  (P0)--(Q1) (P1)--(Q3) (P2)--(Q5) (P2)--(Q6) (P2)--(Q7) (P3)--(Q7) (P3)--(Q8) (Q3) to [bend left = 10] (P2) (P1) to [bend right = 10] (Q5); 
    
    \draw[black] (ex2) -- (P0) (ex2) -- (P4) (P0) -- (P4) (Q1)--(Q2) (Q3)--(Q4)--(Q5)--(Q3) (P3) -- (Q9);
    

    \draw[black] (B1) -- (B2) (B1) -- (P4) (B1) -- (P5) (B1) -- (P10) (B2) -- (P10) (B2) -- (P5) (B2) -- (P4) (P5) -- (P10) -- (P4);
    
    \draw[black] (P5) -- (P11);
    
    \draw[ultra thick,black] (P6) -- (B3) (B3) -- (P11);
    \draw[black] (P8) -- (P7) (P7)--(B6)-- (P13);
    \draw[black] (P6) -- (B4) (P6) -- (B5) (B5) -- (P12) (B3) --(B4) (B4) -- (B5) (B3) -- (B5) (P11) -- (B4) (P12) -- (B4) (B5) -- (P11);
    
    \draw[black] (P7) -- (P12);
    
    \draw[black] (P7) -- (B6) (P7) -- (B7) (P8) -- (B7) (B7)-- (P13) (P13) -- (B6) (P13) -- (P8) (B6) -- (B7);
    
    \draw[black] (P8) -- (B8) (P13) -- (B8);

%

\draw[ultra thick,black] (P11) -- (C4) (C4) -- (C5);
\draw[] (P10) -- (C6);
\draw[] (C6) -- (P11) (C6) -- (C7);
\draw[] (P11) -- (C5) (P11) -- (C4) (P11) -- (C3) (P11) -- (C2) (C5) -- (C2) (C5) -- (C3) (C5) -- (C4) (C4) -- (C3) (C3) -- (C2) (C4) to [bend left = 5] (C2);
\draw (P12) -- (C2) (P12) -- (C1) (P13) -- (C1) (C1) -- (C2);

    \draw[black] (A1)  -- (P3);
    \draw[black] (A1) -- (P2) (A1) -- (A2) (A1) -- (A3) (A1) -- (A4) (A2) -- (P2) (A2) -- (A3) (A2) -- (A4) (A3) -- (A4) (A4) -- (P2);
     \draw[ultra thick,black]      (A3) -- (P2)      (A3) -- (A4);
    \draw[black] (A4) -- (A5)  -- (P1) -- (A4);
    \draw[black] (A5) -- (P0);

\draw[thick,black, path fading=east] (Q9) -- ($(Q9)+(-30:0.6)$);
\draw[thick,black, path fading=south] (Q9) -- ($(Q9)+(-100:0.5)$);
\draw[thick,black, path fading=south] (B8) -- ($(B8)+(-70:0.4)$);
\draw[thick,black, path fading=east] (B8) -- ($(B8)+(-30:0.3)$);
\draw[thick,black, path fading=east] (P13) -- ($(P13)+(10:0.5)$);
\draw[thick,black, path fading=east] (C2) -- ($(C2)+(-30:0.55)$);
\draw[thick,black, path fading=north] (C2) -- ($(C2)+(80:0.7)$);
\draw[ultra thick,black, path fading=north] (C5) -- ($(C5)+(60:1)$);
\draw[thick,black, path fading=north] (C7) -- ($(C7)+(50:1)$);
\draw[thick,black, path fading=north] (C7) -- ($(C7)+(145:1)$);
\draw[thick,black, path fading=north] (P10) -- ($(P10)+(75:1)$);
\draw[thick,black, path fading=west] (P10) -- ($(P10)+(170:2)$);
\draw[thick,black, path fading=south] (A1) -- ($(A1)+(-75:1)$);
\draw[thick,black, path fading=west] (A1) -- ($(A1)+(-140:1)$);
\draw[ultra thick,black, path fading=west] (A4) -- ($(A4)+(-130:1)$);
\draw[thick,black, path fading=west] (A5) -- ($(A5)+(-160:1)$);
\draw[thick,black, path fading=north] (A5) -- ($(A5)+(120:1.5)$);

\end{scope}
\end{scope}

\end{tikzpicture}}
\caption{An illustration of a (flat) railed annulus and a linkage (depicted in red) that is combed through some prescribed vertices of the railed annulus  (depicted in orange).}
\label{fig_flatannulus}
\end{figure}
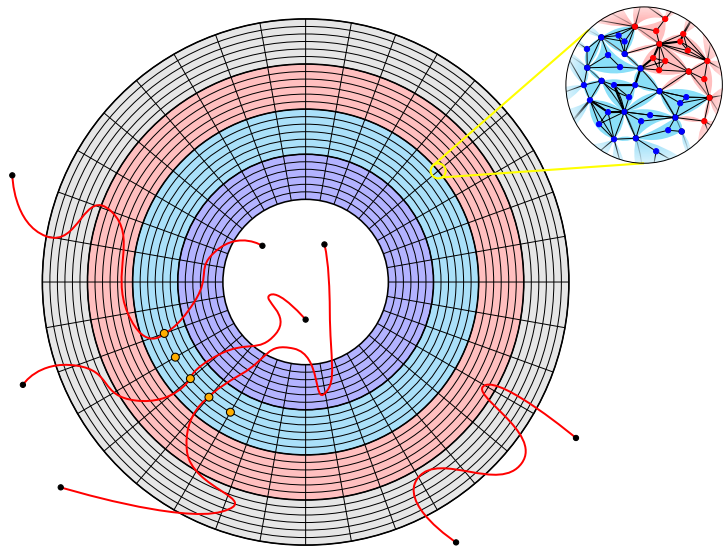

\paragraph{Routing linkages through railed annuli.}
Recall that, the pattern of a boundaried graph $(G,v_1,\ldots, v_r)$, apart from the edges between $v_1,\ldots, v_r$, also
encodes the existence or not of disjoint paths between pairs of the boundary vertices $v_1,\ldots, v_r$.
We are interested in the ways such a pattern may cross a cycle $C$ of the railed annulus.
To define a notion of ``partial'' pattern (i.e., the one ``cropped'' by a ``central-enough'' cycle $C$),
we have to deal with the possibe ways that  $v_1,\ldots, v_r$ can be connected through disjoint paths crossing $C$ and this can be seen as a question about the variety of \emph{linkages}, i.e., collections of disjoint paths, that can be routed between boundary vertices and the way they may cross $C$.
Clearly, a linkage may transverse $C$ in a quite entagled way,
so that the partial pattern of $(G,v_1,\ldots, v_r)$ cropped by $C$
may have an unbounded number of additional boundary vertices on $C$.
In order to deal with this situation, we make use of a special result for handling linkages,
that is the Linkage Combing Lemma (\autoref{prop_combinglemma}) proved in~\cite{GolovachST22comb} (which has appeared before in~\cite{GolovachST20hitti}).
This result is applied on partially annulus-embedded graphs, i.e., graphs that contain some subgraph $K$ that is embedded in an annulus $\Delta$ and $\Delta$ separates the two other parts of the graph obtained after removing the vertices inside $\Delta$.
The Linkage Combing Lemma intuitively says that in the presence of an annulus-embedded railed annulus $\mathcal{A}$ inside 
a partially annulus-embedded graph $G$, every linkage $L$ of $G$ whose terminals are outside the annulus
can be ``combed'', 
i.e., it can be routed throught the cycle $C$ of $\mathcal{A}$ in a predefined number of vertices.
This allows us to {\sl represent} every linkage encoded in the partial pattern
by a ``combed'' linkage with a few predefined additional terminals on $C$.


\paragraph{Applicability of the Linkage Combing Lemma.}
To encode the recursive patterns of a graph, we have to keep the pattern of the graph for every choice of the boundary tuple. Therefore, for each tuple of vertices, we have to encode all linkages between pairs of vertices in the tuple.
The \emph{trace} of each tuple indicates which (sub)annulus of the flat railed annulus remains ``terminal-free'' (see~\autoref{fig_flatannulus}).
This railed annulus is flat but we can consider its \emph{leveling}
(see~\autoref{subsec_levelings}), that is a {\sl planar representation} of it, obtained by contracting the interior of the flaps to single vertices.
We observe that for every linkage whose terminals are outside the compass of the flat railed annulus, there is an equivalent linkage in the graph obtained after replacing the flat railed annulus with its leveling (see~\autoref{lem_levelingpaths}).
Therefore, the question for linkages between pairs of boundary vertices
  can be translated to an equivalent one in the partially annulus-embedded graph obtained after considering the leveling (this graph is different for every different trace).
This allows for the application of the Linkage Combing Lemma.

\paragraph{Partial patterns.}
Using the Linkage Combing Lemma we can assume that, for each tuple of vertices, every linkage between these vertices has an appropriate part of the initial (flat) railed annulus where it is combed.
This implies that the pattern of a boundaried graph can be ``separated'' to two parts,
the one on the ``inner'' part and one on the ``outer'' part of the annulus corresponding to the trace of the boundary tuple.
Therefore, we can encode partial patterns by considering boundaried graphs whose boundary vertices are the vertices of the rails where we comb the linkages 
and some vertices that are on the ``inner side'' of this annulus.
Vertices can also be on the ``outer side'' of this annulus and thus we have to suitably encode their absence from the tuple (see~\autoref{subsec_combing_in_levelings}).

\paragraph{Dealing with apices.}
Having sketched how to route linkages (that correspond to disjoint path queries of $\FOL{\sf +DP}$) in flat walls,
it remains to discuss how the presence of apices, that are few vertices that can have neighbors inside the wall in a completely uncontrolled way.
The high-level idea here is to define a way to syntactically interpret these (few) apices by adding some extra colors in the structure for the neighborood of each apex vertex, remove the edges between apices and the rest of the graph and translate a disjoint path query in the original graph to some disjoint path queries in the new colored graph, encoding possible ways that the original paths could enter/exit the apex set.
This idea originates in~\cite{FlumG01fixe} for \FOL\ (without annotation) and an analogous version it was given in~\cite{FominGSST21acomp}.


\paragraph{Apex-projections of sentences.}
To define this translation, we first need to define the ``projection'' of a sentence with respect to an apex-tuple ${\bf a}$
(see~\autoref{subsec_apices}).
In fact, first we define the ``projection'' of the graph $G$ with respect to an apex-tuple ${\bf a}$, by removing all edges between the vertices in ${\bf a}$ and the rest of the vertices of the graph and ``coloring'' the neighborhood of each $a\in {\bf a}$ by a specific color.
This transformation of the graph allows us to ``isolate'' the apex-tuple and encode its adjacencies with the rest of the graph using colors.
In terms of sentences, we define the \emph{apex-projection} $\varphi^l$ of a sentence $\varphi\in\FOL{\sf +DP}$, where $l = |{\bf a}|$,
as a sentence that uses these new colors in order to interpret the original sentence $\varphi$ in the ``projected'' graph.
We stress that the apex-projected sentence $\varphi^l$ has larger quantifier rank, i.e., we quantify more variables, since
one has to guess which (colored) vertices are the entry and exit points of the paths using vertices in the apex-tuple (see~\autoref{obs_quantifierrank}).

\paragraph{Partial signatures.}
The next step is to build ``meta-collections'' of partial patterns.
For this, in~\autoref{subsec_signatures},
we define the notion of \emph{partial signatures}.
This is a recursive definition, that in the base case is the pattern of a boundaried graph (whose boundary is a tuple of vertices together with the ``combing-points'' of the railed annulus corresponding to the trace of the given tuple) and every recursive step asks for all possible partial signatures that can be obtained after fixing another (boundary) vertex (inside the annotated set) or its absence.
Intuitively, given an annotated graph $(G,R_1,\ldots,R_r)$ that contains a sufficiently big flat railed annulus,
the partial signature of $G$ encodes the recursive collection of (partial) patterns in the ``inner'' part of the railed annulus.

\paragraph{Exchangability of graphs with the same partial signature.}
After defining partial signatures, we prove that in the presence of a railed annulus flatness pair inside an (annotated) graph $(H,\hat{R}_1,\ldots,\hat{R}_r)$,
two (annotated) graphs $(G,R_1,\ldots,R_r)$, 
$(G',R_1',\ldots,R_r')$
that yield the same partial signature when ``glued in the inner part'' of $(H,\hat{R}_1,\ldots,\hat{R}_r)$,
have the same (global) signature when we additionally glue another (annotated) colored graph 
$(F,R^\star_1,\ldots,R^\star_r)$
in the outer part of $(H,\hat{R}_1,\ldots,\hat{R}_r)$ (see~\autoref{lem_equirep}).
The main idea of the proof of~\autoref{lem_equirep} is to show that equivalence of partial signatures implies equivalence of (global) signatures.
Here, we have to demand equality of partial signatures of ``larger depth'' than the one that we aim for, in order to be able to deal with the extra quantifiers
obtained after apex-projecting the quantifier-free formulas given by the pattern. This translation (to the apex-projection) is essential, in order to obtain a flat structure (without edges to the apiecs) where our rerouting arguments for the linkages can work.
In particular, proving equivalency of signatures 
boils down to formalizing
all arguments mentioned in the above paragraph
that use the Linkage Combing Lemma, using a double inductive
argument, in order to deal with the recursive nature of the definitions.
The main importance of ~\autoref{lem_equirep} is that it guarantees exchangeability (in terms of signatures) of graphs with the same partial signature
and allows to shift our pursue of graphs of same signature to pursue of graphs of same partial signature.

\paragraph{Computing representatives.}
Since the partial signature of an annotated graph depends only on adjacency and linkage questions inside the ``inner'' part of a given flat railed annulus of the given graph, we can express the partial signature of
the annotated graph in \MSOL.
Also, as we mentioned before, we can ask the part of the graph, where we want to compute partial signature, to have bounded treewidth.
Therefore,
we can define equivalence between vertices of the same partial signatures and using Courcelle's Theorem we can compute \emph{representatives} of these equivalence relations (see~\autoref{subsec_sigrep}).
Therefore, given an annotated graph $(G,R_1,\ldots,R_r)$, an apex-tuple ${\bf a}$ of $G$, and a ``big enough'' flat railed annulus of $G\setminus V({\bf a})$, where $V({\bf a})$ is the set of vertices in ${\bf a}$, that ``crops'' a bounded treewidth graph,
we can compute sets $R_1',\ldots,R_r'$ such that $R_i'\subseteq R_i$ and the size of the intersection of $R_i'$ with the ``left cropped part'' of the railed annulus 
is upper-bounded by a function of $k$, and moreover $(G,R_1,\ldots,R_r)$ and $(G,R_1',\ldots,R_r')$ have the same partial signature (see~\autoref{lem_repre}).
Therefore, in the presence of a big enough flat wall (of bounded treewidth) inside an annotated graph, we can
bound the number of annotated vertices in the inner part of the railed annulus, which, in turn, permits us to declare annotation-irrelevant the vertices of some ``central'' part of the flat wall that contains the railed annulus (\autoref{lem_new_ext_rep}).

\paragraph{Wrapping-up the proof of~\autoref{thm_main}.}
As mentioned in the beginning of this section,
the most demanding part of our proof is reducing the annotation set, i.e., to show~\autoref{lem_new_ext_rep}.
Then, given that the sets $R_1,\ldots,R_r$, in which we build the signatures, leave some big enough bidimensional area ``annotation-irrelevant'', it remains to compute
a vertex inside this area whose removal does not affect the existence (or not) of any linkage between the interpretations of the variables in the definition of signature. This is done using known irrelevant vertex technique arguments (see~\autoref{subsec_reduce} and in particular~\autoref{lem_removingirr}).
Therefore, we can find a vertex that is irrelevant for the signature (\autoref{corol_redu}).
This produces an equivalent instance $(G',R_1',\ldots,R_r')$ as
required for the application of the main procedure of our algorithm.

%

\paragraph{From disjoint to scattered paths.}
To deal with $\FOL{\sf +SDP}$ and prove \autoref{thm_induced},
we need a Linkage Combing Lemma for scattered linkages.
Such kind of a result is known when we restrict ourselves to graphs of bounded Euler genus and is proved in~\cite{GolovachST22comb} (\autoref{proposition_combinginducedlinkages}; see also~\cite{KawarabayashiK12alin,Mazoit13asin}).
Then, using this result, we follow the same approach as in the proof of~\autoref{thm_main}, and
we prove tractability of model-checking for $\FOL{\sf +SDP}$ for graphs of embedded in some fixed surface (\autoref{thm_induced}). We refer the reader to~\autoref{sec_logicscattered} for more details.

\section{Preliminaries}
\label{sec_preliminaries}
In this section we present some basic definitions and we state our main result (\autoref{thm_main}).
In~\autoref{subsec_basic}, we start with some definitions on graphs and in~\autoref{subsec_log} we define first-order and monadic second-order logic.
Then, in~\autoref{subsec_dplogic} we define the extension of \FOL\ with disjoint paths predicates and we state our main result (\autoref{thm_main}).
\subsection{Basic definitions}\label{subsec_basic}

\paragraph{Integers, sets, and tuples.}
We denote by $\mathbb{N}$ the set of non-negative integers.
Given two integers $p$ and
$q,$ the set $[p,q]$ refers to the set of every integer $r$ such that $p \leq r \leq q.$
For an integer $p\geq 1,$ we set $[p]=[1,p]$ and $\mathbb{N}_{\geq p}=\mathbb{N}\setminus [0,p-1].$
Given a non-negative integer $x,$
we denote by ${\sf odd}(x)$ the minimum odd number that is not smaller than $x.$
For a set $S,$ we denote by $2^{S}$ the set of all subsets of $S$ and, given an integer $r\in[|S|],$
we denote by $\binom{S}{r}$ the set of all subsets of $S$ of size $r.$
Given two sets $A,B$ and a function $f: A\to B,$
for a subset $X\subseteq A$ we use $f(X)$ to denote the set $\{f(x)\mid x\in X\}.$
Given a set $X$ and an $r\in\mathbb{N}$, we use $X^r$ to denote the product $X\times\cdots\times X$ of $r$ copies of $X$.

Let $\mathcal{S}$ be a collection of objects where the operations $\cup$ and $\cap$ are defined.
Given two tuples ${\bf x} = (x_1, \ldots, x_l)$ and ${\bf y} = (y_1, \ldots, y_l),$ where $x_i, y_i \in \mathcal{S},$ we denote ${\bf x}\cup{\bf y} = (x_1 \cup y_1, \ldots, x_l \cup y_l)$ and ${\bf x}\cap{\bf y} = (x_1 \cap y_1, \ldots, x_l \cap y_l).$
Also, we denote $\cupall \mathcal{S} = \bigcup_{X\in \mathcal{S}}X.$

\paragraph{Basic concepts on graphs.}
All graphs considered in this paper are undirected, finite, and without loops or multiple edges.
We use standard graph-theoretic notation and we refer the reader to
\cite{Diestel10grap} for any undefined terminology.
Let $G$ be a graph.
Given a vertex $v\in V(G),$ we denote by $N_{G}(v)$ the set of vertices of $G$ that are adjacent to $v$ in $G.$
For $S \subseteq V(G),$ we set $G[S]=(S,E\cap\binom{S}{2} )$ and use the shortcut $G \setminus S$ to denote $G[V(G) \setminus S].$
The \emph{length} of a path $P$ of $G$ is the number of edges of $P$.

\paragraph{Minors.}
The \emph{contraction} of an edge $e = \{u,v\}$ of a simple graph $G$ results in a simple graph $G'$
obtained from $G \setminus \{u,v\}$ by adding a new vertex $uv$ adjacent to all the vertices
in the set $N_G(u) \cup N_G(v)\setminus \{u,v\}.$
A graph $G'$ is a \emph{minor} of a graph $G,$ denoted by $G'\prem G,$
if $G'$ can be obtained from $G$ by a sequence of vertex removals, edge removals, and edge contractions.
Given a finite collection of graphs $\mathcal{F}$ and a graph $G,$ we use the notation $\mathcal{F}\prem G$ to denote that some graph in $\mathcal{F}$ is a minor of $G.$
Given a set of graphs $\mathcal{F},$ we denote by $\excl(\mathcal{F})$ the set containing every graph that excludes all graphs in $\mathcal{F}$ as minors.
A graph class $\mathcal{G}$ is \emph{minor-closed} if every minor of a graph in $\mathcal{G}$ is also a member of $\mathcal{G}.$
The \emph{Hadwiger number} of a graph $G,$ denoted by $\hw(G),$  is the minimum $k$ where $G\in \excl(\{K_{k}\})$ and $K_{k}$ is the complete graph on $k$ vertices.
A minor-closed graph class is called \emph{proper} if it is not the class of all graphs.
%

\paragraph{Treewidth.}
A \emph{tree decomposition} of a graph~$G$
is a pair~$(T,\chi)$ where $T$ is a tree and $\chi: V(T)\to 2^{V(G)}$
such that
\begin{itemize}
	\item $\bigcup_{t \in V(T)} \chi(t) = V(G),$
	\item for every edge~$e$ of~$G$ there is a $t\in V(T)$ such that
	      $\chi(t)$
	      contains both endpoints of~$e,$ and
	\item for every~$v \in V(G),$ the subgraph of~${T}$
	      induced by $\{t \in V(T)\mid {v \in \chi(t)}\}$ is connected.
\end{itemize}
The \emph{width} of $(T,\chi)$ is equal to $\max\big\{\left|\chi(t)\right|-1 \bigmid t\in V(T)\big\}$ and the \emph{treewidth} of $G$ is the minimum width over all tree decompositions of $G.$

\subsection{First-order logic and monadic second-order logic}\label{subsec_log}

In this subsection, we give the definition of first-order and monadic second-order logic.
We define these logics on {\sl relational} vocabularies with constant symbols and we work with structures of these vocabularies.

\paragraph{Structures.}
A \emph{vocabulary} is a finite set of relation and constant symbols (we do not use function symbols).
Every relation symbol ${\sf R}$ is associated with a positive integer that is called the \emph{arity} of ${\sf R},$ which we denote ${\sf ar}({\sf R}).$
A \emph{structure $\mathfrak{A}$ of vocabulary $\tau$}, in short a \emph{$\tau$-structure}, consists of a non-empty set $V(\mathfrak{A}),$ called the \emph{universe} of $\mathfrak{A},$
 an $r$-ary relation ${\sf R}^{\mathfrak{A}}\subseteq {V(\mathfrak{A})}^r$ for each relation symbol ${\sf R}\in \tau$ of arity $r\geq 1,$
and an element\footnote{{We stress that we allow constant symbols to be interpreted as the element $\mathspace$, where $\mathspace$ is an element that is not in $V(\mathfrak{A})$.
Throughout this paper, we assume that the universe of every given structure is extended by adding the extra element $\mathspace$, while all relation symbols are interpreted as tuples of elements of $V(\mathfrak{A})$, not containing $\mathspace$.
Moreover, we assume that for every formula that we consider, quantified first-order variables are interpreted as elements of the original universe of the structure (and not $\mathspace$).}}
 ${\sf c}^{\mathfrak{A}}\in \{\mathspace\}\cup V(\mathfrak{A})$ for each constant symbol ${\sf c}\in\tau$.
We refer to ${\sf R}^\mathfrak{A}$ (resp. ${\sf c}^{\mathfrak{A}}$) as the \emph{interpretation of the symbol ${\sf R}$ (resp. ${\sf c}$) in the structure $\mathfrak{A}$}.
A structure $\mathfrak{A}$ is \emph{finite} if its universe $V(\mathfrak{A})$ is a finite set.
We denote by $\mathbb{STR}[\tau]$ the set of all finite $\tau$-structures.\medskip

We say that a $\tau$-structure $\mathfrak{A}$ is \emph{isomorphic} to a $\tau$-structure $\mathfrak{B}$ if there is a bijection  $V(\mathfrak{A})\cup\{\mathspace\}$ to $V(\mathfrak{B})\cup\{\mathspace\},$ such that $\pi(\mathspace) = \mathspace$ and for every $k\geq 1,$ every relation symbol ${\sf R}\in \tau$ of arity $k,$ and every $(a_1, \ldots, a_k)\in {V(\mathfrak{A})}^k,$ it holds that $(a_1,\ldots, a_k)\in {\sf R}^{\mathfrak{A}} \iff (\pi(a_1), \ldots, \pi(a_k))\in {\sf R}^{\mathfrak{B}}$ and for every constant symbol ${\sf c}\in \tau,$ it holds that $\pi({\sf c}^{\mathfrak{A}}) = {\sf c}^{\mathfrak{B}}.$\medskip

Given two $\tau$-structures $\mathfrak{A}$ and $\mathfrak{B}$,
where for every constant symbol ${\sf c}\in \tau$ either ${\sf c}^{\mathfrak{A}}={\sf c}^{\mathfrak{B}}$ or ${\sf c}^{\mathfrak{A}} = \mathspace \lor {\sf c}^{\mathfrak{B}} = \mathspace$, we define the \emph{disjoint union} of $\mathfrak{A}$ and $\mathfrak{B},$ and we denote it by $\mathfrak{A}\dot\cup\mathfrak{B},$ as
the $\tau$-structure where $V(\mathfrak{A}\dot\cup\mathfrak{B})$ is the disjoint union of $V(\mathfrak{A})\setminus\{\mathspace\}$, $V(\mathfrak{B})\setminus\{\mathspace\}$ and $\{\mathspace\}$, for every relation symbol ${\sf R}\in \tau,$ ${\sf R}^{\mathfrak{A}\dot\cup\mathfrak{B}}= {\sf R}^{\mathfrak{A}}\cup {\sf R}^{\mathfrak{B}},$
and
for every constant symbol ${\sf c}\in \tau$, if  ${\sf c}^{\mathfrak{A}}={\sf c}^{\mathfrak{B}}$, then
${\sf c}^{\mathfrak{A}\dot\cup\mathfrak{B}}= {\sf c}^{\mathfrak{A}} ={\sf c}^{\mathfrak{B}}$, and if ${\sf c}^{\mathfrak{A}} = \mathspace$ (resp. ${\sf c}^{\mathfrak{B}} = \mathspace$), then ${\sf c}^{\mathfrak{A}\dot\cup\mathfrak{B}} = {\sf c}^{\mathfrak{B}}$ (resp. ${\sf c}^{\mathfrak{A}\dot\cup\mathfrak{B}} = {\sf c}^{\mathfrak{A}}$).

An undirected graph  without loops can be seen as an $\{{\sf E}\}$-structure $\mathfrak{G} = (V(\mathfrak{G}),{\sf E}^{\mathfrak{G}}),$ where ${\sf E}^{\mathfrak{G}}$ is a binary relation that is symmetric and anti-reflexive.

The \emph{Gaifman graph} $G_{\mathfrak{A}}$ of a $\tau$-structure $\mathfrak{A}$ is the graph whose vertex set is $V(\mathfrak{A})$ and two vertices $x,y$ are adjacent if there is an ${\sf R}\in\tau$ and a  $\bar{v}\in {\sf R}^{\mathfrak{A}}$ such that both $x$ and $y$ are elements of $\bar{v}$.

\paragraph{First-order logic and monadic second-order logic.}
We now define the syntax and the semantics of first-order logic and monadic second-order logic of a vocabulary $\tau.$
We assume the existence of a countable infinite set of \emph{first-order variables},
usually denoted by lowercase symbols ${\sf x}_1,{\sf x}_2,\ldots,$
and of a countable infinite set of \emph{set variables},
usually denoted by uppercase symbols ${\sf X}_1,{\sf X}_2, \ldots.$
A \emph{first-order term} is either a first-order variable or a constant symbol.
A \emph{first-order logic formula}, in short \emph{\FOL-formula}, of vocabulary $\tau$ is built from atomic formulas ${\sf x}={\sf y}$
and $({\sf x}_1, \ldots, {\sf x}_r)\in {\sf R},$ where ${\sf R}\in \tau$ and has arity $r\geq 1,$ on first-order terms ${\sf x},{\sf y},{\sf x}_1,\ldots, {\sf x}_r,$
by using the logical connectives $\vee,$ $\wedge,$ $\neg$ and the
quantifiers $\forall, \exists$ on first-order variables.
We denote by $\FOL[\tau]$ the set of all \FOL-formulas of vocabulary $\tau.$\medskip

A \emph{monadic second-order logic formula}, in short \emph{{\sf MSOL}-formula}, of vocabulary $\tau$ is obtained by enhancing
the syntax of \FOL-formulas by allowing the atomic formulas ${\sf x}\in {\sf X},$
for some first-order term ${\sf x}$ and some set variable ${\sf X},$ and allowing quantification on both first-order and set variables.
We denote by ${\sf MSOL}[\tau]$ the set of all {\sf MSOL}-formulas of vocabulary $\tau.$
We make clear that what we call here {\sf MSOL} is what is commonly referred in the literature as {\sf MSO}$_1$, in which, for the vocabulary of graphs, first-order variables are interpreted as vertices and set variables are interpreted as sets of vertices.\medskip

The formulas in $\FOL[\tau]$ and $\MSOL[\tau]$ are evaluated on $\tau$-structures by interpreting every
symbol in $\tau$ as its interpretation in the structure and every first-order (resp. set) variable as an element
(resp. set of elements) of the universe of the structure.
Given a formula $\varphi,$ the \emph{free variables} of $\varphi$ are its variables
that are not in the scope of any quantifier.
We write $\varphi({\sf x}_1,\ldots, {\sf x}_k)$ to indicate that the free variables of the formula $\varphi$ are ${\sf x}_1, \ldots,{\sf x}_k.$
A \emph{sentence} is a formula without free variables.
Let $\varphi = Q_1 {\sf x}_1\ldots Q_r {\sf x}_r \psi({\sf x}_1,\ldots, {\sf x}_r)$ be a sentence, where for each $i\in[r]$,
$Q_i\in \{\forall, \exists\}$, ${\sf x}_1,\ldots, {\sf x}_r$ are first-order variables,
and $\psi({\sf x}_1,\ldots, {\sf x}_r)$ is a quantifier-free formula with free variables ${\sf x}_1,\ldots, {\sf x}_r$.
We call $r$ the \emph{quantifier rank} of $\varphi$.\medskip

Given a $\tau$-structure $\mathfrak{A},$ a formula $\varphi({\sf x}_1, \ldots, {\sf x}_k)\in \FOL[\tau],$
and $a_1,\ldots, a_k$ in $V(\mathfrak{A}),$ we write $\mathfrak{A}\models \varphi(a_1, \ldots, a_k)$ to denote that $\varphi({\sf x}_1,\ldots, {\sf x}_k)$ holds in $\mathfrak{A}$ if, for every $i\in[k],$ the variable ${\sf x}_i$ is interpreted as $a_i.$
Given a $k\in\mathbb{N}$, two formulas $\varphi({\sf x}_1, \ldots, {\sf x}_k),$ $\psi({\sf x}_1,\ldots, {\sf x}_k)\in \FOL[\tau]$ are \emph{equivalent} if for every $\tau$-structure $\mathfrak{A}$ and every $a_1, \ldots, a_k\in V(\mathfrak{A}),$ we have $\mathfrak{A}\models \varphi(a_1, \ldots, a_k) \iff \mathfrak{A}\models \psi(a_1, \ldots, a_k).$
We call the set $\{\mathfrak{A}\in \mathbb{STR}[\tau]\mid \mathfrak{A}\models\varphi\}$ the set of \emph{models of $\varphi$} and we denote it by $\Mod(\varphi).$

\subsection{Disjoint paths logic}\label{subsec_dplogic}
For the rest of this paper,
we deal with colored graphs, i.e.,
we fix a vocabulary $\tau$ that contains a binary relation symbol ${\sf E}$ that is always interpreted as a symmetric and anti-reflexive binary relation (corresponding to the edges of the graph) and a collection of unary relation symbols ${\sf Y}_1,\ldots, {\sf Y}_h$ (corresponding to colors on the vertices of the graph).
We call such a vocabulary a \emph{colored-graph vocabulary}.
Also, we always assume that the interpretations of ${\sf Y}_1,\ldots, {\sf Y}_c$ are always pairwise disjoint (if not, then we introduce extra unary relations symbols for each intersection).

\paragraph{Disjoint-paths logic.}
We define the $2k$-ary predicate $\DP_k({\sf x}_1, {\sf y}_1, \ldots,{\sf x}_k, {\sf y}_k)$, which evaluates true in a $\tau$-structure $\mathfrak{G}$  if and only if there
are paths $P_1,\ldots, P_k$ of $(V(\mathfrak{G}),{\sf E}^{\mathfrak{G}})$ of length at least $2$ between (the interpretations of) ${\sf x}_i$ and ${\sf y}_i$ for all $i\in[k]$ such that for every $i,j\in[k]$, $i\neq j$, $V(P_i)\cap V(P_j)=\emptyset.$
We let $\tau + \DP := \tau \cup \{\DP_k\mid k\geq 1\}$, where each
$\DP_k$ is a $2k$-ary relation symbol.
We use $\DP$ instead of $\DP_k$ when $k$ is clear from the context.
Our main result is the following (slightly more general) version of \autoref{@definitionen}:

\begin{theorem}\label{thm_main}
For every colored-graph vocabulary $\tau$ and every sentence $\varphi\in\FOL[\tau+\DP]$, there exists an algorithm that, given a $\tau$-structure $G$ of size $n$, outputs whether $G\models \varphi$ in time $\mathcal{O}_{|\varphi|,{\bf hw}(G)}(n^2)$.
\end{theorem}

\section{An alternative view to first-order logic model-checking}
\label{sec_alternative}
In this section we present an alternative way to interpret first-order logic model-checking,
by embedding a given graph to a (rooted) tree.
In~\autoref{subsec_treeembed}, we translate the problem of whether a  graph satisfies a formula $\varphi$ to 
the search of a subtree of the tree in which the graph is embedded, such that the bifurcations of this subtree correspond to the quantifiers of $\varphi$ and the vertices collected in each root-to-leaf path satisfy the quantifier-free ``tail'' of $\varphi$. 
We also formally define the notion of \emph{signature} of a graph that encodes in a tree-like way how tuples of vertices satisfy quantifier-free formulas of $\FOLDP$.
In~\autoref{subsec_patt}, we define the notion of {\sl pattern} of a colored graph $\mathfrak{G}$ together with some vertices $v_1,\ldots,v_r\in V(\mathfrak{G})$.
This notion is used to encode the way all quantifier-free $\FOLDP$-formulas can be satisfied by $\mathfrak{G}$ when interpreting the variables as $v_1,\ldots,v_r$.
This way, we also define the respective notion of patterns of quantifier-free $\FOLDP$-formulas and, in~\autoref{subsec_exprpatt}, we formulate model-checking in these terms (\autoref{obs_grleafs}).
Finally, in~\autoref{subsec_equiv}, we prove the main result of this section (\autoref{lemma_reducing}), which intuitively states that two graphs that have the same signatures (or, in other words, the same recursive patterns) satisfy the same sentences.

\subsection{Embedding model-checking to trees}\label{subsec_treeembed}

In this subsection we present a way to embed graphs to trees and how to trace the satisfaction of a sentence from the given graph, in terms of subtrees of the original tree.

\paragraph{Rooted trees.}
Let $(T,t_0)$ be a rooted tree.
Given a node $x$ of $T$, we denote by $T_x$ the subtree of $T$ rooted at $x$.
We use $L(T)$ to denote the leaves of $T$.
For every $i\in[0,r]$, where $r\in \mathbb{N}$ is the height of $T$,
we use $D_i(T)$ to denote the set of nodes of $T$ that are at distance $i$ from $t_0$.
We say that a node $t\in V(T)$ has \emph{depth} $i$ if $t\in D_i (T),i\in[0,r]$.
In this paper, for every rooted tree $(T,t_0)$ of height $r$ that we consider,
we assume that $L(T)= D_r(T)$.
We use $\textsf{Paths}(T)$ to denote the set of all root-to-leaf paths of $T$.
We denote by $\textsf{children}_T (t)$ the set of all children of $t$ in $T$.
Also, when it is clear from the context, we use $(t_0,\ldots, t_r)$ to denote a path $P\in \textsf{Paths}(T)$ such that $V(P)=\{t_0,\ldots,t_r\}$ and $E(P)=\bigcup_{i\in[0,r-1]}\{\{t_i,t_{i+1}\}\}$.
A rooted tree $(T,t_0)$ of height $r$ is called \emph{$n$-ary},
if each node of depth at most $r-1$ has $n$ children.

\paragraph{Trees expressing quantification of formulas.}
We now express how, given a rooted tree $(T,t_0)$ and a sentence $\varphi$, where the height of the tree and the quantifier rank of $\varphi$ are the same, use the quantifier alternation of $\varphi$ to construct a subtree $T'$ of $T$.
This is done in the following recursive way: the subtree starts from $t_0$ and, for every $t\in V(T')$,
if the considered quantifier of $\varphi$ is the universal one, then $T'$ spans to {\sl all} children of $t$, while if we have the existential one, $T'$ arbitrarily choses {\sl one} child of $t$.
We proceed to formalize this idea.

Let $\tau$ be a colored-graph vocabulary.
Let $\varphi = Q_1 {\sf x}_1\ldots Q_r {\sf x}_r \psi({\sf x}_1,\ldots, {\sf x}_r)$ be an $\FOL[\tau+\DP]$-sentence and
let $(T,t_0)$ be a rooted tree of height $r$ (where each root-to-leaf path has $r+1$ nodes).
A \emph{$\varphi$-spanning triple} of $(T,t_0)$ is a triple $(T,t_0,\mathcal{U}_\varphi)$ such that
$\mathcal{U}_\varphi=\{U_0,U_1,\ldots, U_r\}$ where $U_0 =\{t_0\}$
and for every $i\in [r]$, $U_i$ is a subset of $D_i(T)$ such that
\begin{itemize}
\item if $Q_i = \forall$, then for each node $v\in U_{i-1}$,
we add in $U_i$ \textbf{all} nodes $u\in\textsf{children}_T(v)$ and
\item if $Q_i = \exists$, then for each node $v\in U_{i-1}$,
we add in $U_i$ \textbf{one} node $u\in\textsf{children}_T(v)$.
\end{itemize}
A \emph{$\varphi$-spanning subtree} of $T$ is a rooted subtree $(T',t_0)$ of $(T,t_0)$
where $T'=T[\bigcup_{i=0}^r U_i]$ for a $\varphi$-spanning triple
$(T,t_0,\{U_0,U_1,\ldots, U_r\})$ of $(T,t_0)$. See~\autoref{figure_spanningsubtree}.
Note that, given a rooted tree $(T,t_0)$ of height zero and
a quantifier-free formula $\psi$, $(T,t_0)$ is the unique $\psi$-spanning subtree of $T$.

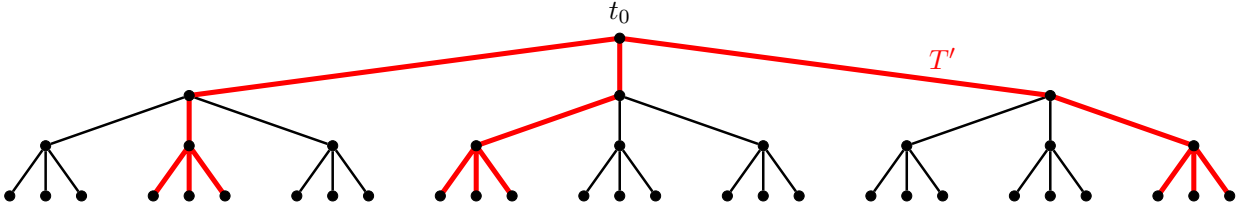
\begin{figure}[ht]
\scalebox{0.954}{
\begin{tikzpicture}
\node[treenode] (v000) at (0,0) {};
\node[treenode] (mid00) at (0.5,0) {};
\node[treenode] (v001) at (1,0) {};
\node[treenode] (v010) at (2,0) {};
\node[treenode] (mid01) at (2.5,0) {};
\node[treenode] (v011) at (3,0) {};
\node[treenode] (v020) at (4,0) {};
\node[treenode] (mid02) at (4.5,0) {};
\node[treenode] (v021) at (5,0) {};
\node[treenode] (v100) at (6,0) {};
\node[treenode] (mid10) at (6.5,0) {};
\node[treenode] (v101) at (7,0) {};
\node[treenode] (v110) at (8,0) {};
\node[treenode] (mid11) at (8.5,0) {};
\node[treenode] (v111) at (9,0) {};
\node[treenode] (v120) at (10,0) {};
\node[treenode] (mid12) at (10.5,0) {};
\node[treenode] (v121) at (11,0) {};
\node[treenode] (v200) at (12,0) {};
\node[treenode] (mid20) at (12.5,0) {};
\node[treenode] (v201) at (13,0) {};
\node[treenode] (v210) at (14,0) {};
\node[treenode] (mid21) at (14.5,0) {};
\node[treenode] (v211) at (15,0) {};
\node[treenode] (v220) at (16,0) {};
\node[treenode] (mid22) at (16.5,0) {};
\node[treenode] (v221) at (17,0) {};

\node[treenode] (v00) at (0.5,0.7) {};
\node[treenode] (v01) at (2.5,0.7) {};
\node[treenode] (v02) at (4.5,0.7) {};
\node[treenode] (v10) at (6.5,0.7) {};
\node[treenode] (v11) at (8.5,0.7) {};
\node[treenode] (v12) at (10.5,0.7) {};
\node[treenode] (v20) at (12.5,0.7) {};
\node[treenode] (v21) at (14.5,0.7) {};
\node[treenode] (v22) at (16.5,0.7) {};

\node[treenode] (v0) at (2.5,1.4) {};
\node[treenode] (v1) at (8.5,1.4) {};
\node[treenode] (v2) at (14.5,1.4) {};

\node[treenode, label={$t_0$}] (t0) at (8.5,2.2) {};
\node[label={[red]$T'$}] () at (13,1.5) {};

\begin{scope}[on background layer]
\draw[line width=1pt] (t0) -- (v0) (t0) -- (v1) (t0) -- (v2);
\draw[line width=1pt] (v0) -- (v00) (v0) -- (v01) (v0) -- (v02);
\draw[line width=1pt] (v1) -- (v10) (v1) -- (v11) (v1) -- (v12);
\draw[line width=1pt] (v2) -- (v20) (v2) -- (v21) (v2) -- (v22);
\draw[line width=1pt] (v00) -- (v000) (v00) -- (v001) (v00) -- (mid00);
\draw[line width=1pt] (v01) -- (v010) (v01) -- (v011) (v01) -- (mid01);
\draw[line width=1pt] (v02) -- (v020) (v02) -- (v021) (v02) -- (mid02);
\draw[line width=1pt] (v10) -- (v100) (v10) -- (v101) (v10) -- (mid10);
\draw[line width=1pt] (v11) -- (v110) (v11) -- (v111) (v11) -- (mid11);
\draw[line width=1pt] (v12) -- (v120) (v12) -- (v121) (v12) -- (mid12);
\draw[line width=1pt] (v20) -- (v200) (v20) -- (v201) (v20) -- (mid20);
\draw[line width=1pt] (v21) -- (v210) (v21) -- (v211) (v21) -- (mid21);
\draw[line width=1pt] (v22) -- (v220) (v22) -- (v221) (v22) -- (mid22);

\draw[line width=2pt, red] (t0.center) -- (v0.center) (t0.center) -- (v1.center) (t0.center) -- (v2.center);
\draw[line width=2pt, red] (v0.center) -- (v01.center) (v1.center) -- (v10.center) (v2.center) -- (v22.center);
\draw[line width=2pt, red]
(v01.center) -- (v010.center)
(v01.center) -- (mid01.center)
(v01.center) -- (v011.center)
(v10.center) -- (v100.center)
(v10.center) -- (mid10.center)
(v10.center) -- (v101.center)
(v22.center) -- (v220.center)
(v22.center) -- (mid22.center)
(v22.center) -- (v221.center);
\end{scope}
\end{tikzpicture}}
\caption{A rooted tree $(T,t_0)$ and a $\varphi$-spanning subtree $T'$ of $T$ (depicted in red), for a sentence $\varphi=\forall{\sf x}_1 \exists {\sf x}_2 \forall {\sf x}_3 \psi({\sf x}_1,{\sf x}_2, {\sf x}_3)$.}
\label{figure_spanningsubtree}
\end{figure}


\paragraph{Signatures of structures.}
Let $\tau$ be the vocabulary of graphs of $h$ colors and $\ell$ roots (i.e., constants).
We use $\Psi_{\FOL[\tau+\DP]}^{r,h,\ell}$ to denote the set of all quantifier-free $\FOL[\tau+\DP]$-formulas with $r$ free variables.
We treat equivalent formulas as equal (and choose one representative for each equivalence class, which is possible for quantifier-free formulas).
Then the size of $\Psi_{\FOL[\tau+\DP]}^{r,h,\ell}$ is upper-bounded by some constant depending only on $r$, $h$, and~$\ell$.
Let $r\in\mathbb{N}$, let $\mathfrak{G}$ be a $\tau$-structure, let $R_1,\ldots,R_r\subseteq V(\mathfrak{G})$.
We set $\bar{R}=(R_1,\ldots,R_r)$.
The \emph{atomic type} of a tuple $(v_1,\ldots,v_r)\in (V(G)\cup\{\mathspace\})^r$
is the set of all atomic formulas that are true for $(v_1,\ldots,v_r)$ in~$\mathfrak{G}$.
Given $(v_1,\ldots,v_r)\in (V(G)\cup\{\mathspace\})^r$, we define $\mathsf{sig}^0(\mathfrak{G},\bar{R},v_1,\ldots,v_r)$
to be the atomic type of $(v_1,\ldots,v_r)$.
Also, for each $i\in[r-1]$ and every $v_1,\ldots,v_{r-i}\in V(G)\cup\{\mathspace\}$,
we define
\[\mathsf{sig}^i(\mathfrak{G},\bar{R},v_1,\ldots,v_{r-i}) = \big\{\mathsf{sig}^{i-1}(\mathfrak{G},\bar{R},v_1,\ldots,v_{r-i},u)\mid u\in R_{r-i+1}\cup\{\mathspace\}\big\}\]
Finally, we define
\[\mathsf{sig}^r(\mathfrak{G},\bar{R}) = \big\{\mathsf{sig}^{r-1}(\mathfrak{G},\bar{R},v)\mid v\in R_1\cup\{\mathspace\}\big\}.\]
In the case where $\bar{R}=V(\mathfrak{G})^r$, we omit $\bar{R}$ from the notation, i.e., we write $\mathsf{sig}^i(\mathfrak{G},v_1,\ldots,v_{r-i})$ and $\mathsf{sig}^r(\mathfrak{G})$ instead of $\mathsf{sig}^i(\mathfrak{G},\bar{R},v_1,\ldots,v_{r-i})$ and $\mathsf{sig}^r(\mathfrak{G},\bar{R})$.
Also, we set $\mathcal{B}^{r,h,\ell}_\mathsf{sig} := \{\mathsf{sig}^r(\mathfrak{G})\mid \mathfrak{G}\in\mathbb{STR}[\tau]\}$.

\begin{observation}
Let $\tau$ be the vocabulary of graphs of $h$ colors and $\ell$ roots and let $r\in\mathbb{N}$.
For every $\beta\in \mathcal{B}^{r,h,\ell}_\mathsf{sig}$, there is a $\varphi_\beta\in\FOL[\tau+\DP]$ such that for every $\tau$-structure $\mathfrak{G}$,
$\mathsf{sig}^r(\mathfrak{G}) = \beta$ if and only if $\mathfrak{G}\models \varphi_\beta$.
\end{observation}

\paragraph{Assigning graphs to rooted trees.}
Let $r\in\mathbb{N}$.
Given a graph $G$, we construct a rooted tree $(T,t_0)$ (starting from a single root $t_0$) and a function $\lambda: V(T)\setminus\{t_0\}\to V(G)$ as follows.
First, for every $\alpha\in \mathsf{sig}^{r}(\mathfrak{G})$, we consider a vertex $v_\alpha\in V(G)$ such that $\mathsf{sig}^{r-1}(\mathfrak{G},v) = \alpha$.
We add $|\mathsf{sig}^r(\mathfrak{G})|$ children to $t_0$, while we define $\lambda_{t_0}$ to be a bijection from $\textsf{children}_T (t_0)$ to $\{v_\alpha\mid \alpha\in \mathsf{sig}^{r}(\mathfrak{G})\}$.
Then, for every $t\in V(T)$, if $t_0,t_1,\ldots,t_d=t$ is a path of $T$ and $(u_1,\ldots,u_d) = \lambda_{t_0}(t_1),\ldots,\lambda_{t_{d-1}}(t_d))$,
then we add $|\mathsf{sig}^{r-d}(\mathfrak{G},u_1,\ldots,u_d)|$-many children to $t$.
Also, for every $\alpha\in \mathsf{sig}^{r-d}(\mathfrak{G},u_1,\ldots,u_d)$, we consider a vertex $v_\alpha\in V(G)$ such that $\mathsf{sig}^{r-d-1}(\mathfrak{G},u_1,\ldots,u_d,v) = \alpha$ and we define $\lambda_t$ to be a bijection from $\textsf{children}_T (t)$
to $\{v_\alpha\mid \alpha\in \mathsf{sig}^{r-d}(\mathfrak{G},u_1,\ldots,u_d)\}$.
We also define the function $\lambda: V(T)\setminus\{t_0\}\to V(G)$ such that for every $t\in V(T)\setminus\{t_0\}$ $\lambda(t) = \lambda_{t'}(t)$, where $t'$ is the parent of $t$.
We call $\lambda$ an \emph{assignment of $G$ to $(T,t_0)$}.

We conclude this subsection by formulating model-checking of sentences in $\FOL[\tau+\DP]$
in terms of assignments and $\varphi$-spanning trees.

\begin{observation}\label{obs_treeformula}
Let $r\in\mathbb{N}$.
Let $\tau$ be a colored rooted graph vocabulary.
For every $\tau$-structure $\mathfrak{G}$ and every sentence
$\varphi = Q_1 {\sf x}_1\ \ldots Q_r {\sf x}_r \ \psi({\sf x}_1,\ldots, {\sf x}_r)$ where $\psi({\sf x}_1,\ldots, {\sf x}_r)\in \FOL[\tau+\DP]$,
we have that $\mathfrak{G}\models \varphi$ if and only if there is
\begin{itemize}
\item[-] an assignement $\lambda$ of $G$ to rooted tree $(T,t_0)$ of height $r$ and
\item[-] a $\varphi$-spanning subtree $(T',t_0)$ of $(T,t_0)$
\end{itemize}
such that for every $(t_0,t_1,\ldots, t_r)\in \textsf{Paths}(T')$,
it holds that
$(\mathfrak{G},\lambda(t_1),\ldots,\lambda(t_r))\models \psi({\sf x}_1,\ldots, {\sf x}_r)$,
where ${\sf x}_1,\ldots, {\sf x}_r$ are interpreted as $\lambda(t_1),\ldots,\lambda(t_r)$.
\end{observation}
We say that a $\varphi$-spanning tree  $(T',t_0)$ as above
\emph{certifies} that $\mathfrak{G}\models \varphi$.

\subsection{Patterns of boundaried colored graphs}\label{subsec_patt}
In this subsection,
we define the notion of a \emph{pattern} of a colored rooted graph.
The pattern aims to encode all information of this colored graph
that concern the boundary vertices:
1) if some elements of the ``boundary'' tuple are the same,
2) which boundary vertices are root-vertices,
3) what are the colors of boundary vertices,
4) what is the graph induced by the boundary vertices,
and
5) which sets of pairs of boundary vertices can be connected with internally vertex-disjoint paths of lenght at least two.
In other words,
the pattern is defined in a way that, having it in hand, we can check which quantifier-free formulas of $\FOL[(\tau+\DP)\cup{\bf c}]$
are satisfied from $(\mathfrak{G},{\bf a})$ if we interpret their free variables as the boundary vertices.

\paragraph{Apex-tuples of structures.}
Let $\tau$ be a vocabulary, let $\mathfrak{A}$ be a $\tau$-structure, and let $l\in \mathbb{N}.$
A tuple ${\bf a}=(a_{1},\ldots, a_{l})$ where each $a_{i}$ is either an element of $V(\mathfrak{A})$ or {$\mathspace$},
is called a \emph{apex-tuple} of $\mathfrak{A}$ of size $l.$
We use $V({\bf a})$ for the set containing the {non-$\mathspace$}
elements in ${\bf a}.$
Also, if $S\subseteq V(\mathfrak{A}),$ we define ${\bf  a}\cap S=(a_{1}',\ldots,a_{l}')$
so that if $a_{i}\in S,$ then $a_i'=a_i,$  and otherwise  $a_i'=\mathspace.$
We also define ${\bf  a}\setminus S={\bf  a}\cap  (V(\mathfrak{A})\setminus S).$
For every apex-tuple ${\bf a}$, we always assume that all non-$\mathspace$ elements in ${\bf a}$ are distinct.

Intuitively, an apex-tuple of a graph is a tuple consisting of vertices and empty entries and every choice of vertices (and empty entries) can be seen as an apex-tuple of appropriate size. All definitions and results in this subsection are stated for general apex-tuples, although in our proofs, we will consider apex-tuples of a particular type, i.e., arbitrary orderings of the apex sets $A$ given by the algorithmic version of the Flat Wall Theorem (\autoref{prop_flatwallbdtw}) presented in~\autoref{sec_flatwalls}.

\paragraph{Boundaried colored graphs.}
Let $t,h,l\in\mathbb{N}$.
A \emph{$t$-boundaried $(h,l)$-colored graph} is a tuple ${\bf G} =(G,X_1,\ldots, X_h,{\bf a},v_1,\ldots,v_t)$ where $(G,X_1,\ldots, X_h)$ is a colored graph, ${\bf a}$ is an apex-tuple of $G$ of size $l$,
and $v_1,\ldots,v_t\in V(G)\cup\{\mathspace\}$.
Intuitively, we use $\mathspace$ to encode the absence of a vertex of $G$ of a certain index.
The vertices $v_i, i\in[t]$ that belong to $V(G)$ are called \emph{boundary vertices}.

\begin{figure}[ht]
\centering
\scalebox{0.9}{
\begin{tikzpicture}
\node[simple,minimum size=5pt,fill=red, line width= 1pt, label={left:$v_5$}] (v5) at (0,0) {};
\node[simple,minimum size=5pt,fill=blue, line width= 1pt, label={left:$v_4$}] (v4) at (0,1) {};
\node[simple,minimum size=5pt,fill=red, line width= 1pt, label={left:$v_2$}] (v2) at (0,2) {};
\node[simple,minimum size=5pt,fill=green, line width= 1pt, label={left:$v_1$}] (v1) at (0,3) {};
\node[simple, fill=red] (u1) at (1,0.5) {};
\node[simple, fill=blue] (u2) at (1.4,2) {};
\node[simple, fill=red] (u3) at (1.5,3.2) {};
\node[simple, fill=blue] (u4) at (1.7,1) {};
\node[simple, fill=red] (u5) at (2.3,2.3) {};
\node[simple, fill=red, label={above:$a_1$}] (u6) at (2.8,3.5) {};
\node[simple, fill=green] (u7) at (3,-0.5) {};
\node[simple, fill=red] (u8) at (3.4,1) {};
\node[simple, fill=blue] (u9) at (3.6,2.5) {};
\node[simple, fill=green] (u10) at (4,1.5) {};
\node[simple, fill=blue, label={above:$a_3$}] (u11) at (5,3) {};
\node[simple, fill=blue] (u12) at (6,0.5) {};
\node[simple, fill=green] (u13) at (7,1.5) {};

\draw[-] (v1) -- (u2) (v1) -- (u3);
\draw[-] (v2) -- (u3) (v2) -- (u4);
\draw[-] (v4) -- (u2) (v4) -- (u4) (v4) -- (u1);
\draw[-] (v5) -- (u3) (v5) -- (u1);
\draw[-] (u3) -- (u5) (u3) -- (u6) (u6) -- (u11) (u6) -- (u9) (u9)-- (u5) (u9) -- (u10) (u6) -- (u8) (u2) -- (u5) (u5) -- (u10) (u1) -- (u7) (u7) -- (u8) (u8) -- (u12) (u12) -- (u13) (u8) -- (u9) (u11) -- (u10) (u11) -- (u9) (u5) -- (u8) (u1) -- (u4) (u4) -- (u7) (u4) -- (u9) (u7) -- (u12);
\end{tikzpicture}}
\caption{An example of a $5$-boundaried $(3,3)$-colored graph $(G,X_1,X_2,X_3, a_1,\mathspace, a_3, v_1,v_2,\mathspace, v_4,v_5)$, where $X_1$ is the set of vertices depicted in red, $X_2$ is the set of vertices depicted in blue, and $X_3$ is the set of vertices depicted in green.}
\label{figure_boundariedgraph}
\end{figure}
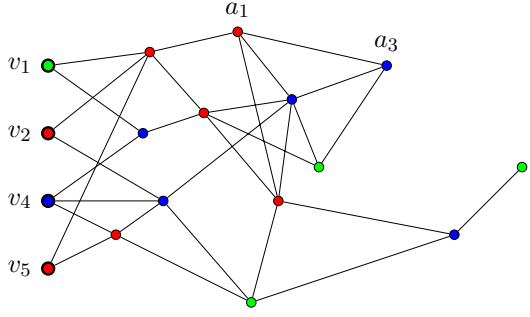

Given two $t$-boundaried $(h,l)$-colored graphs ${\bf G}_1 = (G_1,X_1,\ldots,X_h,{\bf a},v_1,\ldots,v_t)$ and ${\bf G}_2 = (G_2,X_1',\ldots, X_h',{\bf a}',u_1,\ldots, u_t)$, we say that ${\bf G}_1$ and ${\bf G}_2$ are \emph{isomorphic} if
$G_1$ is isomorphic to $G_2$ via a bijection $\eta: V(G_1)\cup\{\mathspace\}\to V(G_2)\cup\{\mathspace\}$ such that $\eta(\mathspace) = \mathspace$, for every $i\in[t]$, $\eta(v_i) = u_i$,  for every $i\in[l]$, $\eta(a_i) =\eta(a_i')$, and  for every $i\in[h]$, $\eta(X_i) = X_i'$.
We denote by $\mathcal{B}^{(t,h,l)}$
the set of all (pairwise non-isomorphic) $t$-boundaried $(h,l)$-colored graphs and we set $\mathcal{B}^{(h,l)} = \bigcup_{t\in\mathbb{N}}\mathcal{B}^{(t,h,l)}$.

\paragraph{Patterns.}
Let $(G,X_1,\ldots, X_h,{\bf a},v_1,\ldots, v_r)\in \mathcal{B}^{(r,h,l)}$.
Given a set $I\subseteq [r]$, we denote $${\sf Ind}_G(I)=(I,\{\{a,b\}\in I\times I\mid \{v_a,v_b\}\in E(G)\}).$$

We set $I=\{i\in[r]\mid v_i\neq\mathspace\}$.
We define the \emph{pattern} of $(G,X_1,\ldots, X_h,{\bf a},v_1,\ldots, v_r)$, denoted by ${\sf pattern}(G,X_1,\ldots, X_h,{\bf a},v_1,\ldots, v_r)$, to be the quintuple $(V,\kappa,\delta,{H}^e,\mathcal{H}^P)$, where
\begin{itemize}
\item $V$ is the partition of $I$ into sets such that for every $A\in I$ and every $i,j\in A$, $v_i = v_j$,
\item $\kappa:I\to [l]$ is the partial function mapping each $i\in I$ to the integer $j\in [l]$ such that $v_i =a_j$,
\item $\delta:I\to [h]$ is the partial function mapping each $i\in I$ to the integer $j\in [h]$ such that $v_i\in X_j$,
\item ${H}^{e}={\sf Ind}_G (I)$, and
\item $\mathcal{H}^{P}=\bigcup_{E\subseteq I\times I} \left\{(I,E)\ \middle\vert \begin{array}{c}G \text{ contains vertex-disjoint paths of length at least}\\
\text{two between the vertices $v_{i}, v_{j}$ for all $\{i,j\}\in E$}
\end{array}\right\}.$
\end{itemize}
Intuitively, the pattern of $(G,X_1,\ldots, X_h,v_1,\ldots, v_r)$ encodes a partition $V$ of the set $I$ formed by grouping indices that correspond to the same vertices, a partial function $\kappa$ mapping each index in $I$ to the index of the apex vertex to which the vertex indexed $i$ corresponds, and a partial function $\delta$ mapping every index in $I$ to the index of the color class among $X_1,\ldots, X_h$ in which the corresponding vertex belongs.
Also, it encodes some graphs with $I$ as vertex set.
These are (1) the graph $H^e$ that corresponds to the graph induced by the vertices $v_i, i\in I$
and (2) the collection $\mathcal{H}^{P}$ that has all graphs with vertex set $I$ whose edge set corresponds to the existence of internally vertex-disjoint paths between the respective vertices.
We set
\[\mathcal{G}_{\sf pat}^{(r,h,l)}=\{{\sf pattern}(G,X_1,\ldots, X_h,{\bf a},v_1,\ldots, v_r)\mid (G,X_1,\ldots, X_h,{\bf a},v_1,\ldots, v_r)\in \mathcal{B}^{(r,h,l)}\}.\]
Also, by the definition of a pattern, we observe the following.
\begin{observation}\label{obs_pattbound}
There is a function $\newfun{@desembarazadamente}:\mathbb{N}^3\to\mathbb{N}$ such that, for every $r\in\mathbb{N}$, $|\mathcal{G}_{\sf pat}^{(r,h,l)}|\leq \funref{@desembarazadamente}(r,h,l)$.
\end{observation}

\subsection{Expressing satisfiability of sentences using patterns}
\label{subsec_exprpatt}
As mentioned in the previous subsection, patterns of boundaried colored graphs can be seen as encodings of the set of quantifier-free formulas of $\FOL[(\tau+\DP)\cup{\bf c}]$ that the given boundaried colored graph satisfies.
Following this line, in this subsection, we also encode quantifier-free formulas of $\FOL[(\tau+\DP)\cup{\bf c}]$ in the setting of patterns.
This will allow us to formulate model-checking questions in terms of pattern realization (see~\autoref{obs_leaftransl} and~\autoref{obs_grleafs}).
Under this viewpoint, model-checking can be formulated in purely graph-theoretical terms (through the information encoded in patterns) of boundaried colored graphs, whose boundary is
rescursively obtained, following the assignment of the given graph to a rooted tree.
\smallskip

We start with some additional definitions on formulas.

\paragraph{Atomic formulas and literals.}
Let $\tau = \{{\sf E},{\sf Y}_1,\ldots,{\sf Y}_h\}$ be a colored-graph vocabulary and let a collection ${\bf c}=\{{\sf c}_1,\ldots, {\sf c}_l\}$ of constant symbols.
An \emph{atomic formula} is a formula of the form ${\sf x}_i = {\sf x}_j$,
or ${\sf x}_i = {\sf c}_j$,
or ${\sf x}_i\in {\sf Y}_j$,
or ${\sf E}({\sf x}_i,{\sf x}_j)$ or $\DP_k({\sf s}_1,{\sf t}_1,\ldots,{\sf s}_k,{\sf t}_k)$ for some $k\in[r(r+1)/2]$, where ${\sf x}_i,{\sf x}_j, {\sf s}_1,\ldots, {\sf s}_k, {\sf t}_1, \ldots, {\sf t}_k$ are first-order variables.
A \emph{literal} is an atomic formula or the negation of an atomic formula.

\paragraph{Disjunctive normal form and full clauses.}
Let $\psi({\sf x}_1,\ldots, {\sf x}_r)$ be a quantifier-free $\FOL[(\tau+\DP)\cup{\bf c}]$-formula.
We say that $\psi({\sf x}_1,\ldots, {\sf x}_r)$ is in \emph{disjunctive normal form}
if there is some $k\in\mathbb{N}$ such that
$\psi = c_1\vee \cdots \vee c_k$,
where for each $i\in[k]$, $c_i$ is a conjunction of literals.
We call each $c_i$ a \emph{clause} of $\psi$.
We say $\psi$ is in \emph{full disjunctive normal form} if
every $\psi$ is in disjunctive normal form and, additionally, for every clause $c$ of $\psi$ it holds that
\begin{itemize}
\item for every distinct  $i,j\in[r]$, either ${\sf x}_i = {\sf x}_j$ or its negation are literals of $c$,
\item for every $i\in[r]$ and every $j\in[l]$, either ${\sf x}_i = {\sf c}_j$ or its negation are literals of $c$,
\item for every $i\in[r]$ and every $j\in[h]$, either ${\sf x}_i\in {\sf Y}_j$ or its negation are literals of $c$, 
\item for every distinct $i,j\in[r]$, either ${\sf E}({\sf x}_i,{\sf x}_j)$ or its negation are literals of $c$, and
\item for every $\ell\in[r(r+1)/2]$ and every $i_1,i_2,\ldots,i_\ell,j_1,j_2,\ldots,j_\ell\in [r]$, either the atomic formula $\DP({\sf x}_{i_1},{\sf x}_{j_1}, \ldots, {\sf x}_{i_\ell},{\sf x}_{j_\ell})$ or its negation are literals of $c$.
\end{itemize}
We say that a quantifier-free  $\FOL[(\tau+\DP)\cup{\bf c}]$-formula $\psi'$ \emph{extends} $\psi$ if every clause of $\psi$ is a sub-formula of a clause of $\psi'$.
We now prove that, for every quantifier-free $\FOL[(\tau+\DP)\cup{\bf c}]$-formula $\psi$, we can
construct an equivalent quantifier-free formula $\psi'$ of the same vocabulary
that is in full disjunctive normal form and extends $\psi$.
 
\begin{lemma}\label{lem_fulldisjunctive}
For every quantifier-free $\FOL[(\tau+\DP)\cup{\bf c}]$-formula $\psi$, there is a quantifier-free  $\FOL[(\tau+\DP)\cup{\bf c}]$-formula $\psi'$ that is in full disjunctive normal form, extends $\psi$, and is equivalent to $\psi$.
\end{lemma}

\begin{proof}
Given a quantifier-free  $\FOL[(\tau+\DP)\cup{\bf c}]$-formula $\psi({\sf x}_1,\ldots, {\sf x}_r)$, we construct the formula $\psi'$ as follows.
For every clause $c$ of $\psi$,
if there are $i,j\in[r]$ such that
neither ${\sf x}_i={\sf x}_j$ nor ${\sf x}_i\neq {\sf x}_j$ appear in $c$, then
we replace $c$ by $$\big(c\wedge {\sf x}_i={\sf x}_j\big)\vee \big(c\wedge {\sf x}_i\neq {\sf x}_j\big).$$
By recursively applying this procedure for  ${\sf x}_i = {\sf c}_j$ for every $i\in[r]$ and every $j\in[l]$, for ${\sf x}_i\in {\sf Y}_j$ for every $i\in[r]$ and every $j\in[h]$, for ${\sf E}({\sf x}_i,{\sf x}_j)$ for every distinct $i,j\in[r]$, and for $\DP_\ell$, for every $\ell\in[r(r+1)/2]$ and every $i_1,i_2,\ldots,i_\ell,j_1,j_2,\ldots,j_\ell\in [r]$, we obtain $\psi'$.
\end{proof}

In the rest of this paper we assume that every quantifier-free formula
$\psi({\sf x}_1,\ldots, {\sf x}_r)$ in
$\FOL[(\tau+\DP)\cup{\bf c}]$ is in disjunctive normal form.

\paragraph{Patterns of quantifier-free formulas.}
We now define a notion of pattern of quantifier-free formulas.
Let $\psi({\sf x}_1,\ldots, {\sf x}_r)$ be a quantifier-free $\FOL[(\tau+\DP)\cup{\bf c}]$-formula in full disjunctive normal form.
We assume that $\psi$ is satisfiable, i.e., $\Mod(\psi)\neq\emptyset$.
For every clause $c$ of $\psi$, we define the \emph{pattern} of $c$, denoted by $H_c$, as the quintuple
$(V_c, \kappa_c,\delta_c,H^{e}_c, \mathcal{H}^{P}_c)$, where
\begin{itemize}
\item $V_c$ is a partition of $[r]$ into sets such that for every $A\in V_H$ and every $i,j\in A$ either ${\sf x}_i = {\sf x}_j $ appears as a literal of $c$ or  $i=j$,
\item $\kappa_c = \{(i,j)\in[r]\times[l]\mid \text{the atomic formula ${\sf x}_i = {\sf c}_j$ appears as a literal of $c$}\}$,
\item $\delta_c = \{(i,j)\in[r]\times[h]\mid \text{the atomic formula ${\sf x}_i\in {\sf Y}_j$ appears as a literal of $c$}\}$,
\item ${H}^{e}_c =([r],\{\{i,j\}\in [r]\times [r]\mid \text{the atomic formula } {\sf E}({\sf x}_i, {\sf x}_j) \text{ appears as a literal of $c$}\})$, and
\item $\mathcal{H}^{P}_c =\bigcup_{E\subseteq [r]\times[r]}
 \left\{([r],E)\ \middle\vert \begin{array}{l} \text{if }E = \{(i_1,j_1),\ldots, (i_\ell, j_\ell)\} \text{ then the atomic formula}\\
  \DP ({\sf x}_{i_1},{\sf x}_{j_1}, \ldots,{\sf x}_{i_\ell},{\sf x}_{j_\ell}) \text{ appears as a literal of $c$}
  \end{array}\right\}$.
\end{itemize}

Notice that the fact that $V_H$ is a partition of $[r]$ follows from the fact that $\psi$ is satisfiable.
Intuitively, satisfiability of $\psi$ is asked so as to rule out the cases where, for example, all three atomic formulas ${\sf x}_1={\sf x}_2$, ${\sf x}_2 = {\sf x}_3$, and $\neg ({\sf x}_2 = {\sf x}_3)$ appear as literals in the same clause.
Also, satisfiability of $\psi$ and the fact that it is in full disjunctive normal form implies that for every $\{i,j\}\in I\times I\setminus E(H^{e}_c)$, $\neg {\sf E}({\sf x}_i, {\sf x}_j)$ appears as a literal in $c$.
If $\psi$ is not satisfiable, then we set $H_c$ to be $\emptyset$.

Let $\psi$ be a quantifier-free  $\FOL[(\tau+\DP)\cup{\bf c}]$-formula.
We define $${\sf ext}(\psi)=\{c\mid \text{ $c$ is a clause of a formula $\psi'$ in full disjunctive normal form that extends $\psi$}\}.$$
Following~\autoref{lem_fulldisjunctive}, we get the following observation.

\begin{observation}\label{obs_nonemptyext}
For every quantifier-free  $\FOL[(\tau+\DP)\cup{\bf c}]$-formula $\psi$, the set ${\sf ext}(\psi)$ is non-empty.
\end{observation}

We set $\mathcal{H}_\psi$ to be the collection of patterns of all $c\in  {\sf ext}(\psi)$, i.e., $$\mathcal{H}_\psi = \{H_c\mid c\in {\sf ext}(\psi)\},$$ and we call it the \emph{set of patterns of $\psi$}.
Having defined the set  $\mathcal{H}_\psi$ of patterns of $\psi$,
we now define when a boundaried colored graph {\sl realizes} an element of  $\mathcal{H}_\psi$.

\paragraph{Realizing a pattern.}
Let $\psi({\sf x}_1,\ldots, {\sf x}_r)$ be a quantifier-free  $\FOL[(\tau+\DP)\cup{\bf c}]$-formula, let $c\in{\sf ext}(\psi)$, and let $H_c$ be the pattern of $c$.
Given a colored rooted graph $\mathfrak{G}$ and vertices $v_1,\ldots, v_r\in V(\mathfrak{G})$, we say that $(\mathfrak{G},v_1,\ldots, v_r)$ \emph{realizes} $H_c$
if  ${\sf pattern}(\mathfrak{G},v_1,\ldots, v_r)=H_c$.
Note that, equivalently, $(\mathfrak{G},v_1,\ldots, v_r)$ realizes $H_c$
if $(\mathfrak{G},v_1,\ldots, v_r)\models c({\sf x}_1,\ldots, {\sf x}_r)$ (where ${\sf x}_i$ is interpreted as $v_i$, for every $i\in[r]$).
Keep in mind that every $(\mathfrak{G},v_1,\ldots, v_r)$ realizes at most one  $H_c\in \mathcal{H}_\psi$.

Due to~\autoref{obs_nonemptyext} and the definition of realization of an element of  $\mathcal{H}_\psi$, we obtain the following result.

\begin{observation}\label{obs_leaftransl}
Let $r\in\mathbb{N}$ and let $\psi({\sf x}_1,\ldots,{\sf x}_r)$ be a quantifier-free  $\FOL[(\tau+\DP)\cup{\bf c}]$-formula.
Then $(\mathfrak{G},v_1,\ldots, v_r)\models \psi({\sf x}_1,\ldots,{\sf x}_r)$ if and only if $(\mathfrak{G},v_1,\ldots,v_r)$ realizes an element of $\mathcal{H}_\psi$.
\end{observation}

To conclude this subsection,
we explain how to revisit the approach presented in~\autoref{subsec_treeembed}, and,
in particular, how to restate~\autoref{obs_treeformula} using patterns.

\paragraph{Pattern-coloring.}
Let $r,h,l\in\mathbb{N}$. 
Let $\mathfrak{G}$ be an $h$-colored $l$-rooted graph and let $\lambda$ be an assignment of $G$ to a rooted tree $(T,t_0)$ of height $r$.
We define the function
${\sf pc}_\lambda: L(T)\to\mathcal{G}_{\sf pat}^{(r,h,l)}$ such that for every $t\in L(T)$, if $(t_0,t_1,\ldots, t_r = t)\in \textsf{Paths}(T)$, then
\begin{equation*}
{\sf pc}_\lambda(t) = {\sf pattern}(\mathfrak{G},\lambda(t_1),\ldots, \lambda(t)).
\end{equation*}
We call ${\sf pc}_\lambda$ the \emph{pattern-coloring of $L(T)$ with respect to $\lambda$}.
\medskip

Using the definition of the pattern-coloring and~\autoref{obs_leaftransl},
we can now restate~\autoref{obs_treeformula} as follows.

\begin{observation}\label{obs_grleafs}
Let $r,l\in\mathbb{N}$.
Let $\tau$ be a colored rooted graph vocabulary.
For every $\tau$-structure $\mathfrak{G}$
and every sentence $\varphi =   Q_1 {\sf x}_1\ \ldots Q_r {\sf x}_r \ \psi({\sf x}_1,\ldots, {\sf x}_r)$ where $\psi({\sf x}_1,\ldots, {\sf x}_r)\in \FOL[\tau+\DP]$,
we have that $\mathfrak{G}\models \varphi$ if and only if there is
\begin{itemize}
\item[-] an assignment $\lambda$ of $G$ to a rooted tree $(T,t_0)$ of height $r$ and
\item[-] a $\varphi$-spanning subtree $(T',t_0)$ of $(T,t_0)$
\end{itemize}
such that for every $t\in L(T')$, ${\sf pc}_\lambda(t)\in \mathcal{H}_{\psi}$.
\end{observation}

\subsection{Graphs with the same patterns satisfy the same sentences}\label{subsec_equiv}
In this subsection, we aim to prove that if two colored graphs give the same signatures, then these two colored graphs
satisfy the same (annotated) sentences.

\begin{lemma}\label{lemma_reducing}
Let $\tau$ be a colored rooted graph vocabulary and let $\mathfrak{G},\mathfrak{G}'$ be two $\tau$-structures.
For every $r\in\mathbb{N}$, if $\mathsf{sig}^r(\mathfrak{G})=\mathsf{sig}^r(\mathfrak{G}')$, then for every sentence
$\varphi\in\FOL[\tau+\DP]$ of quantifier rank at most $r$,
it holds that
$$\mathfrak{G}\models \varphi\iff \mathfrak{G}'\models \varphi.$$
\end{lemma}

\begin{proof}
Let $\varphi =   Q_1 {\sf x}_1\ \ldots Q_r {\sf x}_r \ \psi({\sf x}_1,\ldots, {\sf x}_r)$,
where $\psi({\sf x}_1,\ldots, {\sf x}_r)$ is a quantifier-free formula in $\FOL[\tau+\DP]$.
Let $\lambda$ be an assignment of $\mathfrak{G}$ to some rooted tree $(T,t_0)$ of height $r$ and let $\lambda'$ be an assignment of $\mathfrak{G}'$ to some rooted tree $(\tilde{T},\tilde{t}_0)$ of height $r$.
Due to~\autoref{obs_grleafs},
to prove that
$\mathfrak{G}\models \varphi\iff \mathfrak{G}'\models \varphi$,
it suffices to show that
\begin{eqnarray}
& \text{there is a $\varphi$-spanning subtree $(T',t_0)$ of $(T,t_0)$ certifying that $\mathfrak{G}\models \varphi$} \nonumber\\
& \iff \label{@romantically}\\ 
& \text{there is a $\varphi$-spanning subtree $(\tilde{T}',\tilde{t}_0)$ of $(\tilde{T},\tilde{t}_0)$ certifying that $\mathfrak{G}'\models \varphi$}. \nonumber
\end{eqnarray}

For every $i\in[0,r]$, we denote by $\varphi^i({\sf x}_1,\ldots, {\sf x}_i)$
the formula $Q_{i+1}{\sf x}_{i+1}\ldots Q_r {\sf x}_r \psi({\sf x}_1,\ldots, {\sf x}_r)$.
We first prove an analogue of~\eqref{@romantically} for the quantifier-free formula $\psi({\sf x}_1,\ldots, {\sf x}_r)$.
Recall that since $\psi({\sf x}_1,\ldots, {\sf x}_r)$ is quantifier-free,
given a $t\in L(T)$, $T_t$ has a unique $\psi({\sf x}_1,\ldots, {\sf x}_r)$-spanning subtree, that is $T_t$ itself.
Nevertheless, we formulate the next statement in terms of $\psi({\sf x}_1,\ldots, {\sf x}_r)$-spanning subtrees in order to
use it as the base case of a recursive argument built in the course of the proof of~\eqref{@romantically}.
\begin{claim}\label{claim_base}
For every $(t_0,t_1,\ldots, t_{r})\in\textsf{Paths}(T)$ and every $(\tilde{t}_0,\tilde{t}_1,\ldots, \tilde{t}_{r})\in\textsf{Paths}(\tilde{T})$,
such that for every $i\in[r]$, $\mathsf{sig}^{r-i}(\mathfrak{G},\lambda(t_1),\ldots,\lambda(t_i))=\mathsf{sig}^{r-i}(\mathfrak{G}',\lambda'(\tilde{t}_1),\ldots, \lambda'(\tilde{t}_{i}))$, it holds that 
\begin{eqnarray*}
& \text{there is a $\psi({\sf x}_1,\ldots, {\sf x}_r)$-spanning subtree $T_{t_r}'$ of $T_{t_r}$ certifying that}\\
& \text{$(\mathfrak{G},\lambda(t_1),\ldots, \lambda(t_{r}))\models \psi({\sf x}_1,\ldots, {\sf x}_r)$}\\
& \iff\\
& \text{there is a $\varphi^r({\sf x}_1,\ldots, {\sf x}_r)$-spanning subtree $\tilde{T}_{\tilde{t}_r}'$ of $\tilde{T}_{\tilde{t}_r}$ certifying that}\\
& \text{$\big(\mathfrak{G}',\lambda'(\tilde{t}_1),\ldots, \lambda'(\tilde{t}_r)\big)\models \psi({\sf x}_1,\ldots, {\sf x}_r)$.}
\end{eqnarray*}
\end{claim}
\medskip

\noindent\emph{Proof of~\autoref{claim_base}}:
Since $\psi({\sf x}_1,\ldots, {\sf x}_r)$ is a quantifier-free formula,
$T_{t_r}$ has a unique $\psi({\sf x}_1,\ldots, {\sf x}_r)$-spanning subtree,
that is $T_{t_r}$, and the same holds for $\tilde{T}_{\tilde{t}_r}$.
By assumption, $\mathsf{sig}^{0}(\mathfrak{G},\lambda(t_1),\ldots,\lambda(t_r))=\mathsf{sig}^{0}(\mathfrak{G}',\lambda'(\tilde{t}_1),\ldots, \lambda'(\tilde{t}_{r}))$. Therefore, ${\sf pc}_\lambda(t_r) = {\sf pc}_{\lambda'}(\tilde{t}_r)$, which in turn implies that ${\sf pc}_\lambda(t_r)\in \mathcal{H}_{\psi}\iff {\sf pc}_{\lambda'}(\tilde{t}_r)\in \mathcal{H}_{\psi}.$
\hfill $\diamond$
\medskip\medskip

We now prove the following. The case where $i=0$ proves~\eqref{@romantically}.
\begin{claim}\label{claim_step}
For every $i\in[0,r-1]$, for every $(t_0,t_1,\ldots, t_{i})\in\textsf{Paths}(T)$ and every $(\tilde{t}_0,\tilde{t}_1,\ldots, \tilde{t}_{i})\in\textsf{Paths}(\tilde{T})$, such that for every $j\in[i]$, $\mathsf{sig}^{r-j}(\mathfrak{G},\lambda(t_1),\ldots,\lambda(t_j))=\mathsf{sig}^{r-j}(\mathfrak{G}',\lambda'(\tilde{t}_1),\ldots, \lambda'(\tilde{t}_{j}))$, it holds that 
\begin{eqnarray*}
& \text{there is a $\varphi^i({\sf x}_1,\ldots, {\sf x}_i)$-spanning subtree $T_{t_i}'$ of $T_{t_i}$ certifying that}\\
& \text{$(\mathfrak{G},\lambda(t_1),\ldots, \lambda(t_{i}))\models \varphi^i({\sf x}_1,\ldots, {\sf x}_i)$}\\
& \iff\\
& \text{there is a $\varphi^i({\sf x}_1,\ldots, {\sf x}_i)$-spanning subtree $\tilde{T}_{\tilde{t}_i}'$ of $\tilde{T}_{\tilde{t}_i}$ certifying that}\\
& \text{$\big(\mathfrak{G}',\lambda'(\tilde{t}_1),\ldots, \lambda'(\tilde{t}_i)\big)\models \varphi^i({\sf x}_1,\ldots, {\sf x}_i)$.}
\end{eqnarray*}
\end{claim}

\noindent\emph{Proof of~\autoref{claim_step}}:
In~\autoref{claim_base}, we already proved the statement for $i=r$.
Let $i\in[0,r-1]$.
We assume that the statement holds for $i+1$.
We will prove that it holds for $i$.
We distinguish two cases, depending on whether $Q_{i+1} =\exists$ or $Q_{i+1} =\forall$.
\medskip\medskip\medskip

\noindent\fbox{\bf Case 1}: $Q_{i+1}=\exists$.
\medskip

\noindent$(\Leftarrow)$
Suppose that there is a $\varphi^i({\sf x}_1,\ldots, {\sf x}_i)$-spanning subtree
$\tilde{T}_{\tilde{t}_i}'$ of $\tilde{T}_{\tilde{t}_i}$ certifying that
$$\big(\mathfrak{G}',\lambda'(\tilde{t}_1),\ldots, \lambda'(\tilde{t}_i)\big)\models \varphi^i({\sf x}_1,\ldots, {\sf x}_i)$$
(in the case where $i=0$, we suppose that $\tilde{T}_{\tilde{t}_i}' = \tilde{T}'$ certifies that $\mathfrak{G}'\models \varphi$).

Since $Q_{i+1}=\exists$,
there exists a node $z\in\textsf{children}_{\tilde{T}}(\tilde{t}_i)$ that belongs to $V(\tilde{T}_{\tilde{t}_i}')$.
Also, observe that $\tilde{T}_{z}'$ is a $\varphi^{i+1}({\sf x}_1,\ldots, {\sf x}_{i+1})$-spanning subtree of $\tilde{T}_{z}$ certifying that
$$\big(\mathfrak{G}',\lambda'(\tilde{t}_1),\ldots, \lambda'(\tilde{t}_i), \lambda'(z)\big)\models \varphi^{i+1}({\sf x}_1,\ldots, {\sf x}_{i+1}).$$
Notice that, since $\mathsf{sig}^{r-i}(\mathfrak{G},\lambda(t_1),\ldots,\lambda(t_i))=\mathsf{sig}^{r-i}(\mathfrak{G}',\lambda'(\tilde{t}_1),\ldots, \lambda'(\tilde{t}_{i}))$,
there exists a $t_{i+1}\in\textsf{children}_T(t_i)$ such that
$\mathsf{sig}^{r-i-1}(\mathfrak{G},\lambda(t_1),\ldots,\lambda(t_i),\lambda(t_{i+1}))
=
\mathsf{sig}^{r-i-1}(\mathfrak{G}',\lambda'(\tilde{t}_1),\ldots, \lambda'(\tilde{t}_{i}),\lambda'(z))$.
Following our recursive assumption,
there is a $\varphi^{i+1}({\sf x}_1,\ldots, {\sf x}_{i+1})$-spanning subtree $T_{t_{i+1}}'$ of $T_{t_{i+1}}$ certifying that
$$\big(\mathfrak{G},\lambda(t_1),\ldots,\lambda(t_{i}), \lambda(t_{i+1})\big)\models \varphi^{i+1}({\sf x}_1,\ldots, {\sf x}_{i+1}).$$
Then, the graph $T_{t_i}' :=T[\{t_i\}\cup V(T_{t_{i+1}}')]$
is a $\varphi^i({\sf x}_1,\ldots, {\sf x}_i)$-spanning subtree of $T_{t_i}$ certifying that
$$\big(\mathfrak{G},\lambda(t_1),\ldots,\lambda(t_{i})\big)\models \varphi^i({\sf x}_1,\ldots, {\sf x}_i).$$
In the case where $i=0$, we have that $T_{t_i}' = T'$ certifies that $\mathfrak{G}\models \varphi$.
\medskip

\noindent$(\Rightarrow)$
Suppose that there is a $\varphi^i({\sf x}_1,\ldots, {\sf x}_i)$-spanning subtree $T_{t_i}'$ of $T_{t_i}$ certifying that
$$\big(\mathfrak{G},\lambda(t_1),\ldots,\lambda(t_{i})\big)\models \varphi^i({\sf x}_1,\ldots, {\sf x}_i)$$
(in the case where $i=0$, we suppose that $T_{t_i}' = T'$ certifies that $\mathfrak{G}\models \varphi$).

Since $Q_{i+1}=\exists$,
there exists a node $t_{i+1}\in \textsf{children}_T (t_i)$ that belongs to $V(T_{t_i}')$.
Also, observe that $T_{t_{i+1}} '$ is a $\varphi^{i+1}({\sf x}_1,\ldots, {\sf x}_{i+1})$-spanning subtree of $T_{t_{i+1}}$ certifying that
$$\big(\mathfrak{G},\lambda(t_1),\ldots,\lambda(t_{i}), \lambda(t_{i+1})\big)\models \varphi^{i+1}({\sf x}_1,\ldots, {\sf x}_{i+1}).$$
Since $\mathsf{sig}^{r-i}(\mathfrak{G},\lambda(t_1),\ldots,\lambda(t_i))=\mathsf{sig}^{r-i}(\mathfrak{G}',\lambda'(\tilde{t}_1),\ldots, \lambda'(\tilde{t}_{i}))$,
there is a node $\tilde{t}_{i+1}$ in ${\sf children}_{\tilde{T}}(\tilde{t}_i)$ such that $$\mathsf{sig}^{r-i-1}(\mathfrak{G},\lambda(t_1),\ldots,\lambda(t_i),\lambda(t_{i+1}))
=
\mathsf{sig}^{r-i-1}(\mathfrak{G}',\lambda'(\tilde{t}_1),\ldots, \lambda'(\tilde{t}_{i}),\lambda'(\tilde{t}_{i+1})).$$
Following our recursive assumption, 
there is a $\varphi^{i+1}({\sf x}_1,\ldots, {\sf x}_{i+1})$-spanning subtree
$\tilde{T}_{\tilde{t}_{i+1}}'$ of $\tilde{T}_{\tilde{t}_{i+1}}$ certifying that
$$\big(\mathfrak{G}',\lambda'(\tilde{t}_1),\ldots, \lambda'(\tilde{t}_i), \lambda'(\tilde{t}_{i+1})\big)\models \varphi^{i+1}({\sf x}_1,\ldots, {\sf x}_{i+1}).$$
Now, observe that the graph $\tilde{T}_{\tilde{t}_i}' = \tilde{T}[\{\tilde{t}_i\}\cup V(\tilde{T}_{\tilde{t}_{i+1}}')]$
is a $\varphi^i({\sf x}_1,\ldots, {\sf x}_i)$-spanning subtree of $\tilde{T}_{\tilde{t}_i}$ certifying that
$$\big(\mathfrak{G}',\lambda'(\tilde{t}_1),\ldots, \lambda'(\tilde{t}_i)\big)\models \varphi^i({\sf x}_1,\ldots, {\sf x}_i).$$
In the case where $i=0$, we have that $\tilde{T}_{\tilde{t}_i}' = \tilde{T}'$ certifies that $\mathfrak{G}\models \varphi$.
\medskip\medskip\medskip

\noindent\fbox{\bf Case 2}: $Q_{i+1}=\forall$.
\medskip

\noindent$(\Leftarrow)$
Suppose that there is a $\varphi^i({\sf x}_1,\ldots, {\sf x}_i)$-spanning subtree
$\tilde{T}_{\tilde{t}_i}'$ of $\tilde{T}_{\tilde{t}_i}$ certifying that
$$\big(\mathfrak{G}',\lambda'(\tilde{t}_1),\ldots, \lambda'(\tilde{t}_i)\big)\models \varphi^i({\sf x}_1,\ldots, {\sf x}_i)$$
(in the case where $i=0$, we suppose that $\tilde{T}_{\tilde{t}_i}' = \tilde{T}'$ certifies that $\mathfrak{G}'\models \varphi$).

Since $Q_{i+1}=\forall$,
every $z\in\textsf{children}_{\tilde{T}}(\tilde{t}_i)$ belongs to $V(\tilde{T}_{\tilde{t}_i}')$.
Now observe that,
for every $z\in\textsf{children}_{\tilde{T}}(\tilde{t}_i)$,
 $\tilde{T}_{z}'$ is a $\varphi^{i+1}({\sf x}_1,\ldots, {\sf x}_{i+1})$-spanning subtree of $\tilde{T}_{z}$ certifying that
$$\big(\mathfrak{G}',\lambda'(\tilde{t}_1),\ldots, \lambda'(\tilde{t}_i), \lambda'(z)\big)\models \varphi^{i+1}({\sf x}_1,\ldots, {\sf x}_{i+1}).$$
Also, since $\mathsf{sig}^{r-d}(\mathfrak{G},\lambda({t}_1),\ldots,\lambda({t}_i))= \mathsf{sig}^{r-d}(\mathfrak{G}',\lambda'(\tilde{t}_1),\ldots,\lambda'(\tilde{t}_i))$,
for every $z\in\textsf{children}_{\tilde{T}}(\tilde{t}_i)$
the set
$$\{t_{i+1}\in \textsf{children}_T(t_i)
\mid
\mathsf{sig}^{r-i-1}(\mathfrak{G},\lambda(t_1),\ldots,\lambda(t_i),\lambda(t_{i+1}))
=
\mathsf{sig}^{r-i-1}(\mathfrak{G}',\lambda'(\tilde{t}_1),\ldots, \lambda'(\tilde{t}_{i}),\lambda'(z))
\}$$
is non-empty and, also,
$\bigcup_{z\in \textsf{children}_{\tilde{T}}(\tilde{t}_i)}\{t_{i+1}\in \textsf{children}_T(t_i)\mid \mathsf{sig}^{r-i-1}(\mathfrak{G},\lambda(t_1),\ldots,\lambda(t_i),\lambda(t_{i+1}))
=
\mathsf{sig}^{r-i-1}(\mathfrak{G}',\lambda'(\tilde{t}_1),\ldots, \lambda'(\tilde{t}_{i}),\lambda'(z))\} = \textsf{children}_T(t_i).$
Following our recursive assumption,
for every $t_{i+1}\in \textsf{children}_T(t_i)$,
there is a $\varphi^{i+1}({\sf x}_1,\ldots, {\sf x}_{i+1})$-spanning subtree $T_{t_{i+1}}'$ of $T_{t_{i+1}}$ certifying that
$$\big(\mathfrak{G},\lambda(t_1),\ldots,\lambda(t_{i}), \lambda(t_{i+1})\big)\models \varphi^{i+1}({\sf x}_1,\ldots, {\sf x}_{i+1}).$$
Then, the graph $$T[\{t_i\}\cup \bigcup_{t_{i+1}\in\textsf{children}_T(t_i)}V(T_{t_{i+1}}')]$$
is a $\varphi^i({\sf x}_1,\ldots, {\sf x}_i)$-spanning subtree of $T_{t_i}$ certifying that
$$\big(\mathfrak{G},\lambda(t_1),\ldots,\lambda(t_{i})\big)\models \varphi^i({\sf x}_1,\ldots, {\sf x}_i)$$
In the case where $i=0$, we have that $T_{t_i}' = T'$ certifies that $\mathfrak{G}\models \varphi$.
\medskip

\noindent$(\Rightarrow)$
Suppose that there is a $\varphi^i({\sf x}_1,\ldots, {\sf x}_i)$-spanning subtree $T_{t_i}'$ of $T_{t_i}$ certifying that
$$\big(\mathfrak{G},\lambda(t_1),\ldots,\lambda(t_{i})\big)\models \varphi^i({\sf x}_1,\ldots, {\sf x}_i)$$
(in the case where $i=0$, we suppose that $T_{t_i}' = T'$ certifies that $\mathfrak{G}\models \varphi$).

Since $Q_{i+1}=\forall$,
every $t_{i+1}\in \textsf{children}_T (t_i)$ belongs to $V(T_{t_i}')$.
Now observe that,
for every $t_{i+1}\in \textsf{children}_T (t_i)$,
 $T_{t_{i+1}} '$ is a $\varphi^{i+1}({\sf x}_1,\ldots, {\sf x}_{i+1})$-spanning subtree of $T_{t_{i+1}}$ certifying that
$$\big(\mathfrak{G},\lambda(t_1),\ldots,\lambda(t_{i}), \lambda(t_{i+1})\big)\models \varphi^{i+1}({\sf x}_1,\ldots, {\sf x}_{i+1}).$$
Since $\mathsf{sig}^{r-d}(\mathfrak{G},\lambda({t}_1),\ldots,\lambda({t}_i))= \mathsf{sig}^{r-d}(\mathfrak{G}',\lambda'(\tilde{t}_1),\ldots,\lambda'(\tilde{t}_i))$,
for every $t_{i+1}\in {\sf children}_T (t_i)$ there exists a node $z\in {\sf children}_{\tilde{T}}(\tilde{t}_i)$
such that $$\mathsf{sig}^{r-i-1}(\mathfrak{G},\lambda(t_1),\ldots,\lambda(t_i),\lambda(t_{i+1}))
=
\mathsf{sig}^{r-i-1}(\mathfrak{G}',\lambda'(\tilde{t}_1),\ldots, \lambda'(\tilde{t}_{i}),\lambda'(z))
.$$
Following our recursive assumption,
for every $z\in \textsf{children}_T (t_{i})$
there is a $\varphi^{i+1}({\sf x}_1,\ldots, {\sf x}_{i+1})$-spanning subtree
$\tilde{T}_{\tilde{t}_{i+1}}'$ of $\tilde{T}_{\tilde{t}_{i+1}}$ certifying that
$$\big(\mathfrak{G},\lambda(t_1),\ldots,\lambda(t_i),\lambda(t_{i+1})\big)\models \varphi^{i+1}({\sf x}_1,\ldots, {\sf x}_{i+1}).$$
Also, note that for every $\tilde{t}_{i+1}\in {\sf children}_{\tilde{T}}(\tilde{t}_i)$, there exists a node $t_{i+1}\in {\sf children}_T (t_i)$
such that $\mathsf{sig}^{r-i-1}(\mathfrak{G},\lambda(t_1),\ldots,\lambda(t_i),\lambda(z))
=
\mathsf{sig}^{r-i-1}(\mathfrak{G}',\lambda'(\tilde{t}_1),\ldots, \lambda'(\tilde{t}_{i}),\lambda'(\tilde{t}_{i+1}))
$.
Restating what we mentioned above for
the existence of a  $\varphi^{i+1}({\sf x}_1,\ldots, {\sf x}_{i+1})$-spanning subtree of ${T}_{t_{i+1}}$ for each
$t_{i+1}\in \textsf{children}_T(t_i)$, we get that
for every $z\in \textsf{children}_{\tilde{T}}(\tilde{t}_i)$,
there is a $\varphi^{i+1}({\sf x}_1,\ldots, {\sf x}_{i+1})$-spanning subtree $\tilde{T}_{z}'$ of $\tilde{T}_{z}$ certifying that
$$\big(\mathfrak{G},\lambda'(\tilde{t}_1),\ldots, \lambda'(\tilde{t}_i), \lambda'(z)\big)\models \varphi^{i+1}({\sf x}_1,\ldots, {\sf x}_{i+1}).$$
Now, observe that the graph $$\tilde{T}[\{\tilde{t}_i\}\cup \bigcup_{z\in \textsf{children}_{\tilde{T}}(\tilde{t}_i)} V(\tilde{T}_{z}')]$$ is a $\varphi^i({\sf x}_1,\ldots, {\sf x}_i)$-spanning subtree of $\tilde{T}_{\tilde{t}_i}$ certifying that
$$\big(\mathfrak{G}',\lambda'(\tilde{t}_1),\ldots, \lambda'(\tilde{t}_i)\big)\models \varphi^i({\sf x}_1,\ldots, {\sf x}_i).$$
In the case where $i=0$, we have that $\tilde{T}_{\tilde{t}_i}' = \tilde{T}'$ certifies that $\mathfrak{G}'\models \varphi$.
This concludes the proof of the Claim. \hfill $\diamond$
\bigskip

Following~\autoref{claim_step} for $i=0$, we deduce~\eqref{@romantically}, which completes the proof of the lemma.
\end{proof}

\paragraph{A comment on the use of~\autoref{lemma_reducing}.}
Before concluding this section,
we wish to stress the following.
Suppose that we are given an (annotated) colored graph and one can find an (annotated) colored graph that has the same ``meta-collection'' of patterns (in the sense of equivalence of leaf-labeled trees, as in the statement of~\autoref{lemma_reducing}).
\autoref{lemma_reducing} indicates that
model-checking for the original annotated colored graph can be reduced to model-checking for the second annotated colored graph and this is safe for {\sl every} formula in $\FOL[\tau+\DP]$.
Therefore, if we were able to find a way to compute such an ``equivalent'' graph and if its size (or, the size of its annotated set) was smaller than the original one's, we could reduce the problem of model-checking to smaller instances and therefore report some progress.
However, this is not straightforward and in the next three sections, i.e.,~\autoref{sec_reroutinginannuli},~\autoref{sec_annotated}, and~\autoref{sec_signaturesexchengability}, we explain how to deal with this situation.
From a high-level point of view, our approach considers expressing partial patterns (and consequently, partial satisfaction of formulas in $\FOL[\tau+\DP]$) inside a bounded treewidth part of the given graph. Then, using Courcelle's Theorem, we will be able to compute {\sl representatives} of the vertices inside this bounded treewidth part with respect to equivalence of (partial) patterns.
Using the fact that we know which vertices in this part ``represent''
the variety of patterns, we find some ``irrelevant'' vertices to discard from the annotation, without changing the signature.

\section{Routing linkages through railed annuli}
\label{sec_reroutinginannuli}

In~\autoref{sec_alternative}, we presented how to encode questions expressed in $\FOLDP$ in purely graph-theoretical terms, using patterns.
Recall that a pattern of a boundaried colored graph encodes, by the collection $\mathcal{H}^P$, what disjoint paths can be routed through the boundary vertices.
 This is the only ``non-local'' information encoded in the pattern, as all other information can be determined by inspecting only the boundaried vertices and the adjacencies between them.
Aiming to define a notion of ``partial'' pattern,
we have to deal with the possible ways that boundary vertices can be connected through disjoint paths and this can be seen as a question about the variety of linkages that can be routed through the boundaried vertices.
A crucial tool for handling linkages is the Linkage Combing Lemma (\autoref{prop_combinglemma}) proved in~\cite{GolovachST22comb} (see also~\cite{GolovachST20hitti}).
This result is applied in the presence of a \emph{partially annulus-embedded graph} and an \emph{annulus-embedded railed annulus}, notions that are defined in~\autoref{subsec_disks}.
The definition of linkages and the Linkage Combing Lemma of~\cite{GolovachST22comb} are presented in~\autoref{subsec_combin}.
Finally, in~\autoref{subsec_tm}, we define linkages of boundaried graphs and we describe how to ``encode'' models of boundaried graphs using patterns.

\subsection{Graphs partially embedded on an annulus and railed annuli}\label{subsec_disks}
We say that a pair $(L,R)\in 2^{V(G)}\times 2^{V(G)}$ is a \emph{separation} of $G$
if $L\cup R=V(G)$ and there is no edge in $G$ between a vertex in $L\setminus R$ and a vertex in $R\setminus L.$
We say that two separations  $(X_1,Y_1)$ and $(X_2,Y_2)$ of a graph $G$ are \emph{laminar} if $Y_1\subseteq Y_2$ and $X_2\subseteq X_1$.

\paragraph{Disks and annuli.}
A \emph{cycle} is a set homeomorphic to the set $\{(x,y)\in \mathbb{R}^2 \mid x^2+y^2 = 1\}$.
We define a \emph{closed disk} (resp. \emph{open disk}) to be a set homeomorphic to the set $\{(x,y)\in \mathbb{R}^2\mid x^2+y^2\leq 1\}$
(resp. $\{(x,y)\in \mathbb{R}^2\mid x^2+y^2< 1\}$) and
a \emph{closed annulus} (resp. \emph{open annulus}) to be a set homeomorphic to the set $\{(x,y)\in \mathbb{R}^2 \mid 1\leq x^2+y^2\leq 2\}$ (resp. $\{(x,y)\in \mathbb{R}^2 \mid 1< x^2+y^2< 2\}$).
Given a closed disk or a closed annulus $X$, we use $\bd(X)$ to denote the boundary of $X$ (i.e., the set of points of $X$ for which every neighborhood around them contains some point not in $X$).
Notice that if $X$ is a closed disk then $\bd(X)$ is a cycle, while if $X$ is a closed annulus then $\bd(X)=C_{1}\cup C_{2}$ where $C_{1}, C_{2}$ are the two unique connected components of $\bd(X)$ and $C_1,C_2$ are two disjoint cycles.
We call $C_1$ and $C_2$ \emph{boundaries} of $X$.
We call $C_1$ the \emph{left boundary} of $X$ and $C_2$ the \emph{right boundary} of $X$.
Also given a closed disk (resp. closed annulus) $X$, we use $\inter(X)$ to denote the open disk (resp. open annulus) $X\setminus \bd(X)$. When we embed a graph $G$ in the plane, in a closed disk, or in a closed annulus, we treat G as a set of points.
This permits us to make set operations between graphs and sets of points.

\paragraph{Partially annulus-embedded graphs.}
Let $\Delta$ be a closed annulus.
We say that a graph $G$ is \emph{partially $\Delta$-embedded}, 
if there is some subgraph $K$ of $G$ that is embedded in $\Delta$
such that $\bd(\Delta)$ is the disjoint union of two cycles of $K$ and there are two laminar separations $(X_1,Y_1)$ and $(X_2,Y_2)$ of $G$ such that $X_1\cap Y_2 =V(G)\cap \Delta$.
We also call the graph $K$
\emph{compass}
of the partially $\Delta$-embedded graph $G$ and we always assume that we accompany
a partially $\Delta$-embedded graph $G$ together with an embedding of its compass in $\Delta$ that is the set $G\cap \Delta$.
See~\autoref{figure_partially} for an illustration of a partially annulus-embedded graph.
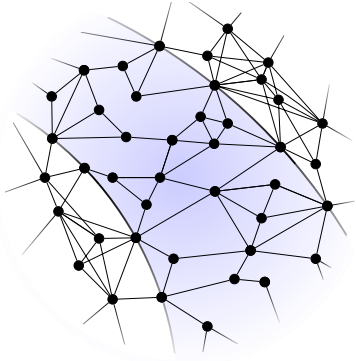
\begin{figure}[ht]
\centering
\scalebox{0.9}{
\begin{tikzpicture}

\clip (0,2.8) circle (2.6cm);

\begin{scope}

    \path (-1.2,2.7)  arc (40:77:5)
    node[terminal,pos=0.1] (P1) {}
    node[terminal,pos=0.3] (P0) {}
    (90:5)--(90:5);
    \draw [thick,black,path fading=north, fading angle =40] (-1.2,2.7)  arc (40:64:5)  (100:4)--(100:4);
    
    \path (-1.2,2.7) arc (40:-15:5)
    node[terminal,pos=0.2] (P2) {}
    node[terminal,pos=0.4] (P3) {}
    node[pos=1] (A) {}
    (-15:5)--(-15:5);
    \draw [thick,black,path fading=south] (-1.2,2.7) arc (40:5:5)    (7:5)--(7:5);
    
    \path (0.2,3)  arc (40:82:6.5)
    node[terminal,pos=0.1] (P6) {}
    node[terminal,pos=0.45] (P4) {}
    (85:7)--(85:7);
    
    \path (0.2,3) arc (40:-5:6.5)
    node[terminal,pos=0.1] (P7) {}
    node[terminal,pos=0.3] (P8) {}
    node[terminal,pos=0.4] (P9) {}
    node[pos=1] (B) {}
    (-5:7)--(-5:7);
    
    \path (1.7,3)  arc (40:85:8.5)
    node[pos=1] (D) {}
    node[terminal,pos=0.4] (P10) {}
    node[terminal,pos=0.25] (P11) {}
    node[terminal,pos=0.05] (P12) {}
    (90:9)--(90:9);
    \draw [thick,black,path fading=west] (1.7,3)  arc (40:69:8.5)   (80:9)--(80:9);
    
    \path (1.7,3) arc  (40:5:8.5)
    node[terminal,pos=0.15] (P13) {}
    node[terminal,pos=0.35] (P14) {}
    node[terminal,pos=0.45] (P15) {}
    node[pos=1] (C) {}
    (5:9)--(5:9);
    \draw [thick,black,path fading=south] (1.7,3) arc  (40:25:8.5)   (25:9)--(25:9);
    
    \foreach \x in {0,...,3} \node[terminal,black,circle] () at (P\x) {};
    \foreach \x in {4,6,7,8,9} \node[terminal,black,circle] () at (P\x) {};
    \foreach \x in {10,...,15} \node[terminal,black,circle] () at (P\x) {};

    \node[] (A1) at ($(A)+(-45:1)$) {$A_{\bar{w}0}$};
    \node[] (A2) at ($(B)+(-45:1)$) {$A_{\bar{w}1}$};

    \node[terminal,black,circle](Q9) at (0.4,.6) {};
    \node[terminal,black,circle](Q8) at (0.8,1.3) {};
    \node[terminal,black,circle](Q7) at (-0.1,1.6) {};
    \node[terminal,black,circle](Q5) at (-.5,2.4) {};
    \node[terminal,black,circle](Q4) at (-0.3,2.8) {};        
    \node[terminal,black,circle](Q3) at (-1,2.8) {};    
    \node[terminal,black,circle](Q2) at (-0.8,3.4) {};
    \node[terminal,black,circle](Q1) at (-1.2,3.8) {};
    
    \node[terminal,black,circle](ex2) at (-1.9,4) {};

    \node[terminal,black,circle](B8) at (2,1.6) {};
    \node[terminal,black,circle](B7) at (1.2,2.2) {};
    \node[terminal,black,circle](B6) at (1.4,2.7) {};
    \node[terminal,black,circle](B5) at (0.5,3.3) {};
    \node[terminal,black,circle](B4) at (0.7,3.6) {};
    \node[terminal,black,circle](B3) at (0.3,3.7) {};    
    \node[terminal,black,circle](B2) at (-.65,4) {};
    \node[terminal,black,circle](B1) at (-.85,4.45) {};

    \node[terminal,black,circle] (C7) at (0.7,5) {};
    \node[terminal,black,circle] (C6) at (0.4,4.6) {};
    \node[terminal,black,circle] (C5) at (1.3,4.5) {};
    \node[terminal,black,circle] (C4) at (1.2,4.25) {};
    \node[terminal,black,circle] (C3) at (1.45,3.95) {};
    \node[terminal,black,circle] (C2) at (2.1,3.6) {};
    \node[terminal,black,circle] (C1) at (2,3) {};

    \node[terminal,black,circle] (A5) at (-2,2.8) {};
    \node[terminal,black,circle] (A4) at (-1.8,2.3) {};
    \node[terminal,black,circle] (A3) at (-1.2,1.9) {};
    \node[terminal,black,circle] (A2) at (-1.5,1.5) {};
    \node[terminal,black,circle] (A1) at (-1,1) {};

\begin{scope}[on background layer]
\fill[white] (0,2.8) circle (2.6cm);
\clip (C)  arc (5:185:8.5);
 \draw[draw=white,fill=white] (0,2.8)  circle (2.7cm);
 \draw[fill=blue!20!white,path fading=fade out](0,2.8)  circle (2.8cm);
     \begin{scope}
         \clip (A)  arc (-15:195:5);
         \draw[draw=white,fill=white](0,2.8)  circle (2.7cm);
     \end{scope}
\clip (0,2.8) circle (2.6cm);

\end{scope}

    \draw[black,path fading= west] (ex2) -- ($(ex2)+(145:0.4)$);

    \draw[black] (P4)--(Q1) (P6) -- (Q4) (P7)--(Q4) (P7)--(P2)  (P8)--(Q7) (P8)--(Q8) (P9)--(Q8);
    \draw[black] (Q3) -- (Q4) (Q4) -- (P6);
    
    \draw[black] (P0) -- (Q2) (Q2)-- (P6);

    \draw[black]  (P0)--(Q1) (P1)--(Q3) (P2)--(Q5)  (P2)--(Q7) (P3)--(Q7) (P3)--(Q8); 
    
    \draw[black] (ex2) -- (P0) (ex2) -- (P4) (P0) -- (P4) (Q1)--(Q2) (Q3)--(Q4)--(Q5)--(Q3) (P3) -- (Q9);
    

    \draw[black] (B1) -- (B2) (B1) -- (P4) (B1) -- (P10) (B2) -- (P10);
    
    \draw[black] (B2) -- (P11);
    
    \draw[black] (P6) -- (B3) (B3) -- (P11);
    \draw[black] (P8) -- (P7) (P7)--(B6)-- (P13);
    \draw[black] (P6) -- (B5) (B5) -- (P12) (B3) --(B4) (B4) -- (B5) (B3) -- (B5) (P11) -- (B4) (P12) -- (B4);
    
    \draw[black] (B5) -- (Q4);
    
    \draw[black] (P7) -- (P12);
    
    \draw[black] (P7) -- (B6) (P7) -- (B7) (P8) -- (B7) (B7)-- (P13) (P13) -- (B6) (P13) -- (P8) (B6) -- (B7);
    
    \draw[black] (P8) -- (B8) (P13) -- (B8);

%

\draw[black] (P11) -- (C4) (C4) -- (C5);
\draw[] (P10) -- (C6) (C6) -- (C5) (C7) -- (C4);
\draw[] (C6) -- (P11) (C6) -- (C7);
\draw[] (P11) -- (C5) (P11) -- (C4) (P11) -- (C3) (P11) -- (C2) (C5) -- (C2) (C5) -- (C3) (C5) -- (C4) (C4) -- (C3) (C3) -- (C2) (C4) to [bend left = 5] (C2) (C1)--(C3) (P12) -- (C3) (P12) -- (C4) (C6) -- (C4) (P11) -- (C7) (C7)--(C5);
\draw (P12) -- (C2) (P12) -- (C1) (P13) -- (C1) (C1) -- (C2);

    \draw[black] (A1)  -- (P3);
    \draw[black] (A1) -- (P2) (A1) -- (A2) (A1) -- (A3) (A1) -- (A4) (A2) -- (P2) (A2) -- (A3) (A2) -- (A4) (A3) -- (A4) (A4) -- (P2);
     \draw[black]      (A3) -- (P2)      (A3) -- (A4);
    \draw[black] (A4) -- (A5)  -- (P1) -- (A4);
    \draw[black] (A5) -- (P0) (A5)-- (A3);

\draw[black, path fading=east] (Q9) -- ($(Q9)+(-30:0.6)$);
\draw[black, path fading=south] (Q9) -- ($(Q9)+(-100:0.5)$);
\draw[black, path fading=north] (P4) -- ($(P4)+(160:0.6)$);
\draw[black, path fading=south] (B8) -- ($(B8)+(-70:0.4)$);
\draw[black, path fading=east] (B8) -- ($(B8)+(-30:0.3)$);
\draw[black, path fading=south] (P9) -- ($(P9)+(-70:0.7)$);
\draw[black, path fading=east] (P13) -- ($(P13)+(10:0.5)$);
\draw[black, path fading=east] (C2) -- ($(C2)+(-30:0.55)$);
\draw[black, path fading=north] (C2) -- ($(C2)+(80:0.7)$);
\draw[black, path fading=north] (C5) -- ($(C5)+(60:1)$);
\draw[black, path fading=north] (C7) -- ($(C7)+(50:1)$);
\draw[black, path fading=north] (C7) -- ($(C7)+(145:1)$);
\draw[black, path fading=north] (P10) -- ($(P10)+(75:1)$);
\draw[black, path fading=west] (P10) -- ($(P10)+(170:2)$);
\draw[black, path fading=south] (A1) -- ($(A1)+(-75:1)$);
\draw[black, path fading=west] (A1) -- ($(A1)+(-140:1)$);
\draw[black, path fading=west] (A4) -- ($(A4)+(-130:1)$);
\draw[black, path fading=west] (A5) -- ($(A5)+(-160:1)$);
\draw[black, path fading=north] (A5) -- ($(A5)+(120:1.5)$);
\end{scope}
\end{tikzpicture}}
\caption{An illustration of a partially annulus-embedded graph.}
\label{figure_partially}
\end{figure}

Let $\Delta$ be a closed annulus with left boundary $B_1$ and right boundary $B_2$.
Also, let $G$ be a partially $\Delta$-embedded graph.
We denote by ${\sf Left}_\Delta (G)$  the connected component of $G\setminus \inter(\Delta)$ that contains $B_1$ and by ${\sf Right}_\Delta (G)$ the graph $G\setminus ((V(G)\cap \Delta)\cup V({\sf Left}_\Delta (G)))$.

\paragraph{Parallel cycles.}
Let $\Delta$ be a closed annulus  with left boundary $B_1$ and right boundary $B_2$ and let $G$ be a partially $\Delta$-embedded graph.
Also, let $\mathcal{C} = [C_1,\ldots, C_r]$, $r\geq 2$ be a collection of vertex disjoint cycles of $G$ that are embedded in $\Delta$.
We say that $\mathcal{C}$ is a \emph{$\Delta$-parallel sequence of cycles} of $G$ if
$C_1 = B_1$, $C_r = B_2$ and, for every $i\in [2,r]$, $C_1$ and $C_i$ are the boundaries of a closed annulus, denoted by $\Delta_i$, that is a subset of $\Delta$ such that $\Delta_2\subseteq \cdots \subseteq \Delta_{r}=\Delta$.
We call $C_1$ the \emph{leftmost} and $C_r$ the \emph{rightmost} cycle of $\mathcal{C}$.
From now on, each $\Delta$-parallel sequence $\mathcal{C}$ of cycles will be accompanied with the sequence $[\Delta_2, \ldots, \Delta_{r}]$ of the corresponding closed annuli.
Given $i,j\in[2,r]$, where $i\leq j$, we call the set $\Delta_i\setminus {\sf int}(\Delta_j)$ \emph{$(i,j)$-annulus} of $\mathcal{C}$ and we denote it by $\ann(\mathcal{C},i,j)$. Also, for every $i\in[2,r]$, we set $\Delta_i$ to be the \emph{$(1,i)$-annulus} of $\mathcal{C}$ and we also denote it by $\ann(\mathcal{C},1,i)$.

\paragraph{Railed annuli.}
Let $G$ be a graph.
Also, let $r\in\mathbb{N}_{\geq 3}$ and $q\in \mathbb{N}_{\geq 3}$ and assume that $r$ is an odd number.
An \emph{$(r,q)$-railed annulus} of $G$ is a pair $\mathcal{A}=(\mathcal{C},\mathcal{P})$ where 
$\mathcal{C}=[C_{1},\ldots,C_{r}]$  is a sequence of cycles of $G$ and $\mathcal{P}=[P_{1},\ldots,P_{q}]$ is a  collection of pairwise vertex-disjoint 
 paths in $G$ such that 
 \begin{itemize}
\item  
For every $j\in[q],$ the endpoints of $P_{j}$ are vertices of $C_r$ and $C_1$ and
 
\item for every $(i,j)\in[r]\times[q],$   $C_{i}\cap P_{j}$ is  a non-empty path, that we denote $P_{i,j}$.
\end{itemize}
See~\autoref{@incomprehensible} for an example. 
We refer to the paths of $\mathcal{P}$ as the \emph{rails} of $\mathcal{A}$ and to the cycles of $\mathcal{C}$ as the \emph{cycles} of $\mathcal{A}$.
We use $V(\mathcal{A})$ to denote the vertex set $\bigcup_{i\in[r]}V(C_i)\cup \bigcup_{j\in[q]}V(P_j)$ and
$E(\mathcal{A})$ to denote the edge set $\bigcup_{i\in[r]}E(C_i)\cup \bigcup_{j\in[q]}E(P_j)$.
We can see each path $P_{j}$ in $\mathcal{P}$ as being oriented towards the ``inner'' part of ${cal A}$,
i.e., starting from a vertex of  $C_{p}$ and finishing to a vertex of $P_{1,j}$.
For every $(i,j)\in[r]\times[q]$, we denote by $s_{i,j}$ (resp. $t_{i,j}$) the first (resp. last) vertex of $P_{i,j}$ when traversing $P_j$ according to this orientation.
If $(i,i')\in[r]^2$ with $i<i'$ then we define $\mathcal{A}_{i,i'}$ to be the railed annulus $([C_{i},\cdots,C_{i'}],P_1',\ldots, P_q'),$
where  for every $j\in[q]$, $P_j'$ is the subpath of $P_j$ between $s_{i,j}$ and $t_{i',j}$.

\paragraph{Annulus-embedded railed annuli.}
Let $\Delta$ be a closed annulus
and let $G$ be a partially $\Delta$-embedded graph.
Also, let $r\in\mathbb{N}_{\geq 3}$ and $q\in \mathbb{N}_{\geq 3}$ and assume that $r$ is an odd number.
A $(r,q)$-railed annulus $\mathcal{A}=(\mathcal{C},\mathcal{P})$ of $G$ is called \emph{$\Delta$-embedded} if $\mathcal{C}=[C_{1},\ldots,C_{r}]$  is a $\Delta$-parallel sequence of cycles of $G$.
We use $\ann(\mathcal{A})$ to denote $\ann(\mathcal{C},1,r)$.

The following proposition~\cite[Proposition 5.1]{BasteST20acomp} states that large railed annuli can be found inside a slightly larger wall  and will be used in the next section. For the definition of a wall see~\autoref{label_mistreatment}.

\begin{proposition}\label{label_simultaneously}
If $x,z \geq 3$ are odd integers, $y\geq 1,$ and $W$  is an   ${\sf odd}(2x+\max\{z,\frac{y}{4}-1\})$-wall,  then
\begin{itemize}
\item there is a collection $\mathcal{P}$ of $y$ paths in $W$ such that if $\mathcal{C}$ is the collection of the first $x$ layers of $W,$ then $(\mathcal{C},\mathcal{P})$ is an $(x,y)$-railed annulus of $W$ where the first cycle of $\mathcal{C}$ is the perimeter of $W,$ and
\item the open disk defined by the $x$-th cycle of $\mathcal{C}$ contains the vertices of the compass of the central $z$-subwall of $W.$
\end{itemize}
\end{proposition}

\subsection{Combing linkages}\label{subsec_combin}
In this subsection we define linkages and we present the Linkage Combing Lemma from~\cite{GolovachST22comb} (\autoref{prop_combinglemma}) -- see also~\cite{GolovachST20hitti}.

\paragraph{Linkages.} 
A \emph{linkage} in a graph $G$ is a subgraph $L$ of $G$ whose connected components are non-trivial paths.
The \emph{paths} of a linkage are its connected components and we denote them by $\mathcal{P}(L)$.
We call $|\mathcal{P}|$ the \emph{size} of $\mathcal{P}(L)$.
The \emph{terminals} of a linkage $L$, denoted by $T(L)$, are the endpoints of the paths of $L$,
and the \emph{pattern} of $L$ is the set $\{\{s,t\}\mid \mathcal{P}(L)\mbox{  contains some  }(s,t)\mbox{-path}\}.$
Two linkages $L_{1}, L_{2}$ of $G$ are \emph{equivalent} if they have the same pattern and we denote this fact by $L_{1}\equiv L_{2}$.
Let $\Delta$ be a closed annulus or a closed disk,
let $G$ be a partially $\Delta$-embedded graph, $L$ be a linkage of $G$, and $D$ be a subset of $\Delta$.
We say that $L$ is \emph{$D$-avoiding} if $T(L)\cap D=\emptyset$
(see \autoref{@incomprehensible}). 

\begin{figure}[ht]
	\centering
	\begin{tikzpicture}[scale=.51]
		
	\foreach \x in {2,2.5,...,6}{
		\draw[line width =0.6pt] (0,0) circle (\x cm);
	}
	\begin{scope}[on background layer]
	\fill[blue!10!white] (0,0) circle (6 cm);
	\fill[white] (0,0) circle (2 cm);
	\end{scope}
	
	\node[black node] (P11) at (45:6) {};
	\node[black node] (m21) at (40:5.5) {};
	\node[black node] (P21a) at (30:5) {};
	\node[black node] (P21b) at (40:5) {};
	\node[black node] (m31) at (40:4.5) {};
	\node[black node] (P31a) at (35:4) {};
	\node[black node] (P31b) at (50:4) {};
	\node[black node] (m41) at (45:3.5) {};
	\node[black node] (P41a) at (45:3) {};
	\node[black node] (P41b) at (25:3) {};
	\node[black node] (m51) at (25:2.5) {};
	\node[black node] (P51) at (40:2) {};
	
	\draw[line width=1pt] (P11) -- (m21) -- (P21a) -- (P21b) -- (m31) --(P31a) -- (P31b) -- (m41) -- (P41a) -- (P41b) -- (m51) -- (P51);
	
	\node[black node] (P12) at (70:6) {};
	\node[black node] (m22) at (80:5.5) {};
	\node[black node] (P22a) at (80:5) {};
	\node[black node] (P22b) at (75:5) {};
		\node[black node] (m32) at (80:4.5) {};
	\node[black node] (P32a) at (90:4) {};
	\node[black node] (P32b) at (75:4) {};
		\node[black node] (m42) at (80:3.5) {};
	\node[black node] (P42a) at (85:3) {};
	\node[black node] (P42b) at (70:3) {};
		\node[black node] (m52) at (70:2.5) {};
	\node[black node] (P52) at (75:2) {};
	
	\draw[line width=1pt] (P12) -- (m22) -- (P22a) -- (P22b) -- (m32)-- (P32a) -- (P32b) -- (m42) -- (P42a) -- (P42b) -- (m52) -- (P52);
	
	\node[black node] (P13a) at (120:6) {};
	\node[black node] (P13b) at (110:6) {};
		\node[black node] (m23) at (105:5.5) {};
	\node[black node] (P23a) at (110:5) {};
	\node[black node] (P23b) at (115:5) {};
		\node[black node] (m33) at (120:4.5) {};
	\node[black node] (P33a) at (120:4) {};
	\node[black node] (P33b) at (130:4) {};
		\node[black node] (m43) at (130:3.5) {};
	\node[black node] (P43a) at (135:3) {};
	\node[black node] (P43b) at (120:3) {};
		\node[black node] (m53) at (120:2.5) {};
	\node[black node] (P53) at (125:2) {};
	
	\draw[line width=1pt] (P13a) -- (P13b) -- (m23) -- (P23a) -- (P23b) --  (m33) -- (P33a) -- (P33b) -- (m43) -- (P43a) -- (P43b) -- (m53) -- (P53);
	
	\node[black node] (P14a) at (170:6) {};
	\node[black node] (P14b) at (160:6) {};
			\node[black node] (m24) at (160:5.5) {};
	\node[black node] (P24) at (155:5) {};
			\node[black node] (m34) at (153:4.5) {};
	\node[black node] (P34) at (160:4) {};
			\node[black node] (m44) at (163:3.5) {};
	\node[black node] (P44a) at (165:3) {};
	\node[black node] (P44b) at (180:3) {};
			\node[black node] (m54) at (175:2.5) {};
	\node[black node] (P54) at (170:2) {};
	
	\draw[line width=1pt] (P14a) -- (P14b) -- (m24) -- (P24) -- (m34) --  (P34) -- (m44) -- (P44a) -- (P44b) -- (m54) -- (P54);
	
	\node[black node] (P18a) at (200:6) {};			
		\node[black node] (m28) at (197:5.5) {};
	\node[black node] (P28a) at (195:5) {};
	\node[black node] (P28b) at (220:5) {};
		\node[black node] (m38) at (215:4.5) {};
	\node[black node] (P38a) at (210:4) {};
	\node[black node] (P38b) at (225:4) {};
		\node[black node] (m48) at (215:3.5) {};
	\node[black node] (P48a) at (205:3) {};
	\node[black node] (P48b) at (220:3) {};
		\node[black node] (m58) at (210:2.5) {};
	\node[black node] (P58) at (200:2) {};
	
	\draw[line width=1pt] (P18a) -- (m28) -- (P28a) to [bend right=10]  (P28b) -- (m38) --  (P38a) -- (P38b) -- (m48) -- (P48a) -- (P48b) -- (m58) -- (P58);

	\node[black node] (P15) at (235:6) {};
		\node[black node] (m25) at (235:5.5) {};
	\node[black node] (P25a) at (240:5) {};
	\node[black node] (P25b) at (250:5) {};
		\node[black node] (m35) at (245:4.5) {};
	\node[black node] (P35b) at (245:4) {};
		\node[black node] (m45) at (247:3.5) {};
	\node[black node] (P45a) at (250:3) {};
	\node[black node] (P45b) at (240:3) {};
		\node[black node] (m55) at (236:2.5) {};
	\node[black node] (P55) at (235:2) {};
	
	\draw[line width=1pt] (P15) -- (m25)-- (P25a)  -- (P25b) -- (m35) --  (P35b) -- (m45) -- (P45a) -- (P45b) -- (m55) -- (P55);
	
	\node[black node] (P16a) at (290:6) {};
	\node[black node] (P16b) at (300:6) {};
		\node[black node] (m26) at (300:5.5) {};
	\node[black node] (P26a) at (295:5) {};
	\node[black node] (P26b) at (310:5) {};
		\node[black node] (m36) at (305:4.5) {};
	\node[black node] (P36a) at (300:4) {};
	\node[black node] (P36b) at (290:4) {};
		\node[black node] (m46) at (300:3.5) {};
	\node[black node] (P46a) at (310:3) {};
	\node[black node] (P46b) at (325:3) {};
		\node[black node] (m56) at (325:2.5) {};
	\node[black node] (P56) at (320:2) {};
	
	\draw[line width=1pt] (P16a) -- (P16b) -- (m26) -- (P26a) -- (P26b) --  (m36) -- (P36a) -- (P36b) -- (m46) -- (P46a) -- (P46b) -- (m56) -- (P56);
	
	\node[black node] (P17a) at (10:6) {};
	\node[black node] (P17b) at (-5:6) {};
		\node[black node] (m27) at (-5:5.5) {};
	\node[black node] (P27a) at (0:5) {};
	\node[black node] (P27b) at (-10:5) {};
		\node[black node] (m37) at (-10:4.5) {};
	\node[black node] (P37a) at (-15:4) {};
	\node[black node] (P37b) at (0:4) {};
		\node[black node] (m47) at (5:3.5) {};
	\node[black node] (P47a) at (10:3) {};
	\node[black node] (P47b) at (-5:3) {};
		\node[black node] (m57) at (0:2.5) {};
	\node[black node] (P57) at (-5:2) {};
	
	\draw[line width=1pt] (P17a) -- (P17b) -- (m27) --  (P27b)  -- (P27a) -- (m37) --  (P37a) -- (P37b) -- (m47) --  (P47a) -- (P47b) -- (m57) -- (P57);
	
	\draw[red, line width=1.5pt] plot [smooth, tension=1.2] coordinates { (120:1) (140:2) (120:3.5) (100:4.5) (145:5) (150:6.5)};
	\draw[red, line width=1.5pt] plot [smooth, tension=1] coordinates { (245:6.5) (230:5) (230:3) (240:1.5) (270:2.5) (285:4) (280:6.5)};
	\draw[red, line width=1.5pt] plot [smooth, tension=1.2] coordinates { (40:6.5) (20:5) (-10:2.5) (-40:6.5)};
	\node[black node]  () at (120:1) {};
	\node[black node]  () at (150:6.5) {};
	\node[black node]  () at (245:6.5) {};
	\node[black node]  () at (280:6.5) {};
	\node[black node]  () at (40:6.5) {};
	\node[black node]  () at (-40:6.5) {};

	\end{tikzpicture}
	\caption{An example of a $\Delta$-embedded railed annulus $\mathcal{A}$
		and a linkage $L$ (depicted in red) that is $\Delta$-avoiding.}
	\label{@incomprehensible}
\end{figure}
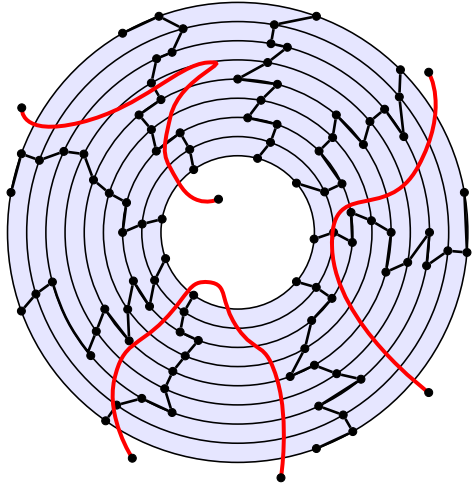

\paragraph{Linkages confined in annuli.}
Let $t\in\mathbb{N}_{\geq 1}$, let  $p=2t+1$, and let $s\in[p]$ where $s=2t'+1$.
Also, let $\Delta$ be a closed annulus and $\mathcal{A}=(\mathcal{C},\mathcal{P})$ be a  $\Delta$-embedded $(p,q)$-railed annulus of a partially $\Delta$-embedded graph $G$.
Given some $I\subseteq [q]$, we say that a linkage  $L$ of $G$ is \emph{$(s,I)$-confined in $\mathcal{A}$} if 
$$L\cap \ann(\mathcal{C},t+1-t',t+1+t')\subseteq \bigcup_{i\in I}P_{i}.$$

%
%

Intuitively, the above definition demands that $L$ traverses the ``middle'' $(s,q)$-annulus 
by intersecting it only at the rails of $\mathcal{A}$.
\medskip

We now state the Linkage Combing Lemma from~\cite{GolovachST22comb} (see also~\cite{GolovachST20hitti}).
Intuitively, it says that in the presence of a ``big enough'' $\Delta$-embedded railed annulus $\mathcal{A}$ in a partially $\Delta$-embedded graph $G$, where $\Delta$ is a closed annulus, every linkage of $G$ can be ``combed'' through the rails of $\mathcal{A}$ in some central buffer inside $\mathcal{A}$.

\begin{proposition}[Linkage Combing]
\label{prop_combinglemma}

There exist two functions  $\newfun{@norteamericana}, \newfun{@heteronomously}:\mathbb{N}\to\mathbb{N}$, where
the images of  $\funref{@heteronomously}$ are even,
such that 
for every odd $s\in \mathbb{N}_{\geq  1}$ and every $k\in\mathbb{N}$, if
\begin{itemize}
\item $\Delta$ is a closed annulus,
\item $G$ is a graph that is partially $\Delta$-embedded,
\item $\mathcal{A}=(\mathcal{C},\mathcal{P})$ is a  $\Delta$-embedded $(p,q)$-railed 
annulus of $G$, where  $p\geq \funref{@heteronomously}(k)+s$ and $q\geq  5/2\cdot\funref{@norteamericana}(k)$,
\item $L$ is a $\Delta$-avoiding linkage of size at most $k$, and
\item $I\subseteq [q]$, where $|I|> \funref{@norteamericana}(k)$,
\end{itemize}
then
$G$ contains a linkage
$\tilde{L}$ where $\tilde{L}\equiv L$, $\tilde{L}\setminus \Delta \subseteq L \setminus \Delta$, and  $\tilde{L}$ is $(s,I)$-confined in $\mathcal{A}$.
Moreover, $ \funref{@heteronomously}(k)= \mathcal{O}((\funref{@norteamericana}(k))^2)$.
\end{proposition}

\subsection{Linkages in boundaried graphs}\label{subsec_tm}
In this subsection, we aim to define the set of \emph{models} of a boundaried graph, that is all linkages that can be routed inside this graph and contain the boundary vertices as terminals.
This collection of graphs, can be encoded in abstract terms of collections of binary relations between indices of the boundaried vertices, representing the existence of (disjoint) paths between the corresponding boundary vertices.
This encoding is, in fact, present in the encoding of the pattern (see~\autoref{obs_translatemodelstopattern}).

\paragraph{Linkages of boundaried graphs.}
A \emph{pairing} is a $t$-boundaried graph  ${\bf L} = (L,v_1,\ldots, v_t)$ where $L$ is a linkage and $T(L)\subseteq \{v_1,\ldots,v_t\}$.
We use $\mathcal{B}_{\sf pair}^{(t)}$ to denote the set of all (pairwise non-isomorphic) $t$-boundaried graphs that are pairings.
A path of a linkage $L$ is \emph{non-trivial} if it is not a single edge.
Given a $k$-boundaried graph $(G,u_1,\ldots, u_k)$, a linkage $L$ of $G$ and
some $v_1,\ldots, v_t\in T(L)\cup\{\mathspace\}$ such  that $\{u_1,\ldots, u_k\}\subseteq \{v_1,\ldots, v_t\}$,
we say that $(L,v_1,\ldots, v_t)$ is a \emph{boundaried linkage of $(G,u_1,\ldots, u_k)$}.

Let $(G,v_1,\ldots,v_t)$ be a $t$-boundaried graph.
We define the set of \emph{pairings} of $(G,v_1,\ldots, v_t)$, denoted by $\Models(G,v_1,\ldots, v_t)$,
to be the set $$\Models (G,v_1,\ldots, v_t)=\left\{(L,v_1,\ldots,v_t) \in\mathcal{B}_{\sf pair}^{(t)}\ \middle\vert
\begin{array}{l} (L,v_1,\ldots,v_t) \text{ is a boundaried linkage}\\
\text{of $(G,v_1,\ldots,v_t)$ and every path of $L$}\\
\text{is non-trivial}
\end{array}\right\}.$$

\paragraph{Imprint of linkages.}
Let $r,h,l\in\mathbb{N}$.
Let ${\bf G}=(G,X_1,\ldots, X_h,{\bf a},v_1,\ldots, v_r)\in \mathcal{B}^{(r,h,l)}$.
Given an ${\bf L} \in \Models(G,v_1,\ldots, v_r)$, we
define the \emph{imprint} of ${\bf L}$, denoted by $\imprint({\bf L})$,
to be the graph whose vertex set is $I=\{i\in [r]\mid v_i\neq \mathspace\}$ and
two vertices $i,j\in I$ are adjacent if $\{v_i,v_j\}$ belongs to the pattern of $L$.
We define the \emph{compression} of ${\bf G}$, denoted by ${\sf compression}({\bf G})$ to be the quintuple $(\mathcal{V}, \kappa,\delta,H^e, \mathcal{H}^P)$, where
\begin{itemize}
\item $\mathcal{V}$ is the partition of $I$ to sets of pairwise equal vertices,
\item $\kappa$ is the partial function mapping each $i\in I$ to $j\in[l]$ such that $v_i = a_j$,
\item $\delta$ is the partial function mapping each $i\in I$ to $j\in [h]$ such that $v_i\in X_j$,
\item $H^e ={\sf Ind}_G(I)$, and
\item $\mathcal{H}^P = \{\imprint({\bf L})\mid {\bf L}\in \Models(G,v_1,\ldots, v_r)\}$.
\end{itemize}
It is easy to observe the following.
\begin{observation}\label{obs_translatemodelstopattern}
For every ${\bf G}\in \mathcal{B}^{(r,h,l)}$,  ${\sf compression}({\bf G}) = {\sf pattern}({\bf G})$.
\end{observation}
We stress that the only difference between the collection $\mathcal{H}^P$ of imprints of all pairings of $(G,v_1,\ldots, v_r)$ and the set $\Models(G,v_1,\ldots, v_r)$ is that the first is a collection of graphs whose vertex set is the set of indices $I\subseteq [r]$, while graphs in $\Models(G,v_1,\ldots, v_r)$ are subgraphs of $G$.
The distinction between the two is essential in order to 1) encode the patterns of linkages ``abstractly'' (in terms of graphs on indices) and 2) decode the presence of same variety of linkages inside graphs with the same pattern.
\medskip

Before concluding this section, we present some additional definitions on linkages of boundaried colored graphs and the reason is the following.
In~\autoref{sec_signaturesexchengability}, we describe how, given a colored graph
partially embedded in a ``big enough'' railed annulus, construct (a series of) boundaried graphs whose boundary vertices will also be some vertices of the rails of the annuli.
These boundary vertices will be chosen to be the ``few'' vertices in which every linkage can be combed, due to~\autoref{prop_combinglemma}.
Therefore, as we are about to prove in~\autoref{sec_signaturesexchengability} (in particular,~\autoref{lem_colomodelsrerout}),
after combing, linkages can be separated on a ``left'' and a ``right'' part.
Therefore, we need to define a way to \emph{glue} pairings.

\paragraph{Gluing pairings.}
We now define compatibility between pairings.
Let $\ell,r\in\mathbb{N}$.
Let ${\bf L}= (L,u_1,\ldots, u_{r+\ell})$ and ${\bf L}'= (L',v_1,\ldots, v_{r+\ell})$ be two pairings.
We say that ${\bf L}$ and ${\bf L}'$ are \emph{$\ell$-compatible} if
$u_{r+1},\ldots, u_{r+\ell}\in T(L)$,
 $v_{r+1},\ldots, v_{r+\ell}\in T(L')$,
 for every $i\in[r]$, $u_i=\mathspace$ if and only if $v_i \neq\mathspace$,
 and for every $i,j\in[\ell]$, $\{v_{r+i},v_{r+j}\}$ is in the pattern of $L$ if and only if $\{v_{r+i},v_{r+j}\}$ is not in the pattern of $L'$.
Given two families $\mathcal{F}_1, \mathcal{F}_2\subseteq \mathcal{B}^{(r+\ell)}$ of pairings,
we say that $\mathcal{F}_1$ and $\mathcal{F}_2$ are \emph{$\ell$-compatible}
if for every ${\bf L}_1\in \mathcal{F}_1$ and every ${\bf L}_2\in \mathcal{F}_2$,
${\bf L}_1$ and ${\bf L}_2$ are $\ell$-compatible.

Let ${\bf L}_1 = (L_1,u_1,\ldots, u_{r+\ell})$ and ${\bf L}_2= (L_2,v_1,\ldots, v_{r+\ell})$ be two $\ell$-compatible pairings.
We denote by $(L_1,u_1,\ldots, u_{r+\ell})\oplus_\ell (L_2,v_1,\ldots, v_{r+\ell})$ the pairing
$(L,c_1,\ldots, c_r)$,
where
\begin{itemize}
\item $L$ is the linkage obtained from the disjoint union of $L_1$ and $L_2$ after identifying, for every $i\in[r+1,\ell]$, the vertices $u_i$ and $v_i$, and
\item $c_i = u_i$, if $u_i\in T(L_1)$, while $c_i = v_i$, if $v_i\in T(L_2)$ (recall that by definition of compatibility, for every $i\in[r]$, at least one of $u_i$ and $v_i$ is equal to $\mathspace$ but not both).
\end{itemize}
Given two $\ell$-compatible families  $\mathcal{F}_1, \mathcal{F}_2\subseteq \mathcal{B}^{(r+\ell)}$ of pairings,
we denote by $\mathcal{F}_1\oplus_\ell \mathcal{F}_2$ the collection $\{{\bf L}_1\oplus_\ell {\bf L}_2 \mid {\bf L}_1\in \mathcal{F}_1 \text{ and } {\bf L}_2\in \mathcal{F}_2\}.$

\section{Dealing with apices}
\label{sec_annotated}
In this section, we show how to deal with apex vertices that can be adjacent to vertices in a flat area of the graph in a completely arbitrary way.
In~\autoref{subsec_apices}, 
we define a certain ``transformation'' from a colored graph to another colored graph (with more colors), interpreting neighborhoods of some predetermined vertices as new colors. Moreover, in~\autoref{subsec_apex_formula},
we describe how to get an equivalent version of a formula in this ``projected'' setting (see~\autoref{lem_interpretapex}).
This will allow us to define an equivalent version of any given formula in $\FOLDP$ that will transform a question from a flat graph with apices to a flat graph without apices (where adjacencies with apices are interpreted as colors in the graph).

\subsection{Projections of graphs with respect to some apex-tuple}\label{subsec_apices}
In this subsection, we define a way to ``project'' a given colored graph with respect to a given apex-tuple and a way to define the ``projected'' version of a sentence such that the initial colored graph satisfies the initial sentence if and only if the ``projected'' colored graph satisfies the ``projected'' sentence (see~\autoref{lem_interpretapex}).
This transformation will allow us to work with the projected colored graphs, where the absence of apices in a graph of big enough treewidth implies a flat bidimensional structure (see~\autoref{prop_flatwallbdtw} in~\autoref{sec_flatwalls}).
Flatness is particularly critical for our techniques, as it will be explained in~\autoref{sec_signaturesexchengability}.
\medskip

We proceed to formalize the idea of ``projecting'' a structure with respect to an apex-tuple.
First, we describe what is the vocabulary of the ``projected'' structure.

\paragraph{Constant-projections of vocabularies.}
Let $\tau$ be a colored-graph vocabulary,
let $l\in\mathbb{N}$, and let ${\bf c}=\{{\sf c}_1,\ldots, {\sf c}_l\}$ be a collection of $l$ constant symbols.
We define the \emph{constant-projection} $\tau^{\langle {\bf c}\rangle}$ of $(\tau\cup{\bf c})$ to be the vocabulary
$(\tau \cup {\bf c} \cup \{{\sf C}_1,\ldots, {\sf C}_l\})$, where ${\sf C}_1,\ldots, {\sf C}_l$ are unary relation symbols not contained in $\tau$.
\medskip

Given a colored-graph vocabulary $\tau$ and some collection of constant symbols  ${\bf c}$, we proceed to define a way to construct a $\tau^{\langle {\bf c}\rangle}$-structure from a given $(\tau\cup{\bf c})$-structure $(\mathfrak{G},{\bf a})$. The obtained $\tau^{\langle {\bf c}\rangle}$-structure is the ``projection'' of  $(\mathfrak{G},{\bf a})$
with respect to the apex-tuple ${\bf a}$.

\paragraph{Projecting a colored graph with respect to an apex-tuple.}
Let $\tau$ be a colored-graph vocabulary,
let $l\in\mathbb{N}$, and let ${\bf c}=\{{\sf c}_1,\ldots, {\sf c}_l\}$ be a collection of $l$ constant symbols.
Let also $\tau^{\langle {\bf c}\rangle}$ be the constant-projection of $(\tau\cup{\bf c})$.
Given a $(\tau\cup{\bf c})$-structure $(\mathfrak{G},{\bf a})$, where ${\bf a} =(a_1,\ldots, a_l)$ is an apex-tuple of $G$ of size $l$ and, for every $i\in[l]$, ${\sf c}_i^{(\mathfrak{G},{\bf a})} =a_i$, we define the structure
${\sf ap}_{\bf c}(\mathfrak{G},{\bf a})$ to be the $\tau^{\langle {\bf c}\rangle}$-structure obtained as follows:

\begin{itemize}
\item $V({\sf ap}_{\bf c}(\mathfrak{G},{\bf a})) = V(G),$
\item for every $i\in [l]$ ${\sf c}_i^{{\sf ap}_{\bf c}(\mathfrak{G},{\bf a})} = a_i,$
\item ${\sf E}^{{\sf ap}_{\bf c}(\mathfrak{G},{\bf a})} = {\sf E}^\mathfrak{G} \cap ((V(G)\setminus V({\bf a}))^2 \cup V({\bf a})^2),$
\item every ${\sf R}\in \tau\setminus\{{\sf E}\}$ is interpreted in ${\sf ap}_{\bf c}(\mathfrak{G},{\bf a})$ as in $\mathfrak{G},$ and
\item for every $i\in[l]$, ${\sf C}_i$ is interpreted in ${\sf ap}_{\bf c}(\mathfrak{G},{\bf a})$ as $N_{G}(a_i)\setminus V({\bf a}).$
\end{itemize}
Notice that if $a_i = \mathspace$, ${\sf C}_i$ is interpreted in ${\sf ap}_{\bf c}(\mathfrak{G},{\bf a})$ as the empty set.
Intuitively, we introduce a color $C_i,i\in[a]$ for each $a_i$.
We keep the same universe,
we keep only edges that are either between apices or between non-apices, and
we color the neighbors of $a_i$ by color $C_i$. See~\autoref{fig_apexprojection} for an example.
\begin{figure}[ht]
\centering
\begin{tikzpicture}
\node[simple] (v1) at (0,-0.2) {};
\node[simple] (v2) at (0.3,1) {};
\node[simple] (v3) at (0.6,0.5) {};
\node[simple] (v4) at (1.2,0) {};
\node[simple] (v5) at (1.7,1.2) {};
\node[simple] (v6) at (1.8,-0.4) {};
\node[simple] (v7) at (2.2,0.6) {};
\node[simple] (v8) at (2.8,0.9) {};
\node[simple] (v9) at (2.9,0) {};
\node[simple] (v10) at (3.3,1.4) {};
\node[simple] (v11) at (3.5, 0.3) {};
\node[simple] (v12) at (3.8, 0.7) {};

\node[simple, label={$a_1$}] (a1) at (0.8,2.5) {};
\node[simple, label={$a_2$}] (a2) at (1.4,2.5) {};
\node[simple, label={$a_3$}] (a3) at (2.1,2.5) {};
\node[simple, label={$a_4$}] (a4) at (2.7,2.5) {};

\begin{scope}[on background layer]
\node[circle, minimum width=2pt,diagonal fill={cyan}{orange!80!yellow}] () at (v2) {};
\node[circle, minimum width=6pt,fill=orange!80!yellow] () at (v9) {};
\node[circle, minimum width=6pt,fill=orange!80!yellow] () at (v10) {};
\node[circle, minimum width=6pt, diagonal fill={green!80!blue}{red}] () at (v5) {};
\node[circle, minimum width=6pt,fill=red] () at (v11) {};
\node[circle, minimum width=5pt,fill=cyan] () at (v6) {};
\node[circle, minimum width=5pt,fill=cyan] () at (v8) {};
\node[circle, minimum width=6pt,fill=green!80!blue] () at (v3) {};
\node[circle, minimum width=6pt,fill=green!80!blue] () at (v12) {};
\end{scope}

\draw[white!80!blue] (a1) -- (v2) (a1) -- (v6) (a1) -- (v8);
\draw[white!80!blue] (a2) -- (v3) (a2) -- (v5) (a2) -- (v12);
\draw[white!80!blue] (a3) -- (v2) (a3) -- (v9) (a3) -- (v10);
\draw[white!80!blue] (a4) -- (v5) (a4) -- (v11);

\draw[-] (v1) -- (v2) (v1) -- (v3) (v1) -- (v4) (v2) -- (v3) (v2) -- (v5) (v3)-- (v4) (v4) -- (v6) (v5) -- (v7) (v5) -- (v8) (v7)--(v9) (v9) -- (v11) (v4)--(v7) (v10)-- (v8) (v10) -- (v12) (v11)--(v12) (v8)--(v11);

\draw[-] (a1) -- (a2) (a2) -- (a3) (a2) to [bend right = 30] (a4)  (a1) to [bend right = 30] (a3);

\end{tikzpicture}
\caption{An example of a graph $G$ with an apex-tuple ${\bf a}=(a_1,\ldots,a_4)$.
The set $C_1$ contains all cyan vertices (the neighbors of $a_1$), the set $C_2$ contains all green vertices, the set $C_3$ contains all orange vertices and the set $C_4$ contains all red vertices. Vertices that are neighbors to more than one $a_i$ are depicted with multiple colors.
The structure ${\sf ap}_{\bf c}(\mathfrak{G},{\bf a})$ is obtained from $G$ after introducing colors $C_i, i\in[4]$ and after removing all edges between a vertex in $\{a_1,\ldots, a_4\}$ and a vertex in $V(G)\setminus \{a_1,\ldots, a_4\}$.}
\label{fig_apexprojection}
\end{figure}
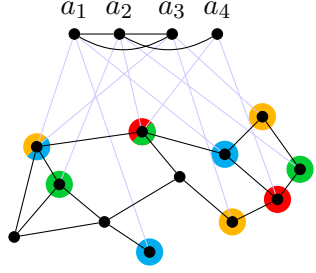

\subsection{Apex-projected sentences}\label{subsec_apex_formula}

Having defined the structure ${\sf ap}_{\bf c}(\mathfrak{G},{\bf a})$,
we now define for every sentence $\varphi\in\FOLDP$ its
\emph{$l$-apex-projected sentence} $\varphi^l$.
This will be a sentence in $\FOL[\tau^{\langle {\bf c}\rangle}+\DP]$ (see \autoref{obs_projectingstructure}) and can be seen as the equivalent (to $\varphi$) question that is asked to be satisfied by ${\sf ap}_{\bf c}(\mathfrak{G},{\bf a})$.

Let $\tau$ be a colored-graph vocabulary, let $l\in \mathbb{N}$, and let ${\bf c}=\{{\sf c}_1,\ldots, {\sf c}_l\}$ be a collection of $l$ constant symbols.
Given a set $S$ and an $\ell\in\mathbb{N}_{\geq 1}$, we use $\mathcal{P}_\ell (S)$ to denote all partitions of $S$ into $\ell$ parts, i.e., all collections of $\ell$ pairwise disjoint subsets $S_1,\ldots,S_\ell$ of $S$ such that $\bigcup_{i\in[\ell]}S_i = S$.\medskip

For every sentence $\varphi\in\FOLDP$, we define its \emph{$l$-apex-projected sentence} $\varphi^l$ to be the sentence obtained from $\varphi$ by replacing each atomic formula ${\sf E}({\sf x},{\sf y})$ by
\[{\sf E}({\sf x},{\sf y})\vee \bigvee_{i\in[l]} \big(({\sf x}={\sf c}_i \wedge {\sf y}\in  {\sf C}_i)\vee ({\sf y}={\sf c}_i \wedge{\sf x}\in  {\sf C}_i )\big)\]
and each atomic formula $\DP({{\sf s}}_{1},  {{\sf t}}_{1},\ldots, {{\sf s}}_{k}, {{\sf t}}_{k})$ by the formula
$\zeta_{\DP}({{\sf s}}_{1},  {{\sf t}}_{1},\ldots, {{\sf s}}_{k}, {{\sf t}}_{k})$ that we proceed to define.
We set ${\sf distinct}({\sf x}_1,\ldots, {\sf x}_t):=\bigwedge_{i,j\in[t], i\neq j} ({\sf x}_i\neq {\sf x}_j)$.
We define $\zeta_{\DP}({{\sf s}}_{1},  {{\sf t}}_{1},\ldots, {{\sf s}}_{k}, {{\sf t}}_{k})$ as follows
 (see next paragraph for an intuitive explanation):\medskip
 
\begin{eqnarray}
& &\bigvee_{B\subseteq [l]}
\bigvee_{d\in[k]}
\bigvee_{(I_1,\ldots, I_d)\in\mathcal{P}_d (B)}
\bigvee_{\substack{\text{all injective functions}\\ \rho: \{I_1,\ldots, I_d\}\to [k]}}
\label{@protagonista}\\
& &~ \Bigg(
\exists\ {\sf x}_1^{I_1}, \ldots, {\sf x}_{|I_1|}^{I_1},
\ldots,
{\sf x}_1^{I_d}, \ldots, {\sf x}_{|I_d|}^{I_d}: {\sf distinct}({\sf x}_1^{I_1}, \ldots, {\sf x}_{|I_1|}^{I_1},
\ldots,
{\sf x}_1^{I_d}, \ldots, {\sf x}_{|I_d|}^{I_d})
\label{@bandolerisme}\\
& &~~
\wedge \bigwedge_{\substack{i,i'\in[d],\\ i\neq i'}}\bigwedge_{j\in[2,|I_i|-1]}({\sf x}_j^{I_i}\neq {\sf s}_{\rho(I_{i'})} \wedge{\sf x}_j^{I_i}\neq {\sf t}_{\rho(I_{i'})})
\label{@unpardonably}\\
& &~~~
\wedge\bigwedge_{i\in[d]}\bigvee_{\substack{\text{all bijections}\\\lambda_{I_i}:[|I_i|]\to I_i}}
\bigg( \bigwedge_{i\in [d]}\bigwedge_{j\in[|I_i|]} {\sf x}_j^{I_i} = {\sf c}_{\lambda_{I_i}(j)}
\label{@desenvolvimiento}\\
& &~~~~\wedge \bigwedge_{i\in[d]}
\bigvee_{J_i\subseteq  [|I_i|-1]}
\Big(
\bigwedge_{i\in[d]}
\bigwedge_{j\in [|I_i|-1]\setminus J_i}
{\sf E} ({\sf x}_{j}^{I_i},{\sf x}_{j+1}^{I_i}) 
\label{@avoidability}
\\
& &~~~~~ \wedge\
\exists_{j\in J_1} {\sf y}_j^{I_1}, \ldots,\exists_{j\in J_d} {\sf y}_{j}^{I_d},
\exists_{j\in J_1} {\sf z}_j^{I_1}, \ldots,\exists_{j\in J_d} {\sf z}_{j}^{I_d},
%
\label{@neoclassicism}
\\
& &~~~~~~ {\sf distinct}(({\sf y}_j^{I_1})_{j\in J_1},\ldots,
({\sf y}_j^{I_d})_{j\in J_d}, {{\sf s}}_{1},  {{\sf t}}_{1},\ldots, {{\sf s}}_{k}, {{\sf t}}_{k})
\label{@deficiencies}
\\
& &~~~~~~~
\wedge\ {\sf distinct}(({\sf z}_j^{I_1})_{j\in J_1},\ldots,
({\sf z}_j^{I_d})_{j\in J_d}, {{\sf s}}_{1},  {{\sf t}}_{1},\ldots, {{\sf s}}_{k}, {{\sf t}}_{k})
\label{@amorosamente}
\\
& &~~~~~~~~~~~~~~
\wedge\bigwedge_{\substack{i,i'\in[d],\\ i\neq i'}}\bigwedge_{\substack{j\in J_i,\\j'\in J_{i'}}} ({\sf y}_j^{I_i}\neq {\sf z}_{j'}^{I_{i'}})\wedge\bigwedge_{i\in[d]}\bigwedge_{\substack{j,j'\in J_i,\\j\neq j'}} ({\sf y}_j^{I_i}\neq {\sf z}_{j'}^{I_{i}})
\label{@unsuspicious}
\\
& &~~~~~~~~~~~~~~\wedge\big(\bigwedge_{i\in[d]}
\bigwedge_{j\in J_i}
({\sf y}_j^{I_i}\in {\sf C}_{\lambda_{I_i}(j)} \wedge{\sf z}_{j}^{I_i}\in  {\sf C}_{\lambda_{I_i}(j+1)})\wedge
\label{@externalized}\\
& &~~~~~~~~~~~~~~~~~~~~ \bigvee_{(X_i,Y_i,Z_i)\in \mathcal{P}_3(J_i)}(
\bigwedge_{j\in X_i} ({\sf y}_j^{I_i} = {\sf z}_{j}^{I_i})
\wedge
\bigwedge_{j\in Y_i} {\sf E}({\sf y}_j^{I_i},{\sf z}_{j}^{I_i}) \wedge\ \psi_{I_i,\rho,Z_i})\big)
\Big)
\bigg)
\Bigg),\label{@unbelievability}
\end{eqnarray}
where $\psi_{I_i,\rho,Z_i}$ is the formula
\begin{align*}
\big({\sf s}_{\rho(I_i)} = {\sf x}_{1}^{I_i} \wedge {\sf t}_{\rho(I_i)} = {\sf x}_{|I_i|}^{I_i}\wedge\
\DP(\xi_{I_i,\rho,Z_i}({{\sf s}}_{1},  {{\sf t}}_{1},\ldots,{{\sf s}}_{k}, {{\sf t}}_{k},{\sf c}_1,{\sf c}_1,\ldots,{\sf c}_l,{\sf c}_l))\big)\\
\vee\\
\big({\sf s}_{\rho(I_i)} \neq {\sf x}_{1}^{I_i} \wedge {\sf t}_{\rho(I_i)} =  {\sf x}_{|I_i|}^{I_i}\\
\wedge\  \exists\ \tilde{{\sf s}}_{\rho(I_i)} \in  {\sf C}_{\lambda_{I_i}(1)}
\wedge{\sf distinct}(\tilde{\sf s}_{\rho(I_i)}, ({\sf s}_j,{\sf t}_j)_{j\neq \rho(I_i)})
\wedge\ \DP(\kappa_{I_i,\rho,Z_i}({{\sf s}}_{1},  {{\sf t}}_{1},\ldots, {{\sf s}}_{k}, {{\sf t}}_{k},{\sf c}_1,{\sf c}_1,\ldots,{\sf c}_l,{\sf c}_l))\big)\\
\vee\\
\big({\sf s}_{\rho(I_i)} = {\sf x}_{1}^{I_i} \wedge {\sf t}_{\rho(I_i)} \neq  {\sf x}_{|I_i|}^{I_i}\\
\wedge\ \exists\ \tilde{{\sf t}}_{\rho(I_i)} \in {\sf C}_{\lambda_{I_i}(|I_i|)}
\wedge{\sf distinct}(\tilde{\sf t}_{\rho(I_i)}, ({\sf s}_j,{\sf t}_j)_{j\neq \rho(I_i)})
\wedge\ \DP(\delta_{I_i,\rho,Z_i}({{\sf s}}_{1},  {{\sf t}}_{1},\ldots, {{\sf s}}_{k}, {{\sf t}}_{k},{\sf c}_1,{\sf c}_1,\ldots,{\sf c}_l,{\sf c}_l))\big)\\ 
\vee\\
\big({\sf s}_{\rho(I_i)} \neq {\sf x}_{1}^{I_i} \wedge {\sf t}_{\rho(I_i)} \neq  {\sf x}_{|I_i|}^{I_i}\\
\wedge\  \exists\ \tilde{{\sf s}}_{\rho(I_i)} \in {\sf C}_{\lambda_{I_i}(1)}
\wedge{\sf distinct}(\tilde{\sf s}_{\rho(I_i)}, ({\sf s}_j,{\sf t}_j)_{j\neq \rho(I_i)})
\\
\wedge\ \exists\ \tilde{{\sf t}}_{\rho(I_i)}\in {\sf C}_{\lambda_{I_i}(|I_i|)}\
\wedge{\sf distinct}(\tilde{\sf t}_{\rho(I_i)}, ({\sf s}_j,{\sf t}_j)_{j\neq \rho(I_i)})
\wedge \DP(\zeta_{I_i,\rho,Z_i}({{\sf s}}_{1},  {{\sf t}}_{1},\ldots, {{\sf s}}_{k}, {{\sf t}}_{k},{\sf c}_1,{\sf c}_1,\ldots,{\sf c}_l,{\sf c}_l))\big),
\end{align*}
and 
\begin{itemize}
\item $\xi_{I_i,\rho,Z_i}({{\sf s}}_{1},  {{\sf t}}_{1},\ldots, {{\sf s}}_{k}, {{\sf t}}_{k},{\sf c}_1,{\sf c}_1,\ldots,{\sf c}_l,{\sf c}_l)$
is used to denote the tuple obtained from the tuple $({{\sf s}}_{1},  {{\sf t}}_{1},\ldots,{{\sf s}}_{k}, {{\sf t}}_{k},{\sf c}_1,{\sf c}_1,\ldots,{\sf c}_l,{\sf c}_l)$
after removing, for every $i\in[\ell]$, $({\sf s}_{\rho(I_i)},{\sf t}_{\rho(I_i)})$
and adding $({\sf y}_j^{I_i},{\sf z}_{j+1}^{I_i})$, for all $j\in Z_i$.
\item $\kappa_{I_i,\rho,Z_i}({{\sf s}}_{1},  {{\sf t}}_{1},\ldots, {{\sf s}}_{k}, {{\sf t}}_{k},{\sf c}_1,{\sf c}_1,\ldots,{\sf c}_l,{\sf c}_l)$
is used to denote the tuple obtained from the tuple $({{\sf s}}_{1},  {{\sf t}}_{1},\ldots,{{\sf s}}_{k}, {{\sf t}}_{k},{\sf c}_1,{\sf c}_1,\ldots,{\sf c}_l,{\sf c}_l)$
after removing, for every $i\in[\ell]$, $({\sf s}_{\rho(I_i)},{\sf t}_{\rho(I_i)})$
and adding $({\sf s}_{\rho(I_i)},\tilde{{\sf s}}_{\rho(I_i)})$ and
$({\sf y}_j^{I_i},{\sf z}_{j+1}^{I_i})$, for all $j\in Z_i$.
\item $\delta_{I_i,\rho,Z_i}({{\sf s}}_{1},  {{\sf t}}_{1},\ldots, {{\sf s}}_{k}, {{\sf t}}_{k},{\sf c}_1,{\sf c}_1,\ldots,{\sf c}_l,{\sf c}_l)$
is used to denote the tuple obtained from the tuple $({{\sf s}}_{1},  {{\sf t}}_{1},\ldots,{{\sf s}}_{k}, {{\sf t}}_{k},{\sf c}_1,{\sf c}_1,\ldots,{\sf c}_l,{\sf c}_l)$
after removing, for every $i\in[\ell]$, $({\sf s}_{\rho(I_i)},{\sf t}_{\rho(I_i)})$
and adding $({\sf t}_{\rho(I_i)},\tilde{{\sf t}}_{\rho(I_i)})$ and
$({\sf y}_j^{I_i},{\sf z}_{j+1}^{I_i})$, for all $j\in Z_i$.
\item 
 $\zeta_{I_i,\rho,Z_i}({{\sf s}}_{1},  {{\sf t}}_{1},\ldots, {{\sf s}}_{k}, {{\sf t}}_{k},{\sf c}_1,{\sf c}_1,\ldots,{\sf c}_l,{\sf c}_l)$
is used to denote the tuple obtained from the tuple $({{\sf s}}_{1},  {{\sf t}}_{1},\ldots,{{\sf s}}_{k}, {{\sf t}}_{k},{\sf c}_1,{\sf c}_1,\ldots,{\sf c}_l,{\sf c}_l)$ after removing, for every $i\in[\ell]$, $({\sf s}_{\rho(I_i)},{\sf t}_{\rho(I_i)})$ and adding 
$({\sf s}_{\rho(I_i)},\tilde{{\sf s}}_{\rho(I_i)}),({\sf t}_{\rho(I_i)},\tilde{{\sf t}}_{\rho(I_i)}),$ and  $({\sf y}_j^{I_i},{\sf z}_{j+1}^{I_i})$, for all $j\in Z_i$.
\end{itemize}

\paragraph{Intuitive explanation of the above formulas.}
We decode step by step the intuition behind the above formulas.
First, we replace ${\sf E}({\sf x},{\sf y})$ by
${\sf E}({\sf x},{\sf y})\vee \bigvee_{i\in[l]} \big(({\sf x}={\sf c}_i \wedge {\sf y}\in  {\sf C}_i)\vee ({\sf y}={\sf c}_i \wedge{\sf x}\in  {\sf C}_i )\big)$ in order to encode that, in the colored graph obtained after the removal of all edges between apex-vertices and the rest vertices of the graph, two vertices $x,y$ are adjacent if either the edge $\{x,y\}$ is present in the modified graph or
one of $x$ and $y$ is an apex-vertex $a_i$ (interpreting ${\sf c}_i$) and the other is colored by the corresponding color $C_i$, that is the color that all neighbors of $a_i$ receive.

For the the atomic formula $\zeta_{\DP}({{\sf s}}_{1},  {{\sf t}}_{1},\ldots,{{\sf s}}_{k}, {{\sf t}}_{k})$, the intuition is the following.
We want to separate the formula into many parts, i.e., many questions for disjoint paths or adjacencies, guessing whether the variables ${{\sf s}}_{1},  {{\sf t}}_{1},\ldots, {{\sf s}}_{k}, {{\sf t}}_{k}$ are assigned to apex-vertices and/or whether the apex-vertices are internal vertices of the disjoint paths between ${{\sf s}}_{1},  {{\sf t}}_{1},\ldots, {{\sf s}}_{k}, {{\sf t}}_{k}$.

For this reason, in line~\eqref{@protagonista} of the above definition, and in particular in ``$\bigvee_{B\subseteq [l]}$'', we start by guessing the subset $B$ of apex-vertices that are part of the disjoint paths (either as endpoints or as internal vertices). We will refer to this set as \emph{active apices}.
Then, with ``$\bigvee_{d\in[k]}$'', we guess how many among the $k$ disjoint paths contain active apices (and we call them \emph{active paths}) and then, with ``$\bigvee_{(I_1,\ldots, I_d)\in\mathcal{P}_d (B)}$'' we guess how $B$ is partitioned in $d$ sets, each set corresponding to active apices that belong to the same active path and with
``$\bigvee_{\substack{\text{all injective functions } \rho: \{I_1,\ldots, I_d\}\to [k]}}$'', we guess which active apices belong to each active path.

Having made all these guesses, in line~\eqref{@bandolerisme}, we ask for variables
(``$\exists\ {\sf x}_1^{I_1}, \ldots, {\sf x}_{|I_1|}^{I_1},\ldots, {\sf x}_1^{I_d}, \ldots, {\sf x}_{|I_d|}^{I_d}$'')
that will be interpreted as the active apices and we ask these variables to be interpreted as pairwise disjoint vertices  (``${\sf distinct}({\sf x}_1^{I_1}, \ldots, {\sf x}_{|I_1|}^{I_1},
\ldots,
{\sf x}_1^{I_d}, \ldots, {\sf x}_{|I_d|}^{I_d})$'').
Also, in line~\eqref{@unpardonably}, we ask that all
vertices that interpret ${\sf x}_j^{I_i}$ for $i\in[d]$ and $j\in[2,|I_i|-1]$ are different from the endpoints of all the other active paths (``$\bigwedge_{\substack{i,i'\in[d], i\neq i'}}\bigwedge_{j\in[2,|I_i|-1]}({\sf x}_j^{I_i}\neq {\sf s}_{\rho(I_{i'})} \wedge{\sf x}_j^{I_i}\neq {\sf t}_{\rho(I_{i'})})$'').
The two later properties are necessary for the disjointness of the demanded paths.
For the moment, we {\sl do not} demand the interpretations of ${\sf x}_1^{I_i}$ and ${\sf x}_{|I_i|}^{I_i}$, for any $i\in[d]$,
to be distinct from the endpoints of the paths.

To express that these variables correspond to apices, in line~\eqref{@desenvolvimiento},
first we guess to which element of $B$ each ${\sf x}_j^{I_i}$ corresponds
(``$\bigwedge_{i\in[d]}\bigvee_{\substack{\text{all bijections }\lambda_{I_i}:[|I_i|]\to I_i}}$'').
We stress that the order of ${\sf x}_j^{I_i}$ is fixed and implicitly corresponds to the order that active apices
are transversed by the corresponding active path.
For each $i\in[d]$, the bijection $\lambda_{I_i}$ is used to correspond the ascending indices $j$ of the variables ${\sf x}_j^{I_i}$ to the actual active apices.
Then, with ``$\bigwedge_{i\in[d]}\bigwedge_{j\in[|I_i|]} {\sf x}_j^{I_i} = {\sf c}_{\lambda_{I_i}(j)}$'',
we check whether this guess indeed corresponds to an interpretation of each  ${\sf x}_j^{I_i} $ with the appropriate 
active apex $a_{\lambda_{I_i}(j)}$.

In line~\eqref{@avoidability}, we partition active apices into two sets. First, we have the active apices whose next neighbor on the corresponding active path is {\sl not} an apex-vertex (we orient paths according to the ordering given by the ascending ordering of the indices of ${\sf x}_j^{I_i}$).
These active apices are guessed using ``$\bigwedge_{\in[d]}\bigvee_{J_i\subseteq  [|I_i|-1]}$''.
These, we call them \emph{shifting active apices}.
For the remaining ones (``$\bigwedge_{i\in[d]}
\bigwedge_{j\in [|I_i|-1]\setminus J_i}$''),
we check if indeed their next neighbor on the corresponding active path {\sl is} an apex-vertex (``${\sf E} ({\sf x}_{j}^{I_i},{\sf x}_{j+1}^{I_i})$'').

Now, in line~\eqref{@neoclassicism},
we deal with the shifting active apices (apices that would be transversed by path that would be routed through at least one non-apex vertex before entering again the set of apex vertices).
For each shifting active apex, we ask for the existence of two extra vertices, corresponding to the first vertex after this shifting active apex and the last vertex before the next apex in a supposed path.
This is done in ``$\exists_{j\in J_1} {\sf y}_j^{I_1}, \ldots,\exists_{j\in J_d} {\sf y}_{j}^{I_d},
\exists_{j\in J_1} {\sf z}_j^{I_1}, \ldots,\exists_{j\in J_d} {\sf z}_{j}^{I_d}$''.

The variables ${\sf y}_j^{I_i}$, $i\in[d]$, $j\in J_i$ and the variables ${\sf z}_j^{I_i}$, $i\in[d]$, $j\in J_i$ should be interpreted as pairwise disjoint vertices that are also different from the endpoints of all the other (active or not) paths.
We check this in lines~\eqref{@deficiencies} and~\eqref{@amorosamente} using ``${\sf distinct}(({\sf y}_j^{I_1})_{j\in J_1},\ldots,
({\sf y}_j^{I_d})_{j\in J_d}, {{\sf s}}_{1},  {{\sf t}}_{1},\ldots, {{\sf s}}_{k}, {{\sf t}}_{k})$'' and ``${\sf distinct}(({\sf z}_j^{I_1})_{j\in J_1},\ldots,
({\sf z}_j^{I_d})_{j\in J_d}, {{\sf s}}_{1},  {{\sf t}}_{1},\ldots, {{\sf s}}_{k}, {{\sf t}}_{k})$''.
Also, we demand that variables ${\sf y}_j^{I_i}$ and ${\sf z}_{j'}^{I_{i'}}$ are interpreted as disjoint vertices if they correspond to different paths (i.e., $i\neq i'$) or if they belong to the same path but they are not consecutive (i.e., $j\neq j'$). This is done in line~\eqref{@unsuspicious}.

For every $i\in[d]$, the vertices interpreting ${\sf y}_j^{I_i}, j\in J_i$ should appear in the corresponding active path directly {\sl after} the corresponding shifting active apex $a_{\lambda_{I_i}(j)}$ and the vertices interpreting ${\sf z}_{j}^{I_i},j\in J_i$ should appear in the corresponding active path directly {\sl before} the next corresponding apex $a_{\lambda_{I_i}(j+1)}$.
Since edges between apices and non-apices are no longer present in the graph, we encode this ``succession'' by using the colors of the neighborhood of the apices. This is done in ``$\bigwedge_{i\in[d]}
\bigwedge_{j\in J_i}
({\sf y}_j^{I_i}\in {\sf C}_{\lambda_{I_i}(j)} \wedge{\sf z}_{j}^{I_i}\in  {\sf C}_{\lambda_{I_i}(j+1)})$'', in line~\eqref{@externalized}.

At this point, we have dealt with the internal part of the paths (we have not yet discussed what happens with the endpoints; we will do so in the next paragraph) for what concerns the apices. In fact, we already explained how to guess which part of the apices will be part of the supposed disjoint paths (lines~\eqref{@protagonista}-\eqref{@desenvolvimiento}), how to guess which part of the paths is routed {\sl only} through apices (line~\eqref{@avoidability}), and where the paths ``exit'' and ``enter'' the set of apices (lines~\eqref{@neoclassicism}-\eqref{@externalized}).
What remains is to describe how to formulate the question on the graph without the apices and how we deal with the endpoints of the paths.

In fact, we already mentioned that, for every $i\in[d]$ and every $j\in[d]$, ${\sf y}_j^{I_i}$ and ${\sf z}_j^{I_i}$ are interpreted as two vertices that are not apices and for the supposed path corresponding to index $i$, the part between the interpretations of ${\sf y}_j^{I_i}$ and ${\sf z}_j^{I_i}$ is a maximal path that does not contain apices.
Having this in mind, in line~\eqref{@unbelievability}, we guess what is the relation between ${\sf y}_j^{I_i}$ and ${\sf z}_{j}^{I_i}$, for all $j\in J_i$.
We partition $J_i$ to three sets $X_i, Y_i,$ and $Z_i$ (``$\bigvee_{(X_i,Y_i,Z_i)\in \mathcal{P}_3(J_i)}$'').
The set $X_i$ contains all indices $j$ for the demanded path that passes through the apices $a_{\lambda_{I_i}(j)}$ and $a_{\lambda_{I_i}(j+1)}$ has to be routed through a single vertex between  $a_{\lambda_{I_i}(j)}$ and $a_{\lambda_{I_i}(j+1)}$, or, in other words,
${\sf y}_j^{I_i}$ and ${\sf z}_{j}^{I_i}$ are asked to be interpreted as the same vertex.
 (``$\bigwedge_{j\in X_i} ({\sf y}_j^{I_i} = {\sf z}_{j}^{I_i})$'').
The set $Y_i$ contains all indices $j$ for which the demanded path that passes through the apices $a_{\lambda_{I_i}(j)}$ and $a_{\lambda_{I_i}(j+1)}$ has to be routed through an edge connecting the interpretations of ${\sf y}_j^{I_i}$ and ${\sf z}_{j}^{I_i}$ (``$\bigwedge_{j\in Y_i} {\sf E}({\sf y}_j^{I_i},{\sf z}_{j}^{I_i})$'').
Finally, $Z_i$ contains all remaining $j\in J_i$, i.e., all indices $j$ for which the demanded path has to be routed through path of length two between the interpretations of ${\sf y}_j^{I_i}$ and ${\sf z}_{j}^{I_i}$.

To finish the description of the formula
$\zeta_{\DP}({{\sf s}}_{1},  {{\sf t}}_{1},\ldots, {{\sf s}}_{k}, {{\sf t}}_{k})$, we have to discuss how the final new demands for disjoint paths are formulated.
This is encoded in the formula $\psi_{I_i,\rho,Z_i}$.
There, for each active path, say indexed by $\rho(I_i)$, we have to ``update'' the demand for a path from ${\sf s}_{\rho(I_i)}$ to ${\sf t}_{\rho(I_i)}$ to the demand for paths between ${\sf y}_j^{I_i}$ and ${\sf z}_{j+1}^{I_i})$, for all $j\in Z_i$ and this update has to be done for all active paths.

In the last argument of the previous paragraph, we omitted some important detail.
The aforementioned update is correct (in the sense that the two questions for disjoint paths are equivalent) only
if, for all active paths, its first and last active apex are its endpoints, i.e., if ${\sf s}_{\rho(I_i)} = {\sf x}_{1}^{I_i}$ and ${\sf t}_{\rho(I_i)} = {\sf x}_{|I_i|}^{I_i}$ are true.
This is why the first disjunctive term of $\psi_{I_i,\rho,Z_i}$ is ``${\sf s}_{\rho(I_i)} = {\sf x}_{1}^{I_i} \wedge {\sf t}_{\rho(I_i)} = {\sf x}_{|I_i|}^{I_i}\wedge\  \DP(\xi_{i,\rho,Z_i}({{\sf s}}_{1},  {{\sf t}}_{1},\ldots,{{\sf s}}_{k}, {{\sf t}}_{k},{\sf c}_1,{\sf c}_1,\ldots,{\sf c}_l,{\sf c}_l))$'', where $\xi_{I_i,\rho,Z_i}({{\sf s}}_{1},  {{\sf t}}_{1},\ldots, {{\sf s}}_{k}, {{\sf t}}_{k},{\sf c}_1,{\sf c}_1,\ldots,{\sf c}_l,{\sf c}_l)$ is the tuple obtained from $({{\sf s}}_{1},  {{\sf t}}_{1},\ldots,{{\sf s}}_{k}, {{\sf t}}_{k},{\sf c}_1,{\sf c}_1,\ldots,{\sf c}_l,{\sf c}_l)$ after removing, for every $i\in[\ell]$, $({\sf s}_{\rho(I_i)},{\sf t}_{\rho(I_i)})$ and adding $({\sf y}_j^{I_i},{\sf z}_{j+1}^{I_i})$, for all $j\in Z_i$.
Since all non-shifting active apices are adjacent (see line~\eqref{@avoidability}) and
for all $j\in J_i\setminus Z_i$, the interpretations of ${\sf y}_j^{I_i}$ and ${\sf z}_{j+1}^{I_i})$ are either identical or adjacent (see line~\eqref{@unbelievability}), what remains is to find disjoint paths between the interpretations of ${\sf y}_j^{I_i}$ and ${\sf z}_{j+1}^{I_i})$, for every $j\in Z_i$.
Of course,
having dealt with the case that paths contains apices,
we want the new paths that we search to be disjoint from all apices.
To express this, we add ``${\sf c}_1,{\sf c}_1,\ldots,{\sf c}_l,{\sf c}_l$'' to the above tuples, asking that the rest (disjoint) paths are also disjoint from the path of zero lenght that starts and finishes to the intepretation of ${\sf c}_i$, for every $i\in[l]$.

To deal with the case where, for an active paths, either its first or last active apex are not its endpoints, we have to guess the existence of an extra ``entering'' or ``exiting'' point, respectively, and ask for some supplementary disjoint path between this new guessed vertex and the corresponding endpoint (see last three disjunctive terms of $\psi_{I_i,\rho,Z_i}$). This concludes the intuitive explanation of the
formula  $\zeta_{\DP}({{\sf s}}_{1},  {{\sf t}}_{1},\ldots, {{\sf s}}_{k}, {{\sf t}}_{k})$.
\medskip

Note that if $\varphi\in\FOLDP$ and ${\bf c}$ is a collection
of $l$ constant symbols, then $\varphi^l\in\FOL[\tau^{\langle {\bf c}\rangle}+\DP]$.
\begin{observation}\label{obs_projectingstructure}
Let $\tau$ be a colored-graph vocabulary.
For every $\varphi\in\FOLDP$, every $l\in\mathbb{N}$, and every collection ${\bf c}$ of $l$ constant symbols,
$\varphi^l\in\FOL[\tau^{\langle {\bf c}\rangle}+\DP]$.
\end{observation}

Also, note that in the definition of $\varphi^l$, we add some extra quantified first-order variables (see the definition of $\zeta_{\DP}({{\sf s}}_{1},{{\sf t}}_{1},\ldots,{{\sf s}}_{k}, {{\sf t}}_{k})$; line~\eqref{@bandolerisme}, line~\eqref{@neoclassicism}, and the definition of $\psi_{I_i,\rho,Z_i}$).
\begin{observation}\label{obs_quantifierrank}
Let $\tau$ be a colored-graph vocabulary and let $r,l\in\mathbb{N}$.
There is a function $\newfun{@transversales}:\mathbb{N}^2\to\mathbb{N}$ such that for every $\varphi\in\FOLDP$,
if $\varphi$ has quantifier rank $r$, then 
$\varphi^l$ has quantifier rank $\funref{@transversales}(r,l)$.
\end{observation}

The definition of the $l$-apex-projected sentence $\varphi^l$ implies the following lemma, which can be seen as a generalization of~\cite[Lemma 26]{FlumG01fixe} that deals with graphs to also ``interpret'' the vertex-disjoint paths predicates.

\begin{lemma}
\labels{lem_interpretapex}
Let $\tau$ be a colored-graph vocabulary, let $l\in\mathbb{N},$ and let ${\bf c}$ be a collection of $l$ constant symbols.
For every $\varphi\in\FOLDP,$ every $\tau$-structure $\mathfrak{G},$ and every apex-tuple  ${\bf a}$ of $\mathfrak{G}$ of size $l,$ it holds that $\mathfrak{G}\models \varphi \iff {\sf ap}_{\bf c}(\mathfrak{G},{\bf a})\models \varphi^l$ (where ${\bf c}$ is interpreted as ${\bf a}$).
\end{lemma}

\section{Partial signatures and exchangability}
\label{sec_signaturesexchengability}
As explained in the end of~\autoref{subsec_equiv}, after proving~\autoref{lemma_reducing} our strategy is, given an annotated colored graph to find a way to construct another annotated colored graph that is equivalent to the original one in the sense that
they have the same signature.
In this section, we describe how, in the presence of a big enough flat railed annulus in the given graph, we can define a way to encode only patterns of ``one side'' of the annulus.
In fact, in~\autoref{subsec_hierarchi}, we define {\sl stamps} of vertices with respect to some given railed annulus, which will help us to group vertices in a way that encodes their relative position to the cycles of the railed annulus.
Then, in~\autoref{subsec_conventions}, in order to facilitate reading, we establish some conventions for boundaried colored graphs in flat railed annuli.
Using the notation introduced in these first two subsections,
in~\autoref{subsec_combing_in_levelings}, we show how to reformulate~\autoref{prop_combinglemma} and  how linkages of the given graph (that contains a big enough railed annulus) are ``combed'' inside boundaried graphs whose boundary vertices are either some particular vertices of the railed annulus, or vertices of appropriate stamps (\autoref{lem_colomodelsrerout}).
Then, to capture the ``pattern-behavior'' of the considered series of boundaried graphs in the railed annulus,
we define the notion of {\sl partial signature} of a graph as a ``meta-collection'' of patterns of boundaried graphs (as the ones described in~\autoref{subsec_conventions}) for boundary vertices of particular stamps.
Finally, in~\autoref{subsec_exchangability}, we show that boundaried (annotated) colored graphs with the same partial signature can be ``replaced'',
maintaining the same (global) signature (see~\autoref{lem_equirep}).

\subsection{Stamps of vertices with respect to annuli}
\label{subsec_hierarchi}
In this subsection, we describe how to attribute {\sl stamps} to the vertices of a given graph $G$ with a {\sl flat} railed annulus $\mathcal{A}$ that encode the relative position of these vertices with respect to the cycles of $\mathcal{A}$.
The definition of a flat railed annulus is given in~\autoref{sec_flatwalls}, and, in particular, in~\autoref{label_exceptionalness}. Intuitively,
it is the analogue of flat walls but in terms of railed annuli.
\medskip

Let $r\in\mathbb{N}$.
Let $G$ be a graph, let $(\mathcal{A},\frR)$ be a $(p,q)$-railed annulus flatness pair of $G$, where $p = 2^r +1$.
For every $\bar{w}\in\{0,1\}^r$, we denote by $C_{\bar{w}}$ the cycle $C_{n_{\bar{w}}}$ of $\mathcal{C}$, where  $n_{\bar{w}}=1+\sum_{i\in[r]} w_i 2^{r-i+1}$.

Given an $i\in[r]$ and a $\bar{w}\in \{0,1\}^{i-1}$,
for every vertex $v\in V(G)$
we define the \emph{$\bar{w}$-stamp} of $v$ as follows:
$$
\textsf{stamp}_{\bar{w}}(v)=
\begin{cases}
(0,\bullet), & \text{if $v\in V(\textsf{Influence}_{\mathfrak{R}} (C_{\bar{w}00^{r-i}}))$},\\
(0,\circ), & \text{if $v\in V(\textsf{Influence}_{\mathfrak{R}} (C_{\bar{w}11^{r-i}}))\setminus V(\textsf{Influence}_{\mathfrak{R}} (C_{\bar{w}10^{r-i}}))$},\\
(1,\bullet), & \text{if $v\in V(\textsf{Influence}_{\mathfrak{R}} (C_{\bar{w}10^{r-i}}))\setminus V(\textsf{Influence}_{\mathfrak{R}} (C_{\bar{w}00^{r-i}}))$, and}\\
(1,\circ), & \text{if $v\notin V(\textsf{Influence}_{\mathfrak{R}} (C_{\bar{w}11^{r-i}}))$}.
\end{cases}
$$

\begin{figure}[ht]
\centering
\scalebox{0.9}{
\begin{tikzpicture}
\node[label={above left:$C_{\bar{w}00^{r-i}}$}] () at ($(160:2.6)+(0,2.7)$) {};
\node[label={above left:$C_{\bar{w}10^{r-i}}$}] () at ($(130:2.6)+(0,2.7)$) {};
\node[label={above left:$C_{\bar{w}11^{r-i}}$}] () at ($(100:2.6)+(0,2.7)$) {};
\clip (0,2.8) circle (2.6cm);

\begin{scope}

    \path (-1.2,2.7)  arc (40:77:5)
    node[terminal,pos=0.1] (P1) {}
    node[terminal,pos=0.3] (P0) {}
    (90:5)--(90:5);
    \draw [line width = 2pt,black,path fading=north, fading angle =40] (-1.2,2.7)  arc (40:64:5)  (100:4)--(100:4);
    
    \path (-1.2,2.7) arc (40:-15:5)
    node[terminal,pos=0.2] (P2) {}
    node[terminal,pos=0.4] (P3) {}
    node[pos=1] (A) {}
    (-15:5)--(-15:5);
    \draw [line width = 2pt,black,path fading=south] (-1.2,2.7) arc (40:5:5)    (7:5)--(7:5);
    
    \path (0.2,3)  arc (40:82:6.5)
    node[terminal,pos=0.1] (P6) {}
    node[terminal,pos=0.25] (P5) {}
    node[terminal,pos=0.45] (P4) {}
    (85:7)--(85:7);
    \draw [line width = 2pt,black!50!white] (0.2,3)  arc (40:59:6.5)  (85:7)--(85:7);
    \draw [line width = 2pt,black,path fading=west] (0.2,3)  arc (40:66:6.5)  (85:7)--(85:7);

    \path (0.2,3) arc (40:-5:6.5)
    node[terminal,pos=0.1] (P7) {}
    node[terminal,pos=0.3] (P8) {}
    node[terminal,pos=0.4] (P9) {}
    node[pos=1] (B) {}
    (-5:7)--(-5:7);
    \draw [line width = 2pt,black,path fading=south] (0.2,3) arc (40:13:6.5)
    (15:7)--(15:7);
    
    \path (1.7,3)  arc (40:85:8.5)
    node[pos=1] (D) {}
    node[terminal,pos=0.4] (P10) {}
    node[terminal,pos=0.25] (P11) {}
    node[terminal,pos=0.05] (P12) {}
    (90:9)--(90:9);
    \draw [line width = 2pt,black,path fading=west] (1.7,3)  arc (40:69:8.5)   (80:9)--(80:9);
    
    \path (1.7,3) arc  (40:5:8.5)
    node[terminal,pos=0.15] (P13) {}
    node[terminal,pos=0.35] (P14) {}
    node[terminal,pos=0.45] (P15) {}
    node[pos=1] (C) {}
    (5:9)--(5:9);
    \draw [line width = 2pt,black,path fading=south] (1.7,3) arc  (40:25:8.5)   (25:9)--(25:9);
    
    \foreach \x in {0,...,3} \node[terminal,circle, draw=black, fill=white, line width=0.6pt] () at (P\x) {};
    \foreach \x in {4,...,9} \node[terminal,blue,circle] () at (P\x) {};
    \foreach \x in {11,...,12} \node[terminal,red,circle] () at (P\x) {};
 \node[terminal,blue,circle] () at (P10) {};
 \node[terminal,blue,circle] () at (P13) {};

    \node[] (A1) at ($(A)+(-45:1)$) {$A_{\bar{w}0}$};
    \node[] (A2) at ($(B)+(-45:1)$) {$A_{\bar{w}1}$};

    \node[terminal, circle, draw=black, fill=white, line width=0.6pt](Q9) at (0.4,.6) {};
    \node[terminal,blue,circle](Q8) at (0.8,1.3) {};
    \node[terminal, circle, draw=black, fill=white, line width=0.6pt](Q7) at (-0.1,1.6) {};
    \node[terminal, circle, draw=black, fill=white, line width=0.6pt](Q6) at (0.2,1.8) {};
    \node[terminal,circle, draw=black, fill=white, line width=0.6pt](Q5) at (-.5,2.4) {};
    \node[terminal,blue,circle](Q4) at (-0.3,2.8) {};        
    \node[terminal,circle, draw=black, fill=white, line width=0.6pt](Q3) at (-1,2.8) {};    
    \node[terminal,blue,circle](Q2) at (-0.8,3.4) {};
    \node[terminal,blue,circle](Q1) at (-1.2,3.8) {};
    
    \node[terminal,blue,circle](ex2) at (-1.9,4) {};

    \node[terminal,red,circle](B8) at (2,1.6) {};
    \node[terminal,blue,circle](B7) at (1.2,2.2) {};
    \node[terminal,blue,circle](B6) at (1.4,2.4) {};
    \node[terminal,red,circle](B5) at (0.5,3.3) {};
    \node[terminal,red,circle](B4) at (0.7,3.6) {};
    \node[terminal,red,circle](B3) at (0.3,3.7) {};    
    \node[terminal,blue,circle](B2) at (-.65,4.25) {};
    \node[terminal,blue,circle](B1) at (-.85,4.45) {};

    \node[terminal, red,circle] (C7) at (0.7,5) {};
    \node[terminal, red,circle] (C6) at (0.4,4.6) {};
    \node[terminal, circle, draw=black, fill=green!81!blue, line width=0.6pt] (C5) at (1.3,4.5) {};
    \node[terminal, circle, draw=black, fill=green!81!blue, line width=0.6pt] (C4) at (1.2,4.25) {};
    \node[terminal, circle, draw=black, fill=green!81!blue, line width=0.6pt] (C3) at (1.45,3.95) {};
    \node[terminal, red,circle] (C2) at (2.1,3.6) {};
    \node[terminal, red,circle] (C1) at (2,3) {};

    \node[terminal, circle, draw=black, fill=white, line width=0.6pt] (A5) at (-2,2.8) {};
    \node[terminal, circle, draw=black, fill=white, line width=0.6pt] (A4) at (-1.8,2.3) {};
    \node[terminal, circle, draw=black, fill=white, line width=0.6pt] (A3) at (-1.2,1.9) {};
    \node[terminal, circle, draw=black, fill=white, line width=0.6pt] (A2) at (-1.5,1.5) {};
    \node[terminal, circle, draw=black, fill=white, line width=0.6pt] (A1) at (-1,1) {};

\begin{scope}[on background layer]
\fill[white] (0,2.8) circle (2.6cm);

\clip (0,2.8) circle (2.6cm);

    \begin{scope}
        \fill[cyan] (P9.center) to [bend left= 20] ($(P9)+(-70:1)$) to [bend left= 20] (P9.center);
        
        \fill[black!30!white] (P3.center) to [bend right= 10] ($(P3)+(-90:1)$) to ($(P3)+(-25:2)$) to [bend right=10] (P3.center);
        
        \fill[black!30!white] (A1.center) to [bend left= 20] ($(A1)+(-80:2)$) to ($(A1)+(-100:1)$) to [bend right=10] (A1.center);
        
        \fill[black!30!white] (A1.center) to [bend left= 10] ($(A1)+(-170:1)$) to ($(A1)+(-120:1)$) to [bend right=10] (A1.center);
        
        \fill[black!30!white] (A4.center) to [bend left= 10] ($(A4)+(-160:1)$) to ($(A4)+(-120:1)$) to [bend right=10] (A4.center);
        
        \fill[black!30!white] (A5.center) to [bend right= 10] ($(A5)+(-170:1)$) to ($(A5)+(-140:1)$) to [bend right=10] (A5.center);
  
        \fill[black!30!white] (P0.center) to [bend right= 10] ($(P0)+(140:1)$) to [bend right= 20]
        (A5.center) to [bend right=30] (P0.center);
        
        \fill[cyan] (ex2.center) to [bend left= 20] ($(ex2)+(150:1)$) to [bend left=30] (ex2.center);
        
        \fill[cyan] (P4.center) to [bend left= 20] ($(P4)+(150:1)$) to [bend left=30] (P4.center);
        
        \fill[red] (P10.center) to [bend left= 10] ($(P10)+(170:2)$) to ($(P10)+(120:1)$) to [bend right=10] (P10.center);
    
        \fill[green!81!blue] (P10.center) to [bend left= 20] ($(P10)+(85:1)$) to [bend left =20](C7.center) to [bend right=30] (P10.center);

        \fill[green!81!blue] (C7.center) to [bend right= 20] ($(C7)+(45:1)$) to [bend right =20]($(C7)+(70:1)$) to [bend right=15] (C7.center);
        
        \fill[green!81!blue] (C5.center) to [bend right= 20] ($(C5)+(45:1)$) to [bend right =20]($(C5)+(70:1)$) to [bend right=15] (C5.center);
        
        \fill[green!81!blue] (C2.center) to [bend right= 20] ($(C2)+(45:1)$) to [bend right =20]($(C2)+(90:2)$) to [bend right=15] (C2.center);
        
        \fill[green!81!blue] (C2.center) to [bend right= 20] ($(C2)+(-40:1)$) to [bend right =20]($(C2)+(-10:1)$) to [bend right=15] (C2.center);
        
        \fill[red] (P13.center) to [bend right= 20] ($(P13)+(-60:1)$) to [bend right =20]($(P13)+(30:1)$) to [bend right=15] (P13.center);
        
         \fill[red] (B8.center) to [bend right= 20] ($(B8)+(-100:1)$) to [bend right =20]($(B8)+(-10:1)$) to [bend right=15] (B8.center);

        \fill[white,path fading=fade in] (0,2.8) circle (3cm);
    \end{scope}

\fill[black!30!white,opacity=0.5] (P1.center) to [bend right= 20] (P2.center) to [bend right= 20] (Q4.center) to [bend right=20] (P1.center);

\fill[cyan!70!white,opacity=0.5] (P0.center) to [bend right= 30] (P4.center) to [bend right= 20] (P6.center) to [bend left=20] (P0.center);
   
\fill[cyan!70!white,opacity=0.5] (P5.center) to [bend right= 20] (P6.center) to [bend right=20] (P5.center);

\fill[black!30!white,opacity=0.5] (P1.center) to [bend left= 20] (A4.center) to [bend left= 20] (A5.center) to [bend left=20] (P1.center);

    \fill[black!30!white,opacity=0.5] (P1.center) to [bend right= 20] (P0.center) to [bend right=20] (P1.center);
    
    \fill[cyan!70!white,opacity=0.5] (P6.center) to [bend left= 20] (P7.center) to [bend left= 20] (Q4.center) to [bend left=40] (P6.center);
    
    \fill[cyan!70!white,opacity=0.5] (P2.center) to [bend left= 20] (P7.center) to [bend left= 20] (P2.center);

    \fill[black!30!white,opacity=0.5] (A4.center) to [bend left= 20] (P2.center) to [bend left= 20] (A1.center) to [bend left=40] (A4.center);
    
    \fill[black!30!white,opacity=0.5] (P3.center) to [bend left= 10] (P2.center) to [bend left= 20] (P8.center) to [bend right=20] (P3.center);
    
    \fill[cyan!70!white,opacity=0.5] (P3.center) to [bend right= 5] (P8.center) to [bend left= 40] (P9.center) to [bend left=20] (P3.center);
 
    \fill[black!30!white,opacity=0.5] (A1.center) to [bend right= 20] (P3.center) to [bend right=20] (A1.center);
    
   \fill[cyan!70!white,opacity=0.5] (P0) to [bend left= 20] (ex2) to [bend left =20](P4.center) to [bend left=15] (P0);
   
   \fill[cyan!70!white,opacity=0.5] (P4.center) to [bend left= 20] (P10.center) to [bend left =20](P5.center) to [bend left=15] (P4.center);
   
   \fill[red!50!white,opacity=0.5] (P10.center) to [bend right= 10] (C7.center) to [bend left =20](P11.center) to [bend left=15] (P10.center);
   
    \fill[red!50!white,opacity=0.5] (P5.center) to [bend right= 10] (P11.center) to [bend right=15] (P5.center);
   
   \fill[green!81!blue,opacity=0.5] (P11.center) to [bend left= 20] (C5.center) to [bend left =20](C2.center) to [bend left=10] (P11.center);
   
   \fill[red!50!white,opacity=0.5] (P11.center) to [bend right= 20] (P6.center) to [bend right =20](P12.center) to [bend right=20] (P11.center);
   
   \fill[red!50!white,opacity=0.5] (P12.center) to [bend left= 20] (C2.center) to [bend left =20](P13.center) to [bend left=15] (P12.center);
   
   \fill[cyan!70!white,opacity=0.5] (P13.center) to [bend right= 20] (P7.center) to [bend right =20](P8.center) to [bend right=15] (P13.center);
   
   \fill[red!50!white,opacity=0.5] (P13.center) to [bend left= 20] (B8.center) to  [bend left=15] (P13.center);
   
   \fill[red!50!white,opacity=0.5] (P8.center) to [bend left= 20] (B8.center) to [bend left=15] (P8.center);
   
   \fill[red!50!white,opacity=0.5] (P7.center) to [bend left= 20] (P12.center) to [bend left=15] (P7.center);
\end{scope}

    \draw[black,path fading= west] (ex2) -- ($(ex2)+(145:0.4)$);

    \draw[black] (P4)--(Q1) (P6) -- (Q4) (P7)--(Q4) (P7)--(P2) (P8)--(Q6) (P8)--(Q7) (P8)--(Q8) (P9)--(Q8);
    \draw[black] (Q3) -- (Q4) (Q4) -- (P6);
    
    \draw[black] (P0) -- (Q2) (Q2)-- (P6);
    \draw[black] (P3) -- (P2) (P2)-- (Q6) (Q6)-- (P8);

    \draw[black]  (P0)--(Q1) (P1)--(Q3) (P2)--(Q5) (P2)--(Q6) (P2)--(Q7) (P3)--(Q7) (P3)--(Q8) (Q3) to [bend left = 10] (P2) (P1) to [bend right = 10] (Q5); 
    
    \draw[black] (ex2) -- (P0) (ex2) -- (P4) (P0) -- (P4) (Q1)--(Q2) (Q3)--(Q4)--(Q5)--(Q3) (P3) -- (Q9);
    

    \draw[black] (B1) -- (B2) (B1) -- (P4) (B1) -- (P5) (B1) -- (P10) (B2) -- (P10) (B2) -- (P5) (B2) -- (P4) (P5) -- (P10) -- (P4);
    
    \draw[black] (P5) -- (P11);
    
    \draw[thick,black] (P6) -- (B3) (B3) -- (P11);
    \draw[black] (P8) -- (P7) (P7)--(B6)-- (P13);
    \draw[black] (P6) -- (B4) (P6) -- (B5) (B5) -- (P12) (B3) --(B4) (B4) -- (B5) (B3) -- (B5) (P11) -- (B4) (P12) -- (B4) (B5) -- (P11);
    
    \draw[black] (P7) -- (P12);
    
    \draw[black] (P7) -- (B6) (P7) -- (B7) (P8) -- (B7) (B7)-- (P13) (P13) -- (B6) (P13) -- (P8) (B6) -- (B7);
    
    \draw[black] (P8) -- (B8) (P13) -- (B8);

%

\draw[thick,black] (P11) -- (C4) (C4) -- (C5);
\draw[] (P10) -- (C6);
\draw[] (C6) -- (P11) (C6) -- (C7);
\draw[] (P11) -- (C5) (P11) -- (C4) (P11) -- (C3) (P11) -- (C2) (C5) -- (C2) (C5) -- (C3) (C5) -- (C4) (C4) -- (C3) (C3) -- (C2) (C4) to [bend left = 5] (C2);
\draw (P12) -- (C2) (P12) -- (C1) (P13) -- (C1) (C1) -- (C2);

    \draw[black] (A1)  -- (P3);
    \draw[black] (A1) -- (P2) (A1) -- (A2) (A1) -- (A3) (A1) -- (A4) (A2) -- (P2) (A2) -- (A3) (A2) -- (A4) (A3) -- (A4) (A4) -- (P2);
     \draw[thick,black]      (A3) -- (P2)      (A3) -- (A4);
    \draw[black] (A4) -- (A5)  -- (P1) -- (A4);
    \draw[black] (A5) -- (P0);

\draw[thick,black, path fading=east] (Q9) -- ($(Q9)+(-30:0.6)$);
\draw[thick,black, path fading=south] (Q9) -- ($(Q9)+(-100:0.5)$);
\draw[thick,black, path fading=south] (B8) -- ($(B8)+(-70:0.4)$);
\draw[thick,black, path fading=east] (B8) -- ($(B8)+(-30:0.3)$);
\draw[thick,black, path fading=east] (P13) -- ($(P13)+(10:0.5)$);
\draw[thick,black, path fading=east] (C2) -- ($(C2)+(-30:0.55)$);
\draw[thick,black, path fading=north] (C2) -- ($(C2)+(80:0.7)$);
\draw[thick,black, path fading=north] (C5) -- ($(C5)+(60:1)$);
\draw[thick,black, path fading=north] (C7) -- ($(C7)+(50:1)$);
\draw[thick,black, path fading=north] (C7) -- ($(C7)+(145:1)$);
\draw[thick,black, path fading=north] (P10) -- ($(P10)+(75:1)$);
\draw[thick,black, path fading=west] (P10) -- ($(P10)+(170:2)$);
\draw[thick,black, path fading=south] (A1) -- ($(A1)+(-75:1)$);
\draw[thick,black, path fading=west] (A1) -- ($(A1)+(-140:1)$);
\draw[thick,black, path fading=west] (A4) -- ($(A4)+(-130:1)$);
\draw[thick,black, path fading=west] (A5) -- ($(A5)+(-160:1)$);
\draw[thick,black, path fading=north] (A5) -- ($(A5)+(120:1.5)$);

\end{scope}
\end{tikzpicture}}
\caption{An illustration of some vertices of a graph which contains a flat railed annulus $\mathcal{A}=({\cal C},\mathcal{P})$, colored with respect to their stamps. The white-colored vertices have $\bar{w}$-stamp equal to $(0,\bullet)$, the blue-colored vertices have  $\bar{w}$-stamp  equal to $(1,\bullet)$, the red-colored vertices have  $\bar{w}$-stamp equal to $(0,\circ)$, and the green-colored vertices have  $\bar{w}$-stamp  equal to $(1,\circ)$.}
\label{figure_stamps}
\end{figure}
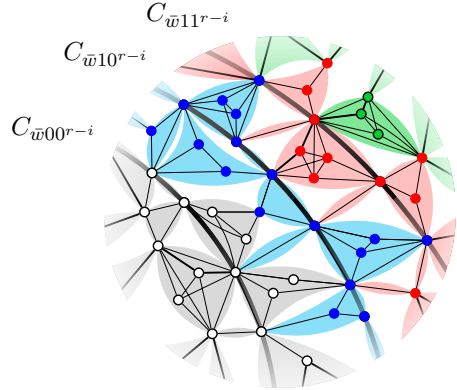

\paragraph{Trace of tuples of vertices with respect to annuli.}
Let $r\in\mathbb{N}$.
Let $G$ be a graph, let $(\mathcal{A},\frR)$ be a $(p,q)$-railed annulus flatness pair, where $p = 2^r +1$.
Given a tuple $(v_1,\ldots, v_r)$ of vertices of $G$, we define the \emph{trace of $(v_1,\ldots, v_r)$ with respect to $(\mathcal{A},\frR)$}, denoted by ${\sf trace}_{(\mathcal{A},\frR)}(v_1,\ldots, v_r)$, to be the pair $(w_1\ldots w_r,\omega_1\ldots \omega_r)$ where $(w_1,\omega_1) = \textsf{stamp}_\varepsilon (v_1)$ and for every $i\in[2,r]$, we set $(w_i,\omega_i) =\textsf{stamp}_{w_1\ldots w_{i-1}}(v_i)$.

\subsection{Some conventions for boundaried graphs in flat railed annuli}\label{subsec_conventions}

We proceed to define a series of boundaried graphs in a given railed annulus flatness pair $(\mathcal{A},\mathfrak{R})$.
In fact, we define these boundaried graphs in the {\sl leveling} ${\sf Leveling}_{(\mathcal{A}, \mathfrak{R})} (G)$ of $(\mathcal{A},\mathfrak{R})$, that is the ``planar representation'' of $(\mathcal{A},\mathfrak{R})$,
as defined in~\autoref{subsec_levelings}.
In order $(\mathcal{A},\mathfrak{R})$ to has this ``planar representation''
property, it has to be {\sl well-aligned} (see also~\autoref{subsec_levelings}).

Let $G$ be a graph and let $(\mathcal{A},\mathfrak{R})$ be a well-aligned $(p,q)$-railed annulus flatness pair of $G$.
We consider the graph ${\sf Leveling}_{(\mathcal{A}, \mathfrak{R})} (G)$ and keep in mind that ${\sf Leveling}_{(\mathcal{A}, \mathfrak{R})} (G)$ contains the representation $R_\mathcal{A}$ of $\mathcal{A}$, that is a $\Delta$-embedded $(p,q)$-railed annulus.
Let $R_\mathcal{A} = (\mathcal{C},\mathcal{P})$ and keep in mind that $|\mathcal{C}|=p$ and $p$ is an odd integer in $\mathbb{N}_{\geq 3}$.
Intuitively, the cycle $C_{(p+1)/2}$ is the ``middle'' cycle of $\mathcal{C}$.
We can see each path $P_{j}$ in $\mathcal{P}$ as being oriented towards the ``inner'' part of $R_\mathcal{A}$, i.e., starting from an endpoint of  $P_{p,j}$ and finishing to an endpoint of $P_{1,j}$.
For every $j\in[q],$
we define $r_{j}$ as the first vertex of $P_{j}$ that appears in $P_{(p+1)/2,j}$ (recall that $P_{(p+1)/2,j}$ is the intersection of the cycle $C_{(p+1)/2}$ and the rail $P_j$)
while traversing $P_{j}$ according to this orientation. 
Given a $t\in [q]$, 
we define the $t$-boundaried graph $${\bf G}_\mathcal{A}^{(t)} = (G',r_{1}\,\ldots,r_{t}),$$
where $G'$ is the graph ${\sf Leveling}_{(\mathcal{A}, \mathfrak{R})} (G)\setminus V({\sf Right}_{\Delta_{(p+1)/2}}({\sf Leveling}_{(\mathcal{A}, \mathfrak{R})} (G)))$ (recall that $\Delta_{(p+1)/2}$ is the closed annulus cropped by $C_1$ and $C_{(p+1)/2}$).
We call $r_1,\ldots, r_t$ the \emph{boundary vertices} of ${\bf G}_\mathcal{A}^{(t)}$.

Assume now that $\mathcal{A}$ is a well-aligned $(2^r \cdot(p-1) +1,q)$-railed annulus flatness pair of $G$.
We set $\mathcal{C}^{(p)} = C_1',\ldots, C_{2^r+1}'$, where for every $i\in [2^r+1]$, $C_i'= C_{ 1 + (p-1)\cdot (i-1)}$.
Intuitively,
we pick $\mathcal{C}^{(p)}$ in a way that every two consecutive cycles $C_i'$ and $C_{i+1}'$ of $\mathcal{C}^{(p)}$ crop
a $(p,q)$-railed annulus.
For every $\bar{w}\in \{0,1\}^r$, we use $\mathcal{A}_{\bar{w}}$ to denote the $(p,q)$-annulus that is cropped by the cycles $C_{n_{\bar{w}}}'$ and $C_{n_{\bar{w}}+1}'$ of $\mathcal{C}^{(p)}$, where $n_{\bar{w}}=1+\sum_{i\in[r]} w_i 2^{r-i+1}$ (i.e., $\mathcal{A}_{\bar{w}} = \mathcal{A}_{1 + (p-1)\cdot (n_{\bar{w}}-1), 1+ (p-1)\cdot n_{\bar{w}}})$.
See~\autoref{figr_nested} for an example.
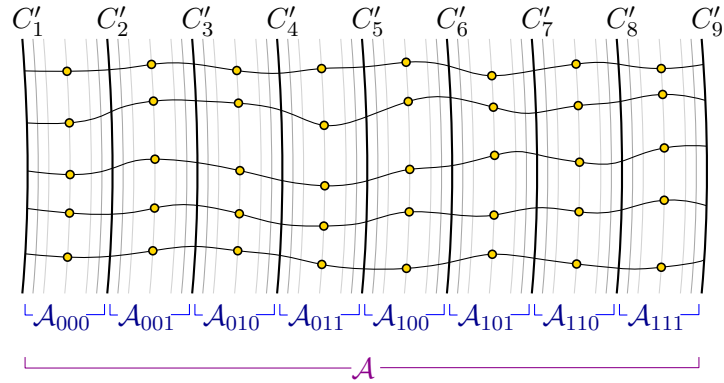
\begin{figure}[ht]
\centering
\begin{tikzpicture}[ipe stylesheet]
  \draw[lightgray]
    (232, 784)
     .. controls (234.6667, 749.3333) and (234.6667, 717.3333) .. (232, 688);
  \draw[lightgray]
    (248, 784)
     .. controls (250.6667, 749.3333) and (250.6667, 717.3333) .. (248, 688);
  \node[ipe node, text=darkblue]
     at (116, 676) {$\mathcal{A}_{000}$};
  \node[ipe node, text=darkblue]
     at (148, 676) {$\mathcal{A}_{001}$};
  \node[ipe node, text=darkblue]
     at (180, 676) {$\mathcal{A}_{010}$};
  \node[ipe node, text=darkblue]
     at (212, 676) {$\mathcal{A}_{011}$};
  \node[ipe node, text=darkblue]
     at (244, 676) {$\mathcal{A}_{100}$};
  \node[ipe node, text=darkblue]
     at (276, 676) {$\mathcal{A}_{101}$};
  \node[ipe node, text=darkblue]
     at (308, 676) {$\mathcal{A}_{110}$};
  \node[ipe node, text=darkblue]
     at (340, 676) {$\mathcal{A}_{111}$};
  \node[ipe node, text=darkmagenta]
     at (236, 656) {$\mathcal{A}$};
  \draw[shift={(113.035, 684.003)}, xscale=0.1236, yscale=1.0352, blue]
    (0, 0)
     -- (0, -4)
     -- (24, -4);
  \draw[shift={(136, 680.008)}, xscale=0.2905, yscale=0.9979, blue]
    (0, 0)
     -- (24, 0)
     -- (24, 4);
  \draw[shift={(145.035, 684.003)}, xscale=0.1236, yscale=1.0352, blue]
    (0, 0)
     -- (0, -4)
     -- (24, -4);
  \draw[shift={(168, 680.008)}, xscale=0.2905, yscale=0.9979, blue]
    (0, 0)
     -- (24, 0)
     -- (24, 4);
  \draw[shift={(177.035, 684.003)}, xscale=0.1236, yscale=1.0352, blue]
    (0, 0)
     -- (0, -4)
     -- (24, -4);
  \draw[shift={(200, 680.008)}, xscale=0.2905, yscale=0.9979, blue]
    (0, 0)
     -- (24, 0)
     -- (24, 4);
  \draw[shift={(209.035, 684.003)}, xscale=0.1236, yscale=1.0352, blue]
    (0, 0)
     -- (0, -4)
     -- (24, -4);
  \draw[shift={(232, 680.008)}, xscale=0.2905, yscale=0.9979, blue]
    (0, 0)
     -- (24, 0)
     -- (24, 4);
  \draw[shift={(241.035, 684.003)}, xscale=0.1236, yscale=1.0352, blue]
    (0, 0)
     -- (0, -4)
     -- (24, -4);
  \draw[shift={(264, 680.008)}, xscale=0.2905, yscale=0.9979, blue]
    (0, 0)
     -- (24, 0)
     -- (24, 4);
  \draw[shift={(273.035, 684.003)}, xscale=0.1236, yscale=1.0352, blue]
    (0, 0)
     -- (0, -4)
     -- (24, -4);
  \draw[shift={(296, 680.008)}, xscale=0.2905, yscale=0.9979, blue]
    (0, 0)
     -- (24, 0)
     -- (24, 4);
  \draw[shift={(305.035, 684.003)}, xscale=0.1236, yscale=1.0352, blue]
    (0, 0)
     -- (0, -4)
     -- (24, -4);
  \draw[shift={(328, 680.008)}, xscale=0.2905, yscale=0.9979, blue]
    (0, 0)
     -- (24, 0)
     -- (24, 4);
  \draw[shift={(337.035, 684.003)}, xscale=0.1236, yscale=1.0352, blue]
    (0, 0)
     -- (0, -4)
     -- (24, -4);
  \draw[shift={(360, 680.008)}, xscale=0.2905, yscale=0.9979, blue]
    (0, 0)
     -- (24, 0)
     -- (24, 4);
  \draw[shift={(113.036, 664.003)}, xscale=5.1236, yscale=1.0352, darkmagenta]
    (0, 0)
     -- (0, -4)
     -- (24, -4);
  \draw[shift={(246.857, 660.008)}, xscale=5.0048, yscale=0.9979, darkmagenta]
    (0, 0)
     -- (24, 0)
     -- (24, 4);
  \node[ipe node]
     at (108, 788) {$C_1'$};
  \node[ipe node]
     at (140, 788) {$C_2'$};
  \node[ipe node]
     at (172, 788) {$C_3'$};
  \node[ipe node]
     at (204, 788) {$C_4'$};
  \node[ipe node]
     at (236, 788) {$C_5'$};
  \node[ipe node]
     at (268, 788) {$C_6'$};
  \node[ipe node]
     at (300, 788) {$C_7'$};
  \node[ipe node]
     at (332, 788) {$C_8'$};
  \node[ipe node]
     at (364, 788) {$C_9'$};
  \draw[ipe pen heavier]
    (112, 784)
     .. controls (114.6667, 749.3333) and (114.6667, 717.3333) .. (112, 688);
  \draw[ipe pen heavier]
    (144, 784)
     .. controls (146.6667, 749.3333) and (146.6667, 717.3333) .. (144, 688);
  \draw[ipe pen heavier]
    (176, 784)
     .. controls (178.6667, 749.3333) and (178.6667, 717.3333) .. (176, 688);
  \draw[ipe pen heavier]
    (208, 784)
     .. controls (210.6667, 749.3333) and (210.6667, 717.3333) .. (208, 688);
  \draw[ipe pen heavier]
    (240, 784)
     .. controls (242.6667, 749.3333) and (242.6667, 717.3333) .. (240, 688);
  \draw[ipe pen heavier]
    (272, 784)
     .. controls (274.6667, 749.3333) and (274.6667, 717.3333) .. (272, 688);
  \draw[ipe pen heavier]
    (304, 784)
     .. controls (306.6667, 749.3333) and (306.6667, 717.3333) .. (304, 688);
  \draw[ipe pen heavier]
    (336, 784)
     .. controls (338.6667, 749.3333) and (338.6667, 717.3333) .. (336, 688);
  \draw[ipe pen heavier]
    (368, 784)
     .. controls (370.6667, 749.3333) and (370.6667, 717.3333) .. (368, 688);
  \draw[darkgray]
    (116, 784)
     .. controls (118.6667, 749.3333) and (118.6667, 717.3333) .. (116, 688);
  \draw[darkgray]
    (140, 784)
     .. controls (142.6667, 749.3333) and (142.6667, 717.3333) .. (140, 688);
  \draw[darkgray]
    (148, 784)
     .. controls (150.6667, 749.3333) and (150.6667, 717.3333) .. (148, 688);
  \draw[darkgray]
    (172, 784)
     .. controls (174.6667, 749.3333) and (174.6667, 717.3333) .. (172, 688);
  \draw[darkgray]
    (180, 784)
     .. controls (182.6667, 749.3333) and (182.6667, 717.3333) .. (180, 688);
  \draw[darkgray]
    (204, 784)
     .. controls (206.6667, 749.3333) and (206.6667, 717.3333) .. (204, 688);
  \draw[darkgray]
    (212, 784)
     .. controls (214.6667, 749.3333) and (214.6667, 717.3333) .. (212, 688);
  \draw[darkgray]
    (236, 784)
     .. controls (238.6667, 749.3333) and (238.6667, 717.3333) .. (236, 688);
  \draw[darkgray]
    (244, 784)
     .. controls (246.6667, 749.3333) and (246.6667, 717.3333) .. (244, 688);
  \draw[darkgray]
    (268, 784)
     .. controls (270.6667, 749.3333) and (270.6667, 717.3333) .. (268, 688);
  \draw[darkgray]
    (276, 784)
     .. controls (278.6667, 749.3333) and (278.6667, 717.3333) .. (276, 688);
  \draw[darkgray]
    (300, 784)
     .. controls (302.6667, 749.3333) and (302.6667, 717.3333) .. (300, 688);
  \draw[darkgray]
    (308, 784)
     .. controls (310.6667, 749.3333) and (310.6667, 717.3333) .. (308, 688);
  \draw[darkgray]
    (332, 784)
     .. controls (334.6667, 749.3333) and (334.6667, 717.3333) .. (332, 688);
  \draw[darkgray]
    (340, 784)
     .. controls (342.6667, 749.3333) and (342.6667, 717.3333) .. (340, 688);
  \draw[darkgray]
    (364, 784)
     .. controls (366.6667, 749.3333) and (366.6667, 717.3333) .. (364, 688);
  \draw[lightgray]
    (120, 784)
     .. controls (122.6667, 749.3333) and (122.6667, 717.3333) .. (120, 688);
  \draw[lightgray]
    (136, 784)
     .. controls (138.6667, 749.3333) and (138.6667, 717.3333) .. (136, 688);
  \draw[lightgray]
    (152, 784)
     .. controls (154.6667, 749.3333) and (154.6667, 717.3333) .. (152, 688);
  \draw[lightgray]
    (168, 784)
     .. controls (170.6667, 749.3333) and (170.6667, 717.3333) .. (168, 688);
  \draw[lightgray]
    (184, 784)
     .. controls (186.6667, 749.3333) and (186.6667, 717.3333) .. (184, 688);
  \draw[lightgray]
    (200, 784)
     .. controls (202.6667, 749.3333) and (202.6667, 717.3333) .. (200, 688);
  \draw[lightgray]
    (216, 784)
     .. controls (218.6667, 749.3333) and (218.6667, 717.3333) .. (216, 688);
  \draw[lightgray]
    (264, 784)
     .. controls (266.6667, 749.3333) and (266.6667, 717.3333) .. (264, 688);
  \draw[lightgray]
    (280, 784)
     .. controls (282.6667, 749.3333) and (282.6667, 717.3333) .. (280, 688);
  \draw[lightgray]
    (296, 784)
     .. controls (298.6667, 749.3333) and (298.6667, 717.3333) .. (296, 688);
  \draw[lightgray]
    (312, 784)
     .. controls (314.6667, 749.3333) and (314.6667, 717.3333) .. (312, 688);
  \draw[lightgray]
    (328, 784)
     .. controls (330.6667, 749.3333) and (330.6667, 717.3333) .. (328, 688);
  \draw[lightgray]
    (344, 784)
     .. controls (346.6667, 749.3333) and (346.6667, 717.3333) .. (344, 688);
  \draw[lightgray]
    (360, 784)
     .. controls (362.6667, 749.3333) and (362.6667, 717.3333) .. (360, 688);
  \draw
    (113.718, 752.3)
     .. controls (137.756, 751.022) and (141.617, 755.038) .. (147.5298, 757.4563)
     .. controls (153.4427, 759.8747) and (161.4073, 760.6953) .. (166.7312, 760.9212)
     .. controls (172.055, 761.147) and (174.738, 760.778) .. (181.419, 760.4457)
     .. controls (188.1, 760.1133) and (198.779, 759.8177) .. (206.671, 757.8772)
     .. controls (214.563, 755.9367) and (219.668, 752.3513) .. (225.0337, 751.4912)
     .. controls (230.3993, 750.631) and (236.0257, 752.496) .. (241.8147, 754.6738)
     .. controls (247.6037, 756.8517) and (253.5553, 759.3423) .. (258.8302, 760.7962)
     .. controls (264.105, 762.25) and (268.703, 762.667) .. (275.0532, 761.6333)
     .. controls (281.4033, 760.5997) and (289.5057, 758.1153) .. (294.8897, 756.8903)
     .. controls (300.2737, 755.6653) and (302.9393, 755.6997) .. (307.1983, 756.3813)
     .. controls (311.4573, 757.063) and (317.3097, 758.392) .. (322.6207, 758.9723)
     .. controls (327.9317, 759.5527) and (332.7013, 759.3843) .. (339.2822, 760.7762)
     .. controls (345.863, 762.168) and (354.255, 765.12) .. (369.404, 760.803);
  \draw
    (114, 734.181)
     .. controls (130.793, 731.826) and (138.3965, 732.951) .. (143.5272, 734.4208)
     .. controls (148.658, 735.8907) and (151.316, 737.7053) .. (156.6472, 738.3513)
     .. controls (161.9783, 738.9973) and (169.9827, 738.4747) .. (177.9867, 737.3165)
     .. controls (185.9907, 736.1583) and (193.9943, 734.3647) .. (199.3278, 733.08)
     .. controls (204.6613, 731.7953) and (207.3247, 731.0197) .. (211.5522, 730.2345)
     .. controls (215.7797, 729.4493) and (221.5713, 728.6547) .. (226.9042, 728.582)
     .. controls (232.237, 728.5093) and (237.111, 729.1587) .. (242.0268, 730.4877)
     .. controls (246.9427, 731.8167) and (251.9003, 733.8253) .. (257.236, 734.67)
     .. controls (262.5717, 735.5147) and (268.2853, 735.1953) .. (275.1303, 736.4153)
     .. controls (281.9753, 737.6353) and (289.9517, 740.3947) .. (295.2803, 741.2197)
     .. controls (300.609, 742.0447) and (303.29, 740.9353) .. (308.6347, 739.7635)
     .. controls (313.9793, 738.5917) and (321.9877, 737.3573) .. (327.3245, 736.9083)
     .. controls (332.6613, 736.4593) and (335.3267, 736.7957) .. (340.7598, 738.8506)
     .. controls (346.193, 740.9055) and (354.394, 744.679) .. (369.922, 743.535);
  \draw
    (113.827, 720.046)
     .. controls (145.701, 715.737) and (154.973, 718.702) .. (161.8795, 720.2032)
     .. controls (168.786, 721.7043) and (173.327, 721.7417) .. (178.5415, 721.115)
     .. controls (183.756, 720.4883) and (189.644, 719.1977) .. (194.9335, 717.7438)
     .. controls (200.223, 716.29) and (204.914, 714.673) .. (211.2622, 713.9337)
     .. controls (217.6103, 713.1943) and (225.6157, 713.3327) .. (230.9585, 713.5913)
     .. controls (236.3013, 713.85) and (238.9817, 714.229) .. (242.62, 715.2945)
     .. controls (246.2583, 716.36) and (250.8547, 718.112) .. (256.2052, 718.6435)
     .. controls (261.5557, 719.175) and (267.6603, 718.486) .. (273.1653, 717.8785)
     .. controls (278.6703, 717.271) and (283.5757, 716.745) .. (288.9265, 717.4253)
     .. controls (294.2773, 718.1057) and (300.0737, 719.9923) .. (305.8722, 720.2933)
     .. controls (311.6707, 720.5943) and (317.4713, 719.3097) .. (322.789, 718.6885)
     .. controls (328.1067, 718.0673) and (332.9413, 718.1097) .. (339.7442, 720.0001)
     .. controls (346.547, 721.8905) and (355.318, 725.629) .. (369.852, 721.078);
  \draw
    (113.114, 702.939)
     .. controls (128.703, 701.26) and (136.8645, 701.3725) .. (143.9848, 701.9223)
     .. controls (151.105, 702.472) and (157.184, 703.459) .. (162.565, 704.3015)
     .. controls (167.946, 705.144) and (172.629, 705.842) .. (177.7098, 705.7153)
     .. controls (182.7907, 705.5887) and (188.2693, 704.6373) .. (193.5702, 704.0467)
     .. controls (198.871, 703.456) and (203.994, 703.226) .. (209.3697, 702.3042)
     .. controls (214.7453, 701.3823) and (220.3737, 699.7687) .. (225.6435, 698.7862)
     .. controls (230.9133, 697.8037) and (235.8247, 697.4523) .. (240.7075, 697.2647)
     .. controls (245.5903, 697.077) and (250.4447, 697.053) .. (255.801, 697.3877)
     .. controls (261.1573, 697.7223) and (267.0157, 698.4157) .. (272.459, 699.58)
     .. controls (277.9023, 700.7443) and (282.9307, 702.3797) .. (288.2863, 702.7255)
     .. controls (293.642, 703.0713) and (299.325, 702.1277) .. (304.8495, 701.281)
     .. controls (310.374, 700.4343) and (315.74, 699.6847) .. (321.0332, 699.078)
     .. controls (326.3263, 698.4713) and (331.5467, 698.0077) .. (335.479, 697.618)
     .. controls (339.4113, 697.2283) and (342.0557, 696.9127) .. (349.4321, 697.5458)
     .. controls (356.8085, 698.179) and (368.917, 699.761) .. (368.985, 700.825);
  \draw
    (112.827, 771.929)
     .. controls (144.872, 771.169) and (153.1345, 773.28) .. (159.8064, 774.2593)
     .. controls (166.4783, 775.2387) and (171.5597, 775.0863) .. (176.475, 774.5402)
     .. controls (181.3903, 773.994) and (186.1397, 773.054) .. (191.5213, 772.2598)
     .. controls (196.903, 771.4657) and (202.917, 770.8173) .. (208.3208, 771.0268)
     .. controls (213.7247, 771.2363) and (218.5183, 772.3037) .. (223.8257, 772.7397)
     .. controls (229.133, 773.1757) and (234.954, 772.9803) .. (240.889, 773.4403)
     .. controls (246.824, 773.9003) and (252.873, 775.0157) .. (258.189, 775.2952)
     .. controls (263.505, 775.5747) and (268.088, 775.0183) .. (272.7532, 773.7597)
     .. controls (277.4183, 772.501) and (282.1657, 770.54) .. (287.5218, 770.17)
     .. controls (292.878, 769.8) and (298.843, 771.021) .. (305.7875, 772.2358)
     .. controls (312.732, 773.4507) and (320.656, 774.6593) .. (325.9667, 775.0257)
     .. controls (331.2773, 775.392) and (333.9747, 774.916) .. (339.5753, 773.9963)
     .. controls (345.176, 773.0765) and (353.68, 771.713) .. (368.677, 774.363);
  \draw[lightgray]
    (128, 784)
     .. controls (130.6667, 749.3333) and (130.6667, 717.3333) .. (128, 688);
  \draw[lightgray]
    (160, 784)
     .. controls (162.6667, 749.3333) and (162.6667, 717.3333) .. (160, 688);
  \draw[lightgray]
    (192, 784)
     .. controls (194.6667, 749.3333) and (194.6667, 717.3333) .. (192, 688);
  \draw[lightgray]
    (224, 784)
     .. controls (226.6667, 749.3333) and (226.6667, 717.3333) .. (224, 688);
  \draw[lightgray]
    (256, 784)
     .. controls (258.6667, 749.3333) and (258.6667, 717.3333) .. (256, 688);
  \draw[lightgray]
    (288, 784)
     .. controls (290.6667, 749.3333) and (290.6667, 717.3333) .. (288, 688);
  \draw[lightgray]
    (320, 784)
     .. controls (322.6667, 749.3333) and (322.6667, 717.3333) .. (320, 688);
  \draw[lightgray]
    (352, 784)
     .. controls (354.6667, 749.3333) and (354.6667, 717.3333) .. (352, 688);
  \pic[fill=gold]
     at (128.834, 771.781) {ipe fdisk};
  \pic[fill=gold]
     at (129.718, 752.325) {ipe fdisk};
  \pic[fill=gold]
     at (129.998, 732.826) {ipe fdisk};
  \pic[fill=gold]
     at (129.781, 718.311) {ipe fdisk};
  \pic[fill=gold]
     at (129.04, 701.682) {ipe fdisk};
  \pic[fill=gold]
     at (160.676, 774.382) {ipe fdisk};
  \pic[fill=gold]
     at (161.415, 760.536) {ipe fdisk};
  \pic[fill=gold]
     at (161.982, 738.631) {ipe fdisk};
  \pic[fill=gold]
     at (161.831, 720.193) {ipe fdisk};
  \pic[fill=gold]
     at (161.181, 704.085) {ipe fdisk};
  \pic[fill=gold]
     at (192.819, 772.073) {ipe fdisk};
  \pic[fill=gold]
     at (193.449, 759.776) {ipe fdisk};
  \pic[fill=gold]
     at (194.001, 734.315) {ipe fdisk};
  \pic[fill=gold]
     at (193.773, 718.057) {ipe fdisk};
  \pic[fill=gold]
     at (193.18, 704.091) {ipe fdisk};
  \pic[fill=gold]
     at (224.773, 772.811) {ipe fdisk};
  \pic[fill=gold]
     at (225.745, 751.393) {ipe fdisk};
  \pic[fill=gold]
     at (225.975, 728.602) {ipe fdisk};
  \pic[fill=gold]
     at (225.619, 713.413) {ipe fdisk};
  \pic[fill=gold]
     at (224.862, 698.936) {ipe fdisk};
  \pic[fill=gold]
     at (256.624, 775.192) {ipe fdisk};
  \pic[fill=gold]
     at (257.422, 760.389) {ipe fdisk};
  \pic[fill=gold]
     at (258, 734.783) {ipe fdisk};
  \pic[fill=gold]
     at (257.792, 718.766) {ipe fdisk};
  \pic[fill=gold]
     at (256.761, 697.451) {ipe fdisk};
  \pic[fill=gold]
     at (288.935, 770.109) {ipe fdisk};
  \pic[fill=gold]
     at (289.512, 758.238) {ipe fdisk};
  \pic[fill=gold]
     at (289.969, 740.073) {ipe fdisk};
  \pic[fill=gold]
     at (289.758, 717.54) {ipe fdisk};
  \pic[fill=gold]
     at (289.103, 702.768) {ipe fdisk};
  \pic[fill=gold]
     at (320.668, 774.502) {ipe fdisk};
  \pic[fill=gold]
     at (321.486, 758.838) {ipe fdisk};
  \pic[fill=gold]
     at (321.988, 737.495) {ipe fdisk};
  \pic[fill=gold]
     at (321.794, 718.811) {ipe fdisk};
  \pic[fill=gold]
     at (320.872, 699.097) {ipe fdisk};
  \pic[fill=gold]
     at (352.791, 697.867) {ipe fdisk};
  \pic[fill=gold]
     at (353.896, 723.15) {ipe fdisk};
  \pic[fill=gold]
     at (353.933, 742.872) {ipe fdisk};
  \pic[fill=gold]
     at (353.302, 763.044) {ipe fdisk};
  \pic[fill=gold]
     at (352.774, 772.817) {ipe fdisk};
\end{tikzpicture}
\caption{An illustration of the railed  annuli $\mathcal{A}_{\bar{w}}$, for every $\bar{w}\in\{0,1\}^3$, of a railed annulus $\mathcal{A}$.
For each $\bar{w}\in\{0,1\}^3$, the yellow vertices inside $\mathcal{A}_{\bar{w}}$ are the boundary vertices of ${\bf G}_{\bar{w}}^{(t)}$.}
\label{figr_nested}
\end{figure}

Given a $\bar{w}\in \{0,1\}^{r}$,
we define 
${\bf G}_{\bar{w}}^{(t)}$ to be the graph ${\bf G}_{\mathcal{A}_{\bar{w}}}^{(t)}$.
In the rest of the paper, given an $r\in\mathbb{N}$ and a $\bar{w}\in\{0,1\}^r$,
we denote by ${\bf G}_{\bar{w}}$ the $\ell$-boundaried graph ${\bf G}_{\bar{w}}^{(\ell)}$, where $\ell:= \funref{@norteamericana}(\binom{r}{2})+1$ (where $\funref{@norteamericana}$ is the first of the two functions of~\autoref{prop_combinglemma}), we use $G_{\bar{w}}$ to denote its underlying graph, and $v_{r+1},\ldots, v_{r+\ell}$ to denote its boundary vertices.
Also,
we denote by $G_{\bar{w}}^{\sf out}$ the graph $G\setminus(V(G_{\bar{w}})\setminus\{v_{r+1},\ldots, v_{r+\ell}\})$
and by ${\bf G}_{\bar{w}}^{\sf out}$ the $\ell$-boundaried graph $(G_{\bar{w}}^{\sf out}, v_{r+1},\ldots, v_{r+\ell})$.

\subsection{Combing linkages in levelings of flat annuli}
\label{subsec_combing_in_levelings}

We now formulate~\autoref{prop_combinglemma} in terms of pairings of the graphs
 ${\bf G}_{\bar{w}}$ and  ${\bf G}_{\bar{w}}^{\sf out}$.
 That is, for every linkage with terminals $v_1,\ldots,v_r$, we
 prove that another equivalent linkage can be found, being combed through the extra boundary vertices of 
${\bf G}_{\bar{w}}$ and  ${\bf G}_{\bar{w}}^{\sf out}$ (see~\autoref{lem_colomodelsrerout}).
Before presenting~\autoref{lem_colomodelsrerout} and its proof, we introduce some additional notation.

Let $G$ be a graph and let $(\mathcal{A},\mathfrak{R})$ be a well-aligned $(2^r \cdot(p-1) +1,q)$-railed annulus flatness pair of $G$, for some odd $p\in\mathbb{N}_{\geq 3}$, and some $q\in\mathbb{N}_{\geq 3}$.
Let $v_1,\ldots, v_r\in V(G)$ and let $(\bar{w},\omega_1\ldots \omega_r)$ be the trace of $(v_1,\ldots, v_r)$.
We use $\bar{v}^\bullet$ to denote the tuple $(v_1^\bullet,\ldots, v_r^\bullet)$, where for every $i\in[r]$, $v_i^\bullet = v_i$, if $\omega_i = \bullet$, and $v_i^\bullet=\mathspace$, if $\omega_i=\circ$.
Also, we use $\bar{v}^\circ$ to denote the tuple $(v_1^\circ,\ldots, v_r^\circ)$, where for every $i\in[r]$, $v_i^\circ = v_i$, if $\omega_i = \circ$, and $v_i^\circ=\mathspace$, if $\omega_i=\bullet$.
We use $({\bf G}_{\bar{w}}, \bar{v}^\bullet)$ to denote $(G_{\bar{w}},\bar{v}^\bullet,v_{r+1}, \ldots, v_{r+\ell})$ and
$({\bf G}_{\bar{w}}^{\sf out}, \bar{v}^{\circ})$ to denote $(G_{\bar{w}}^{\sf out},\bar{v}^{\circ}, v_{r+1}, \ldots, v_{r+\ell})$.
\medskip

\begin{lemma}\label{lem_colomodelsrerout}
There is a function $\newfun{@aristocracias}:\mathbb{N}\to\mathbb{N}$ such that for every $r\in\mathbb{N}$, if $G$ is a graph, $(\mathcal{A},\mathfrak{R})$ is a well-aligned $(\funref{@aristocracias}(r),\ell)$-railed annulus flatness pair of $G$,
and $(L,v_1,\ldots, v_r)\in \Models({G},v_1,\ldots,v_r)$, where $v_1,\ldots, v_r\in V(G)$,
then there exists a $(\tilde{{L}},v_1,\ldots, v_r)\in \Models({G},v_1,\ldots,v_r)$ such that
$L\equiv\tilde{L}$ and
 $(\tilde{{L}},v_1,\ldots,v_r)\in \Models({\bf G}_{\bar{w}},\bar{v}^\bullet)\oplus \Models({\bf G}_{\bar{w}}^{\sf out}, \bar{v}^\circ)$,
 where $(\bar{w},\bar{\omega}) = {\sf trace}_{(\mathcal{A},\mathfrak{R})}(v_1,\ldots, v_r)$. 
\end{lemma}

\begin{proof}
We set $\funref{@aristocracias}(r) = 2^r \cdot(p-1) +1$, where $p=\funref{@heteronomously}(\binom{r}{2})+5$.
Recall that ${\bf G}_{\bar{w}}$ and ${\bf G}_{\bar{w}}^{\sf out}$ are $(r+\ell)$-boundaried graphs where $\ell = \funref{@norteamericana}(\binom{r}{2})+1$ and whose boundary vertices are $v_{1},\ldots,v_{r+\ell}$.
Let $(L,v_1,\ldots, v_r)\in \Models({G},v_1,\ldots,v_r)$.

We set $(\bar{w},\omega_1\ldots \omega_r)={\sf trace}_{(\mathcal{A},\frR)}(v_1,\ldots, v_r)$ and we consider the $(p,q)$-railed annulus $\mathcal{A}_{\bar{w}}$ of $G$.
Observe that $v_1,\ldots, v_r\notin {\sf Influence}_\frR (\mathcal{A}_{\bar{w}})$
and therefore, by~\autoref{lem_levelingpaths}, we have that ${\sf Leveling}_{(\mathcal{A}_{\bar{w}},\frR)}(G)$
contains a linkage $\hat{L}$ that is equivalent to $L$ and it is $\ann(R_{\mathcal{A}_{\bar{w}}})$-avoiding.
We set $\Delta = \ann(R_{\mathcal{A}_{\bar{w}}})$ and we note that $R_{\mathcal{A}_{\bar{w}}}$ is a $\Delta$-embedded $(p,q)$-railed annulus
of ${\sf Leveling}_{(\mathcal{A},\frR)}(G)$.
By applying~\autoref{prop_combinglemma} for the $(p,q)$-railed annulus $R_{\mathcal{A}_{\bar{w}}}$,
for $s=3$, and for $I = [\ell]$,
we have that ${\sf Leveling}_{(\mathcal{A}_{\bar{w}},\frR)}(G)$
contains a linkage $\tilde{L}$ that is equivalent to $\hat{L}$  (and, therefore, equivalent to $L$)
and is $(3,I)$-confined in $R_{\mathcal{A}_{\bar{w}}}$.

We set $\tilde{L}^{\bullet}:=\tilde{L}\cap G_{\bar{w}}$
and 
$\tilde{L}^\circ:=\tilde{L}\cap G_{\bar{w}}^{\sf out}$.
Observe that 
$(\tilde{{L}}^\bullet,\bar{v}^\bullet, v_{r+1},\ldots, v_{r+\ell})\in \Models({\bf G}_{\bar{w}}, \bar{v}^\bullet)$ and
$(\tilde{{L}}^\circ, \bar{v}^\circ, v_{r+1},\ldots, v_{r+\ell})\in \Models({\bf G}_{\bar{w}}^{\sf out}, \bar{v}^\circ)$.
Also, the pairings $(\tilde{{L}}^\bullet, \bar{v}^\bullet, v_{r+1},\ldots, v_{r+\ell})$
and $(\tilde{{L}}^\circ, \bar{v}^\circ, v_{r+1},\ldots, v_{r+\ell})$
are $\ell$-compatible and
$$(\tilde{{L}},v_1,\ldots, v_r) =
(\tilde{{L}}^\bullet, \bar{v}^\bullet, v_{r+1},\ldots, v_{r+\ell})
\oplus
(\tilde{{L}}^\circ, \bar{v}^\circ, v_{r+1},\ldots, v_{r+\ell}).$$
Since
$(\tilde{{L}}^\bullet, \bar{v}^\bullet, v_{r+1},\ldots, v_{r+\ell})\in  \Models({\bf G}_{\bar{w}}, \bar{v}^\bullet)$
and
$(\tilde{{L}}^\circ, \bar{v}^\circ, v_{r+1},\ldots, v_{r+\ell})\in \Models({\bf G}_{\bar{w}}^{\sf out}, \bar{v}^\circ)$,
we have that
$(\tilde{{L}},v_1,\ldots, v_r)\in \Models({\bf G}_{\bar{w}}, \bar{v}^\bullet)\oplus \Models({\bf G}_{\bar{w}}^{\sf out}, \bar{v}^\circ).$
\end{proof}

\subsection{Partial signatures of tuples of vertices}\label{subsec_signatures}
We now present the definition of partial signatures.
It is a recursive definition, that in the base case captures the pattern of a boundaried graph and for every recursive step, asks for all possible partial signatures that can be obtained after fixing another (boundary) vertex or its absence.
Intuitively, the symbol ``$\mathspace$'' expresses the absence of a vertex for this entry, or, in other words, that this vertex should be picked inside some boundaried graph that should be glued to our considered boundaried graph.
In this sense, partial signatures express ``partial'' patterns (see~\autoref{obs_bound}).\medskip

Let $\tau$ be a colored-graph vocabulary and let ${\bf c}=\{{\sf c}_1,\ldots,{\sf c}_l\}$ be a collection of $l$ constant symbols.
Let $r\in\mathbb{N}$.
Recall that  $\ell = \funref{@norteamericana}(\binom{r}{2})+1$, where $\funref{@norteamericana}$ is the second function of~\autoref{prop_combinglemma}.
Let $\mathfrak{G}$ be a $\tau$-structure, let ${R}_1,\ldots,R_r\subseteq V(G)$, and let $(\mathcal{A},\mathfrak{R})$ be a well-aligned $(2^r+1,\ell)$-railed annulus flatness pair of $G$.
For every $i\in[r]$ and every $\bar{w}\in \{0,1\}^{i}$, we define
$$R^{\bar{w}}_i = \{v\in R_i\mid \textsf{stamp}_{w_1\ldots w_{i-1}}(v) =(w_i,\bullet)\}$$
and
$$V^{\bar{w}} = \{v\in V(G)\mid \textsf{stamp}_{w_1\ldots w_{i-1}}(v) =(w_i,\bullet)\}.$$

\paragraph{Partial signature.}
Let $d\in\mathbb{N}$ and $r\in[d-1]$.
Given a $\bar{w}\in\{0,1\}^{d}$ and $v_1,\ldots, v_{d}$
such that for every $i\in[d]$, $v_i\in V^{w_1\ldots w_{i}}\cup\{\mathspace\}$
we define
\begin{equation*}
\mathsf{partial\text{-}sig}^{0}_r(\mathfrak{G},\bar{R},\bar{w},v_1,\ldots, v_{d}) =  {\sf pattern}({\bf G}_{\bar{w}},\mathspace^l,v_1,\ldots, v_{d}).
\end{equation*}

Also, for each $i\in[d-r]$, every $\bar{w}\in\{0,1\}^{d-i}$ and
every $v_1,\ldots,v_{d-i}\in V(G)\cup\{\mathspace\}$ such that
for every 
$j\in[d-i]$,
$v_j\in V^{w_1\ldots w_{j}}\cup\{\mathspace\}$,
we define
\begin{eqnarray*}
\mathsf{partial\text{-}sig}^{i}_r(\mathfrak{G},\bar{R},\bar{w},v_1,\ldots, v_{d-i})
& = &
\{\mathsf{partial\text{-}sig}^{i-1}_r(\mathfrak{G},\bar{R},\bar{w}0,v_1,\ldots, v_{d-i},v)\mid v\in V^{\bar{w}0}\}\\
& &\cup\  \{\mathsf{partial\text{-}sig}^{i-1}_r(\mathfrak{G},\bar{R},\bar{w}1,v_1,\ldots, v_{d-i},v)\mid v\in V^{\bar{w}1}\}\\
& & \cup\ \{\mathsf{partial\text{-}sig}^{i-1}_r(\mathfrak{G},\bar{R},\bar{w}0,v_1,\ldots, v_{d-i},\mathspace)\}\\
& & \cup\ \{\mathsf{partial\text{-}sig}^{i-1}_r(\mathfrak{G},\bar{R},\bar{w}1,v_1,\ldots, v_{d-i},\mathspace)\}
\end{eqnarray*}

Also, for each $i\in[d-r+1,d-1]$, every $\bar{w}\in\{0,1\}^{d-i}$ and
every $v_1,\ldots,v_{d-i}\in V(G)\cup\{\mathspace\}$ such that
for every $j\in[d-i]$,
$v_j\in V^{w_1\ldots w_{j}}\cup\{\mathspace\}$,
we define
\begin{eqnarray*}
\mathsf{partial\text{-}sig}^{i}_r(\mathfrak{G},\bar{R},\bar{w},v_1,\ldots, v_{d-i})
& = &
\{\mathsf{partial\text{-}sig}^{i-1}_r(\mathfrak{G},\bar{R},\bar{w}0,v_1,\ldots, v_{d-i},v)\mid v\in R^{\bar{w}0}_{d-i+1}\}\\
& &\cup\  \{\mathsf{partial\text{-}sig}^{i-1}_r(\mathfrak{G},\bar{R},\bar{w}1,v_1,\ldots, v_{d-i},v)\mid v\in R^{\bar{w}1}_{d-i+1}\}\\
& & \cup\ \{\mathsf{partial\text{-}sig}^{i-1}_r(\mathfrak{G},\bar{R},\bar{w}0,v_1,\ldots, v_{d-i},\mathspace)\}\\
& & \cup\ \{\mathsf{partial\text{-}sig}^{i-1}_r(\mathfrak{G},\bar{R},\bar{w}1,v_1,\ldots, v_{d-i},\mathspace)\}
\end{eqnarray*}
Finally, we define
\begin{eqnarray*}
\mathsf{partial\text{-}sig}^{d}_r(\mathfrak{G},\bar{R}) 
& = &
 \{\mathsf{partial\text{-}sig}^{d-1}_r(\mathfrak{G},\bar{R},0,v)\mid v\in R^{0}_1\}\\
& &\cup\  \{\mathsf{partial\text{-}sig}^{d-1}_r(\mathfrak{G},\bar{R},1,v)\mid v\in R^{1}_1\}\\
& & \cup\ \{\mathsf{partial\text{-}sig}^{d-1}_r(\mathfrak{G},\bar{R},0,\mathspace)\}\\
& & \cup\ \{\mathsf{partial\text{-}sig}^{d-1}_r(\mathfrak{G},\bar{R},1,\mathspace)\}
\end{eqnarray*}

We use $\mathbb{P}(A)$ to denote the powerset of a set $A$ and for every $d\in\mathbb{N}$,
we use $\mathbb{P}^d(A)$ to denote the
set $\underbrace{\mathbb{P}(\mathbb{P}(\cdots\mathbb{P}}_{d\text{-times}}(A)))$.
We observe the following.

\begin{observation}\label{obs_bound}
Let $d\in\mathbb{N}$ and $r\in[d-1]$.
For every $i\in[0,d-1]$, every $\bar{w}\in\{0,1\}^{d-i}$, and every $\bar{v}\in (V(G)\cup\{\mathspace\})^{r-i}$, where
for every $j\in[d-i]$, $v_j\in V^{w_1\ldots w_{j}}\cup\{\mathspace\}$,
it holds that
$\mathsf{partial\text{-}sig}^{i}_r(\mathfrak{G},\bar{R},\bar{w},\bar{v}) \in \mathbb{P}^{i}(\mathcal{G}_{\sf pat}^{(d+\ell,h,l)})$ and $\mathsf{partial\text{-}sig}^{d}_r (\mathfrak{G},\bar{R}) \in\mathbb{P}^d(\mathcal{G}_{\sf pat}^{(d+\ell,h,l)})$.
\end{observation}

Following the above definition, we also define (global) signatures where some vertices, up to some index, are asked to belong to the corresponding $R_i$ while the rest of them can belong in $V(G)$. More formally, given some $d\in\mathbb{N}$ and some $r\in[0,d]$, for every $(v_1,\ldots,v_d)\in (V(G)\cup\{\mathspace\})^d$, we define $\mathsf{sig}^0_r(\mathfrak{G},\bar{R},v_1,\ldots,v_d)$
to be the atomic type of $(v_1,\ldots,v_d)$.
Also, for each $i\in[d-r]$, and every $v_1,\ldots,v_{d-i}\in V(G)\cup\{\mathspace\}$,
we define
\[\mathsf{sig}^i_r(\mathfrak{G},\bar{R},v_1,\ldots,v_{d-i}) = \big\{\mathsf{sig}^{i-1}_r(\mathfrak{G},\bar{R},v_1,\ldots,v_{d-i},u)\mid u\in V(G)\cup\{\mathspace\}\big\}.\]
Also, for each $i\in[d-r+1,d-1]$ and every $v_1,\ldots,v_{d-i}\in V(V(G)\cup\{\mathspace\}$,
we define
\[\mathsf{sig}^i_r(\mathfrak{G},\bar{R},v_1,\ldots,v_{d-i}) = \big\{\mathsf{sig}^{i-1}_r(\mathfrak{G},\bar{R},v_1,\ldots,v_{d-i},u)\mid u\in V(G)\cup\{\mathspace\}\big\},\]
while for each 
$i\in[d-r]$ and every $v_1,\ldots,v_{d-i}\in V(V(G)\cup\{\mathspace\}$,
we define
\[\mathsf{sig}^i_r(\mathfrak{G},\bar{R},v_1,\ldots,v_{d-i}) = \big\{\mathsf{sig}^{i-1}_r(\mathfrak{G},\bar{R},v_1,\ldots,v_{d-i},u)\mid u\in R_{d-i+1}\cup\{\mathspace\}\big\}.\]
Finally, we define
\[\mathsf{sig}^d_r(\mathfrak{G},\bar{R}) = \big\{\mathsf{sig}^{d-1}_r(\mathfrak{G},\bar{R},v)\mid v\in R_1\cup\{\mathspace\}\big\}.\]

\subsection{Exchangeability of graphs with the same partial signature}\label{subsec_exchangability}
The goal of this section is to present~\autoref{lem_equirep} and its proof.
This result states that in the presence of a railed annulus flatness pair in side an (annotated) colored graph $(\mathfrak{D},R^{\diamond}_1,\ldots,R^{\diamond}_r)$,
two (annotated) colored graphs $(\mathfrak{G},R_1,\ldots,R_r)$, 
$(\mathfrak{G}',R_1',\ldots,R_r')$
that yield the same partial signature when ``glued in the inner part'' of $(\mathfrak{D},R^{\diamond}_1,\ldots,R^{\diamond}_r)$,
satisfy the same formulas when we additionally glue another (annotated) colored graph 
$(\mathfrak{F},R^\star_1,\ldots,R^\star_r)$
in the outer part of $(\mathfrak{D},R^{\diamond}_1,\ldots,R^{\diamond}_r)$.
We now formalize the idea of ``gluing''.

\paragraph{Compatible colored graphs.}
Let $\mathfrak{G}=(G,X_1,\ldots, X_h),\mathfrak{G}'=(G',X_1',\ldots, X_h')$ be two colored graphs and let a partial function $\eta: V(G)\to V(G')$.
We say that $\mathfrak{G}$ and $\mathfrak{G}'$ are \emph{$\eta$-compatible} if for every $i\in[h]$ and every $v\in X_i$, $v\in X_i\iff \eta(v)\in X_i'$.

\paragraph{Gluing colored graphs.}
Let $\mathfrak{G},\mathfrak{D}$ be two colored graphs and let a partial function $\eta: V(G)\to V(H)$ such that $\mathfrak{G}$ and $\mathfrak{D}$ are $\eta$-compatible.
We denote by $\mathfrak{G}\oplus_\eta \mathfrak{D}$ the colored graph obtained from the disjoint union of $\mathfrak{G}$ and $\mathfrak{D}$ after identifying vertices $v\in V(G)$ and $u\in V(H)$ if $\eta(v)=u$.

\paragraph{Inner- and outer-compatibility functions.}
Let two colored graphs $\mathfrak{G},\mathfrak{H}$, let ${\bf a}$ be an apex-tuple of $\mathfrak{G}$, and let $(\mathcal{A},\mathfrak{R})$ be a
railed annulus flatness pair of $\mathfrak{G}\setminus V({\bf a})$.
Given that $\mathfrak{R}=(X_1,Y_1,X_2,Y_2,Z_1,Z_2,\Gamma,\sigma,\pi)$,
we call a partial function $\eta: V(X_1\cap Y_1)\cup V({\bf a})\to V(H)$ (resp. $\xi: V(X_2\cap Y_2)\cup V({\bf a})\to V(H)$) such that 
$\mathfrak{G}$ and $\mathfrak{H}$ are $\eta$-compatible (resp. $\xi$-compatible)
an \emph{inner-compatibility (resp. outer-compatibility) function of $\mathfrak{G}$ and $\mathfrak{H}$}.

\begin{lemma}\label{lem_equirep}
Let $\tau$ be a colored-graph vocabulary and let $r,l\in\mathbb{N}$ and let $d=\funref{@transversales}(r,l)$.
Let $(\mathfrak{D},\bar{R}^\diamond)$ be a colored graph, let ${\bf a}$ be an apex-tuple of $\mathfrak{D}$ of size $l$,
and let $(\mathcal{A},\mathfrak{R})$ be a well-aligned $(\funref{@aristocracias}(r),\ell)$-railed annulus flatness pair of $\mathfrak{D}\setminus V({\bf a})$.
Also, let $(\mathfrak{G},\bar{R}), (\mathfrak{G}',\bar{R}')$ be two colored graphs and let two inner-compatibility functions $\eta,\eta'$ of $(\mathfrak{D},\bar{R}^{\diamond})$ and $(\mathfrak{G},\bar{R})$ (resp. $(\mathfrak{G}',\bar{R}')$).
If 
\[\mathsf{partial\text{-}sig}^{d}_r({\sf ap}_{\bf c}((\mathfrak{G},\bar{R})\oplus_{\eta} (\mathfrak{D},\bar{R}^\diamond),{\bf a}))= \mathsf{partial\text{-}sig}^{d}_r({\sf ap}_{\bf c}((\mathfrak{G}',\bar{R}')\oplus_{\eta'}(\mathfrak{D},\bar{R}^\diamond),{\bf a})),
\]
then
for every $(\tau\cup\{{\sf R}_1,\ldots,{\sf R}_r\})$-structure $(\mathfrak{F},\bar{R}^\star)$, every outer-compatibility function $\xi$ of $(\mathfrak{D},\bar{R}^\diamond)$ and $(\mathfrak{F},\bar{R}^\star)$,
it holds that
$$\mathsf{sig}^r \big((\mathfrak{G},\bar{R})\oplus_{\eta} (\mathfrak{D},\bar{R}^\diamond)\oplus_{\xi} (\mathfrak{F},\bar{R}^\star)\big)=\mathsf{sig}^r \big((\mathfrak{G}',\bar{R}')\oplus_{\eta'} (\mathfrak{D},\bar{R}^\diamond)\oplus_{\xi} (\mathfrak{F},\bar{R}^\star)\big).$$
\end{lemma}

See~\autoref{fig_exchangability} for a simplified visualization of the statement of~\autoref{lem_equirep}.

\begin{figure}[ht]
\centering
\scalebox{0.385}{\begin{tikzpicture}[scale=0.8]

\begin{scope}
\begin{scope}[on background layer]
\begin{scope}
\draw[fill=blue!80!green, path fading = fade out] (0,0) circle (8.5cm);

\fill[white] (0:5.8) to [bend right =30]  (10:5.8) to [bend right =30]  (40:5.8) to [bend right =30]  (50:5.8) to [bend right =30]  (80:5.8) to [bend right =30]  (120:5.8) to [bend right =30]  (150:5.8) to [bend right =30] (160:5.8) to [bend right =30] (200:5.8) to [bend right =30] (220:5.8) to [bend right =30] (240:5.8) to [bend right =30] (250:5.8) to [bend right =30]  (280:5.8) to [bend right =30]  (320:5.8) to [bend right =30] 
(330:5.8) to [bend right =30] (0:5.8);

\node[label={\huge $(\mathfrak{F},\bar{R}^\star)$}] () at (-90:8) {};

\end{scope}

\draw[blue] (0:5.8) to [bend right =30]  (10:5.8) to [bend right =30]  (40:5.8) to [bend right =30]  (50:5.8) to [bend right =30]  (80:5.8) to [bend right =30]  (120:5.8) to [bend right =30]  (150:5.8) to [bend right =30] (160:5.8) to [bend right =30] (200:5.8) to [bend right =30] (220:5.8) to [bend right =30] (240:5.8) to [bend right =30] (250:5.8) to [bend right =30]  (280:5.8) to [bend right =30]  (320:5.8) to [bend right =30] 
(330:5.8) to [bend right =30] (0:5.8);
\end{scope}
\foreach \x in {0,10,40,50,80,120,150,160,200,220,240,250,280,320,330}{
\node[small black node] (a\x) at (\x:5.8) {};
}
\end{scope}

\begin{scope}[xshift=17cm]
\begin{scope}[on background layer]
\node[label={\Huge $\oplus_{\xi}$}] () at (-7.3,-1) {}; 

\draw[very thick,black,fill=black!10!white] (0,0) circle (5.8cm);
\draw[very thick,black,fill=white] (0,0) circle (3.4cm);

\foreach \x in {5.8,5.6,...,3.4} \draw (0,0) circle (\x cm);
\foreach \x in {0,10,..., 350} \draw (\x:5.8) -- (\x:3.4);

\draw[red, fill=red!20!white] (0:3.4) to [bend left =30]  (20:3.4) to [bend left =30]  (30:3.4) to [bend left =30]  (60:3.4) to [bend left =30]  (90:3.4) to [bend left =30]  (130:3.4) to [bend left =30]  (150:3.4) to [bend left =30] (180:3.4) to [bend left =30] (190:3.4) to [bend left =30] (210:3.4) to [bend left =30] (240:3.4) to [bend left =30] (270:3.4) to [bend left =30]  (280:3.4) to [bend left =30]  (310:3.4) to [bend left =30] 
(330:3.4) to [bend left =30] (0:3.4);
\node[label={\huge $(\mathfrak{G}',\bar{R}')$}] () at (0,0) {};

\foreach \x in {0,10,..., 350}{
\node[small black node] (h\x) at (\x:3.4) {};
\draw (\x:5.8) -- (\x:3.4);
}
\end{scope}

\node[rectangle,fill=black!10!white, path fading = middle, minimum size=0.8cm] () at (-90:4.5) {\huge $(\mathfrak{D},\bar{R}^\diamond)$};

\foreach \x in {0,10,..., 350}{
\node[small black node] (a\x) at (\x:5.8) {};
}
\end{scope}

\begin{scope}[xshift=36cm]
\begin{scope}[on background layer]
\begin{scope}
\draw[fill=blue!80!green, path fading = fade out] (0,0) circle (8.5cm);

\fill[white] (0:5.8) to [bend right =30]  (10:5.8) to [bend right =30]  (40:5.8) to [bend right =30]  (50:5.8) to [bend right =30]  (80:5.8) to [bend right =30]  (120:5.8) to [bend right =30]  (150:5.8) to [bend right =30] (160:5.8) to [bend right =30] (200:5.8) to [bend right =30] (220:5.8) to [bend right =30] (240:5.8) to [bend right =30] (250:5.8) to [bend right =30]  (280:5.8) to [bend right =30]  (320:5.8) to [bend right =30] 
(330:5.8) to [bend right =30] (0:5.8);

\node[label={\huge $(\mathfrak{F},\bar{R}^\star)$}] () at (-90:8) {};

\end{scope}

\node[label={\Huge $=$}] () at (-10.3,-0.5) {}; 

\draw[very thick,black,fill=black!10!white] (0,0) circle (5.8cm);
\draw[very thick,black,fill=white] (0,0) circle (3.4cm);

\foreach \x in {5.8,5.6,...,3.4} \draw (0,0) circle (\x cm);
\foreach \x in {0,10,..., 350} \draw (\x:5.8) -- (\x:3.4);

\draw[red, fill=red!20!white] (0:3.4) to [bend left =30]  (20:3.4) to [bend left =30]  (30:3.4) to [bend left =30]  (60:3.4) to [bend left =30]  (90:3.4) to [bend left =30]  (130:3.4) to [bend left =30]  (150:3.4) to [bend left =30] (180:3.4) to [bend left =30] (190:3.4) to [bend left =30] (210:3.4) to [bend left =30] (240:3.4) to [bend left =30] (270:3.4) to [bend left =30]  (280:3.4) to [bend left =30]  (310:3.4) to [bend left =30] 
(330:3.4) to [bend left =30] (0:3.4);
\node[label={\huge $(\mathfrak{G}',\bar{R}')$}] () at (0,0) {};

\draw[blue] (0:5.8) to [bend right =30]  (10:5.8) to [bend right =30]  (40:5.8) to [bend right =30]  (50:5.8) to [bend right =30]  (80:5.8) to [bend right =30]  (120:5.8) to [bend right =30]  (150:5.8) to [bend right =30] (160:5.8) to [bend right =30] (200:5.8) to [bend right =30] (220:5.8) to [bend right =30] (240:5.8) to [bend right =30] (250:5.8) to [bend right =30]  (280:5.8) to [bend right =30]  (320:5.8) to [bend right =30] 
(330:5.8) to [bend right =30] (0:5.8);
\foreach \x in {0,10,..., 350}{
\node[small black node] (h\x) at (\x:3.4) {};
\draw (\x:5.8) -- (\x:3.4);
}
\end{scope}

\node[rectangle,fill=black!10!white, path fading = middle, minimum size=0.8cm] () at (-90:4.5) {\huge $(\mathfrak{D},\bar{R}^\diamond)$};

\foreach \x in {0,10,..., 350}{
\node[small black node] (a\x) at (\x:5.8) {};
}
\end{scope}

\begin{scope}[yshift=20cm]
\begin{scope}[on background layer]
\begin{scope}
\draw[fill=blue!80!green, path fading = fade out] (0,0) circle (8.5cm);

\fill[white] (0:5.8) to [bend right =30]  (10:5.8) to [bend right =30]  (40:5.8) to [bend right =30]  (50:5.8) to [bend right =30]  (80:5.8) to [bend right =30]  (120:5.8) to [bend right =30]  (150:5.8) to [bend right =30] (160:5.8) to [bend right =30] (200:5.8) to [bend right =30] (220:5.8) to [bend right =30] (240:5.8) to [bend right =30] (250:5.8) to [bend right =30]  (280:5.8) to [bend right =30]  (320:5.8) to [bend right =30] 
(330:5.8) to [bend right =30] (0:5.8);

\node[label={\huge $(\mathfrak{F},\bar{R}^\star)$}] () at (-90:8) {};
\end{scope}
\draw[blue] (0:5.8) to [bend right =30]  (10:5.8) to [bend right =30]  (40:5.8) to [bend right =30]  (50:5.8) to [bend right =30]  (80:5.8) to [bend right =30]  (120:5.8) to [bend right =30]  (150:5.8) to [bend right =30] (160:5.8) to [bend right =30] (200:5.8) to [bend right =30] (220:5.8) to [bend right =30] (240:5.8) to [bend right =30] (250:5.8) to [bend right =30]  (280:5.8) to [bend right =30]  (320:5.8) to [bend right =30] 
(330:5.8) to [bend right =30] (0:5.8);
\end{scope}
\foreach \x in {0,10,40,50,80,120,150,160,200,220,240,250,280,320,330}{
\node[small black node] (a\x) at (\x:5.8) {};
}

\end{scope}

\begin{scope}[xshift=17cm, yshift=20cm]
\begin{scope}[on background layer]

\node[label={\Huge $\oplus_{\xi}$}] () at (-7.3,-1) {}; 

\draw[very thick,black,fill=black!10!white] (0,0) circle (5.8cm);
\draw[very thick,black,fill=white] (0,0) circle (3.4cm);

\foreach \x in {5.8,5.6,...,3.4} \draw (0,0) circle (\x cm);
\foreach \x in {0,10,..., 350} \draw (\x:5.8) -- (\x:3.4);

\draw[green!50!blue, fill=green!20!white] (0:3.4) to [bend left =30]  (20:3.4) to [bend left =30]  (30:3.4) to [bend left =30]  (60:3.4) to [bend left =30]  (90:3.4) to [bend left =30]  (130:3.4) to [bend left =30]  (150:3.4) to [bend left =30] (180:3.4) to [bend left =30] (190:3.4) to [bend left =30] (210:3.4) to [bend left =30] (240:3.4) to [bend left =30] (270:3.4) to [bend left =30]  (280:3.4) to [bend left =30]  (310:3.4) to [bend left =30] 
(330:3.4) to [bend left =30] (0:3.4);
\node[label={\huge $(\mathfrak{G},\bar{R})$}] () at (0,0) {};

\foreach \x in {0,10,..., 350}{
\node[small black node] (h\x) at (\x:3.4) {};
\draw (\x:5.8) -- (\x:3.4);
}
\end{scope}

\node[rectangle,fill=black!10!white, path fading = middle, minimum size=0.8cm] () at (-90:4.5) {\huge $(\mathfrak{D},\bar{R}^\diamond)$};

\foreach \x in {0,10,..., 350}{
\node[small black node] (a\x) at (\x:5.8) {};
}
\end{scope}

\begin{scope}[xshift=36cm, yshift=20cm]
\begin{scope}[on background layer]
\begin{scope}
\draw[fill=blue!80!green, path fading = fade out] (0,0) circle (8.5cm);

\fill[white] (0:5.8) to [bend right =30]  (10:5.8) to [bend right =30]  (40:5.8) to [bend right =30]  (50:5.8) to [bend right =30]  (80:5.8) to [bend right =30]  (120:5.8) to [bend right =30]  (150:5.8) to [bend right =30] (160:5.8) to [bend right =30] (200:5.8) to [bend right =30] (220:5.8) to [bend right =30] (240:5.8) to [bend right =30] (250:5.8) to [bend right =30]  (280:5.8) to [bend right =30]  (320:5.8) to [bend right =30] 
(330:5.8) to [bend right =30] (0:5.8);

\node[label={\huge $(\mathfrak{F},\bar{R}^\star)$}] () at (-90:8) {};

\end{scope}

\node[label={\Huge $=$}] () at (-10.3,-0.5) {}; 

\draw[very thick,black,fill=black!10!white] (0,0) circle (5.8cm);
\draw[very thick,black,fill=white] (0,0) circle (3.4cm);

\foreach \x in {5.8,5.6,...,3.4} \draw (0,0) circle (\x cm);
\foreach \x in {0,10,..., 350} \draw (\x:5.8) -- (\x:3.4);

\draw[green!50!blue, fill=green!20!white] (0:3.4) to [bend left =30]  (20:3.4) to [bend left =30]  (30:3.4) to [bend left =30]  (60:3.4) to [bend left =30]  (90:3.4) to [bend left =30]  (130:3.4) to [bend left =30]  (150:3.4) to [bend left =30] (180:3.4) to [bend left =30] (190:3.4) to [bend left =30] (210:3.4) to [bend left =30] (240:3.4) to [bend left =30] (270:3.4) to [bend left =30]  (280:3.4) to [bend left =30]  (310:3.4) to [bend left =30] 
(330:3.4) to [bend left =30] (0:3.4);
\node[label={\huge $(\mathfrak{G},\bar{R})$}] () at (0,0) {};

\draw[blue] (0:5.8) to [bend right =30]  (10:5.8) to [bend right =30]  (40:5.8) to [bend right =30]  (50:5.8) to [bend right =30]  (80:5.8) to [bend right =30]  (120:5.8) to [bend right =30]  (150:5.8) to [bend right =30] (160:5.8) to [bend right =30] (200:5.8) to [bend right =30] (220:5.8) to [bend right =30] (240:5.8) to [bend right =30] (250:5.8) to [bend right =30]  (280:5.8) to [bend right =30]  (320:5.8) to [bend right =30] 
(330:5.8) to [bend right =30] (0:5.8);
\foreach \x in {0,10,..., 350}{
\node[small black node] (h\x) at (\x:3.4) {};
\draw (\x:5.8) -- (\x:3.4);
}
\end{scope}

\node[rectangle,fill=black!10!white, path fading = middle, minimum size=0.8cm] () at (-90:4.5) {\huge $(\mathfrak{D},\bar{R}^\diamond)$};

\foreach \x in {0,10,..., 350}{
\node[small black node] (a\x) at (\x:5.8) {};
}
\end{scope}
\end{tikzpicture}}
\caption{A simplified visualization of the statement of \autoref{lem_equirep}. If the upper middle and the lower middle colored graphs have the same partial signature, then for every colored graph $(\mathfrak{F},\bar{R}^\star)$ (on the left) glued to them, the resulting colored graphs (the upper right and the lower right) have the same (global) signature.}
\label{fig_exchangability}
\end{figure}
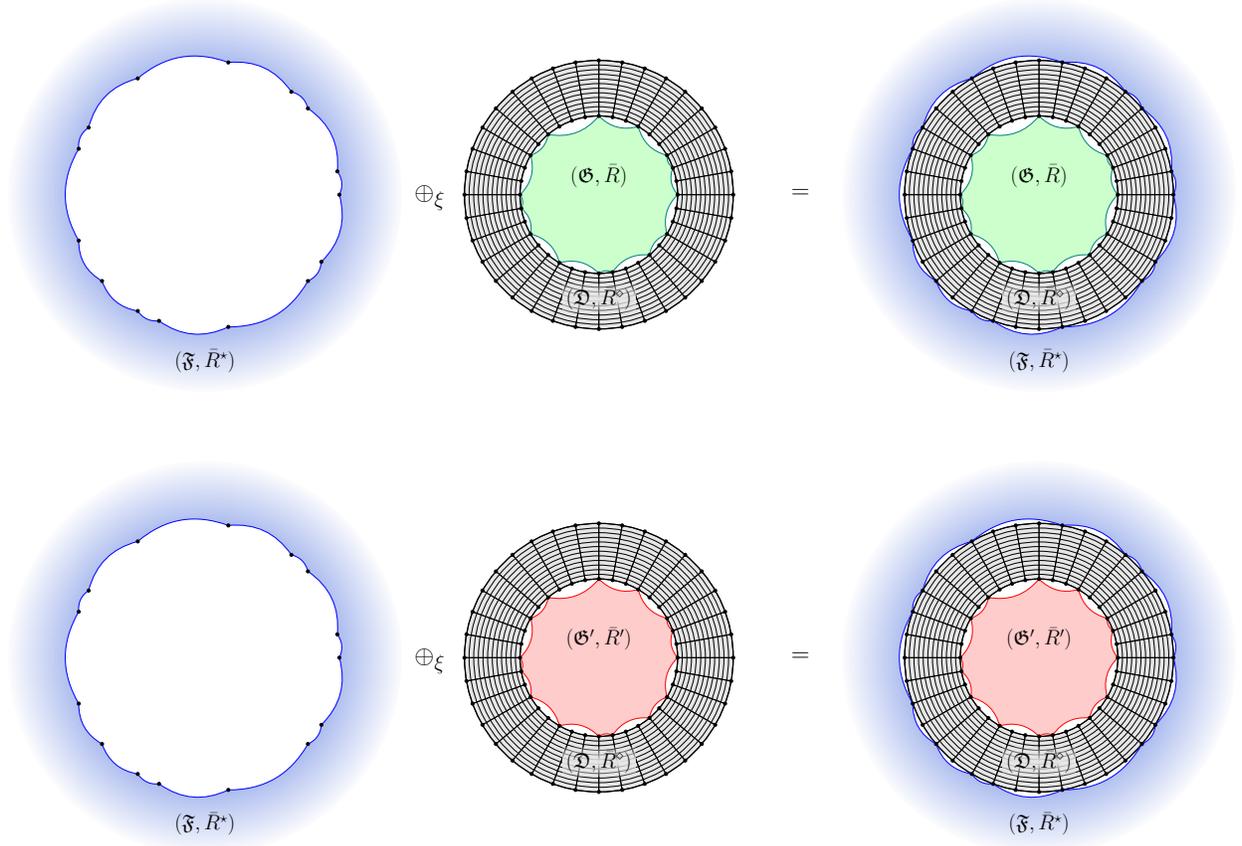

\begin{proof}
Let $\mathfrak{F}$ be a colored graph, let $\bar{R}^\star=(R^\star_1,\ldots,R^\star_r)$, where $R^\star_1,\ldots,R^\star_r\subseteq V(\mathfrak{F})$, and let $\xi$  be an outer-compatibility function of $(\mathfrak{D},\bar{R}^\diamond)$ and $(\mathfrak{F},\bar{R}^\star)$.

Given that $\mathfrak{R}=(X_1,Y_1,X_2,Y_2,Z_1,Z_2,\Gamma,\sigma,\pi)$, we set $\tilde{\mathfrak{R}}=(\tilde{X}_1,\tilde{Y}_1,\tilde{X}_2,\tilde{Y}_2,Z_1,Z_2,\Gamma,\sigma,\pi)$ and
$\tilde{\mathfrak{R}}'=(\tilde{X}_1',\tilde{Y}_1,\tilde{X}_2',\tilde{Y}_2,Z_1,Z_2,\Gamma,\sigma,\pi)$,
where
\begin{itemize}
\item $\tilde{Y}_1$ (resp. $\tilde{Y}_2$) is the vertex set obtained from the union of $Y_1$ (resp. $Y_2$) and $V(\mathfrak{G})$ after identifying the vertices in $(X_1\cap Y_1)\cup V({\bf a})$ with their images via $\eta$,
\item $\tilde{Y}_1'$ (resp. $\tilde{Y}_2'$) is the vertex set obtained from the union of $Y_1$ (resp. $Y_2$) and $V(\mathfrak{G}')$ after identifying the vertices in $(X_1\cap Y_1)\cup V({\bf a})$ with their images via $\eta'$, and
\item $\tilde{X}_1$ (resp. $\tilde{X}_2$) is the vertex set obtained from the union of $X_1$ (resp. $X_2$) and $V(\mathfrak{F})$ after identifying the vertices in $(X_2\cap Y_2)\cup V({\bf a})$ with their images via $\xi$.
\end{itemize}
Observe that $(\mathcal{A},\tilde{\mathfrak{R}})$ is a railed annulus flatness pair of $(\mathfrak{G}\oplus_{\eta} \mathfrak{D}\oplus_\xi\mathfrak{F})\setminus V({\bf a})$ 
and $(\mathcal{A},\tilde{\mathfrak{R}}')$ is a railed annulus flatness pair of  $(\mathfrak{G}'\oplus_{\eta'} \mathfrak{D}\oplus_\xi\mathfrak{F})\setminus V({\bf a})$.

We now consider the $\tau^{\langle {\bf c}\rangle}$-structures ${\sf ap}_{\bf c}(\mathfrak{G}\oplus_{\eta} \mathfrak{D}\oplus_\xi\mathfrak{F},{\bf a})$ and ${\sf ap}_{\bf c}(\mathfrak{G}'\oplus_{\eta'} \mathfrak{D}\oplus_\xi\mathfrak{F},{\bf a})$.
Note that  $(\mathcal{A},\tilde{\mathfrak{R}})$ is also a railed annulus flatness pair of ${\sf ap}_{\bf c}(\mathfrak{G}\oplus_{\eta} \mathfrak{D}\oplus_\xi\mathfrak{F},{\bf a})\setminus V({\bf a})$.
Also, since in  ${\sf ap}_{\bf c}(\mathfrak{G}\oplus_{\eta} \mathfrak{D}\oplus_\xi\mathfrak{F},{\bf a})$ there are no edges between $V({\bf a})$ and $V({\sf ap}_{\bf c}(\mathfrak{G}\oplus_{\eta} \mathfrak{D}\oplus_\xi\mathfrak{F},{\bf a})\setminus V({\bf a}))$, we can update $\tilde{\mathfrak{R}}$ by adding $V({\bf a})$ to $\tilde{X}_2$ and observe that
after this modification of $\tilde{\mathfrak{R}}$,
$(\mathcal{A},\tilde{\mathfrak{R}})$ is a railed annulus flatness pair of ${\sf ap}_{\bf c}(\mathfrak{G}\oplus_{\eta} \mathfrak{D}\oplus_\xi\mathfrak{F},{\bf a})$.
For the same reasons, we can assume that $(\mathcal{A},\tilde{\mathfrak{R}}')$ is a railed annulus flatness pair of ${\sf ap}_{\bf c}(\mathfrak{G}'\oplus_{\eta'} \mathfrak{D}\oplus_\xi\mathfrak{F},{\bf a})$.
We set
\begin{itemize}
\item $(\mathfrak{Z},\tilde{R}_1,\ldots,\tilde{R}_r):=(\mathfrak{G},\bar{R})\oplus_{\eta} (\mathfrak{D},\bar{R}^\diamond)\oplus_{\xi} (\mathfrak{F},\bar{R}^\star)$,
\item $(\mathfrak{Z}',\tilde{R}_1',\ldots,\tilde{R}_r'):=(\mathfrak{G}',\bar{R}')\oplus_{\eta'} (\mathfrak{D},\bar{R}^\diamond)\oplus_{\xi} (\mathfrak{F},\bar{R}^\star)$,
\item $(\mathfrak{H},\tilde{R}_1,\ldots,\tilde{R}_r, {\bf a}):={\sf ap}_{\bf c}(\mathfrak{Z},\tilde{R}_1,\ldots,\tilde{R}_r,{\bf a})$, and
\item $(\mathfrak{H}',\tilde{R}_1',\ldots,\tilde{R}_r', {\bf a}):={\sf ap}_{\bf c}(\mathfrak{Z}',\tilde{R}_1',\ldots,\tilde{R}_r',{\bf a})$.
\end{itemize}

Our goal is to prove that $\mathsf{sig}^r (\mathfrak{Z},\tilde{R}_1,\ldots,\tilde{R}_r)=\mathsf{sig}^r (\mathfrak{Z}',\tilde{R}_1',\ldots,\tilde{R}_r').$
To achieve this,
it will suffice to show that $\mathsf{sig}^{d}_r(\mathfrak{H},\tilde{R}_1,\ldots,\tilde{R}_r, {\bf a}) = \mathsf{sig}^{d}_r (\mathfrak{H}',\tilde{R}_1',\ldots,\tilde{R}_r', {\bf a})$.
Let us first prove that 
\[\text{$\mathsf{sig}^{d}_r(\mathfrak{H},\tilde{R}_1,\ldots,\tilde{R}_r, {\bf a}) = \mathsf{sig}^{d}_r (\mathfrak{H}',\tilde{R}_1',\ldots,\tilde{R}_r', {\bf a})$
implies $\mathsf{sig}^r (\mathfrak{Z},\tilde{R}_1,\ldots,\tilde{R}_r)=\mathsf{sig}^r (\mathfrak{Z}',\tilde{R}_1',\ldots,\tilde{R}_r').$}\]
For this, we will prove that for every quantifier-free formula $\psi(\mathsf{x}_1,\ldots,\mathsf{x}_r)$ of $\FOL[\tau+\DP]$ on $r$ free variables,
the following holds:
if $v_1,\ldots,v_r\in V(\mathfrak{Z})$ and $v_1',\ldots,v_r'\in V(\mathfrak{Z}')$ such that for every $i\in[r-1]$,
\[\mathsf{sig}^i(\mathfrak{Z},\tilde{R}_1,\ldots,\tilde{R}_r,v_1,\ldots,v_{r-i}) = \mathsf{sig}^i(\mathfrak{Z}',\tilde{R}_1',\ldots,\tilde{R}_r',v_1',\ldots,v_{r-i}'),\] 
then it holds that $(\mathfrak{Z},v_1,\ldots,v_r)\models \psi(\mathsf{x}_1,\ldots,\mathsf{x}_r) \iff (\mathfrak{Z},v_1,\ldots,v_r)\models \psi(\mathsf{x}_1,\ldots,\mathsf{x}_r)$.

Let $\psi(\mathsf{x}_1,\ldots,\mathsf{x}_r)$ be such a formula and let $v_1,\ldots,v_r\in V(\mathfrak{Z})$ and $v_1',\ldots,v_r'\in V(\mathfrak{Z}')$ such that $\mathsf{sig}^i(\mathfrak{Z},\tilde{R}_1,\ldots,\tilde{R}_r,v_1,\ldots,v_{r-i}) = \mathsf{sig}^i(\mathfrak{Z}',\tilde{R}_1',\ldots,\tilde{R}_r',v_1',\ldots,v_{r-i}')$, for every $i\in[r-1]$.
Assuming that $\mathsf{sig}^d_r(\mathfrak{H},\tilde{R}_1,\ldots,\tilde{R}_r, {\bf a}) = \mathsf{sig}^d_r (\mathfrak{H}',\tilde{R}_1',\ldots,\tilde{R}_r', {\bf a})$, it holds that
\[\mathsf{sig}^{d-r}_r(\mathfrak{H},\tilde{R}_1,\ldots,\tilde{R}_r, {\bf a}, v_1,\ldots, v_r) = \mathsf{sig}^{d-r}_r (\mathfrak{H}',\tilde{R}_1',\ldots,\tilde{R}_r', {\bf a},v_1',\ldots,v_r').\]
Therefore, by~\autoref{lemma_reducing}, we have $(\mathfrak{H}, {\bf a}, v_1,\ldots, v_r)\models \psi^l \iff (\mathfrak{H}', {\bf a},v_1',\ldots,v_r')\models \psi^l$,
where $\psi^l$ is the apex-projection of $\psi$ (see~\autoref{subsec_apex_formula}).
Also, by~\autoref{lem_interpretapex}, $(\mathfrak{Z},v_1,\ldots,v_r)\models \psi \iff (\mathfrak{H}, {\bf a}, v_1,\ldots, v_r)\models \psi^l$ and $(\mathfrak{Z},v_1,\ldots,v_r)\models \psi \iff (\mathfrak{H}', {\bf a},v_1',\ldots,v_r')\models \psi^l$.
Combining these last three logical equivalences, we get $(\mathfrak{Z},v_1,\ldots,v_r)\models \psi \iff (\mathfrak{Z},v_1,\ldots,v_r)\models \psi$.

We devote the rest of the proof to show that $\mathsf{sig}^{d}_r(\mathfrak{H},\tilde{R}_1,\ldots,\tilde{R}_r, {\bf a}) = \mathsf{sig}^{d}_r (\mathfrak{H}',\tilde{R}_1',\ldots,\tilde{R}_r', {\bf a}).$
Let $\lambda$ be an assignment of $(\mathfrak{H},\tilde{R}_1,\ldots,\tilde{R}_r)$ to a rooted tree $(T,t_0)$ and
let $\lambda'$ be an assignment of $(\mathfrak{H}',\tilde{R}_1',\ldots,\tilde{R}_r')$ to a rooted tree $(T',t_0')$.
%

To prove that $\mathsf{sig}^{d}_r(\mathfrak{H},\tilde{R}_1,\ldots,\tilde{R}_r, {\bf a}) = \mathsf{sig}^{d}_r(\mathfrak{H}',\tilde{R}_1',\ldots,\tilde{R}_r', {\bf a})$, it suffices to prove that
for every $i\in[0,d]$ and every $\gamma\in \mathbb{P}^{d-i}(\mathcal{G}_{\sf pat}^{(t,h)}),$
the two following statements hold:
\begin{itemize}
\item[(i)] For every $(t_0,\ldots,t_{i})\in\mathsf{Paths}(T)$ there exist $(t_0',\ldots,t_{i}')\in\mathsf{Paths} (T')$ such that $$\mathsf{sig}^{d-i}_r(\mathfrak{H},\tilde{R}_1,\ldots,\tilde{R}_r, {\bf a},\lambda(t_1),\ldots,\lambda(t_i)) = \mathsf{sig}^t (\mathfrak{H}',\tilde{R}_1',\ldots,\tilde{R}_r', {\bf a},\lambda'(t_1'),\ldots,\lambda'(t_i')).$$
\item[(ii)] for every $(t_0',\ldots,t_{i}')\in\mathsf{Paths}(T')$ there exist $(t_0,\ldots,t_{i})\in\mathsf{Paths} (T)$ such that $$\mathsf{sig}^{d-i}_r(\mathfrak{H},\tilde{R}_1,\ldots,\tilde{R}_r, {\bf a},\lambda(t_1),\ldots,\lambda(t_i)) = \mathsf{sig}^t (\mathfrak{H}',\tilde{R}_1',\ldots,\tilde{R}_r', {\bf a},\lambda'(t_1'),\ldots,\lambda'(t_i')).$$
\end{itemize}
We will show only a proof for (i), since the proof of (ii) will be totally symmetric to the one of (i).
In fact, we will prove the following statement, which is equivalent to (i).

\begin{claim}\label{claim_palette}
For every $i\in[0,d]$ and every  $\beta\in \mathcal{G}_{\sf pat}^{(d,h)},$
it holds that: 
For every $(t_0,\ldots,t_{i})\in\mathsf{Paths}(T)$ there is a  $(t_0',\ldots,t_{i}')\in\mathsf{Paths} (T')$
such that for every $t_d\in L(T_{t_i})$,
where $(t_0,\ldots,t_d)\in{\sf Paths}(T)$, there is a $t_d'\in L(T_{t_{i}'}')$, where $(t_0',\ldots, t_d')\in{\sf Paths}(T')$,
such that
 $$\mathsf{sig}^0_r(\mathfrak{H}, \tilde{R}_1,\ldots,\tilde{R}_r,{\bf a},\lambda(t_1),\ldots, \lambda(t_{d})) = \mathsf{sig}^0_r(\mathfrak{H}',\tilde{R}_1',\ldots,\tilde{R}_r', {\bf a},\lambda'(t_1'),\ldots, \lambda'(t_{d}')).$$
\end{claim}
\medskip

\noindent\emph{Proof of~\autoref{claim_palette}.}
Let $i\in[0,d]$.
Let $t_0,\ldots, t_i\in {\sf Paths}(T)$, let $(v_1,\ldots,v_i) = (\lambda(t_1),\ldots, \lambda(t_{i}))$.
We will prove that there is a $(t_{1}',\ldots,t_i')\in \textsf{Paths}(T')$ such that 
for every $(t_{i+1},\ldots,t_d)\in \textsf{Paths}(T),$ where $t_{i+1}\in{\sf children}_{T}(t_i)$, there is a $(t_{i+1}',\ldots,t_d')\in \textsf{Paths}(T')$, where $t_{i+1}'\in{\sf children}_{T'}(t_i')$,
such that
$$\mathsf{sig}^0_r(\mathfrak{H}, \tilde{R}_1,\ldots,\tilde{R}_r, {\bf a},\lambda(t_1),\ldots, \lambda(t_{d})) = \mathsf{sig}^0_r(\mathfrak{H}', \tilde{R}_1',\ldots,\tilde{R}_r', {\bf a},\lambda'(t_1'),\ldots, \lambda'(t_{d}')).$$
%
For every $j\in[i]$, let $(w_{j}, \omega_{j})={\sf stamp}_{w_1\ldots w_{j-1}}(v_j)$.
For every $j\in[i]$, we set
$$v_{j}^\bullet =
 \begin{cases}
 v_j, & \text{if $\omega_j=\bullet$}\\
\mathspace, & \text{if $\omega_j=\circ$}.
\end{cases}$$
Since $\mathsf{partial\text{-}sig}^{d}_r({\sf ap}_{\bf c}((\mathfrak{G},\bar{R})\oplus_{\eta} (\mathfrak{D},\bar{R}^\diamond),{\bf a}))= \mathsf{partial\text{-}sig}^{d}_r ({\sf ap}_{\bf c}((\mathfrak{G}',\bar{R}')\oplus_{\eta'}(\mathfrak{D},\bar{R}^\diamond),{\bf a})),$
there is a $t_i'\in D_i (T)$, such that
if $t_0',\ldots, t_i'\in \textsf{Paths}(T')$ and $(u_1,\ldots, u_i) = (\lambda'(t_1'),\ldots, \lambda'(t_{i}'))$,
then it holds that ${\sf trace}(v_1,\ldots, v_i) = {\sf trace}(u_1,\ldots, u_i)$ and
\begin{eqnarray}
& \mathsf{partial\text{-}sig}^{d-i}_r({\sf ap}_{\bf c}((\mathfrak{G},\bar{R})\oplus_{\eta} (\mathfrak{D},\bar{R}^\diamond), {\bf a}),\bar{w},v_1^\bullet,\ldots, v_i^\bullet)\notag\\
& =\label{eq_sig}\\
& \mathsf{partial\text{-}sig}^{d-i}_r({\sf ap}_{\bf c}((\mathfrak{G}',\bar{R}')\oplus_{\eta'}(\mathfrak{D},\bar{R}^\diamond),{\bf a}),\bar{w},u_1^\bullet,\ldots,u_i^\bullet),\notag
\end{eqnarray}
where for every $j\in[i]$, $u_j^\bullet = u_j$ if $\omega_j=\bullet$ and $u_j^\bullet = \mathspace$ if $\omega_j= \circ$.
Let $v_{i+1},\ldots, v_d\in V(\mathfrak{H})$, where for each $j\in[i+1,r]$, $v_j\in \tilde{R}_j$.
For every $j\in[i+1,d]$, let $(w_{j}, \omega_{j})={\sf stamp}_{w_1\ldots w_{j-1}}(v_j)$.
For every $j\in[i+1,d]$, we set
$$v_{j}^\bullet =
 \begin{cases}
 v_j, & \text{if $\omega_j=\bullet$}\\
\mathspace, & \text{if $\omega_j=\circ$}.
\end{cases}$$
Therefore, for every $j\in[i+1,r]$, $v_j^\bullet\in \tilde{R}^{w_1\ldots w_{j}}_j\cup\{\mathspace\}$ and for every $j\in[r+1,d]$, $v_j^\bullet\in V(\mathfrak{H})^{w_1\ldots w_{j}}\cup\{\mathspace\}$.

Following~\eqref{eq_sig},
there are $u_{i+1}^\bullet,\ldots, u_{d}^\bullet\in V({\sf ap}_{\bf c}(\mathfrak{G}\oplus_{\eta} \mathfrak{D}, {\bf a}))\cup\{\mathspace\}$ where for every $j\in[i+1,r]$,
 $u_j^\bullet\in (R_j'\cup R^\diamond_j)^{w_1\ldots w_{j}}\cup\{\mathspace\}$,
 and for every $j\in[r+1,d]$, 
  $u_j^\bullet\in  V({\sf ap}_{\bf c}(\mathfrak{G}\oplus_{\eta} \mathfrak{D}, {\bf a}))
^{w_1\ldots w_{j}}\cup\{\mathspace\}$,
such that
${\sf pattern}({\bf G}_{\bar{w}},\mathspace^l,v_{1}^\bullet,\ldots, v_d^\bullet) = {\sf pattern}({\bf G}_{\bar{w}}',\mathspace^l,u_{1}^\bullet,\ldots, u_d^\bullet).$
Then, by~\autoref{obs_translatemodelstopattern}, we have that
$$ {\sf compression}({\bf G}_{\bar{w}},\mathspace^l,v_{1}^\bullet,\ldots, v_d^\bullet)= {\sf compression}({\bf G}_{\bar{w}}',\mathspace^l,u_{1}^\bullet, \ldots,u_d^\bullet)$$
and therefore, if $I=\{j\in[d]\mid v_j^\bullet\neq\mathspace\}$ and $I' = \{j\in[d]\mid u_j^\bullet\neq\mathspace\}$, then
\begin{enumerate}
\item[(P1)] $I=I'$ and for every $j,\ell\in[d]$ $v_j^\bullet=v_\ell^\bullet$ if and only if $u_j^\bullet=u_\ell^\bullet$,
\item[(P2)] $\kappa:[d]\to [l]$ is the empty function,
\item[(P3)] for every $j\in[d]$, $\delta(j) = \delta'(j)$,
\item[(P4)] ${\sf Ind}_{G_{\bar{w}}}(I) = {\sf Ind}_{G_{\bar{w}}'}(I)$, and
\item[(P5)]  $\{\imprint({\bf L})\mid {\bf L}\in \Models({\bf G}_{\bar{w}},v_{1}^\bullet,\ldots, v_d^\bullet) \}=\{\imprint({\bf L})\mid {\bf L}\in \Models({\bf G}_{\bar{w}}',u_{1}^\bullet,\ldots, u_d^\bullet)\}$.
\end{enumerate}
We now define a sequence of vertices $u_{i+1},\ldots, u_{d}\in V(\mathfrak{H}')$, as follows:
$$\text{ for every } j\in[i+1,r], u_j =\begin{cases} u_j^\bullet,& \text{if }\omega_j=\bullet\\
v_j,& \text{if } \omega_j=\circ,\end{cases}$$
 and we aim to show that
$\mathsf{sig}^0_r(\mathfrak{H}, {\bf a},v_1,\ldots, v_d) =\mathsf{sig}^0_r(\mathfrak{H}', {\bf a},u_1,\ldots, u_{d}).$
To prove this, intuitively, we have to show that (P1)-(P5) also hold for the tuples $(v_1,\ldots,v_{d})$ and $(u_1,\ldots,u_{d})$.
Keep in mind that ${\sf trace}_{(\mathcal{A},\tilde{\frR}')}(u_1,\ldots, u_d) = {\sf trace}_{(\mathcal{A},\tilde{\frR})}(v_1,\ldots, v_d)$.

First observe that, by (P1) and by the definition of $u_j,j\in[d]$,
the partition of $[d]$ with respect to the equal vertices in $(v_1,\ldots,v_d)$ is the same as the one with respect to the equal vertices in $(u_1,\ldots,u_{d})$.

Also, observe that using the definition of $\mathfrak{H}$ and $\mathfrak{H}'$, the fact that ${\bf a}$ is an apex-tuple of $\mathfrak{D}$,
together with the properties (P2) and (P3), we have that, when modifiying the functions $\kappa,\delta$
to map the indices of $(v_1,\ldots,v_{d})$ to $[l]$ and $[h]$ (before they mapped the indices of $(v_{1}^\bullet,\ldots, v_d^\bullet)$)
and
the functions $\kappa',\delta'$ to map the indices of $(u_1,\ldots,u_{d})$
(before they mapped the indices of $(u_{1}^\bullet,\ldots, u_d^\bullet)$),
for every $j\in[d]$, $\kappa(j) = \kappa'(j)$ and for every $j\in[d]$, $\delta(j) = \delta'(j)$.

We next show that
${\sf Ind}_{\mathfrak{H}}([d]) = {\sf Ind}_{\mathfrak{H}'}([d])$.
Recall that, by (P1), $I = \{j\in[d]\mid v_j^\bullet\neq\mathspace\}= \{j\in[d]\mid u_j^\bullet\neq\mathspace\}$ and, by (P4), ${\sf Ind}_{G_{\bar{w}}}(I) = {\sf Ind}_{G_{\bar{w}}'}(I)$.
Also, since ${\sf trace}_{(\mathcal{A},\tilde{\frR}')}(u_1,\ldots, u_d) = {\sf trace}_{(\mathcal{A},\tilde{\frR})}(v_1,\ldots, v_d) = (\bar{w},\bar{\omega})$,
we have that $\{v_1,\ldots, v_d\}\cap V({\sf Influence}_{\tilde{\frR}}(\mathcal{A}_{\bar{w}}))=\emptyset$ and $\{u_1,\ldots, u_d\}\cap V({\sf Influence}_{\tilde{\frR}'}(\mathcal{A}_{\bar{w}}))$.
This implies that ${\sf Ind}_{\mathfrak{H}}(I) = {\sf Ind}_{\mathfrak{H}'}(I)$.
Also, let $J = \{j\in[d]\mid v_j^\bullet =\mathspace\} = \{j\in[d]\mid u_j^\bullet =\mathspace\}$ and observe that $[d]\setminus I = J$.
Also, observe that ${\sf Ind}_{\mathfrak{H}}(J)={\sf Ind}_{\mathfrak{H}'}(J)$ and  that ${\sf Ind}_{\mathfrak{H}}(J)$ is a subgraph of $\mathfrak{D}\oplus_{\xi}\mathfrak{F}$.
The fact that $\{v_1,\ldots, v_d\}\cap V({\sf Influence}_{\tilde{\frR}}(\mathcal{A}_{\bar{w}}))=\emptyset$ and $\{u_1,\ldots, u_d\}\cap V({\sf Influence}_{\tilde{\frR}'}(\mathcal{A}_{\bar{w}}))$ implies that there is no edge neither in $\mathfrak{H}$ nor in $\mathfrak{H}'$ between vertices indexed by $I$ and $J$.
Therefore, both ${\sf Ind}_{\mathfrak{H}}([d])$ and ${\sf Ind}_{\mathfrak{H}'}([d])$ are equal to the disjoint union of ${\sf Ind}_{\mathfrak{H}}(I)$ and ${\sf Ind}_{\mathfrak{H}}(J)$.

We conclude the proof of the claim by showing that
$$\{\imprint({\bf L})\mid {\bf L}\in \Models(\mathfrak{H},v_1,\ldots, v_d) \}
 = \{\imprint({\bf L})\mid {\bf L}\in \Models(\mathfrak{H}',u_1,\ldots,u_d)\}.
$$

Let $({L},v_1,\ldots,v_d)\in \Models(\mathfrak{H},v_1,\ldots, v_d)$.
By~\autoref{lem_colomodelsrerout}, there exists a
$(\tilde{{L}},v_1,\ldots,v_d)\in \Models(\mathfrak{H},v_1,\ldots, v_d)$ such that $L\equiv\tilde{L}$
and
$$(\tilde{{L}},v_1,\ldots,v_d)\in \Models({\bf G}_{\bar{w}},\bar{v}^\bullet)\oplus \Models({\bf G}_{\bar{w}}^{\sf out},\bar{v}^\circ).$$
Recall that, by definition, ${\sf Leveling}_{(\mathcal{A},\tilde{\frR})}(\mathfrak{H}) = {\bf G}_{\bar{w}}\oplus {\bf G}_{\bar{w}}^{\sf out}$.
Note that the vertex set of ${\bf G}_{\bar{w}}^{\sf out}$ is a subset of $V(\mathfrak{D}\oplus_{\xi}\mathfrak{F})$ and therefore,
${\sf Leveling}_{(\mathcal{A},\tilde{\frR}')}(\mathfrak{H}') = {\bf G}_{\bar{w}}'\oplus{\bf G}_{\bar{w}}^{\sf out}$.
Also, since $\bar{v}^\circ = \bar{u}^\circ$, we have that
$\Models({\bf G}_{\bar{w}}^{\sf out},\bar{v}^\circ) = \Models({\bf G}_{\bar{w}}^{\sf out},\bar{u}^\circ).$

We set $\tilde{\bf L}^\bullet = (\tilde{{L}},v_1,\ldots,v_d)\cap ({\bf G}_{\bar{w}},\bar{v}^\bullet).$
By (P5), we have that $\{\imprint({\bf L})\mid {\bf L}\in \Models({\bf G}_{\bar{w}},\bar{v}^\bullet) \}=
\{\imprint({\bf L})\mid {\bf L}\in \Models({\bf G}_{\bar{w}}',\bar{u}^\bullet)\}.$
Therefore, there is an ${\bf L}'\in \Models({\bf G}_{\bar{w}}',\bar{u}^\bullet)$ such that $\imprint(\tilde{\bf L}^\bullet) = \imprint({\bf L}')$.
This, together with the fact that $\Models({\bf G}_{\bar{w}}^{\sf out},\bar{v}^\circ) = \Models({\bf G}_{\bar{w}}^{\sf out},\bar{u}^\circ),$
implies that there exists an $({{L}}',u_1,\ldots,u_d)\in \Models(\mathfrak{H}',u_1,\ldots, u_r)$ such that 
$\imprint(\tilde{{L}},v_1,\ldots,v_d) = \imprint({{L}}',u_1,\ldots,u_d)$.
Since $L\equiv\tilde{L}$ implies that  $\imprint({L},v_1,\ldots, v_d)=\imprint(\tilde{{L}},v_1,\ldots, v_d)$, 
we have that  $\imprint({L},v_1,\ldots, v_d)=\imprint({{L}}',u_1,\ldots,u_d)$.
Therefore, we can conclude that
$\{\imprint({\bf L})\mid {\bf L}\in \Models(\mathfrak{H},v_1,\ldots,v_d)\}\subseteq\{\imprint({\bf L})\mid {\bf L}\in \Models(\mathfrak{H}',u_1,\ldots, u_{d})\}$.

To show that $\{\imprint({\bf L})\mid {\bf L}\in \Models(\mathfrak{H}',u_1,\ldots, u_d)\}\subseteq \{\imprint({\bf L})\mid {\bf L}\in \Models(\mathfrak{H},v_1,\ldots,v_d)\}$, we follow the same arguments.
We consider an $({L}',u_1,\ldots,u_d)\in \Models(\mathfrak{H}',u_1,\ldots, u_d)$, we apply~\autoref{lem_colomodelsrerout} and we obtain a $(\tilde{{L}}',u_1,\ldots,u_d)\in \Models(\mathfrak{H}',u_1,\ldots, u_d)$
such that 
$(\tilde{{L}}',u_1,\ldots,u_d)\in \Models({\bf G}_{\bar{w}}',\bar{u}^\bullet)\oplus \Models({\bf G}_{\bar{w}}^{\sf out \prime},\bar{u}^\circ)$
and $\imprint({L}',u_1,\ldots, u_d)=\imprint(\tilde{{L}}',u_1,\ldots, u_d)$.
This, by (P5), implies the existence of an ${\bf L}\in \Models({\bf G}_{\bar{w}},\bar{v}^\bullet)$ such that $\imprint(\tilde{\bf L}^{\prime \bullet}) = \imprint({\bf L})$, which, in turn, implies that there exists an
$({{L}},v_1,\ldots,v_d)\in \Models(\mathfrak{H},v_1,\ldots, v_d)$ such that 
$\imprint({{L}},v_1,\ldots,v_d) = \imprint(\tilde{{L}}',u_1,\ldots,u_d)$.

Therefore, for every $i\in[0,d]$ and every  $\beta\in \mathcal{G}_{\sf pat}^{(d,h)},$
it holds that: 
For every $(t_0,\ldots,t_{i})\in\mathsf{Paths}(T)$ there is a  $(t_0',\ldots,t_{i}')\in\mathsf{Paths} (T')$
such that for every $t_d\in L(T_{t_i})$,
where $(t_0,\ldots,t_d)\in{\sf Paths}(T)$, there is a $t_d'\in L(T_{t_{i}'}')$, where $(t_0',\ldots, t_d')\in{\sf Paths}(T')$,
such that
$\mathsf{sig}^0_r(\mathfrak{H}, {\bf a},\lambda(t_1),\ldots, \lambda(t_{d})) = \mathsf{sig}^0_r(\mathfrak{H}', {\bf a},\lambda'(t_1'),\ldots, \lambda'(t_{d}')).$
This concludes the proof of the claim.
 \hfill$\diamond$
\medskip

By~\autoref{claim_palette}, we have that
for every $i\in[0,d]$ and every $(t_0,\ldots,t_{i})\in\mathsf{Paths}(T)$ there exist $(t_0',\ldots,t_{i}')\in\mathsf{Paths} (T')$ such that
\[\mathsf{sig}^{d-i}_r (\mathfrak{H},\tilde{R}_1,\ldots,\tilde{R}_r, {\bf a},\lambda(t_1),\ldots,\lambda(t_{i})) = \mathsf{sig}^{d-i}_r(\mathfrak{H}',\tilde{R}_1',\ldots,\tilde{R}_r', {\bf a},\lambda'(t_1'),\ldots,\lambda'(t_{i}')).\]
Symmetrically, one can prove that for every $i\in[0,d]$ and every $(t_0',\ldots,t_{i}')\in\mathsf{Paths}(T')$ there exist $(t_0,\ldots,t_{i})\in\mathsf{Paths} (T)$ such that
\[\mathsf{sig}^{d-i}_r(\mathfrak{H},\tilde{R}_1,\ldots,\tilde{R}_r, {\bf a},\lambda(t_1),\ldots,\lambda(t_i)) = \mathsf{sig}^{d-i}_r (\mathfrak{H}',\tilde{R}_1',\ldots,\tilde{R}_r', {\bf a},\lambda'(t_1'),\ldots,\lambda'(t_i')).\]
These two imply that $\mathsf{sig}^{d}_r(\mathfrak{H},\tilde{R}_1,\ldots,\tilde{R}_r, {\bf a}) = \mathsf{sig}^{d}_r(\mathfrak{H}',\tilde{R}_1',\ldots,\tilde{R}_r', {\bf a})$.
\end{proof}

\section{Proof of~\autoref{thm_main}}
\label{sec_proofoftheorem}
In this section, we aim to show~\autoref{thm_main}.
In this direction, in~\autoref{subsec_sigrep}, we define {\sl representatives} of vertices,
following the notion of partial signatures defined in the previous section.
Using Courcelle's theorem, we are able to compute these representatives and therefore obtain a ``reduced'' colored graph that has the same signature as the initial one (\autoref{lem_repre}).
Using this and in the presence of a big enough flat wall, in~\autoref{subsec_reduce} we argue how to safely remove vertices from a bidimensional area of the graph where no variables of the sentence are quantified and obtain a ``reduced'' equivalent instance of the (annotated) problem.
This will allow us to apply iteratively this procedure in order to reduce the treewidth of the given colored graph.
We wrap-up the proof of~\autoref{thm_main} in~\autoref{subsec_mainproof}.

\subsection{Representatives}\label{subsec_sigrep}
Let $d\in\mathbb{N}$ and $r\in[d-1]$.
Let $(\mathfrak{G},{R}_1,\ldots,R_r)$ be an (annotated) colored graph and let $(\mathcal{A},\mathfrak{R})$ be a well-aligned $(\funref{@aristocracias}(d),\ell)$-railed annulus flatness pair of $G$.
Recall that for every $i\in[d]$ and every $w_1,\ldots, w_i\in \{0,1\}$, $R^{w_1\ldots w_i}_i = \{v\in R_i \mid \textsf{stamp}_{w_1\ldots w_{i-1}}(v) =(w_i,\bullet)\}$
and  $V^{w_1\ldots w_i} = \{v\in V(\mathfrak{G}) \mid \textsf{stamp}_{w_1\ldots w_{i-1}}(v) =(w_i,\bullet)\}.$
Given an $i\in[d]$ and $v_1,\ldots, v_{i-1}\in V(G)\cup\{\mathspace\}$ such that for every $j\in[d-i]$, $v_j\in V^{w_1\ldots w_{j}}\cup\{\mathspace\}$, 
where $(w_1,\ldots, w_{i-1},\bar{\omega}) = {\sf trace}_{(\mathcal{A},\frR)}(v_1,\ldots, v_{i-1})$,
we say that two vertices $v,v'\in V(\mathfrak{G})$ such that $v,v'\in R^{w_1\ldots w_{i-1}0}_i\cup R^{w_1\ldots w_{i-1}1}_i$ if $i\in[\min\{r,d-i\}]$, and 
$v,v'\in V^{w_1\ldots w_{i-1}0}\cup V^{w_1\ldots w_{i-1}1}$ if $i\in[\min\{r,d-i\}+1,d-i]$,
are \emph{$((v_1,\ldots,v_{i-1}),i)$-equivalent},
which we denote by $v\sim_i^{(v_1,\ldots,v_{i-1})} v'$, if $\textsf{stamp}_{w_1\ldots w_{i-1}}(v) = \textsf{stamp}_{w_1\ldots w_{i-1}}(v') = (w_i,\bullet)$
and
$$\textsf{sig}^{d-i}_r(\mathfrak{G},\bar{R},w_1\ldots w_{i-1}w_i,v_1,\ldots,v_{i-1},v) = \textsf{sig}^{d-i}_r(\mathfrak{G},\bar{R},w_1\ldots w_{i-1}w_i,v_1,\ldots, v_{i-1},v').$$

Following~\autoref{obs_pattbound} and~\autoref{obs_bound}, we can easily derive an upper bound to the number of equivalence classes of each equivalence relation above. 

\begin{observation}\label{obs_bigrep}
There is a function $\newfun{@disassociate}:\mathbb{N}\to\mathbb{N}$
such that for every $i\in[d]$ and every $v_1,\ldots, v_{i-1}\in V(G)\cup\{\mathspace\}$,
such that for every $j\in[i-1]$,
 $v_j\in V(G)^{w_1\ldots w_{j}}\cup\{\mathspace\}$, where $(w_1,\ldots, w_{i-1},\bar{\omega}) = {\sf trace}_{(\mathcal{A},\frR)}(v_1,\ldots, v_{i-1})$,
it holds that the number of equivalence classes of $\sim_i^{(v_1,\ldots, v_{i-1})}$ is at most $\funref{@disassociate}(d)$.
\end{observation}


We can now prove the main result of this subsection.

\begin{lemma}\label{lem_repre}
There is a function $\newfun{@interchangeability}:\mathbb{N}\to\mathbb{N}$ and an algorithm that,
given $d,r,{\sf tw}\in\mathbb{N}$, where $r\in[d-1]$, an $n$-vertex colored graph $\mathfrak{G}$,
an apex-tuple ${\bf a}$ of $\mathfrak{G}$ of size $l$,
a well-aligned $(\funref{@aristocracias}(d),\ell)$-railed annulus flatness pair $(\mathcal{A},\mathfrak{R})$ of $\mathfrak{G}\setminus V({\bf a})$, where ${\bf tw}(G[Y_2])\leq {\sf tw}$, a set $Z\subseteq Y_1$, and sets $R_1,\ldots,R_r\subseteq V(G)$,
outputs, in time $\mathcal{O}_{d,{\sf tw}}(n)$, sets $R_1',\ldots,R_r'\subseteq V(\mathfrak{G})$ such that for every $i\in[r]$,
$R_i'\subseteq R_i$, $|R_i'\cap Z|\leq \funref{@interchangeability}(d)$, and
$$\mathsf{partial\text{-}sig}^d_r({\sf ap}_{\bf c}((\mathfrak{G},R_1,\ldots,R_r),{\bf a})) = \mathsf{partial\text{-}sig}^d_r ({\sf ap}_{\bf c}((\mathfrak{G},R_1',\ldots,R_r'),{\bf a})).$$
\end{lemma}

\begin{proof}
We set $\funref{@interchangeability}(d) = \funref{@disassociate}(d)^{d+1}$.
First observe that $\mathsf{partial\text{-}sig}^d_r$ can be expressed in {\sf MSOL}.
Assume that we have computed, for some $i\geq 1$, sets $R_1',\ldots,R_{i-1}'$ as claimed.
We show how to compute $R_i'$.
The fact that ${\bf tw}(G[Y_2])\leq {\sf tw}$ and that for every $\bar{w}\in\{0,1\}^d$, $V({\bf G}_{\bar{w}})$ is a subset of $X_2$ implies that, by using Courcelle's theorem, we can compute,
in time $\mathcal{O}_{d,{\sf tw}}(n)$, a subset $R_i'$ of $R_i$ such that $Y_2\cap X_2\subseteq R_i'$ and
for every $v_1\ldots,v_{i-1}\in V(G)$,
where $v_j\in R_j',j\in [i-1]$,
it holds that for every $v$ in $R_i'$ such that $v\in Z$, if $v\sim_i^{(v_1,\ldots, v_{i-1})}u$, for some $u\in R_i$, then $u=v$.
By keeping one $v$ for every $v_1\ldots,v_{i-1}\in V(G)$,
where $v_j\in R_j',j\in [i-1]$, we obtain a set $R_i'$ as claimed.
Intuitively, we obtain $R_i'$ from $R_i$ by keeping only one representative from each equivalence class (that itself can be expressed in \textsf{MSOL}) inside $Z$.
%
By~\autoref{obs_bigrep}, it follows that $|R_i'\cap Z|\leq \funref{@interchangeability}(d)$.
\end{proof}

To conclude this subsection, we next show how to combine~\autoref{lem_repre} and~\autoref{lem_equirep} in order to compute an equivalent colored graph with reduced annotation.
Let two colored graphs $\mathfrak{G},\mathfrak{H}$, let ${\bf a}$ be an apex-tuple of $\mathfrak{G}$, and let $(W,\mathfrak{R})$ be a flatness pair of $\mathfrak{G}\setminus V({\bf a})$.
Given that $\mathfrak{R}=(X,Y,P,C,\Gamma,\sigma,\pi)$,
we call a partial function $\xi: V(X\cap Y)\cup V({\bf a})\to V(\mathfrak{H})$ such that 
$\mathfrak{G}$ and $\mathfrak{H}$ are $\xi$-compatible 
a \emph{compatibility function of $\mathfrak{G}$ and $\mathfrak{H}$}.

\begin{lemma}
\label{lem_new_ext_rep}
Let $\tau$ be a colored-graph vocabulary.
There is a function $\newfun{samknkal}:\mathbb{N}^3\to\mathbb{N}$
and an algorithm that, given
\begin{itemize}
\item $r,l,q\in\mathbb{N}$,
\item a $\tau$-structure $\mathfrak{G}$,
\item an apex-tuple ${\bf a}$ of $\mathfrak{G}$ of size $l$,
\item a regular flatness pair $(W,\mathfrak{R})$ of $\mathfrak{G}\setminus V({\bf a})$
of height at least $\funref{samknkal}(r,l,q)$ whose compass has treewidth at most ${\sf tw}$,
and
\item sets  $R_1,\ldots,R_r\subseteq V(\mathfrak{G})$,
\end{itemize}
outputs, in time $\mathcal{O}_{r,l,q,{\sf tw}}(n)$, sets $R_1',\ldots,R_r'\subseteq V(\mathfrak{G})$ such that for every $i\in[r]$, $R_i'\subseteq R_i$, and a 
flatness pair $(\tilde{W}',\tilde{\mathfrak{R}}')$ of $G\setminus V({\bf a})$
that is a $W'$-tilt of $(W,\mathfrak{R})$ for some $q$-subwall $W'$ of $W$
such thatfor every $i\in[r]$, $R_i'\cap V({\sf Compass}_{\tilde{\mathfrak{R}}'}(\tilde{W}')) = \emptyset$ and
for every $(\tau\cup\{\mathsf{R}_1,\ldots,{\sf R}_r\})$-structure $(\mathfrak{F},\bar{R}^\star)$
and every compatibility function $\xi$ of $(\mathfrak{G},R_1,\ldots,R_r)$ and $(\mathfrak{F},\bar{R}^\star)$, it holds that
$$\mathsf{sig}^r((\mathfrak{F},\bar{R}^\star)\oplus_\xi (\mathfrak{G},R_1,\ldots,R_r)) = \mathsf{sig}^r((\mathfrak{F},\bar{R}^\star)\oplus_\xi (\mathfrak{G},R_1',\ldots,R_r')).$$
\end{lemma}

\begin{proof}
We set and let $d=\funref{@transversales}(r,l)$, $p=\funref{@aristocracias}(d)$, $y=\ell$, $g=\funref{label_internalization}(\funref{@interchangeability}(d)+1,q)$, and $\funref{samknkal}(d,y) = {\sf odd}(2p+\max\{g,\frac{y}{4}-1\})$.
We first apply ~\autoref{label_simultaneously}, and obtain a $(p,y)$-railed annulus $\mathcal{A}$ of $\mathfrak{G}\setminus V({\bf a})$.
We set  $(X,Y,P,C,\Gamma,\pi,\sigma)=\frR$.
Let $\Sigma$ be the union of all $W^{(g+1)}$-internal cells of $\frR$.
By setting $Y_2:=Y$, $X_2:=X$, $Y_1=\Sigma\cup\cupall\{V(\sigma)\mid \text{$c$ is a $W^{(g+1)}$-marginal cell of $\frR$}\}$, and $X_1 = V(G)\setminus \Sigma$,
we obtain a tuple $\frR' =(X_1,Y_1,X_2,Y_2,Z_1,Z_2,\Gamma',\pi',\sigma')$, where $\Gamma',\pi',\sigma'$ are the restrictions of $\Gamma,\pi,\sigma$, respectively to $\Gamma\setminus \Sigma$.

Observe that since $(W,\mathfrak{R})$ is a flatness pair of $\mathfrak{G}\setminus V({\bf a})$,
then $(\mathcal{A},\frR')$ is a railed annulus flatness pair of $\mathfrak{G}\setminus V({\bf a})$.
Also, since $(W,\mathfrak{R})$ is regular, it is also well-aligned (\autoref{prop_wellaligned}), which implies that $(\mathcal{A},\frR')$ is also well-aligned.
Moreover, since ${\sf Compass}_\frR (W)$ has treewidth at most ${\sf tw}$, we have that
 ${\bf tw}(G[Y_2])\leq {\sf tw}$.
Therefore, by applying the algorithm of~\autoref{lem_repre},
we find, in time $\mathcal{O}_{d,{\sf tw}}(n)$,
sets $R_1',\ldots,R_r'\subseteq V(\mathfrak{G})$ such that for every $i\in[r]$,
$R_i'\subseteq R_i$ and $|R_i'\cap Y_1|\leq \funref{@interchangeability}(d)$, and
$\mathsf{partial\text{-}sig}^d_r({\sf ap}_{\bf c}((\mathfrak{G},R_1,\ldots,R_r),{\bf a})) = \mathsf{partial\text{-}sig}^d_r ({\sf ap}_{\bf c}((\mathfrak{G},R_1',\ldots,R_r'),{\bf a})).$
By using~\autoref{label_stereotypical}, we can compute, in linear time, a $q$-subwall $W'$ of $W$ such that $V({\sf Influence}_\frR(W'))\subseteq Y_1$ and  $R_i' \cap V({\sf Influence}_\frR(W'))= \emptyset$ for every $i\in[r]$.
Then, by applying the algorithm of~\autoref{label_proporcionada} we compute, in linear time,
a $W'$-tilt $(\tilde{W}',\tilde{\mathfrak{R}}')$ of $(W,\frR)$ such that $R_i'\cap V({\sf Compass}_{\tilde{\mathfrak{R}}'}(\tilde{W}')) = \emptyset$, for every $i\in [r]$.

Since \[\mathsf{partial\text{-}sig}^d_r({\sf ap}_{\bf c}((\mathfrak{G},R_1,\ldots,R_r),{\bf a})) = \mathsf{partial\text{-}sig}^d_r ({\sf ap}_{\bf c}((\mathfrak{G},R_1',\ldots,R_r'),{\bf a})),\]
by \autoref{lem_equirep}, we have that for every $(\tau\cup\{{\sf R}_1,\ldots,{\sf R}_r\})$-structure $(\mathfrak{F},\bar{R}^\star)$, every compatibility function $\xi$ of $(\mathfrak{G},R_1,\ldots,R_r)$ and $(\mathfrak{F},\bar{R}^\star)$, it holds that
$\mathsf{sig}^r((\mathfrak{F},\bar{R}^\star)\oplus_\xi (\mathfrak{G},R_1,\ldots,R_r)) = \mathsf{sig}^r((\mathfrak{F},\bar{R}^\star)\oplus_\xi (\mathfrak{G},R_1',\ldots, R_r')).$
\end{proof}

\subsection{Reducing the instance}\label{subsec_reduce}

In this subsection we describe how to remove problem-irrelevant vertices inside a ``big enough'' annotation-irrelevant bidimensional area of our instance.
In order to achieve this, we have to argue that a linkage whose terminals are not intersecting a big enough bidimensional area can be rerouted away from some central part of this area.
This is guaranteed by the following version of the {\sl Unique Linkage Theorem}
that can be derived from~\cite[Theorem 23]{BasteST19hittIV} (see also~\cite{RobertsonS09XXI,RobertsonSGM22,KawarabayashiW10asho,Mazoit13asin,AdlerKKLST17irre}).

\begin{proposition}\label{label_pretendientes}
There exists a function $\newfun{label_encompassing}: \mathbb{N}^3\to\mathbb{N}$
such that,
for every $l,z,k\in\mathbb{N}$,
if $G$ is a graph, ${\bf a}$ is an apex-tuple of $G$ of size $l$,
if $(W,\mathfrak{R})$ is flatness pair of $G\setminus V({\bf a})$ of height $\funref{label_encompassing}(l,z,k)$
then for every linkage $L$ of size at most $k$ such that
$T(L)\cap V({\sf Compass}_{\frR}(W))  = \emptyset$,
there is a linkage $L'$ such that $L\equiv L'$ and $L\cap  V({\sf Compass}_{\tilde{\frR}'}(\tilde{W}')) = \emptyset$,
where $(\tilde{W}',\tilde{\mathfrak{R}}')$ is a $W^{(z)}$-tilt of $(W,\mathfrak{R}).$
\end{proposition}

Using \autoref{label_pretendientes}, we can easily prove the following result that allows us to remove problem-irrelevant vertices inside a ``big enough'' annotation-irrelevant bidimensional area of our instance.
\begin{lemma}\label{lem_removingirr}
Let $\tau$ be a colored-graph vocabulary.
There is a function $\newfun{@psychotechnics}:\mathbb{N}^2\to\mathbb{N}$ such that, if
\begin{itemize}
\item $r,l,z\in\mathbb{N}$,
\item $\mathfrak{G}$ is a $\tau$-structure,
\item ${\bf a}$ is an apex-tuple of $\mathfrak{G}$ of size $l$,
\item $(W,\mathfrak{R})$ is a flatness pair of $\mathfrak{G}\setminus V({\bf a})$ of height $\funref{@psychotechnics}(r,l)+z$,
and
\item sets $R_1,\ldots,R_r\subseteq V(\mathfrak{G})$, where for every $i\in[r]$, $R_i\cap V({\sf Compass}_\frR (W))=\emptyset$,
\end{itemize}
then for every flatness pair $(\tilde{W}',\tilde{\frR}')$ of $\mathfrak{G}\setminus V({\bf a})$ that is a $W'$-tilt of $(W,\frR)$ for some $\funref{@psychotechnics}(q,l)$-internal subwall $W'$ of $W$ of height $z$,
for every $(\tau\cup\{{\sf R}_1,\ldots,{\sf R}_r\})$-structure $(\mathfrak{F},\bar{R}^\star)$, every compatibility function $\xi$ of $(\mathfrak{G},R_1,\ldots,R_r)$ and $(\mathfrak{F},\bar{R}^\star)$,
and every $Y\subseteq V({\sf Compass}_{\tilde{\mathfrak{R}}'}(\tilde{W}'))$, it holds that
$$\mathsf{sig}^r((\mathfrak{F},\bar{R}^\star)\oplus_\xi (\mathfrak{G},R_1,\ldots,R_r)) = \mathsf{sig}^r((\mathfrak{F},\bar{R}^\star)\oplus_\xi (\mathfrak{G}\setminus Y,R_1,\ldots,R_r)).$$
\end{lemma}

\begin{proof}
We set $\funref{@psychotechnics}(r,l)=\funref{label_encompassing}(l,z,r).$
Let $\mathfrak{G}$ be a $\tau$-structure, let ${\bf a}$ be an apex-tuple of $\mathfrak{G}$ of size $l$, let
$(W,\mathfrak{R})$ be regular flatness pair of $\mathfrak{G}\setminus V({\bf a})$ of height $\funref{@psychotechnics}(r,l)+z$,
and $R_1,\ldots,R_r\subseteq V(\mathfrak{G})$, where $R_i\cap V({\sf Compass}_\frR (W))=\emptyset$ for every $i\in[r]$.
Also, let $(\tilde{W}',\tilde{\frR}')$ be a flatness pair of $\mathfrak{G}\setminus V({\bf a})$
that is a $W'$-tilt of $(W,\frR)$ for some $\funref{@psychotechnics}(q,l)$-internal subwall $W'$ of $W$ of height $z$.
Let also
$Y\subseteq V({\sf Compass}_{\tilde{\mathfrak{R}}'}(\tilde{W}'))$.
Observe that,
for every $k\in \lfloor r/2\rfloor$ and every $s_1,t_1\ldots,s_k,t_k\in R$,
by \autoref{label_pretendientes}, if
$G$ is the Gaifman graph of $(\mathfrak{G},R)$, then
$(G,s_1,t_1,\ldots,s_k,t_k)\models {\sf dp}_k({\sf x}_1,{\sf y}_1,\ldots,{\sf x}_k,{\sf y}_k)$ if and only if $(G\setminus Y,s_1,t_1,\ldots,s_k,t_k)\models {\sf dp}_k({\sf x}_1,{\sf y}_1,\ldots,{\sf x}_k,{\sf y}_k)$.
This implies that $\mathsf{sig}^r((\mathfrak{F},\bar{R}^\star)\oplus_\xi (\mathfrak{G},R_1,\ldots,R_r)) = \mathsf{sig}^r((\mathfrak{F},\bar{R}^\star)\oplus_\xi (\mathfrak{G}\setminus Y,R_1,\ldots,R_r)).$
\end{proof}

By applying~\autoref{lem_new_ext_rep} and then~\autoref{lem_removingirr}, we get the following:

\begin{lemma}\label{corol_redu}
Let $\tau$ be a colored-graph vocabulary.
There is a function $\newfun{@cor_individuated}:\mathbb{N}^3\to\mathbb{N}$
and an algorithm that, given $r,l,q\in\mathbb{N}$,
a $\tau$-structure $\mathfrak{G}$, an apex-tuple ${\bf a}$ of $\mathfrak{G}$ of size $l$, a regular flatness pair $(W,\mathfrak{R})$ of $\mathfrak{G}\setminus V({\bf a})$ of height $h\geq \funref{@cor_individuated}(r,l,q)$ whose compass has treewidth at most ${\sf tw}$,
and sets $R_1,\ldots,R_r\subseteq V(G)$,
outputs, in time $\mathcal{O}_{r,l,q,{\sf tw}}(n)$, sets $R_1',\ldots,R_r'\subseteq V(\mathfrak{G})$ and a flatness pair $(\tilde{W}',\tilde{\frR}')$ of $\mathfrak{G}\setminus V({\bf a})$ that is a $W'$-tilt of $(W,\frR)$ for some $q$-subwall $W'$ of $W$ such that for every $i\in[r]$, $R_i'\subseteq R_i$ and $R_i'\cap V({\sf Compass}_{\tilde{\frR}'}(\tilde{W}'))=\emptyset$, 
and for every $(\tau\cup\{{\sf R}_1,\ldots,{\sf R}_r\})$-structure $(\mathfrak{F},\bar{R}^\star)$, every compatibility function $\xi$ of $(\mathfrak{G},R_1,\ldots,R_r)$ and $(\mathfrak{F},\bar{R}^\star)$, and every $Y\subseteq V({\sf Compass}_{\tilde{\frR'}}(\tilde{W}'))$, it holds that 
$$\mathsf{sig}^r((\mathfrak{F},\bar{R}^\star)\oplus_\xi (\mathfrak{G},R_1,\ldots,R_r)) = \mathsf{sig}^r((\mathfrak{F},\bar{R}^\star)\oplus_\xi (\mathfrak{G}\setminus Y,R_1',\ldots,R_r')).$$
\end{lemma}

\subsection{Proof of~\autoref{thm_main}}
\label{subsec_mainproof}
We conclude this section by presenting the proof of~\autoref{thm_main}.

\begin{proof}[Proof of~\autoref{thm_main}]
Given a sentence $\varphi\in \FOL[\tau+\DP]$ of quantifier rank $q$, we set
\begin{eqnarray*}
c&:=&\hw(G),\\
l&:=&\funref{@connaissions}(c)\mbox{ where $\funref{@connaissions}$ is the function of~\autoref{prop_flatwallbdtw}}, and\\
r&:= &\funref{@cor_individuated}(q,l,3).
\end{eqnarray*}
Our algorithm consists of three steps.
\medskip

\noindent{\bf Step 1}:
Run the algorithm of~\autoref{prop_flatwallbdtw} for $G,$ $r,$ and $c.$
This algorithm outputs, in linear time, either
a tree decomposition of $G$ of width at most $\funref{@resplandecientes}(c)\cdot r,$ or
a set $A\subseteq V(G),$  where $|A|\leq l,$ a regular flatness pair $(W,\mathfrak{R})$ of $G\setminus A$ of height $r,$
and a tree decomposition of ${\sf compass}_{\mathfrak{R}}(W)$ of width at most $\funref{@resplandecientes}(c)\cdot r.$
In the first possible output,
i.e., a tree decomposition of $G$ of width at most $\funref{@resplandecientes}(c)\cdot r,$
 proceed to Step 3.
In the second possible output,
proceed to Step 2.
\medskip

\noindent{\bf Step 2}:
We run the algorithm of~\autoref{corol_redu} for $q,l,3,$ $\mathfrak{G},$ $R_1,\ldots,R_q$ ${\bf a},$ and $(W,\mathfrak{R}),$ and we obtain, in linear time, 
sets $R_1',\ldots,R_q'\subseteq V(G)$
and a flatness
pair $(\tilde{W}',\tilde{\mathfrak{R}}')$ of $G\setminus V({\bf a})$ that is a $W'$-tilt of $(W,\frR)$ for some subwall $W'$ of $W$ of height $3$ such that for every $i\in [q]$,
$V({\sf compass}_{\tilde{\mathfrak{R}}'}(\tilde{W}'))\cap R_i'=\emptyset$ and $R_i'\subseteq R_i$,
and
$\mathsf{sig}^q(\mathfrak{G},R_1,\ldots,R_q)= \mathsf{sig}^q(\mathfrak{G}\setminus V({\sf compass}_{\tilde{\mathfrak{R}}'}(\tilde{W}')),R_1',\ldots,R_q').$
Then, we set $\mathfrak{G}:=\mathfrak{G}\setminus V({\sf compass}_{\tilde{\mathfrak{R}}'}(\tilde{W}')),$ for every $i\in[q]$, $R_i:= R_i',$ and we run again Step 1.
\medskip

\noindent{\bf Step 3}:
Given a tree decomposition of $G$ of width at most $\funref{@resplandecientes}(c)\cdot r,$ and since
 $\mathsf{sig}^q$ is expressible in $\MSOL[\tau],$ by using Courcelle's theorem, in linear time we can compute $\mathsf{sig}^q(\mathfrak{G},R_1,\ldots,R_q)$ and check the existence of a $\varphi$-spanning subtree of the tree obtained by $\mathsf{sig}^q(\mathfrak{G},R_1,\ldots,R_q)$ and therefore decide whether $\mathfrak{G}\models \varphi$.
\medskip

Observe that the first and the second step of the algorithm are executed in linear time and they can be repeated no more than a linear number of times.
Therefore, the overall algorithm runs in quadratic time, as claimed.
\end{proof}

%

\section{\texorpdfstring{Logic for $s$-scattered paths}{Logic for s-scattered paths}}
\label{sec_logicscattered}
In this section we define a class of extensions of $\FOL$, in the same spirit as the disjoint paths logic.

\subsection{\texorpdfstring{Definition of $s$-scattered paths logic}{Definition of s-scattered paths logic}}\label{subsec_scattereddefinition}
Let $r\in\mathbb{N}$.
We define the $2k$-ary predicate $s\text{-}{\sf sdp}_k({\sf x}_1, {\sf y}_1, \ldots, {\sf x}_k, {\sf y}_k)$, which evaluates true in a $\tau$-structure $\mathfrak{G}$  if and only if there
are paths $P_1,\ldots, P_k$ of $(V(\mathfrak{G}),{\sf E}^{\mathfrak{G}})$ of length at least $2$ between (the interpretations of) ${\sf x}_i$ and ${\sf y}_i$ for all $i\in[k]$ such that for every $i,j\in[k]$, $j\neq i$, $V(P_i)\cap N_{(V(\mathfrak{G}),{\sf E}^{\mathfrak{G}})}^{(\leq s)}(V(P_j))=\emptyset$.
We let $\tau + s\text{-}{\sf sdp} := \tau \cup \{i\text{-}{\sf sdp}_k\mid k\geq 1, i\leq s\}$, where each
$s\text{-}{\sf sdp}_k$ is a $2k$-ary relation symbol.
We use $s\text{-}{\sf sdp}$ instead of $s\text{-}{\sf sdp}_k$ when $k$ is clear from the context.
It is easy to see that for every $k\in\mathbb{N}$ and every $x_1,y_1,\ldots,x_k,y_k\in V(G)$,
$$G\models 0\text{-}{\sf sdp}_k (x_1,y_1, \ldots, x_k,y_k)\iff
G\models\DP_k (x_1,y_1, \ldots, x_k,y_k).$$
Therefore, we can observe the following.
\begin{observation}
For every colored-graph vocabulary $\tau$, it holds that
$\{\Mod(\varphi)\mid \varphi\in \FOL[\tau+\DP]\} = \{\Mod(\varphi)\mid \varphi\in \FOL[\tau+ 0\text{-}{\sf sdp}]\}.$
\end{observation}

Our main result is the following:

\begin{theorem}\label{thm_induced}
For every colored-graph vocabulary $\tau$, every $s\in\mathbb{N}$, and every sentence $\varphi\in\FOL[\tau+s\text{-}{\sf sdp}]$, there exists an algorithm that, given a $\tau$-structure $G$ of size $n$, outputs whether $G\models \varphi$ in time $\mathcal{O}_{|\varphi|,k}(n^2)$, where $k$ is the Euler genus of $G$.
\end{theorem}

\subsection{Proof of~\autoref{thm_induced}}\label{subsec_sketchscattered}
In this subsection, we sketch how to prove~\autoref{thm_induced}.
Our strategy is the same as for the proof of~\autoref{thm_main} and in what follows we discuss how to recreate definitions and results of the previous sections in the setting of $s$-scattered paths.
We use $s\text{-}\mathsf{sig}$ to denote the signature obtained from $\mathsf{sig}$ where we replace the atomic types concerning disjoint paths predicates with scattered disjoint path predicates.

First of all, we stress that since we work with graphs of bounded Euler genus,
we can avoid the use of the framework of flat walls. In fact,
using~\cite[Lemma 6.2]{GolovachST20hittarxiv} (see also~\cite{BodlaenderDDFLP16,DemaineFHT05sube,DemaineHT04theb,FominGT11cont,DemaineH08line,GeelenRS04embe,CabelloVL16find,DemaineFHT05sube}), we can find a disk-embedded wall of bounded treewidth compass in a surface-embedded of large enough treewidth.
Before presenting this result,
we give some additional definitions.
Given a closed disk $\Delta$ and an integer $q\in\mathbb{N}_{\geq 3}$,
a \emph{$\Delta$-embedded $q$-wall} of $G$ is a $q$-wall $W$ that is embedded on $\Delta$ and whose perimeter is the boundary of $\Delta$. The \emph{compass} of $W$, denoted by ${\sf Compass}(W)$, is the graph $G\cap \Delta$.

\begin{proposition}\label{find_wall_eulergen}
There exists a constant $\newcon{dsdfsffds}$ and an algorithm that given
an $n$-vertex graph $G$ of Euler genus at most $g$ and an integer $q\in \mathbb{N}_{\geq 3}$,
outputs either a closed disk $\Delta$ and a $\Delta$-embedded $q$-wall $W$ of $G$ whose
compass has treewidth at most $\conref{dsdfsffds}\cdot q$ or  
a tree decomposition of $G$ of width at most $\conref{dsdfsffds}\cdot q$.
Moreover, this algorithm runs in $2^{\mathcal{O}_g(q^2)}\cdot n$ time.
\end{proposition}

\paragraph{Patterns for $s$-scattered linkages.}
Next step is to define \emph{patterns} that encode $s$-scattered linkages.
Given an $s\in\mathbb{N}$,
we say that a linkage $L$ of $G$ is \emph{$s$-scattered} if for every $v\in V(L)$ it holds that
$N_G^{(\leq s)}(v)\cap (V(L)\setminus V(C_v)) = \emptyset,$
where $C_v$ is the connected component of $L$ that contains $v$.
Observe that, since $N_G^{(\leq 0)}(v)=\{v\}$, every linkage is $0$-scattered.

To define \emph{patterns} that encode $s$-scattered linkages, the definition of a pattern of a boundaried colored graph in~\autoref{subsec_patt} has to be modified as follows:
the collection $\mathcal{H}^P$ is defined as all graphs $(I,E)$, where $E\subseteq I\times I$, such that
$G$ contains $s$-scattered paths of length at least two between the vertices $v_i,v_j$ for all $\{i,j\}\in E$.

Also, as in~\autoref{subsec_exprpatt}, we define the \emph{pattern} of a quantifier-free formula in $\FOL[\tau+s\text{-}{\sf sdp}]$ by modifying the corresponding definition for formulas in $\FOLDP$.
We do this by replacing, in both definitions of full clauses and in the definition of the collection $\mathcal{H}_c^P$, every appearance of the atomic formula $\DP$ with the atomic formula $s\text{-}{\sf sdp}$.

After the above modifications, \autoref{obs_grleafs} holds also for sentences
$\varphi \in \FOL[\tau+s\text{-}{\sf sdp}]$.
Based on this observation (for $s$-scattered linkages),
we can then prove~\autoref{lemma_reducing} for any sentence $\varphi \in \FOL[\tau+s\text{-}{\sf sdp}]$.

\paragraph{Routing $s$-scattered linkages through railed annuli.}
Having~\autoref{lemma_reducing} in hand, our next goal is to prove~\autoref{lem_colomodelsrerout} for the case of $s$-scattered linkages. Recall that~\autoref{lem_colomodelsrerout} intuitively states that, in the presence of a flat railed annulus $\mathcal{A}$ inside a given graph $G$, every linkage $L$ of $G$ can be combed through some paths of $\mathcal{A}$ in some ``buffer'' (obtained by some hierarchical refinement of $\mathcal{A}$; see~\autoref{subsec_conventions}) of $\mathcal{A}$ corresponding to the position of the terminals of $L$.
In the case of $\FOLDP$, this is essentially a reformulation of~\autoref{prop_combinglemma} to the setting of pairings and flat railed annuli.
Therefore, to generalize~\autoref{lem_colomodelsrerout} to $s$-scattered linkages, one has to prove the analogue of~\autoref{prop_combinglemma} for $s$-scattered linkages.
This was done in~\cite{GolovachST22comb} and is stated below. However, this result is proven for graphs embeddable in some fixed surface and this is the reason why~\autoref{thm_induced} holds up to graphs of bounded-genus.
\begin{proposition}\label{proposition_combinginducedlinkages}
There exist two functions  $\newfun{@preservative}, \newfun{@desaguaderos}:\mathbb{N}^3\to\mathbb{N}$ such that 
for every odd $\ell\in \mathbb{N}$ and every $s,k,g\in\mathbb{N}$,
if 
\begin{itemize}
\item $\Sigma$ is a surface of Euler genus $g$,
\item $\Delta$ is a closed annulus of $\Sigma$,
\item $G$ is a graph embedded in $\Sigma$
\item  $\mathcal{A}=(\mathcal{C},\mathcal{P})$ is a $\Delta$-embedded $(p,q)$-railed annulus of $G$, where
$p\geq \funref{@preservative}(s,k,g)+\ell$ and $q\geq  \frac{2s+5}{2}\cdot\funref{@desaguaderos}(s,k,g)$
\item $L$ is a $\Delta$-avoiding $s$-scattered linkage of size at most $k$, and
\item $I\subseteq [q]$, where $|I|> \funref{@desaguaderos}(s,k,g)\cdot (s+1)$,
\end{itemize}
then
$G$ contains an $s$-scattered linkage
$\tilde{L}$ where $\tilde{L}\equiv L$, $\tilde{L}\setminus \Delta \subseteq L \setminus \Delta$, and  $\tilde{L}$ is $(\ell,I)$-confined in $\mathcal{A}$.
Moreover, $ \funref{@preservative}(s,k,g)= \mathcal{O}((\funref{@desaguaderos}(s,k,g))^2)$ and $\funref{@desaguaderos}(s,k,g)=s\cdot 2^{\mathcal{O}(k+g)}$.
\end{proposition}
Using~\autoref{proposition_combinginducedlinkages} and adjusting the definitions in~\autoref{subsec_tm} for the
$s$-scattered linkages setting
(this is done by just modifying the definition of pairings to consider (boundaried) linkages that are $s$-scattered),
we can prove the analogue of~\autoref{lem_colomodelsrerout}.

\paragraph{Partial signatures and exchangeability for $\FOL[\tau+s\text{-}{\sf sdp}]$.}
 The next important task is to define partial signatures of colored graphs (equipped with a disk-embedded railed annulus) that encode the (recursive) containment of ``meta-collections'' of $s$-scattered linkages in a recursively obtained collection of boundaried (sub)graphs.
This is done as for $\FOLDP$ (see~\autoref{subsec_signatures}) by building the recursive definition of partial signatures using the ``$s$-scattered'' patterns, as defined two paragraphs above, for the base of the recursion.
Using this definition, one can formulate~\autoref{lem_equirep} for $\FOL[\tau+s\text{-}{\sf sdp}]$.
The proof of this analogous version of~\autoref{lem_equirep} is actually the same as the one in~\autoref{subsec_exchangability} and uses the ``$s$-scattered'' versions of~\autoref{lemma_reducing} and~\autoref{lem_colomodelsrerout}.
Let us stress that since~\autoref{proposition_combinginducedlinkages} demands that the given graph $G$ is embedded on a fixed surface,
this also holds for the ``$s$-scattered'' version of~\autoref{lem_colomodelsrerout}
and therefore, in the ``$s$-scattered'' version of~\autoref{lem_equirep},
all considered colored graphs $\mathfrak{G},\mathfrak{G}',\mathfrak{H},$ and $\mathfrak{F}$ should be embedded on some fixed surface and their ``gluing'' should also preserve embeddability.
Also, since we deal with graphs of bounded Euler genus
and using~\autoref{find_wall_eulergen} we can directly obtain a disk-embedded wall if our input graph has large enough treewidth,
there are no apices to deal with.
Therefore, for $\FOL[\tau+s\text{-}{\sf sdp}]$, we do not need to apply the transformations of~\autoref{sec_annotated}.

\paragraph{Finding representatives and proof of~\autoref{thm_induced}.}
From this point on, the steps towards the proof of~\autoref{thm_induced} are completely analogous to the ones in~\autoref{sec_proofoftheorem} for the proof of~\autoref{thm_main}.
With the ``$s$-scattered'' version of \autoref{lem_equirep} in our toolbox,
we have to {\sl find} an (annotated) colored graph with the same partial signature as the original one.
This will allow us to reduce the annotation of the original colored graph.
To do this, we define representatives of vertices with the same recursive ``$s$-scattered'' partial signature as in~\autoref{subsec_sigrep} and we deduce the ``$s$-scattered'' analogue of~\autoref{lem_repre}.
In turn,~\autoref{lem_repre}, when combined with~\autoref{lem_equirep}, implies~\autoref{lem_new_ext_rep} 
for sentences in $\FOL[\tau+s\text{-}{\sf sdp}]$.
Last remaining piece is to prove the following analogue of~\autoref{lem_removingirr} for sentences in $\FOL[\tau+s\text{-}{\sf sdp}]$.
Given two colored graphs $\mathfrak{G},\mathfrak{H}$ and a disk-embedded wall $W$ of $\mathfrak{G}$,
a \emph{perimeter-compatibility function}  of  $\mathfrak{G}$ and $\mathfrak{H}$ is any partial function $\xi: V(D(W))\to V(\mathfrak{H})$ such that $\mathfrak{G}$ and $\mathfrak{H}$ are $\xi$-compatible.

\begin{lemma}\label{lem_scat_ire}
Let $\tau$ be a colored-graph vocabulary.
There is a function $\newfun{fun_sca_ril}:\mathbb{N}^3\to\mathbb{N}$ such that, if
\begin{itemize}
\item $s,g,r,q\in\mathbb{N}$,
\item $\mathfrak{G}$ is a $\tau$-structure embedded in a surface $\Sigma$ of Euler genus $g$,
\item $\Delta$ is a closed disk of $\Sigma$,
\item $W$ is a $\Delta$-embedded $(\funref{fun_sca_ril}(s,r,g)+q)$-wall of $\mathfrak{G}$,
and
\item $R_1,\ldots,R_r\subseteq V(\mathfrak{G})$, where $\cup_{i\in[r]}R_i$ is disjoint from $V({\sf Compass}(W))$,
\end{itemize}
then for each $\funref{fun_sca_ril}(s,r,g)$-internal $q$-subwall $W'$ of $W$,
every $(\tau\cup\{{\sf R}_1,\ldots,{\sf R}_r\})$-structure $(\mathfrak{F},\bar{R}^\star)$, every compatibility function $\xi$ of $(\mathfrak{G},R_1,\ldots,R_r)$ and $(\mathfrak{F},\bar{R}^\star)$,
and every $Y\subseteq V({\sf Compass}(W'))$, it holds that
$$s\text{-}\mathsf{sig}^r((\mathfrak{F},\bar{R}^\star)\oplus_\xi (\mathfrak{G},R_1,\ldots,R_r)) = s\text{-}\mathsf{sig}^r((\mathfrak{F},\bar{R}^\star)\oplus_\xi (\mathfrak{G}\setminus Y,R_1,\ldots,R_r)).$$
\end{lemma}

The proof of~\autoref{lem_scat_ire} is obtained by the proof of \autoref{lem_removingirr} (see~\autoref{subsec_reduce}) by replacing the application of ~\autoref{label_pretendientes} by~\autoref{proposition_combinginducedlinkages} combined with an ``$s$-scattered'' version of \autoref{lem_levelingpaths} (for the proof of the latter, it is easy to observe that it can be directly generalized to $s$-scattered linkages).

Then, using the ``$s$-scattered'' version of~\autoref{lem_new_ext_rep} and~\autoref{lem_scat_ire}, we obtain the analogue of \autoref{corol_redu}, which we state below.

\begin{lemma}
\label{lem_gen_irr_scatt}
There are two functions $\newfun{@_scattindividuated}:\mathbb{N}^3\to\mathbb{N}$ and $\newfun{fun_scattinte}:\mathbb{N}^4\to\mathbb{N}$
and an algorithm that, given
\begin{itemize}
\item $s,g,r,q\in\mathbb{N}$,
\item a $\tau$-structure $\mathfrak{G}$ of Euler genus at most $g$,
\item a closed disk $\Delta$,
\item a $\Delta$-embedded wall $W$ of $\mathfrak{G}$
of height at least $\funref{@_scattindividuated}(s,r,q)$ whose compass has treewidth at most ${\sf tw}$,
and
\item sets  $R_1,\ldots,R_r\subseteq V(\mathfrak{G})$,
\end{itemize}
outputs, in time $\mathcal{O}_{s,g,t,q,{\sf tw}}(n)$,
sets $R_1',\ldots,R_r'\subseteq V(\mathfrak{G})$ and a $q$-subwall $W'$ of $W$ such that for every $Y\subseteq V({\sf Compass}(W'))$, 
every $(\tau\cup\{{\sf R}_1,\ldots,{\sf R}_r\})$-structure $(\mathfrak{F},\bar{R}^\star)$, every compatibility function $\xi$ of $(\mathfrak{G},R_1,\ldots,R_r)$ and $(\mathfrak{F},\bar{R}^\star)$,
\[s\text{-}\mathsf{sig}^r((\mathfrak{F},\bar{R}^\star)\oplus_\xi (\mathfrak{G},R_1,\ldots,R_r))= s\text{-}\mathsf{sig}^r((\mathfrak{F},\bar{R}^\star)\oplus_\xi (\mathfrak{G}\setminus Y,R_1',\ldots,R_r')).\]
\end{lemma}
The proof of~\autoref{thm_induced} is obtained from the one of \autoref{thm_main} (see~\autoref{subsec_mainproof}),
by plugging \autoref{find_wall_eulergen} and  \autoref{lem_gen_irr_scatt}
instead of \autoref{prop_flatwallbdtw} and \autoref{corol_redu}, respectively.

\section{Conclusions and open problems}
\label{sec_conclusion}

In this paper we proved two AMT's for the logic  \FOL{\sf +}{\sf DP} and its newly introduced extension  \FOL{\sf +}{\sf SDP} on graphs of bounded Hadwiger number 
and Euler genus respectively. These two logics can be seen as non-trivial extensions of \FOL, as they may express a wide range of problems (and meta-problems) that are not \FOL-expressible. (See \autoref{sec_problemswesolve} for an exposition of the expressivity potential of \FOL{\sf +}{\sf DP} and \FOL{\sf +}{\sf SDP}.)

\subsection{Open problems}

Recall that \FOL{\sf +}{\sf DP} is an extension of the {\sl separator logic \FOL{\sf +}{\sf conn}}, 
introduced in \cite{Bojanczyk21separ,SchirrmacherSV22first}.
The combinatorial condition given in \cite{PilipczukSSTV22algo} for this logic is having bounded Hajós number. 
As minor excluding graphs classes are also topological-minor excluding classes, the combinatorial condition of  \cite{PilipczukSSTV22algo}
is more general that the one that we give for  \FOL{\sf +}{\sf DP} in this paper.
This makes the AMT of  \cite{PilipczukSSTV22algo} non-comparable 
to ours. Moreover it is shown in \cite{PilipczukSSTV22algo} 
that, under certain complexity assumptions, the bounded Hajós number 
demand is actually demarking the combinatorial horizon of \FOL{\sf +}{\sf conn}. The open question is whether  having bounded Hajós number 
is also the combinatorial horizon of  the more expressive  \FOL{\sf +}{\sf DP}.  We are not in position to make a positive or negative conjecture on this. We wish only to comment that the algorithmic/combinatorial tools that where used in  \cite{PilipczukSSTV22algo} are quite different than the ones used 
in this paper.
\medskip

Another open question is to what extend one may further strengthen the expressibility {\sf  dp}$_k(\cdot)$ (resp.  {\sf $s$\text{-}sdp}$_k(\cdot)$) predicate, while maintaining the 
combinatorial condition of bounded Hadwiger number (resp. bounded Euler genus).
A possible candidate might be to ask for  paths of 
\emph{guided disjointness}, that is to consider the predicate 
{\sf gdp}$_{H}(x_1,y_1,\ldots,x_{k},y_{k})$
where $H$ is a graph where $V(H)=\{1,\ldots,k\}$
and where we ask that, for every edge $(i,j)\in E(H)$,
the $(x_i,y_i)$-path and the 
 $(x_j,y_j)$-path are disjoint. Clearly, ${\sf dp}_{k}(x_1,y_1,\ldots,x_{k},y_{k})={\sf gdp}_{K_{k}}(x_1,y_1,\ldots,x_{k},y_{k})$,
 therefore this would provide a more general logic than  \FOL{\sf +}{\sf DP}. To our knowledge, even the parameterized complexity of the evaluation of  ${\sf dp}_{H}(x_1,y_1,\ldots,x_{k},y_{k})$, when parameterized by $k=|H|$, is
 an interesting open problem.
 
\subsection{Limitations} 
  In the beginning of  \autoref{sec_problemswesolve}
 we comment that \FOL{\sf +}{\sf DP},  on multicolored (by $z$ colors)  graphs,
 may express 
 the predicate $\text{\sf dp}_{k,\lambda}( {\sf s}_{1}, {\sf t}_{1},\ldots, {\sf s}_{k}, {\sf t}_{k})$ equipped with a list function $\lambda: [k]\to 2^{[z]}$, where we demand that, for every $i\in[k]$, $\lambda(i)$ is a subset of the set of all colors assigned to the vertices of the path between the (valuations)
of ${\sf s}_{1}$ and ${\sf t}_{1}$. 
Interestingly this permits us to {\sl demand}
certain colors to be traversed by the disjoint paths. However, on the negative side, we may not expect that \FOL{\sf +}{\sf DP} 
may {\sl exclude} colors. From the empirical point of view, such a demand
obstructs the application of the irrelevant vertex technique.  
Moreover, we may have more solid evidence of this by picking a 
typical example of such a problem. In \autoref{@thesprotians}, we present a colored variant of the {\sc Topological Minor Containment} problem, namely the {\sc  Monochromatic Path Topological Minor} problem that 
we prove (\autoref{@specialisation}) that is {\sf W[1]}-hard on planar 
graphs.
%

%

\paragraph{Acknowledgements.} We would like to thank Anuj Dawar and the anonymous reviewers of previous versions of this paper. Their comments led to an improvement of the presentation of our results.

\newpage

\bibliographystyle{plainurl}

\newpage

\appendix

\section{Problems expressible in  \FOL{\sf +}{\sf DP} and in  \FOL{\sf +}{\sf SDP}}
\label{sec_problemswesolve}

In this section we present some (families) of parameterized problems where \autoref{@definitionen} and \autoref{@geographical} are applicable. 
In all cases, we consider the standard parameterization (by the integer $k$ of the input). Also for each (meta)-problem we comment on their general parameterized complexity status,  their possible classification on the 
families A, B, and C, given in the introduction.

First of all observe that {\sc Disjoint Paths} is the prototypical problem of this category as it is expressed by the ``trivial'' sentence 
$\varphi_k=\text{\sf dp}_{k}( {\sf s}_{1}, {\sf t}_{1},\ldots, {\sf s}_{k}, {\sf t}_{k})$. This problem is  \FPT because of \cite{RobertsonS95GMXIII}. If instead we consider the {\sc Induced Disjoint Paths}, that is $\varphi_k=\text{\sf sdp}_{k}( {\sf s}_{1}, {\sf t}_{1},\ldots, {\sf s}_{k}, {\sf t}_{k})$, then the corresponding parameterized problem becomes {\sf para-NP}-hard because checking whether $(G,{\sf s}_{1}, {\sf t}_{1},{\sf s}_{2}, {\sf t}_{2})\models \varphi_{2}$ is already \NP-complete \cite{KawarabayashiK08thei}.
More generaly, we may also consider problems where the input graph is accompanied 
with a fixed number of colors $X_{1},\ldots,X_{z}$, i.e., we consider structures of the form $(G,X_1,\ldots,X_{z})$, and we may consider the more general predicates $\text{\sf {\rm (}s{\rm )}dp}_{k,\lambda}( {\sf s}_{1}, {\sf t}_{1},\ldots, {\sf s}_{k}, {\sf t}_{k})$ equiped with a list function $\lambda: [k]\to 2^{[z]}$, where we demand that, for every $i\in[k]$, $\lambda(i)$ is a subset of the set of colors of the vertices of the path between the (valuations)
of ${\sf s}_{1}$ and ${\sf t}_{1}$. For instance we may demand that all disjoint paths are {\sl colorful}, i.e., they contain vertices of all available colors.

%
%
%

\subsection{Graph containment problems}
\label{@stubbornness}

We first consider  families of problems expressible by some $\varphi_{k}=\exists \overline{\sf x}\ \psi(\overline{\sf x})\in \FOL{\sf +}{\sf SDP}$.
Such problems are defined by some partial ordering relation $\preceq$ on graphs. We say that $H$ is a \emph{contraction} of $G$ if $H$ can be obtained from $G$ after contracting edges, $H$ is an (\emph{induced}) \emph{minor} of $G$ 
if $H$ is the contraction of an (induced) subgraph of $G$. Finally we say that $H$ is an (\emph{induced}) \emph{topological minor} of $G$ if $G$ contains  a subdivision of $H$ is an (induced) subgaph.

The general setting is the following.

\medskip
\fbox{
\begin{minipage}{15cm}
\noindent{\sc $\preceq$-Contrainment}\\
\noindent{\sl Input}: two graphs $G$ and $H$ where $k=|H|$.\\
\noindent{\sl Question}: $H\preceq G$ ?
\end{minipage}
}
\medskip

Notice that if $\preceq$ is the minor or the topological minor relation, {\sc $\preceq$-Contrainment}
is definable in \FOL{\sf +}{\sf DP} and yields  {\sc Minor Containment} and {\sc Topological Minor Containment}
respectively.  {\sc Minor Containment} belongs in the Category A  as it is \FPT in general graphs because \no-instances are trivially excluding a $K_{|H|}$ minor.   {\sc  Topological Minor Containment} belongs in the  category C as it its \FPT in general graphs however to deal with the question ``{\sl what to do with a clique}'' it needs extra arguments \cite{GroheKMW11find}.  In the case where  $\preceq$ is the induced minor or the induced topological minor relation, {\sc $\preceq$-Contrainment}
is definable in \FOL{\sf +}{\sf SDP} and \autoref{@geographical}  yields that {\sc Induced Minor} and  {\sc Induced Topological Minor Containment}  are \FPT on bounded genus graphs. Moreover as observed in  \cite{KaminskiT12cont}, it is possible ro reduce the {\sc Contraction Containment} problem on bounded genus graphs to {\sc Topological Minor Containment}. These last three problems belong in category C because
they are all \NP-hard for particular instantiations of $H$  due to  the results of \cite{LevequeLMT09dete,FellowsKMP95thec,BrouwerV87contra,LevinPW08,LevinPW08a} (using the parameterized complexity terminology, their standard parmeterizatins are {\sf para-NP}-hard).
Certain rooted variants of all these problems can be also expressed by the corresponding logics
if we ask that the ``models'' certifying each of the aforementioned relations meet certain vertices or sets of vertices (colors) of the input graph.

%
%

\subsection{Linkability problems} 
We now consider families of problems expressible by some $\varphi_{k}=\forall \overline{\sf x}\ \psi(\overline{\sf x})\in \FOL{\sf +}{\sf SDP}$. Such problems involve disjoint path queries  for every choice of terminals in the graph.

\paragraph{Unordered Linkability problems.} 
Given a graph $G$, a set $R\subseteq V(G)$, and a $k\in \mathbb{N}$, we say that $R$ is \emph{$k$-cyclable} in $G$
if every $k$ vertices of $R$ belong in some cycle of $G$ (see \cite{Dirac60inab,WatkinsM67cycl,FlandrinLMW07agen,PlummerG01anin,AldredBHM99cycl
} for the combinatorial properties of $k$-cyclable sets).
The {\sc Cyclability} problem asks, given a triple $G$, $R$, and $k$ as above,
whether $R$ is {$k$-cyclable} in $G$. The algorithmic properties of  {\sc Cyclability} have been 
studied in \cite{GolovachKMT17thep} where it was proven that the standard parametrization of {\sc Cyclability} is \FPT for planar graphs,
while the general problem is {\sf co-W[1]}-hard. 

For this, given a graph $H$, we say that $R$ is \emph{$H$-linkable}
if for every subset $S\in \binom{R}{k}$  there is a $\mathcal{P}\in \binom{S}{2}$
and a collection of internally vertex disjoint 
paths in $G$ between the pairs in  $\mathcal{P}$ that, when contracted to single edges, give a graph that is  isomorphic to $H$. We now consider the following problem.

\medskip
\fbox{
\begin{minipage}{15cm}
\noindent{\sc Unordered Linkability}\\
\noindent{\sl Input}: a graph $G$, $R\subseteq V(G)$, and a graph $H$ where $k=|H|$.\\
\noindent{\sl Question}: is $R$   $H$-linkable in $G$?
\end{minipage}}
\medskip

 Notice that above problem is expressible by a sentence  $\varphi_{H}=\forall \overline{\sf x}\ \psi(\overline{\sf x})$, 
where  $\psi$ consists of $k!$ disjunctions of the {\sf dp}$_{|E(H)|}$ predicate. Moreover,  when $H$ is a cycle of $k$ vertices the above problem yields the {\sc Cyclability} problem, that is already {\sf co-W[1]}-hard.
\autoref{@definitionen} automatically implies that the standard parameterization of {\sc Unordered Linkability}
is \FPT, when restricted to graphs of bounded Hadwiger number.

\paragraph{Ordered Linkability problems.} Let $G$ be a graph and let $R\subseteq V(G)$.
Given a $k\in\mathbb{N}$ we say that  \emph{$R$ is $k$-linked in $G$}  
if for every (ordered) set $\{s_{1},\ldots,s_{k},t_{1},\ldots,t_{k}\}$ of $2k$ distinct vertices in $R$
there are $k$ vertex disjoint paths in $G$ joining the pairs $(s_{i},t_{i}), i\in[k]$.
This notion has introduced by \cite{Watkins68onth} in the late 60s and its graph-theoretical properties 
have been extensively studied in   \cite{LarmanM70onth,Jung70eine,RobertsonS95GMXIII,BollobasT96high,EgawaFGISW00vert,KawarabayashiKY06onsu
}. For instance, Thomas and Wollan proved in \cite{ThomasW05anim} that 
if a graph $G$ is $10k$-connected, then $V(G)$ is $k$-linked in $G$. This notion has been extended 
to the one of a \emph{$H$-linked set} as follows: given some  graph $H$ with  $V(H)=\{x_{1},\ldots,x_{k}\}$,
we say that $G$ is \emph{$H$-linked in $R$} if, {\sl for every  sequence}   $v_{1},\ldots,v_{k}$
of vertices in $R$ there is a $\mathcal{P}\in \binom{\{v_{1},\ldots,v_{k}\}}{2}$
and a collection of internally vertex disjoint 
paths in $G$ between the pairs in  $\mathcal{P}$ that, when contracted to single edges, give a graph 
with vertex set $S$ that is  isomorphic to $H$ via the isomorphism that maps $x_{i}$ to $v_{i}$, $i\in[k]$. 
The combinatorics of $H$-linked sets
has been studied in \cite{GouldW07subd,KostochkaY05anex,EllinghamPY12link,GouldKY06onmi,FerraraGTW06onhl
}. However, to our knowledge, nothing is known about the algorithmic properties of  $k$-linked sets or the more general 
concept of $H$-linked sets. For this, we consider the following general problem.

\medskip
\fbox{
\begin{minipage}{15cm}
\noindent{\sc Ordered Linkability}\\
\noindent{\sl Input}: a graph $G$, $R\subseteq V(G)$, and a graph $H$ where $k=|H|$.\\
\noindent{\sl Question}: is    $R$ $H$-linked in $G$?
\end{minipage}
}
\medskip
 
 It is easy to verify  that the above problem is expressible by a sentence  $\varphi_{H}=\forall \overline{\sf x}\ \psi(\overline{\sf x})$, 
 for some suitable choice of the quantification-free formula $\psi$.
 In \autoref{@imperceptible} we prove that the  {\sc Ordered Linkability} problem, even for the case where 
 $H$ is the disjoint union of $k$ edges, is not \FPT, unless   ${\sf FPT}={\sf W[1]}$ (\autoref{thm_link_hard}).
%
\autoref{@definitionen} automatically implies that the standard parameterization of {\sc Ordered Linkability}
is \FPT, when restricted to graphs of bounded Hadwiger number. \medskip

Clearly both {\sc Unordered Linkability} and {\sc Ordered Linkability}   belong category B and it is an open question whether it is \FPT for other graph classes more general (or different) than those of bounded Hadwiger number. Notice that we may further consider that the disjoint paths in the definition of being $H$-linked 
and $H$-linkable are {\sl induced} paths. This would define the {\sc Induced Unordered Linkability} and {\sc Induced Ordered Linkability} problems whose standard parameterizations are not expected to be \FPT (by easy reductions from the problems 
 {\sc Cyclability} and {\sc Ordered Linkability}) and they are \FPT in  bounded genus graphs, because of \autoref{@geographical}.
%
%
%

\subsection{Vertex deletion problems} 

We give now a wide variety of \FOL{\sf +}{\sf SDP}-expressible problems, typically correspond to formulas of the type $\varphi_{k}=\exists\overline{\sf x}_1\,\forall \overline{\sf x}_2\ \psi(\overline{\sf x}_1,\overline{\sf x}_2)$. We present below those that we consider more relevant.\medskip

\paragraph{Vertex deletion to exclusion.}
This family of problems
can be seen as a natural extension of those mentioned in \autoref{@stubbornness}.
Let $\preceq$ be a  partial ordering relation on graphs and let  $\mathcal{F}$ be a finite set of graphs. 
Given a graph $G$, we say that $\mathcal{F}\preceq G$ if for some $F\in\mathcal{F}$ it holds that $F\preceq G$. We define the following meta-problem.

\medskip
\fbox{
\begin{minipage}{15cm}
\noindent{\sc $\mathcal{F}$-$\preceq$-Deletion}\\
\noindent{\sl Input}: a graph $G$ and  a $k\in \mathbb{N}$.\\
\noindent{\sl Question}: is there an $S\in\binom{V(G)}{k}$ such that  $\mathcal{F}\not\preceq G\setminus S$?
\end{minipage}
}
\medskip

Notice that the above yields the problems {\sc $\mathcal{F}$-Minor-Deletion} and {\sc $\mathcal{F}$-Topological Minor-Deletion} are expressible  by some $\varphi_{k}=\exists{\sf x}_1,\ldots,{\sf x}_k\ \neg \xi({\sf x}_1,\ldots,{\sf x}_k)$
where $\xi({\sf x}_1,\ldots,{\sf x}_k)$ is obtained if, in the corresponding sentences of \autoref{@stubbornness}, we replace every $\text{\sf dp}_{k}( {\sf s}_{1}, {\sf t}_{1},\ldots, {\sf s}_{k}, {\sf t}_{k})$ by  $\text{\sf dp}_{k}( {\sf s}_{1}, {\sf t}_{1},\ldots, {\sf s}_{k}, {\sf t}_{k},{\sf x}_1,{\sf x}_1,\cdots,{\sf x}_k,{\sf x}_k)$. 
It is easy to see that 
\yes-instance of {\sc $\mathcal{F}$-Minor-Deletion} have bounded Hadwiger number, therefore  it belongs in Category A. The standard parameterization of {\sc $\mathcal{F}$-Minor-Deletion} is known to be \FPT in general because of \cite{SauST21kapiII,FominLMS12plan,MarxS07obta,KociumakaP19,JansenLS14anea
}. On the other hand,  the standard parameterization of 
{\sc $\mathcal{F}$-Topological Minor Deletion} was proved to be \FPT in general  \cite{FominLP0Z20hitti
} and belongs to Category C, as extra machinery was used for 1st phase of the irrelevant vertex technique.

By applying \autoref{@definitionen} on bounded genus graphs and combining it with the duality trick of \cite{KaminskiT12cont} it is possible to give an \FPT-reduction of the standard parameterization of  {\sc $\mathcal{F}$-Contraction Deletion} on bounded genus graphs to the standard parameterization of  {\sc $\mathcal{F}$-Topological Minor Deletion}.
Furthermore, {\sc $\mathcal{F}$-Induced Minor Deletion} and {\sc $\mathcal{F}$-Induced Topological Minor Deletion} are expressible in  \FOL{\sf +}{\sf SDP} using $\text{\sf sdp}_{k}$ instead of $\text{\sf dp}_{k}$ in the above expressibility argument. Therefore both parameterized problems are in  \FPT on graphs of bounded genus because of \autoref{@geographical}. It is easy to see that
these last three parameterized problems 
are {\sf para-NP}-hard in general, therefore they are classified in Category B.\medskip

\paragraph{Annotation and subset variants.}~\!\!\!Another direction is to consider {\sc Annotated $\mathcal{F}$-$\preceq$-Deletion}
where the input comes with an annotated set of vertices $R\subseteq V(G)$
and we further demand that $S\subseteq R$. Another variant is the {\sc Subset $\mathcal{F}$-$\preceq$-Deletion} where again the input comes an annotated set $R$ of vertices but now we ask 
that for every $Z\subseteq V(G)$ where $\mathcal{F}\preceq G[Z]$ and $Z\cap R\neq\emptyset$, it holds that 
$Z\cap S\neq \emptyset$. Intuitively, we demand that \emph{only} the certificates of the containment of a graph in $\mathcal{F}$ that intersect $R$ are required to be intersected by the solution $S$.

All results mentioned above for the five aforementioned partial relations on graphs
hold also  for the corresponding {\sc Annotated $\mathcal{F}$-$\preceq$-Deletion} and the  {\sc Subset $\mathcal{F}$-$\preceq$-Deletion} meta-problems
when restricted either to graphs of bounded Hadwiger number or of bounded genus. For general graphs, 
the {\sc Subset $\mathcal{F}$-Minor-Deletion} problem
has been treated   in~\cite{CyganPPW13subs,KawarabayashiK12feed} for $\mathcal{F}=\{K_{3}\}$.

\medskip

\paragraph{General vertex deletion.}
Deviating from the $\exists\ \forall$ scheme of this subsection we wish to mention that 
all above problems can be seen as \emph{modification} problems where we want to achieve some particular target property by removing vertices. The target property is  the exclusion of some pattern graphs (the graphs in $\mathcal{F}$) under some partial relation (the relation $\preceq$).  The application of 
\autoref{@definitionen} and \autoref{@geographical} was made possible because  because the $\preceq$-exclusion is, depending on the choice of $\preceq$, either \FOL{\sf +}{\sf DP} or \FOL{\sf +}{\sf SDP}-expressible. In fact we can see all above results as special cases of the following meta-problem defined given a $\varphi\in \FOL{\sf +}{\sf SDP}$.

\medskip\fbox{\begin{minipage}{15cm}
\noindent{\sc $\varphi$-Deletion}\\
\noindent{\sl Input}: a graph $G$ and  a $k\in \mathbb{N}$.\\
\noindent{\sl Question}: is there an $S\in\binom{V(G)}{k}$ such that  $G\setminus S\models \varphi$?
\end{minipage}}\medskip

According to \autoref{@definitionen} and \autoref{@geographical} 
if $\varphi\in \FOL{\sf +}{\sf DP}$  (resp. $\varphi\in \FOL{\sf +}{\sf SDP}$), then {\sc $\varphi$-Deletion} is \FPT on graphs of bounded Hadwiger number  (resp. bounded genus). The special case where 
$\varphi$ corresponds to the property of being planar and satisfying some \FOL\ property is treated in~\cite{FominGST20analgo}.

\subsection{Amalgamation problems}\label{subsec_amalgamation}

The input of an amalgamation problem consists of two graphs 
and asks for a way to identify their vertices so that the new graph satisfies some particular property.
The notion of amalgamation dates back to \cite{Nesetril79amal} and its combinatorial study includes \cite{Gross11genu,HiltonJRW03amal,LeachR03hami,YangC17thet}. 

Given two graphs $G_{1}$ and $G_{2}$ we define $G_{1}\otimes_{k}G_{2}$
as the set containing every graph obtained if for some  $S_i\in\binom{V(G_i)}{k}, i\in[2]$ and a bijection $\sigma: S_{1}\to S_{2}$ we take the disjoint union of $G_{1}$ and $G_{2}$  and then identify each vertex $v\in S_{1}$ with $\sigma(v)$.
Consider the following problem, defined given a $\varphi\in \FOL{\sf +}{\sf DP}$ ($\varphi\in \FOL{\sf +}{\sf SDP}$).

\medskip\fbox{\begin{minipage}{15cm}
\noindent{\sc $\varphi$-Amalgamation}\\
\noindent{\sl Input}: two graphs $G_1,G_{2}$ and  a $k\in \mathbb{N}$.\\
\noindent{\sl Question}: is there   a graph in  $J\in G_{1}\otimes_{k}G_{2}$ 
where $J\models \varphi$?
\end{minipage}}\medskip

\paragraph{Expressing {\sc $\varphi$-Amalgamation}.}
Let $\varphi\in \FOL{\sf +}{\sf DP}$  (resp. $\varphi\in \FOL{\sf +}{\sf SDP}$).
We claim that there is a sentence $\chi_{k}\in \FOL{\sf +}{\sf DP}$ ($\chi_{k}\in \FOL{\sf +}{\sf SDP}$)
such that $(G_{1},G_{2},k)$ is a \yes-instance of {\sc $\varphi$-Amalgamation} iff 
the tuple $(G,V_1,V_2)\models \chi_{k}$  where $G$ is the disjoint union of $G_{1}$ and $G_{2}$ and $V_{i}=V(G_{i}), i\in[2]$. To see this, consider the sentence  
$$\chi_{k}=\exists {\sf v}_{1}^{1},\ldots,{\sf v}_{k}^{1},{\sf v}_{1}^{2},\ldots,{\sf v}_{k}^{2}\  (\bigwedge_{\substack{i\in[k],j\in[2]}}
{\sf v}_{i}^{j}\in V_{j}\wedge \varphi^\star({\sf v}_{1}^{1},\ldots,{\sf v}_{k}^{1},{\sf v}_{1}^{2},\ldots,{\sf v}_{k}^{2})),$$
where $\varphi^\star$ is a formula with ${\sf v}_{1}^{1},\ldots,{\sf v}_{k}^{1},{\sf v}_{1}^{2},\ldots,{\sf v}_{k}^{2}$ as free variables
obtained from $\varphi$ after replacing each of its atomic formulas as follows:
\begin{itemize}
\item Each atomic formula ${\sf x}={\sf y}$, is replaced by the formula $\zeta_{=}({\sf x},{\sf y})$, defined as $$({\sf x}={\sf y})\vee \bigvee_{i\in[k]}\big(({\sf x}={\sf v}_i^1 \wedge {\sf y}={\sf v}_i^2)
\vee
({\sf y}={\sf v}_i^1 \wedge {\sf x}={\sf v}_i^2)\big).$$
\item Each atomic formula ${\sf E}({\sf x},{\sf y})$ is replaced by the formula $\zeta_{\sf E}({\sf x},{\sf y})$,
defined as
$${\sf E}({\sf x},{\sf y})\vee \bigvee_{i\in[k]} \Big(\big({\sf x}=^\star{\sf v}_i^{1} \wedge {\sf E}({\sf v}_i^{2},{\sf y})\big)\vee \big({\sf y}=^\star {\sf v}_i^{1} \wedge 
 {\sf E}({\sf x},{\sf v}_i^{2})\big)\Big),$$
\item Each atomic formula ${\sf {\rm (}s{\rm )}dp}({\sf s}_1,{\sf t}_1,\ldots,{\sf s}_t,{\sf t}_t)$ is replaced by the formula $\zeta_{\sf {\rm (}s{\rm )}dp}({\sf s}_1,{\sf t}_1,\ldots,{\sf s}_t,{\sf t}_t)$ defined in~\autoref{subsec_apices}, after doing some local replacements in $\zeta_{\sf {\rm (}s{\rm )}dp}({\sf s}_1,{\sf t}_1,\ldots,{\sf s}_t,{\sf t}_t)$ as follows.
\begin{itemize}
\item replace every atomic formula of the form ``${\sf y}_j\in {\sf C}_i$'' by $\zeta_{\sf E}({\sf x}_i,{\sf y}_j)$, and
\item replace all atomic formulas ${\sf x}={\sf y}$ and ${\sf E}({\sf x},{\sf y})$ by $\zeta_{=}({\sf x},{\sf y})$ and $\zeta_{\sf E}({\sf x},{\sf y})$, respectively.
\end{itemize}
\end{itemize}
We make clear that we assume that the collection ${\bf c}$ in the definitions in~\autoref{subsec_apices} is replaced by $({\sf v}_{1}^{1},\ldots,{\sf v}_{k}^{1},{\sf v}_{1}^{2},\ldots,{\sf v}_{k}^{2})$.

Intuitively, in the
sentence $\chi_k$, ``$\exists {\sf v}_{1}^{1},\ldots,{\sf v}_{k}^{1},{\sf v}_{1}^{2},\ldots,{\sf v}_{k}^{2}$''  asks 
for the existence of sets $S_1$ (corresponding to the interpretations of $ {\sf v}_{1}^{1},\ldots,{\sf v}_{k}^{1}$)
and $S_2$  (corresponding to the interpretations of ${\sf v}_{1}^{2},\ldots,{\sf v}_{k}^{2}$) each of size $k$,
and a bijection $\sigma:S_1\to S_2$, given by the ordering of the variables (i.e., mapping the interpretation of ${\sf v}_i^1$ to the interpretation of ${\sf v}_i^2$ for each $i\in[k]$).
Also, for each $i\in[2]$, we demand $S_i$ to be a subset of $V_i$ (``$\bigwedge_{\substack{i\in[k],j\in[2]}}
{\sf v}_{i}^{j}\in V_{j}$'').
Then, to express the identification of each $v\in S_1$ to $\sigma(v)$ and the satisfaction of $\varphi$ from the graph $G_{1}\otimes_{k}G_{2}$, we define the formula $\varphi^\star$ with $ {\sf v}_{1}^{1},\ldots,{\sf v}_{k}^{1},{\sf v}_{1}^{2},\ldots,{\sf v}_{k}^{2}$ as free variables.
For this formula, we use the idea from~\autoref{subsec_apices} for the definition of the apex-projection of a formula,
in order to deal with $S_1$ and $S_2$ separately and ask a modified version of $\varphi$ in the resulting graph.
First feature of $\varphi^\star$ is to consider
the interpretations of ${\sf v}_i^1$ and ${\sf v}_i^2$ as the same vertex.
To incorporate this, we ``re-define'' equality as $\zeta_{=}({\sf x},{\sf y})$.
The respective configuration has to be done to the adjacency predicate ${\sf E}$ and all atomic formulas ${\sf {\rm (}s{\rm )}dp}({\sf s}_1,{\sf t}_1,\ldots,{\sf s}_t,{\sf t}_t)$ in $\varphi$.
For this reason, we define $\zeta_{\sf E}({\sf x},{\sf y})$ in a way that, for example, $v_i^1$ is adjacent to $y$ adjacent if and only if $v_i^2$ is adjacent to $y$.
Then, each atomic formula ${\sf {\rm (}s{\rm )}dp}({\sf s}_1,{\sf t}_1,\ldots,{\sf s}_t,{\sf t}_t)$ is ``splitted'' to the part that concerns ${\sf v}_{1}^{1},\ldots,{\sf v}_{k}^{1},{\sf v}_{1}^{2},\ldots,{\sf v}_{k}^{2}$ and the rest of
the graph, and ``guessing'' where the supposed paths should enter or exit the sets $S_1$ and $S_2$.
This idea is the same as the one in~\autoref{subsec_apices} and for this reason we use
the formula $\zeta_{\sf {\rm (}s{\rm )}dp}({\sf s}_1,{\sf t}_1,\ldots,{\sf s}_t,{\sf t}_t)$.
However, here we {\sl do not} need to remove
the edges between the apex-tuple and the rest of the graph, so we may just ask for adjacencies between the ``apex set'' and the rest of the graph.
This is why we can write $\zeta_{\sf E}({\sf x}_i,{\sf y}_j)$ instead of ``${\sf y}_j\in {\sf C}_i$'' and avoid using the toolbox of ``apex-projection'' and ``backwards-translation''. Also, to be consistent to the identification given by $\sigma$,
we have to replace all atomic formulas ${\sf x}={\sf y}$ and ${\sf E}({\sf x},{\sf y})$ by $\zeta_{=}({\sf x},{\sf y})$ and $\zeta_{\sf E}({\sf x},{\sf y})$, respectively.
Observe that if 
$\varphi\in \FOL{\sf +}{\sf DP}$ (resp. $\varphi\in \FOL{\sf +}{\sf DP}$), then $\chi_{k}\in \FOL{\sf +}{\sf DP}$ (resp. $\chi_{k}\in \FOL{\sf +}{\sf DP}$)

\medskip

According to the above, if $\varphi\in \FOL{\sf +}{\sf DP}$ (resp. $\varphi\in \FOL{\sf +}{\sf DP}$), then  {\sc $\varphi$-Amalgamation} is \FPT on graphs of bounded Hadwiger number (resp. bounded genus).
In \cite{Oliveira18graph}, de Oliveira Oliveira considered, given a pattern graph $H$ on $k$ vertices, the alternative amalgamation operation $G_{1}\otimes_{H}G_{2}$
where the 
subgraphs of $G_{1}$ and $G_{2}$ induced by $S_{1}$ and $S_{2}$ are also asked to be isomorphic to $H$.
Clearly, this operation can also be treated by above machinery by introducing the isomorphism of $G_1[S_{1}]$
and $G_{2}[S_{2}]$ in the formula $\chi_{k}$. 
For this operation,
de Oliveira Oliveira~\cite[Theorem 4.3]{Oliveira18graph} proves that there is an algorithm that, given
a sentence $\varphi\in {\sf CMSOL}$, three (connected) graphs $G_1,G_2,H$ of treewidth at most $t$ and maximum degree at most $\Delta$,
reports whether $G_{1}\otimes_{H}G_{2}\models \varphi$ in time $\mathcal{O}_{|\varphi|,t,\Delta}(n^{\mathcal{O}(t)})$, where $n=|G_1|+|G_2|$.
This result is incomparable with our results.

Note that if the models of $\varphi$ have bounded Hadwiger number, then 
we define problems of Category A that are in \FPT in general. As an example of such a 
problem we mention {\sc Planar Amalgamation}. 

\subsection{Actions and replacements}

In \cite{FominGT19modif} a general local graph modification framework was defined where 
we consider a set of ways, called {\sl actions}, that locally replace small size subgraph patterns 
of a graph. The problem treated in  \cite{FominGT19modif}  is whether such a replacement (or a sequence of such replacements)
may modify the graph so to satisfy some graph property. In \cite{FominGT19modif} this property 
was planarity (and some modifications of it). Here we will consider a way more general setting.

\paragraph{Replacement actions.}
We start with some necessary definitions.  We use the notation $\injection([k],G)$ for all the injections of $[k]$ to the set of vertices of $G$.
A \emph{$k$-numbered-graph} is any graph $H$ where $V(H)=[k],$ i.e., the vertices of $H$ are the numbers $\{1,\ldots,k\}.$
We  denote the set of all $k$-numbered graphs by $\mathcal{H}_{k}$
and we set $\mathcal{H}=\bigcup_{k\in\mathbb{N}}\mathcal{H}_{k}.$
A \emph{replacement action}   is any function 
$\repact:\mathcal{H}\to\mathcal{H},$ where for every $H\in \mathcal{H},$ $|\repact(H)|=|H|,$ i.e.,
graphs in $\mathcal{H}$ are mapped to same-size graphs.

Let $G$ be a graph and let $\injec\in\injection([k],G).$ 
We  set
$\injec^{-1}(G)=([k], \{\injec^{-1}(e)\mid e\in E(G[\injec([k])])\}),$ i.e.,
we see $\injec^{-1}(G)$ is the graph in $\mathcal{H}_{k}$ that is isomorphic, via $\injec,$
to the subgraph of $G$ where $\injec$ applies.

Let $G$ be a graph and let $G'$ be an other  graph where $V(G')\subseteq V(G).$ We  denote $G \sqcup G'=(G\setminus \binom{V(G')}{2})\cup G',$ i.e., $G \sqcup G'$ occurs if we remove from $G$ the edges between vertices in $G'$ and then 
add all edges of $G'.$
Given a graph $G,$ a $\injec\in\injection([k],V(G)),$ and  a $H\in\mathcal{H}_{k},$ we define $\injec(H)=\{\injec([k]),\{\injec(e)\mid e\in E(H)\}\}.$ 
Given a replacement action $\repact:\mathcal{H}\to\mathcal{H},$ we set $\repact_{\injec}(G)=G\sqcup\injec(\repact(\injec^{-1}(G))),$ in other words,
we consider  the part of $G$ that is delimited by $\injec$ and then we replace this part by its image via $\repact$. 

An action $\mathcal{L}$ may be seen as prescribed way to locally change a graph.
It might be the complementation of the edges of a subgraph $G'$ of $G$, their removal, or the addition of a clique on the vertices of $G'$, or the removal of a matching from $G'$.


We now have all ingredients we need for defining a general local replacement problem. Let $\mathcal{L}$ be an action and let  $\varphi\in \FOL{\sf +}{\sf DP}$ ($\varphi\in \FOL{\sf +}{\sf SDP}$). We define the following problem.

\medskip\fbox{\begin{minipage}{15cm}
\noindent{\sc $\mathcal{L}$-$\varphi$-Replacement}\\
\noindent{\sl Input}: a graph $G$ and  a $k\in \mathbb{N}$.\\
\noindent{\sl Question}:  Is there a  $\injec\in\injection([k],V(G))$ such that 
$\repact_{\injec}(G)\models \varphi$?
\end{minipage}}\medskip

\paragraph{Expressing {\sc $\mathcal{L}$-$\varphi$-Replacement}.}
We claim that, for every action $\mathcal{L}$ and every sentence  $\varphi\in \FOL{\sf +}{\sf DP}$ ($\varphi\in \FOL{\sf +}{\sf SDP}$),
 there is a sentence $\xi_k\in  \FOL{\sf +}{\sf DP}$  ($\xi_k\in \FOL{\sf +}{\sf SDP}$) such that
$(G,k)$ is a \yes-instance of {\sc $\mathcal{L}$-$\varphi$-Replacement} if and only if
$G\models \xi_k$
To see this, consider the sentence
$$\xi_k = \exists {\sf v}_1,\ldots, {\sf v}_k\ \hat{\varphi}({\sf v}_1,\ldots, {\sf v}_k),$$
where $\hat{\varphi}$ is the formula obtained from $\varphi^\star$ of~\autoref{subsec_amalgamation}, after replacing
each atomic formula $\zeta_{=}({\sf x},{\sf y})$ by ${\sf x}={\sf y}$,
each atomic formula $\zeta_{\sf E}({\sf x},{\sf y})$ by the formula
\begin{align*}
 \xi_{\sf E}({\sf x},{\sf y}) =&\Bigg(\Big({\sf E}({\sf x},{\sf y})\wedge \bigwedge_{(i,j)\in \binom{[k]}{2}}({\sf x}\neq {\sf v}_i \vee {\sf y}\neq {\sf v}_j)\Big)\\
&~~~~~~~~ \vee  \bigwedge_{H\in \mathcal{H}_k}
\Big(
\bigwedge_{(i,j)\in E(H)}({\sf x}= {\sf v}_i \wedge {\sf y}= {\sf v}_j)
\implies
\bigwedge_{(i,j)\in E(\mathcal{L}(H))}({\sf x}= {\sf v}_i \wedge {\sf y}= {\sf v}_j)
\Big)
\Bigg).
\end{align*}
Intuitively, the sentence $\xi_k$ asks the existence of an $\injec\in\injection([k],V(G))$ (``$ \exists {\sf v}_1,\ldots, {\sf v}_k$'') such that the modified graph $\repact_{\injec}(G)$ satisfies $\varphi$.
We use $v_1,\ldots, v_k$ to denote the interpretation of ${\sf v}_1,\ldots, {\sf v}_k$.
To express the replacement action, we define $\hat{\varphi}$ to be the formula obtained from
$\varphi$ after ``forgetting'' all edges between  $v_1,\ldots, v_k$ (``$\bigwedge_{(i,j)\in \binom{[k]}{2}}({\sf x}\neq {\sf v}_i \vee {\sf y}\neq {\sf v}_j)$'')
and after ``adding'' all edges given by the replacement action $\mathcal{L}$, i.e., for each graph $H\in \mathcal{H}_k$,
if $H= \injec^{-1}(G)$, then ``add'' the edges of $\injec(\repact(\injec^{-1}(G)))$ in $G$.
This is expressed using the formula ``$\bigwedge_{H\in \mathcal{H}_k}
\Big(
\bigwedge_{(i,j)\in E(H)}({\sf x}= {\sf v}_i \wedge {\sf y}= {\sf v}_j)
\implies
\bigwedge_{(i,j)\in E(\mathcal{L}(H))}({\sf x}= {\sf v}_i \wedge {\sf y}= {\sf v}_j)
\Big)
$''.
Additionally, in order to deal with disjoint paths that pass through $v_1,\ldots, v_k$,
we use the trick in~\autoref{subsec_amalgamation} to replace the atomic formulas ${\sf {\rm (}s{\rm )}dp}({\sf s}_1,{\sf t}_1,\ldots,{\sf s}_t,{\sf t}_t)$.
For this reason, we define the formula  $\hat{\varphi}$ as a modified version of $\varphi^\star$ from~\autoref{subsec_amalgamation} by only replacing each atomic formula $\zeta_{=}({\sf x},{\sf y})$ by ${\sf x}={\sf y}$ and each atomic formula  $\zeta_{\sf E}({\sf x},{\sf y})$ (even the ones that appear in the ``translation'' of the atomic formulas ${\sf {\rm (}s{\rm )}dp}$) by $\xi_{\sf E}({\sf x},{\sf y})$. Observe that if $\varphi\in \FOL{\sf +}{\sf DP}$ (resp. $\varphi\in \FOL{\sf +}{\sf SDP}$), then $\xi_k\in \FOL{\sf +}{\sf DP}$ (resp. $\xi_k\in \FOL{\sf +}{\sf SDP}$).
\medskip

According to the above, if $\varphi\in \FOL{\sf +}{\sf DP}$ (resp. $\varphi\in \FOL{\sf +}{\sf SDP}$), then  {\sc $\mathcal{L}$-$\varphi$-Replacement} is \FPT on graphs of bounded Hadwiger number (resp. bounded genus).  Again in case where the models of $\varphi$ have bounded Hadwiger number then 
we obtain problems of Category A that belong \FPT in general and this already includes the results of \cite{FominGT19modif} where the target property was planarity. 
In fact using the above setting we may extend the definition of replacement actions 
so to permit the substitution of subgraphs of the input graph with new graphs of different (but still bounded) sizes.
This would make it possible to define more flexible types of modifications such as edge contractions or \Delta-Y transformations.

\subsection{Elimination distance problems}
Given a graph $G$, we use ${\sf cc}(G)$ to denote the connected components of $G$.
We say that a graph class is \emph{non-trivial} if it contains at least one non-empty graph and does not contain all graphs.
Given a graph class $\mathcal{G}$,
we define the \emph{connected closure} of $\mathcal{G}$ as $\mathcal{C}(\mathcal{G}) = \{G\mid \forall C\in{\sf cc}(G), C\in \mathcal{G}\}$.
Also, we use  $\mathcal{A}(\mathcal{G})$ to denote the set $\{G\mid \exists v\in V(G) G\setminus v\in \mathcal{G}\}$.

Let $\mathcal{G}$ be a non-trivial  graph class.
We say that a graph $G$ has \emph{elimination distance at most $k$} to $\mathcal{G}$ if
\[G\in \overbrace{\mathcal{C}(\mathcal{A}(\cdots \mathcal{C}(\mathcal{A}}^{\text{$k$ times}}(\mathcal{C}(\mathcal{G}))))).\]
The elimination distance from a graph class was defined by Bulian and Dawar in \cite{BulianD16graph}
as an alternative graph modification measure (see~\cite{JansenK021verte,DLindermayrSV20elimi,AgrawalKPRS21anfp,AgrawalR20onthe,HolsKP20elimi,AgrawalKLPRSZ22dele,AgrawalKLPRS21elim,JansenK21fpta} for algorithmic results concerning elimination distance).
Given a $\varphi\in \FOL{\sf +}{\sf DP}$  (resp. $\varphi\in \FOL{\sf +}{\sf SDP}$) we define the following problem: 

\medskip\fbox{\begin{minipage}{15cm}
\noindent{\sc $\varphi$-Elimination distance}\\
\noindent{\sl Input}: a graph $G$ and  a $k\in \mathbb{N}$.\\
\noindent{\sl Question}:  is the elimination distance of $G$ from $\Mod(\varphi)$ at most $k$?
\end{minipage}}\medskip

Bulian and Dawar in \cite{BulianD17fixe} considered the above problem for the case where $\varphi$ expresses the minor-exclusion of some finite set of graphs and they proved that this problem is (constructively) \FPT.
In \cite{FominGT21param} considered {\sc $\varphi$-Elimination distance} when $\varphi\in \FOL$ and proved that 
for particular instantiations of $\varphi$ the problem, parameterized by $k$, is {\sf W[2]}-hard.
According to the recent meta-algorithmic results in \cite{PilipczukSSTV22algo}, when $\varphi\in $\FOL+{\sf conn}, {\sc $\varphi$-Elimination distance} is \FPT for graphs of bounded Hajós number (see~\cite{SchirrmacherSV22first}).
According to our results if $\varphi\in \FOL{\sf +}{\sf SDP}$ (resp. $\varphi\in \FOL{\sf +}{\sf DP}$), then  {\sc $\varphi$-Elimination distance} is \FPT on graphs of bounded Hadwiger number (resp. bounded genus).

\paragraph{The problem {\sc $\varphi$-Block Elimination Distance}.}
Another parameter similar to elimination distance is the \emph{block elimination distance}, introduced in~\cite{DinerGT22block}, that is obtained if we replace the connected closure operator $\mathcal{C}$ in the above definition by the operator $\mathcal{B}$,
defined as $$\mathcal{B}(\mathcal{G})=\{G\mid \text{$\forall B\in{\sf bc}(G)$, } B\in\mathcal{B}\},$$
where ${\sf bc}(G)$ is the set of all blocks of $G$.
Similar to {\sc $\varphi$-Elimination distance}, we can define the problem {\sc $\varphi$-Block Elimination Distance}.
According to our results, if $\varphi\in \FOL{\sf +}{\sf DP}$ (resp. $\varphi\in \FOL{\sf +}{\sf SDP}$), then  {\sc $\varphi$-Block Elimination Distance} is \FPT on graphs of bounded Hadwiger number (resp. bounded genus).
We remark, that when  $\varphi\in $\FOL+{\sf conn}, {\sc $\varphi$-Elimination distance} is \FPT for graphs of bounded Hajós number because of the results   in \cite{PilipczukSSTV22algo}.

\subsection{Reconfiguration  problems}
Intuitively, reconfiguration problems ask, given two feasible solutions $S$ and $T$ of a problem, whether there is a step-by-step transformation between $S$ and $T$
where all intermediate sets are also feasible solutions.

\paragraph{Reconfiguration sequences.}
Let $\varphi$ be a sentence that is satisfied in structures of the form $(G,S)$, i.e., structures of the colored-graph vocabulary $\{{\sf E},{\sf S}\}$.
A sequence $S_1,\ldots, S_\ell,
\ell\in\mathbb{N}$ of subsets of $V(G)$ is called a  \emph{$\varphi$-reconfiguration sequence} if
\begin{itemize}
\item for every $i\in[\ell]$, $(G,S_i)\models \varphi$ and 
\item for every $i\in[\ell]$, there is a $v\in S_{i}$ and a $u\in V(G)\setminus S_i$ such that $S_{i+1} = (S_i\setminus \{v\})\cup \{u\}$.
\end{itemize}
We call $\ell$ the \emph{lenght} of the  $\varphi$-reconfiguration sequence.

Given a $\varphi\in \FOL{\sf +}{\sf DP}$  (resp. $\varphi\in \FOL{\sf +}{\sf SDP}$), whose models are of the form $(G,S)$,
we define the following problem: 

\medskip\fbox{\begin{minipage}{15cm}
\noindent{\sc $\varphi$-Reconfiguration}\\
\noindent{\sl Input}: a graph $G$, two sets $S,T\subseteq V(G)$, and an $\ell\in\mathbb{N}$.\\
\noindent{\sl Question}:  is there a $\varphi$-reconfiguration sequence $S,\ldots, T$ of length at most $\ell$?
\end{minipage}}\medskip

Reconfiguration problems have received a lot of attention in the literature~\cite{ItoDHPSUU11onth,BonsmaC09find,GopalanKMP09thec,HearnD05pspa,ItoKC09reco,Mouawad15onre,Nishimura18intr,Heuvel13thec,MynhardtN20reco,BousquetMNS22asur,MouawadNRSS17onth,LokshtanovMPRS18reco,Siebertz18reco,MouawadNRSS17onth,LokshtanovMPS22onth}).
A vibrant branch of research on reconfiguration problems deals with other reconfiguration models (apart from removing/adding vertices) like token sliding, (perfect) matching flipping, spanning tree flipping~\cite{BonamyDO21domi,BousquetJ21tsre,BelmonteKLMOS21toke,DemaineDFHIOOUY15line,BartierBDLM21ongi,YamadaU21shor,BartierBM22gala,BonamyBHIKMMW19thep,BousquetHIM19shor,BousquetIKMOSW22reco,BousquetIKMOSW20reco,EtoIKOW22reco}.
Also, the tractability of reconfiguration problems has been studied under different structural parameterizations of the input graph~\cite{MouawadNRS18vert,Wrochna18reco,BelmonteHLOO20inde,BodlaenderGS21para,BodlaenderGNS21para}.

\paragraph{Known AMTs for reconfiguration problems.}
Also, there are some konwn algorithmic meta-theorems for reconfiguration problems.
In fact,~\cite{MouawadNRW14reco}  proved that for every $\varphi\in{\sf MSO}_2$, the problem {\sc $\varphi$-Reconfiguration} is \FPT parameterized by $\tw+\ell+|\varphi|$, where $\tw$ is the treewidth of the input graph.
Also, their framework can be used to derive an \FPT algorithm for the parameterization of the problem  {\sc $\varphi$-Reconfiguration} for $\varphi\in{\sf MSO}_1$ by $\ell+{\sf cw}+|\varphi|$, where ${\sf cw}$ is the cliquewidth of the input graph.
For formulas $\varphi\in\MSOL$, in~\cite{GimaIKO22algo},
 they considered two parameterizations of {\sc $\varphi$-Reconfiguration}:
the first is by the neighborhood diversity of the input graph and the second is, when restricted to feasible solutions of size $k$, by the treewidth of the input graph and $k$. For these two parameterized (meta)problems, they give an \FPT algorithm.

For $\varphi\in\FOL$, the variant of {\sc $\varphi$-Reconfiguration} where the size of all sets in the reconfiguration sequence is $k$, is \FPT parameterized by $\ell+k+|\varphi|$ on nowhere dense classes~\cite{LokshtanovMPRS18reco}.
This result also holds for any reconfiguration model (for example, token sliding) that can be expressed by a formula in \FOL.

\paragraph{Expressing {\sc $\varphi$-Reconfiguration}.}
Let $\varphi\in \FOL{\sf +}{\sf DP}$  (resp. $\varphi\in \FOL{\sf +}{\sf SDP}$).
We claim that there is a sentence $\tilde{\varphi}_\ell\in  \FOL{\sf +}{\sf DP}$  ($\tilde{\varphi}_\ell\in \FOL{\sf +}{\sf SDP}$) such that
$(G,S,T,\ell)$ is a \yes-instance of  {\sc $\varphi$-Reconfiguration} if and only if
$(G,S,T)\models \tilde{\varphi}_\ell$.
To see this,
first, we consider the formula $\psi_{{\sf S}}^{(0)}({\sf z}) = ({\sf z}\in {\sf S})$ and, for every $i\in[\ell-1]$, we consider the formula
$$\psi_{{\sf S}}^{(i)}({\sf x}_1,{\sf y}_1,\ldots, {\sf x}_{i},{\sf y}_{i},{\sf z}) = \Big( \psi_{{\sf S}}^{(i-1)}({\sf x}_1,{\sf y}_1,\ldots, {\sf x}_{i-1},{\sf y}_{i-1},{\sf z})\wedge {\sf z}\neq {\sf x}_i\Big)\vee {\sf z} = {\sf y}_i.$$
Intuitively, for every $i\in\mathbb{N}$,
given a graph $G$, a set $S\subseteq V(G)$ that interprets ${\sf S}$, and some vertices $x_1,y_1,\ldots, x_{i},y_{i},z\in V(G)$ that interpret the variables ${\sf x}_1,{\sf y}_1,\ldots, {\sf x}_{i},{\sf y}_{i},{\sf z}$,
$\psi_{{\sf S}}^{(i)}({\sf x}_1,{\sf y}_1,\ldots, {\sf x}_{i},{\sf y}_{i},{\sf z})$
expresses that $z\in S_i$, where $S_0 = S$ and for every $j\in[i]$, $S_{j} = (S_{j-1}\setminus \{x_j\})\cup\{y_j\}$.

Also,
for every $i\in\mathbb{N}$, 
we define $\varphi^{(i)}_{\sf S}({\sf x}_1,{\sf y}_1,\ldots, {\sf x}_i,{\sf y}_i)$
 to be the formula obtained from $\varphi$ after replacing each atomic term ${\sf x}\in {\sf S}$ with the formula $\psi_{{\sf S}}^{(i)}({\sf x}_1,{\sf y}_1,\ldots, {\sf x}_{i},{\sf y}_{i},{\sf x})$.
 Intuitively, the formula $\varphi^{(i)}_{\sf S}$ demands that the variables of $\varphi$ are picked inside the set $S_i$, instead of the set $S$, where $S_i$ is defined as above.
We define $ \tilde{\varphi}_\ell $ as follows.
\begin{align*}
& \tilde{\varphi}_\ell = \exists {\sf v}_1,{\sf u}_1,\ldots, {\sf v}_\ell, {\sf u}_\ell\ \\
&~~\bigwedge_{i\in[\ell]}\Big( \psi_{{\sf S}}^{(i-1)}({\sf v}_1,{\sf u}_1,\ldots, {\sf v}_{i-1},{\sf u}_{i-1},{\sf v}_i)\wedge \neg
\psi_{\sf S}^{(i-1)}({\sf v}_1,{\sf u}_1,\ldots, {\sf v}_{i-1},{\sf u}_{i-1},{\sf u}_i)\Big)
\wedge\varphi^{(i)}_{\sf S}({\sf v}_1,{\sf u}_1,\ldots, {\sf v}_i,{\sf u}_i).
\end{align*}
Observe that if $\varphi\in \FOL{\sf +}{\sf DP}$ (resp. $\varphi\in \FOL{\sf +}{\sf SDP}$), then $\tilde{\varphi}_\ell\in \FOL{\sf +}{\sf DP}$ (resp. $\tilde{\varphi}_\ell\in \FOL{\sf +}{\sf SDP}$).
According to our results, if $\varphi\in \FOL{\sf +}{\sf DP}$ (resp. $\varphi\in \FOL{\sf +}{\sf SDP}$), then  {\sc $\varphi$-Reconfiguration} is \FPT, parameterized by $\ell$ and $|\varphi|$, on graphs of bounded Hadwiger number (resp. bounded genus).

\subsection{Planarizer game}

We would like to finish this section with a new problem that we consider worth mentioning here.
We call it {\sc Planarizer Game} and it is played by two players, the {\sl blocker} and the {\sl planarizer}.
The two players play in rounds. In each round each player places tokens on the vertices of the graph and the  blocker plays first. None of the players can move his/her token on a vertex that is already occupied by a token. The planarizer wins if the removal  from $G$ of the vertices 
that are occupied by his/her tokens yields a planar graph.  The {\sc Planarizer Game}  problem 
asks, given a graph $G$ and a non-negative integer $k$, whether the planarizer has a victory strategy 
against the blocker that uses at most $k$ rounds. 
To prove that {\sc Planarizer Game} is {\sf NP}-hard we conside the {\sc Edge Planarizer} problem 
asking, give a graph $G$ and a non-negative integer $k$, whether there is set of at most $k$ edges of $G$ 
whose removal produces a planar graph.  In \cite{FarciaHM04onth},  Fariaa, Herrera de Figueiredo, Mendonça, proved that this 
problem is \NP-hard for graphs of maximum degree three. We reduce this restricted version of  {\sc Edge Planarizer}  to  {\sc Planarizer Game} as follows.

Let $(G,k)$ be an input of  {\sc Edge Planarizer}.
We transform  $(G,k)$  to an input $(G',k)$ of  {\sc Planarizer Game} by replacing each edge by a path of length $k$
and by removing every vertex $v$ of degree three and making $N_G (v)$ a clique (i.e., a triangle).
%
 Then  $(G,k)$ is a \yes\ instance of {\sc Edge Planarizer} iff  $(G',k)$ is a \yes\ instance of {\sc Planarizer Game} because each edge-choice of  {\sc Edge Planarizer} corresponds to a choice of a vertex of a joining path and the length of joining paths is big enough to ``neutralize'' the potential of the blocker during the game.

According to \autoref{@definitionen} and given that 
\yes-instances of {\sc Planarizer Game}   have bounded Hadwiger number, the problem belongs in Category A and, when parameterized by the number of rounds,  is \FPT. Of course the same result can be generalized if, instead of planarity, the planarizer pursues 
some other graph property whose graphs have bounded Hadwiger number (resp. genus) and is expressible in \FOL{\sf +}{\sf DP} (resp. \FOL{\sf +}{\sf SDP}).

\newpage

\section{Flat walls and flat annuli framework}\label{sec_flatwalls}

Here we present the framework on flat walls that was introduced in~\cite{SauST21amor}.
In~\autoref{label_domesticated} we define walls, subwalls, and other notions related to walls.
Next, in~\autoref{label_inappropriate}, we give the definitions of renditions and paintings, that are used in \autoref{label_exceptionalness} to define flatness pairs.
In~\autoref{label_exceptionalness}, apart from the definition of flatness pairs, we present notions like {\sl influence}, {\sl regularity}, and {\sl tilts}.

%

\subsection{Walls and subwalls}\label{label_domesticated}

\paragraph{Dissolutions and subdivisions.}
Given a vertex $v\in V(G)$ of degree two with neighbors $u$ and $w,$ we define the \emph{dissolution} of $v$
to be the operation of deleting $v$ and, if $u$ and $w$ are not adjacent, adding the edge $\{u,w\}.$
Given an edge $e=\{u,v\}\in E(G),$ we define the \emph{subdivision} of $e$
to be the operation of deleting $e,$ adding a new vertex $w$ and making it adjacent to $u$ and $v.$
Given two graphs $H,G,$ we say that $H$ is a \emph{subdivision} of $G$
if $H$ can be obtained from $G$ after subdividing edges of $G.$

\paragraph{Walls.}\label{label_mistreatment}
Let  $k,r\in\mathbb{N}.$ The
\emph{$(k\times r)$-grid} is the
graph whose vertex set is $[k]\times[r]$ and two vertices $(i,j)$ and $(i',j')$ are adjacent if and only if $|i-i'|+|j-j'|=1.$
An  \emph{elementary $r$-wall}, for some odd integer $r\geq 3,$ is the graph obtained from a
$(2 r\times r)$-grid
with vertices $(x,y)
	\in[2r]\times[r],$
after the removal of the
``vertical'' edges $\{(x,y),(x,y+1)\}$ for odd $x+y,$ and then the removal of
all vertices of degree one.
Notice that, as $r\geq 3,$  an elementary $r$-wall is a planar graph
that has a unique (up to topological isomorphism) embedding in the plane
such that all its finite faces are incident to exactly six
edges.
The \emph{perimeter} of an elementary $r$-wall is the cycle bounding its infinite face,
while the cycles bounding its finite faces are called \emph{bricks}.
Also, the vertices
in the perimeter of an elementary $r$-wall that have degree two are called \emph{pegs},
while the vertices $(1,1), (2,r), (2r-1,1), (2r,r)$ are called \emph{corners} (notice that the corners are also pegs).

An \emph{$r$-wall} is any graph $W$ obtained from an elementary $r$-wall $\bar{W}$
after subdividing edges. A graph $W$ is a \emph{wall} if it is an $r$-wall for some odd $r\geq 3$
and we refer to $r$ as the \emph{height} of $W.$ Given a graph $G,$
a \emph{wall of} $G$ is a subgraph of $G$ that is a wall.
We insist that, for every $r$-wall, the number $r$ is always odd.

We call the vertices of degree three of a wall $W$ \emph{3-branch vertices}.
A cycle of $W$ is a \emph{brick} (resp. the \emph{perimeter}) of $W$
if its 3-branch vertices are the vertices of a brick (resp. the perimeter) of $\bar{W}.$
We denote by $\mathcal{C}(W)$ the set of all cycles of $W.$
We  use $D(W)$ in order to denote the perimeter of the  wall $W.$
A brick of $W$ is \emph{internal} if it is disjoint from $D(W).$

\paragraph{Subwalls.}
Given an elementary $r$-wall $\bar{W},$ some odd $i\in \{1,3,\ldots,2r-1\},$ and $i'=(i+1)/2,$
the \emph{$i'$-th  vertical path} of $\bar{W}$  is the one whose
vertices, in order of appearance, are $(i,1),(i,2),(i+1,2),(i+1,3),
	(i,3),(i,4),(i+1,4),(i+1,5),
	(i,5),\ldots,(i,r-2),(i,r-1),(i+1,r-1),(i+1,r).$
Also, given some $j\in[2,r-1]$ the \emph{$j$-th horizontal path} of $\bar{W}$
is the one whose
vertices, in order of appearance, are $(1,j),(2,j),\ldots,(2r,j).$

A \emph{vertical} (resp. \emph{horizontal}) path of $W$ is one
that is a subdivision of a  vertical (resp. horizontal) path of $\bar{W}.$
Notice that the perimeter of an $r$-wall $W$
is uniquely defined regardless of the choice of the elementary $r$-wall $\bar{W}.$
An \emph{$r'$-subwall} (or simply \emph{subwall}) of $W$ is any subgraph $W'$ of  $W$
that is an $r'$-wall, with $r' \leq r,$ and such the vertical (resp. horizontal) paths of $W'$ are subpaths of the
{vertical} (resp. {horizontal}) paths of $W.$

\paragraph{Layers.}
The \emph{layers} of an $r$-wall $W$  are recursively defined as follows.
The first layer of $W$ is its perimeter.
For $i=2,\ldots,(r-1)/2,$ the $i$-th layer of $W$ is the $(i-1)$-th layer of the subwall $W'$
obtained from $W$ after removing from $W$ its perimeter and
removing recursively all occurring vertices of degree one.
We refer to the $(r-1)/2$-th layer as the \emph{inner layer} of $W.$
The \emph{central vertices} of an $r$-wall $W$ are its two branch vertices that do not belong to any of its layers
and are connected by a path of $W$  that does not intersect any layers of $W$.

\paragraph{Central walls.}
Given an $r$-wall $W$ and an odd $q\in\mathbb{N}_{\geq 3}$ where $q\leq r,$
we define the \emph{central $q$-subwall} of $W,$ denoted by $W^{(q)},$
to be the $q$-wall obtained from $W$ after removing
its first $(r-q)/2$ layers and all occurring vertices of degree one.
Given an $h\in [(r-1)/2]$, a subwall $W'$ of $W$ is called \emph{$h$-internal} if it is a subwall of $W^{(r-2h)}$.

\paragraph{Tilts.}
The \emph{interior} of a wall $W$ is the graph obtained
from $W$ if we remove from it all edges of $D(W)$ and all vertices
of $D(W)$ that have degree two in $W.$
Given two walls $W$ and $\tilde{W}$ of a graph $G,$
we say that $\tilde{W}$ is a \emph{tilt} of $W$ if $\tilde{W}$ and $W$ have identical interiors.

\subsection{Paintings and renditions}
\label{label_inappropriate}
In this subsection we present the notions of renditions and paintings, originating in the work of Robertson and Seymour \cite{RobertsonS95GMXIII}.
The definitions presented here were introduced by Kawarabayashi, Thomas, and Wollan \cite{KawarabayashiTW18anew} (see also~\cite{SauST21amor}).
\paragraph{Paintings.}
Let $\Delta$ be a closed annulus or a closed disk.
A \emph{{$\Delta$}-painting} is a pair $\Gamma=(U,N)$
where
\begin{itemize}
	\item  $N$ is a finite set of points of $\Delta,$
	\item $N \subseteq U \subseteq \Delta,$ and
	\item $U \setminus  N$ has finitely many arcwise-connected  components, called \emph{cells}, where, for every cell $c,$
	      \begin{itemize}
		      \item[$\circ$] the closure $\bar{c}$ of $c$
		            is a closed disk
		            and
		      \item[$\circ$]  $|\tilde{c}|\leq 3,$ where $\tilde{c}:=\bd(c)\cap N.$
	      \end{itemize}
\end{itemize}
We use the  notation $U(\Gamma) := U,$
$N(\Gamma) := N$  and denote the set of cells of $\Gamma$
by $C(\Gamma).$
For convenience, we may assume that each cell  of $\Gamma$ is an open disk of $\Delta.$
Notice that, given a $\Delta$-painting $\Gamma,$
the pair $(N(\Gamma),\{\tilde{c}\mid c\in C(\Gamma)\})$  is a hypergraph whose hyperedges have cardinality at most three and  $\Gamma$ can be seen as a plane embedding of this hypergraph in $\Delta.$

\paragraph{Disk and annulus renditions.}
Let $G$ be a graph and let $\Omega$ be a cyclic permutation of a subset of $V(G)$ that we denote by $V(\Omega).$ By a \emph{disk $\Omega$-rendition} of $G$ we mean a triple $(\Gamma, \sigma, \pi),$ where
\begin{itemize}
	\item[(a)] $\Gamma$ is a $\Delta$-painting for some closed disk $\Delta,$
	\item[(b)] $\pi: N(\Gamma)\to V(G)$ is an injection, and
	\item[(c)] $\sigma$ assigns to each cell $c \in  C(\Gamma)$ a subgraph $\sigma(c)$ of $G,$ such that
	      \begin{enumerate}
		      \item[(1)] $G=\bigcup_{c\in C(\Gamma)}\sigma(c),$
		      \item[(2)]  for distinct $c, c' \in  C(\Gamma),$  $\sigma(c)$ and $\sigma(c')$  are edge-disjoint,
		      \item[(3)] for every cell $c \in  C(\Gamma),$ $\pi(\tilde{c}) \subseteq V (\sigma(c)),$
		      \item[(4)]  for every cell $c \in  C(\Gamma),$
		            $V(\sigma(c)) \cap \bigcup_{c' \in  C(\Gamma) \setminus  \{c\}}V(\sigma(c')) \subseteq \pi(\tilde{c}),$ and
		      \item[(5)]  $\pi(N(\Gamma)\cap \bd(\Delta))=V(\Omega),$ such that the points
		            in $N(\Gamma)\cap \bd(\Delta)$ appear in $\bd(\Delta)$ in the same ordering
		            as their images, via $\pi,$ in $\Omega.$
	      \end{enumerate}
\end{itemize}

Similarly, we define {\sl annulus $(\Omega_1, \Omega_2)$-renditions} as follows.
Let $G$ be a graph and let $X_1,X_2$ be two subsets of $V(G)$,
and let $\Omega_1$ (resp. $\Omega_2$) be a cyclic permutation of $X_1$ (resp. $X_2$). By an \emph{$(\Omega_1, \Omega_2)$-rendition} of $G$ we mean a triple $(\Gamma, \sigma, \pi),$ where $(\Gamma, \sigma, \pi),$ is defined as for disk $\Omega$-renditions but  $\Gamma$ is a $\Delta$-painting for some closed annulus $\Delta$ (instead of a closed disk) and as for item (5), $\pi(N(\Gamma)\cap \bd(\Delta))=X_1\cup X_2,$ such that, if $\bd(\Delta)=B_1\cup B_2$, then the points in $N(\Gamma)\cap B_i$ appear in $B_i$ in the same ordering as their images, via $\pi,$ in $\Omega_i,$ for $i\in\{1,2\}$.

\subsection{Flatness pairs}
\label{label_exceptionalness}
In this subsection we define the notion of a flat wall, originating in the work of Robertson and Seymour \cite{RobertsonS95GMXIII} and later used in \cite{KawarabayashiTW18anew}.
Here, we define flat walls as in \cite{SauST21amor}.

\paragraph{Flat walls.}
Let $G$ be a graph and let $W$ be an $r$-wall  of $G,$ for some odd integer $r\geq 3.$
We say that a pair $(P,C)\subseteq D(W)\times D(W)$ is a \emph{choice
		of pegs and corners for $W$} if $W$ is the subdivision of an  elementary $r$-wall $\bar{W}$
where $P$ and
$C$ are the pegs and the corners of $\bar{W},$ respectively (clearly, $C\subseteq P$).
To get more intuition, notice that a wall $W$ can occur in several ways from the elementary wall $\bar{W},$
depending on the way the vertices in the perimeter of $\bar{W}$ are subdivided.
Each of them gives a different selection $(P,C)$ of pegs and corners of $W.$

We say that $W$ is a \emph{flat $r$-wall}
of $G$ if there is a separation $(X,Y)$ of $G$ and a choice  $(P,C)$
of pegs and corners for $W$ such that:
\begin{itemize}
	\item $V(W)\subseteq Y,$
	\item  $P\subseteq X\cap Y\subseteq V(D(W)),$ and
	\item  if $\Omega$ is the cyclic ordering of the vertices $X\cap Y$ as they appear in $D(W),$
	      then there exists an $\Omega$-rendition $(\Gamma,\sigma,\pi)$ of  $G[Y].$
\end{itemize}
We say that $W$ is a \emph{flat wall}
of $G$ if it is a flat $r$-wall for some odd integer $r \geq 3.$
%

\paragraph{Flatness pairs.}
Given the above, we  say that  the choice of the 7-tuple $\mathfrak{R}=(X,Y,P,C,\Gamma,\sigma,\pi)$
\emph{certifies that $W$ is a flat wall of $G$}.
We call the pair $(W,\mathfrak{R})$ a \emph{flatness pair} of $G$ and define
the \emph{height} of the pair $(W,\mathfrak{R})$ to be the height of $W.$
We use the term \emph{cell of} $\mathfrak{R}$ in order to refer to the cells of $\Gamma.$

We call the graph $G[Y]$ the \emph{$\mathfrak{R}$-compass} of $W$ in $G,$
denoted by ${\sf Compass}_{\mathfrak{R}}(W).$
It is easy to see that there is a connected component of ${\sf Compass}_{\mathfrak{R}}(W)$ that contains the wall $W$ as a subgraph.
We can assume that ${\sf Compass}_{\frR} (W)$ is connected, updating $\frR$ by removing from $Y$ the vertices of all the connected components of ${\sf Compass}_\frR (W)$
except of the one that contains $W$ and including them in $X$ ($\Gamma, \sigma, \pi$ can also be easily modified according to the removal of the aforementioned vertices from $Y$).
We define the  \emph{flaps} of the wall $W$ in $\mathfrak{R}$ as
${\sf flaps}_{\mathfrak{R}}(W):=\{\sigma(c)\mid c\in C(\Gamma)\}.$
Given a flap $F\in {\sf flaps}_{\mathfrak{R}}(W),$ we define its \emph{base}
as $\partial F:=V(F)\cap \pi(N(\Gamma)).$

\paragraph{Flat railed annuli.}
Let $G$ be a graph and let $\mathcal{A}$ be an $(r,q)$-railed annulus of $G,$ for some odd integer $r\geq 3$ and $q\in \mathbb{N}_{\geq 3}$.
We say that $\mathcal{A}$ is a \emph{flat  $(r,q)$-railed annulus}
of $G$ if there are two laminar separations $(X_1,Y_1), (X_2,Y_2)$ of $G$, a set $Z_1$ of degree-two vertices in $V(C_1)$, and a set $Z_2$ of degree-two vertices in $V(C_2)$ such that:
\begin{itemize}
	\item $V(\mathcal{A})\subseteq Y_1\cap X_2,$
	\item  $Z_1\subseteq X_1\cap Y_1\subseteq V(C_1),$
	\item  $Z_2\subseteq X_2\cap Y_2\subseteq V(C_r),$ and
	\item if $\Omega_1$ (resp. $\Omega_2$) is the cyclic ordering of the vertices $X_1\cap Y_1$ (resp. $X_2\cap Y_2$) as they appear in $C_1$ (resp. $C_r$),
then there is an $(\Omega_1,\Omega_2)$-rendition of $G[X_1\cap Y_2]$.
\end{itemize}
We say that $\mathcal{A}$ is a \emph{flat railed annulus}
of $G$ if it is a flat $(r,q)$-railed annulus for some odd integer $r \geq 3$ and some $q\in\mathbb{N}_{\geq 3}$.

\paragraph{Railed annuli flatness pairs.}
Given the above, we  say that  the choice of the 9-tuple $\mathfrak{R}=(X_1,Y_1,X_2,Y_2,Z_1,Z_2,\Gamma,\sigma,\pi)$
\emph{certifies that $\mathcal{A}$ is a flat railed annulus of $G$}.
We call the pair $(\mathcal{A},\mathfrak{R})$ an \emph{railed annulus flatness pair} of $G$.
We use the term \emph{cell of} $\mathfrak{R}$ in order to refer to the cells of $\Gamma.$

We call the graph $G[X_1\cap Y_2]$ the \emph{$\mathfrak{R}$-compass} of $\mathcal{A}$ in $G,$
denoted by ${\sf Compass}_{\mathfrak{R}}(\mathcal{A}).$
It is easy to see that there is a connected component of ${\sf Compass}_{\mathfrak{R}}(\mathcal{A})$ that contains the cycles and the paths of $\mathcal{A}$ as subgraphs.
We can assume that ${\sf Compass}_{\frR} (\mathcal{A})$ is connected, updating $\frR$ by removing from $X_1\cap Y_2$ the vertices of all the connected components of ${\sf Compass}_\frR (\mathcal{A})$
except of the one that contains $\mathcal{A}$ and including them in $X_2$ ($\Gamma, \sigma, \pi$ can also be easily modified according to the removal of the aforementioned vertices from $X_1\cap Y_2$).
We define the  \emph{flaps} of the railed annulus $\mathcal{A}$ in $\mathfrak{R}$ as
${\sf flaps}_{\mathfrak{R}}(\mathcal{A}):=\{\sigma(c)\mid c\in C(\Gamma)\}.$
Given a flap $F\in {\sf flaps}_{\mathfrak{R}}(\mathcal{A}),$ we define its \emph{base}
as $\partial F:=V(F)\cap \pi(N(\Gamma)).$

\subsection{Influence of cycles in flat walls and flat railed annuli}

Let $G$ be a graph and let $(W,\frR)$ be either a flatness pair or a railed annulus flatness pair of $G$.

A  cell $c$ of ${\frR}$ is \emph{untidy} if  $\pi(\tilde{c})$ contains a vertex
$x\in V(W)$ such that two of the edges in $E(W)$ that are incident to $x$ are edges of $\sigma(c).$ Notice that if $c$ is untidy then  $|\tilde{c}|=3.$
A cell $c$ of $\frR$ is \emph{tidy} if it is not untidy.
The notion of tidy/untidy cell as well as the notions that we present in the rest of this subsection have been introduced in~\cite{SauST21amor}.

\paragraph{Cell classification.}
Given a graph $G$ and a set $X\subseteq V(G),$ we denote by $\partial_G (X)$ the set of vertices in $X$ that are adjacent to vertices of $G\setminus X.$

Given a cycle $C$ of ${\sf Compass}_{\mathfrak{R}}(W),$ we say that
$C$ is \emph{$\mathfrak{R}$-normal} if it is {\sl not} a subgraph of a flap $F\in {\sf flaps}_{\mathfrak{R}}(W).$
Given an $\mathfrak{R}$-normal cycle $C$ of ${\sf Compass}_{\mathfrak{R}}(W),$
we call a cell $c$ of $\mathfrak{R}$ \emph{$C$-perimetric} if
$\sigma(c)$ contains some edge of $C.$
Since every $C$-perimetric cell $c$ contains some edge of $C$ and $|\partial\sigma(c)|\leq 3,$ we observe the following.
\begin{observation}\label{label_grundfalscher}
	For every pair $(C,C')$ of $\frR$-normal cycles of ${\sf Compass}_{\frR} (W)$ such that $V(C)\cap V(C')=\emptyset,$ there is no cell of $\frR$ that is both $C$-perimetric and $C'$-perimetric.
\end{observation}
Notice that if $c$ is $C$-perimetric, then $\pi(\tilde{c})$ contains two points $p,q\in N(\Gamma)$
such that  $\pi(p)$ and $\pi(q)$ are vertices of $C$ where one,
say $P_{c}^{\rm in},$ of the two $(\pi(p),\pi(q))$-subpaths of $C$ is a subgraph of $\sigma(c)$ and the other,
denoted by $P_{c}^{\rm out},$  $(\pi(p),\pi(q))$-subpath contains at most one internal vertex of $\sigma(c),$
which should be the (unique) vertex $z$ in $\partial\sigma(c)\setminus\{\pi(p),\pi(q)\}.$
We pick a $(p,q)$-arc $A_{c}$ in $\hat{c}:={c}\cup\tilde{c}$ such that  $\pi^{-1}(z)\in A_{c}$ if and only if $P_{c}^{\rm in}$ contains
the vertex $z$ as an internal vertex.

We consider the circle  $K_{C}=\cupall\{A_{c}\mid \mbox{$c$ is a $C$-perimetric cell of $\mathfrak{R}$}\}$
and we denote by $\Delta_{C}$ the closed disk bounded by $K_{C}$ that is contained in  $\Delta$
 (in the case of a railed annulus flatness pair, $\Delta_{C}$ the closed annulus bounded by $K_{C}$ and $B_1$ that is contained in $\Delta$).
 A cell $c$ of $\mathfrak{R}$ is called \emph{$C$-internal} if $c\subseteq \Delta_{C}$
and is called \emph{$C$-external} if $\Delta_{C}\cap c=\emptyset.$
Notice that  the cells of $\mathfrak{R}$ are partitioned into  $C$-internal,  $C$-perimetric, and  $C$-external cells.

Let $c$ be a tidy $C$-perimetric cell of $\mathfrak{R}$ where $|\tilde{c}|=3.$ Notice that $c\setminus A_{c}$ has two arcwise-connected components and one of them is an open disk $D_{c}$ that is a subset of $\Delta_{C}.$
If the closure $\overline{D}_{c}$  of $D_{c}$ contains only two points of $\tilde{c}$ then we call the cell $c$ \emph{$C$-marginal}.
We refer the reader to \cite{SauST21amor} for figures illustrating the above notions.

%

\paragraph{Influence.}
For every $\mathfrak{R}$-normal cycle $C$ of ${\sf Compass}_{\mathfrak{R}}(W)$ we define the set
$${\sf influence}_{\mathfrak{R}}(C)=\{\sigma(c)\mid \mbox{$c$ is a cell of $\mathfrak{R}$ that is not $C$-external}\}.$$

We conclude this subsection by presenting a corollary of~\cite[Lemma 12]{SauST21kapiI}.

\begin{proposition}\label{label_stereotypical}
There exists a function $\newfun{label_internalization}: \mathbb{N}^3 \to \mathbb{N}$ such that if $z\in\mathbb{N}_{\geq 1},$ $x\in\mathbb{N}_{\geq 3}$ is an odd integer, $G$ is a graph, and $(W,\frR)$ is a flatness pair of $G$ of height at least $\funref{label_internalization}(z,x),$ then
there is a collection ${\cal W}=\{W_1, \ldots, W_z\}$ of $x$-subwalls of $W$ such that
\begin{itemize}
\item for every $i \in [z],$ $\cupall{\sf influence}_{\frR}(D(W_i))$ is a subgraph of ${\sf Compass}_{\frR}(W)$ and
\item for every $i,j\in[z],$ with $i\neq j,$ $V(\cupall{\sf influence}_\frR (D(W_i)))$ and $V(\cupall{\sf influence}_\frR (D(W_j)))$ are pairwise disjoint.
\end{itemize}
Moreover, $\funref{label_internalization}(z,x,p)= {\cal O}(\sqrt{z}\cdot x)$ and ${\cal W}$ can be constructed in linear time.
\end{proposition}

\subsection{Regular flatness pairs and tilts}
Let $(W,\mathfrak{R})$ be a flatness pairs of a graph $G$.
A wall $W'$  of ${\sf Compass}_{\mathfrak{R}}(W)$  is \emph{$\mathfrak{R}$-normal} if $D(W')$ is $\mathfrak{R}$-normal.
Notice that every wall of $W$ (and hence every subwall of $W$) is an $\mathfrak{R}$-normal wall of ${\sf Compass}_{\mathfrak{R}}(W).$ We denote by $\mathcal{S}_{\mathfrak{R}}(W)$ the set of all $\mathfrak{R}$-normal walls of ${\sf Compass}_{\mathfrak{R}}(W).$ Given a wall $W'\in \mathcal{S}_{\mathfrak{R}}(W)$ and a cell $c$ of $\mathfrak{R},$
we say that $c$ is \emph{$W'$-perimetric/internal/external/marginal} if $c$ is  $D(W')$-perimetric/internal/external/marginal, respectively.
We also use $K_{W'},$ $\Delta_{W'},$ ${\sf influence}_{\mathfrak{R}}(W')$ as shortcuts
for $K_{D(W')},$ $\Delta_{D(W')},$ ${\sf influence}_{\mathfrak{R}}(D(W')),$ respectively.

\paragraph{Regular flatness pairs.}
We call a  flatness pair $(W,\mathfrak{R})$ of a graph $G$ \emph{regular}
if none of its cells is $W$-external, $W$-marginal, or untidy.

\paragraph{Tilts of flatness pairs.}
Let $(W,\mathfrak{R})$ and $(\tilde{W}',\tilde{\mathfrak{R}}')$  be two flatness pairs of a graph $G$
and let $W'\in \mathcal{S}_{\mathfrak{R}}(W).$
We assume that ${\mathfrak{R}}=(X,Y,P,C,\Gamma,\sigma,\pi)$
and $\tilde{\mathfrak{R}}'=(X',Y',P',C',\Gamma',\sigma',\pi').$
We say that   $(\tilde{W}',\tilde{\mathfrak{R}}')$   is a \emph{$W'$-tilt}
of $(W,\mathfrak{R})$ if
\begin{itemize}
	\item $\tilde{\mathfrak{R}}'$ does not have $\tilde{W}'$-external cells,
	\item  $\tilde{W}'$ is a tilt of $W',$
	\item  the set of $\tilde{W}'$-internal  cells of  $\tilde{\mathfrak{R}}'$ is the same as the set of $W'$-internal
	      cells of ${\mathfrak{R}}$ and their images via $\sigma'$ and ${\sigma}$ are also the same,
	\item ${\sf Compass}_{\tilde{\mathfrak{R}}'}(\tilde{W}')$ is a subgraph of $\cupall{\sf influence}_{{\mathfrak{R}}}(W'),$ and
	\item if $c$ is a cell in $C(\Gamma') \setminus C(\Gamma),$ then $|\tilde{c}| \leq 2.$
\end{itemize}

The next observation follows from the third item above and the fact that the cells corresponding to flaps
containing a central vertex of $W'$ are all internal (recall that the height of a wall is always at least three).

\begin{observation}\label{label_surreptitiously}
	Let $(W,\frR)$ be a flatness pair of a graph $G$ and $W'\in\mathcal{S}_{\frR}(W).$
	For every $W'$-tilt $(\tilde{W}',\tilde{\frR}')$ of $(W,\frR),$ the central vertices of $W'$ belong to the vertex set of ${\sf Compass}_{\tilde{\frR}'}(\tilde{W}').$
\end{observation}

Also, given a regular flatness pair $(W,\frR)$ of a graph $G$ and a $W'\in \mathcal{S}_{\mathfrak{R}}(W),$
for every $W'$-tilt $(\tilde{W}', \tilde{\frR}')$ of $(W,\mathfrak{R}),$ by definition, none of its cells is $\tilde{W}'$-external, $\tilde{W}'$-marginal, or untidy -- thus, $(\tilde{W}', \tilde{\frR}')$ is regular.
Therefore, regularity of a flatness pair is a property that its tilts ``inherit''.

\begin{observation}\label{label_expressionism}
	If $(W,\mathfrak{R})$ is a regular flatness pair of a graph $G,$ then for every $W'\in \mathcal{S}_{\mathfrak{R}}(W),$ every $W'$-tilt of $(W,\mathfrak{R})$ is also regular.
\end{observation}

We next present one of the two main results of~\cite{SauST21amor} (see~\cite[Theorem 5]{SauST21amor}).

\begin{proposition}
\label{label_proporcionada}
There exists an algorithm that given a graph $G,$ a flatness pair $({W},{\mathfrak{R}})$ of $G,$ and a wall $W'\in \mathcal{S}_{\mathfrak{R}}(W),$ outputs  a  $W'$-tilt of $({W},{\mathfrak{R}})$ in  time $\mathcal{O}(n+m).$
\end{proposition}

We conclude this subsection with the Flat Wall theorem and, in particular, the version proved by Chuzhoy \cite{Chuzhoy15impr}, restated in our framework (see \cite[Proposition 7]{SauST21amor}).

\begin{proposition}\label{label_aldobrandesco}
	There exist two functions  $\newfun{@endurcissement}:\mathbb{N}\to \mathbb{N}$  and
	$\newfun{@connaissions}:\mathbb{N}\to \mathbb{N},$ where the images of $\funref{@endurcissement}$ are odd numbers, such that if $r \in \mathbb{N}_{\geq 3}$ is an odd integer, $t\in\mathbb{N}_{\geq 1},$
	$G$ is a graph that does not contain $K_t$ as a minor,  and  $W$ is an $\funref{@endurcissement}(t)\cdot r$-wall of $G,$
	then there is a set $A\subseteq V(G)$ with $|A|\leq \funref{@connaissions}(t)$
	and a flatness pair $(\tilde{W}',\tilde{\mathfrak{R}}')$ of $G\setminus A$ of height $r.$
	Moreover, $\funref{@endurcissement}(t)=\mathcal{O}(t^{2})$ and $\funref{@connaissions}(t)=t-5.$
\end{proposition}

\subsection{Flat walls with compasses of bounded treewidth}
\label{subsec_flatwalls_boundedtwcompass}

The following result was proved in~\cite[Theorem 8]{SauST21amor}.
It is a version of the Flat Wall theorem, originally proved in~\cite{RobertsonS95GMXIII}.
The proof in~\cite[Theorem 8]{SauST21amor} is strongly based on
the proof of an improved version of the Flat Wall theorem given by of Kawarabayashi, Thomas, and Wollan~\cite{KawarabayashiTW18anew} (see also~\cite{Chuzhoy15impr,GiannopoulouT13opti}).
\begin{proposition}\label{prop_flatwallbdtw}
	There is a function $\newfun{@resplandecientes}:\mathbb{N}\to \mathbb{N}$    and
	an algorithm that receives as  input  a graph $G,$ an odd integer $r\geq 3,$ and a  $t\in\mathbb{N}_{\geq 1},$ and  outputs, in time $2^{\mathcal{O}_{t}(r^2)}\cdot n,$ one of the following:
\begin{itemize}
\item a report  that $K_{t}$ is a minor of $G,$
\item a tree decomposition of $G$ of width at most $\funref{@resplandecientes}(t)\cdot r,$ or
\item a set $A\subseteq V(G),$  where $|A|\leq \funref{@connaissions}(t),$
a regular flatness pair $(W,\mathfrak{R})$ of $G\setminus A$ of height $r,$
and a tree decomposition of the $\mathfrak{R}$-compass of $W$
of width at most $\funref{@resplandecientes}(t)\cdot r.$
(Here $\funref{@connaissions}(t)$ is the function of \autoref{label_aldobrandesco} and
$\funref{@resplandecientes}(t)=2^{\mathcal{O}(t^2 \log t)}.$)
\end{itemize}
\end{proposition}

Given graphs $H$ and $G$,
we say that a subgraph $M$ of $G$ is a \emph{minor-model} of $H$ in $G$ if
there is a partition of the vertex set of $M$ to sets $V_1,\ldots, V_{|V(H)|}$ such that
for every $i\in[|V(H)|]$, $G[V_i]$ is connected and the graph obtained from $G$ after contracting the edges of each $G[V_i]$ is isomorphic to $H$.
Following the version of the Flat Wall theorem in~\cite{KawarabayashiTW18anew},
\autoref{prop_flatwallbdtw} can be modified so as when it reports that $K_t$ is a minor of $G$, it also outputs a minor-model of $K_t$ in $G$.

\subsection{Levelings and well-aligned flatness pairs}
\label{subsec_levelings}
Let $G$ be a graph and let $(W,\mathfrak{R})$ be either a flatness pair or a railed annulus flatness pair of $G.$
If  $(W,\mathfrak{R})$ is a flatness pair of $G$, let  $\mathfrak{R}=(X,Y,P,C,\Gamma,\sigma,\pi),$  where $(\Gamma,\sigma,\pi)$ is an $\Omega$-rendition of $G[Y]$ and $\Gamma=(U,N)$ is a $\Delta$-painting.
If  $(W,\mathfrak{R})$ is a railed annulus flatness pair of $G,$
let $\mathfrak{R}=(X_1,Y_1,X_2,Y_2,Z_1,Z_2,\Gamma,\sigma,\pi),$  where $(\Gamma,\sigma,\pi)$ is an $(\Omega_1,\Omega_2)$-rendition of $G[Y_1\cap X_2]$ and $\Gamma=(U,N)$ is a $\Delta$-painting, for some closed annulus $\Delta$.
In both cases, we define the \emph{ground set} of $W$ in ${\mathfrak{R}}$ to be the set ${\sf ground}_{\mathfrak{R}}(W):=\pi(N(\Gamma))$ and we refer to the vertices of this set as the \emph{ground vertices} of the $\mathfrak{R}$-compass of $W$ in $G.$
Notice  that ${\sf ground}_{\mathfrak{R}}(W)$ may  contain vertices
of ${\sf compass}_{\mathfrak{R}}(W)$ that are not necessarily vertices in $V(W).$

\paragraph{Levelings.}
We define  the $\mathfrak{R}$-\emph{leveling}  of $W$ in $G,$
denoted by ${W}_{\mathfrak{R}},$ as the bipartite graph
where  one part is the ground set of $W$ in $\mathfrak{R},$ the  other part is a set ${\sf vflaps}_{\mathfrak{R}}(W)=\{v_{F}\mid F\in {\sf flaps}_{\mathfrak{R}}(W)\}$ containing one new vertex $v_{F}$ for each flap  $F$ of $W$ in $\mathfrak{R},$
and, given  a pair $(x,F)\in {\sf ground}_{\mathfrak{R}}(W)\times {\sf flaps}_{\mathfrak{R}}(W),$   the set $\{x,v_F\}$ is an edge of ${W}_{\mathfrak{R}}$ if and only if
$x\in \partial F.$ We call the vertices of ${\sf ground}_{\mathfrak{R}}(W)$ (resp. ${\sf vflaps}_{\mathfrak{R}}(W)$) \emph{ground-vertices} (resp. \emph{flap-vertices}) of ${W}_{\mathfrak{R}}.$
Notice that the incidence graph of the plane hypergraph $(N(\Gamma),\{\tilde{c}\mid c\in C(\Gamma)\})$ is isomorphic to ${W}_{\mathfrak{R}}$
via an isomorphism that extends  $\pi$ and, moreover, bijectively corresponds cells to flap-vertices.
If  $(W,\mathfrak{R})$ is a flatness pair of $G$, this permits us to treat ${W}_{\mathfrak{R}}$ as a $D$-embedded graph, for some closed disk $D$, where  $\bd(D)\cap {W}_{\mathfrak{R}}$ is the set $X\cap Y.$
If  $(W,\mathfrak{R})$ is a railed annulus flatness pair of $G,$
we can treat $\mathcal{A}_{\mathfrak{R}}$ as a $\Delta$-embedded graph, for some closed annulus $\Delta$ with boundaries $B_1$ and $B_2$, where  $B_1\cap \mathcal{A}_{\mathfrak{R}}$ is the set $X_1\cap Y_1$ and $B_2\cap \mathcal{A}_{\mathfrak{R}}$ is the set $X_2\cap Y_2$.

\paragraph{Representations in flatness pairs and railed annulus flatness pairs.}
{We denote by $W^{\bullet}$ the  graph obtained from $W$ if we subdivide \emph{once} every
edge of $W$ that is {short} in ${\sf compass}_{\mathfrak{R}}(W).$}
The graph $W^\bullet$  is  a ``slightly richer variant'' of $W$  that is necessary for our definitions and  proofs, namely to be able to associate  every flap-vertex of  an appropriate subgraph of $W_{\mathfrak{R}}$ (that we will denote by $R_{W}$) with  a non-empty path of $W^\bullet,$ as we proceed to formalize.
We say that $(W,\mathfrak{R})$ is \emph{well-aligned} if the following holds:
\begin{quote}
	$W_{\mathfrak{R}}$ contains as a subgraph an $r$-wall $R_{W}$
	where ${D(R_{W})}=D({W}_{\mathfrak{R}})$ and $W^{\bullet}$  is isomorphic to some subdivision of  $R_{W}$
	via an isomorphism that maps each ground vertex to itself.
\end{quote}
Suppose now that the flatness pair $(W,\mathfrak{R})$ is well-aligned.
We call the wall  $R_{W}$ in the above condition  a \emph{representation} of $W$ in $W_{\mathfrak{R}}.$

\begin{proposition}[\cite{SauST21amor}]
	\label{prop_wellaligned}
	If a flatness pair $({W},\mathfrak{{R}})$  is regular, then it is also well-aligned.
	Moreover, there is an {$\mathcal{O}(n)$}-time algorithm that, given $G$ and such a $({W},\mathfrak{{R}}),$ outputs a representation $R_{W}$
	of $W$ in $W_{\mathfrak{R}}.$
\end{proposition}

\paragraph{Well-aligned railed annulus flatness pairs.}
{We denote by $\mathcal{A}^{\bullet}$ the  graph obtained from $\mathcal{A}$ if we subdivide \emph{once} every
	edge in $E(\mathcal{A})$ that is {short} in ${\sf compass}_{\mathfrak{R}}(\mathcal{A}).$}
The graph $\mathcal{A}^\bullet$  is  a ``slightly richer variant'' of $\mathcal{A}$  that is necessary for our definitions and  proofs, namely to be able to associate  every flap-vertex of  an appropriate subgraph of $\mathcal{A}_{\mathfrak{R}}$ (that we will denote by $R_{\mathcal{A}}$) with  a non-empty path of $\mathcal{A}^\bullet,$ as we proceed to formalize.
We say that $(\mathcal{A},\mathfrak{R})$ is \emph{well-aligned} if the following holds:
\begin{quote}
	$\mathcal{A}_{\mathfrak{R}}$ contains as a subgraph an $(r,q)$-railed annulus $R_\mathcal{A}$
	where the first (resp. last) cycle of $\mathcal{A}_{W}$ is the same as the first (resp. last) cycle of $\mathcal{A}_{\mathfrak{R}}$ and $\mathcal{A}^{\bullet}$  is isomorphic to some subdivision of  $R_\mathcal{A}$
	via an isomorphism that maps each ground vertex to itself.
\end{quote}
Suppose now that the railed annulus flatness pair $(\mathcal{A},\mathfrak{R})$ is well-aligned.
We call the wall  $R_\mathcal{A}$ in the above condition  a \emph{representation} of $\mathcal{A}$ in $\mathcal{A}_{\mathfrak{R}}.$ Note that, as $R_\mathcal{A}$ is a subgraph of $\mathcal{A}_{\mathfrak{R}},$ it is bipartite as well.
The above property gives us a way to represent a flat wall by a wall of its leveling
in a way that ground vertices are not altered.

Notice that both $\mathcal{A}_{\mathfrak{R}}$ and its subgraph $R_\mathcal{A}$ can be seen as $\Delta$-embedded graphs where $B_1\cap \mathcal{A}_{\mathfrak{R}}=B_1\cap R_\mathcal{A}$ and these are both subsets of the vertex set of the first cycle of $\mathcal{A}_{\mathfrak{R}}$ (resp. of $R_\mathcal{A}$) and $B_2\cap \mathcal{A}_{\mathfrak{R}}=B_2\cap R_\mathcal{A}$
and these are both subsets of the vertex set of the last cycle of $\mathcal{A}_{\mathfrak{R}}$ (resp. of $R_\mathcal{A}$).
This establishes a bijection $\delta$ from the set of cycles of $\mathcal{A}$  to the set of cycles of $R_\mathcal{A}.$

 Let $G$ be a graph, let $(\mathcal{A},\mathfrak{R})$ be a well-aligned railed annulus flatness pair of $G$.
We set ${\sf Leveling}_{(\mathcal{A},\mathfrak{R})}(G)$ to be the graph obtained from $G$ after replacing ${\sf Compass}_{\mathfrak{R}}(\mathcal{A})$ with $\mathcal{A}_{\mathfrak{R}}$.
We now show the next result.

\begin{lemma}\label{lem_levelingpaths}
Let $G,H$ be graphs and let $(\mathcal{A},\mathfrak{R})$ be a well-aligned railed annulus flatness pair of $G$.
Also, let  $s_1,t_1,\ldots, s_k,t_k\in V(G)$ such that for every $i\in[k]$, $s_i,t_i\notin {\sf Compass}_{\mathfrak{R}}(\mathcal{A})$.
Then there is a linkage $L$ in $G$ such that $T(L)=\{s_1,t_1,\ldots, s_k,t_k\}$ if and only if there is a linkage $L'$ in ${\sf Leveling}_{(\mathcal{A},\mathfrak{R})}(G)$ such that $L\equiv L'$.
\end{lemma}


\begin{proof}
Let $P_1,\ldots, P_k$ be the paths of $L$ in $G$.
Since for every $i\in[k]$, $s_i,t_i\notin {\sf Compass}_{\mathfrak{R}}(\mathcal{A})$, no flap $F\in {\sf flaps}_{\mathfrak{R}}(\mathcal{A})$ does contain any endpoint of $P_1,\ldots, P_k$.
Therefore, for every flap $F\in {\sf flaps}_{\mathfrak{R}}(\mathcal{A})$ 
there is at most one $i\in[k]$ such that $P_i$ contains some vertex of $F\setminus \partial F$.
For every $i\in[k]$, let $\mathcal{Q}_i=\{P_i\cap (F\setminus \partial F)\mid F\in  {\sf flaps}_{\mathfrak{R}}(\mathcal{A})\}.$
Observe that each $Q\in \mathcal{Q}_i$ corresponds to a single vertex $v_Q$ of ${\sf Leveling}_{(\mathcal{A},\mathfrak{R})}(G)$.
Therefore, by replacing, for each $i\in[k]$, each $Q\in \mathcal{Q}_i$ by $v_Q$, we get a linkage $L'$
of paths $P_1',\ldots, P_k'$ in ${\sf Leveling}_{(\mathcal{A},\mathfrak{R})}(G)$ that is equivalent to $L$.
For the reverse implication, notice that every collection $P_1',\ldots, P_k'$ of disjoint paths in ${\sf Leveling}_{(\mathcal{A},\mathfrak{R})}(G)$ corresponds to a collection $P_1,\ldots, P_k$ of disjoint paths in $G$, where for every $i\in[k]$, $P_i$ is obtained by replacing each vertex $v_F$ that is contained in $P_i'$ by a path in $F$ that connects the two neighbors of $v_F$ in $P_i'$ (these neighbors belong to $\partial F$).
\end{proof}

\newpage

\section{Missing complexity proofs}\label{@ressemblerois}
\label{@zeicheneinheiten}

\subsection{Hardness of {\sc Ordered Linkability}}
\label{@imperceptible}

In this subsection we treat the parameterized complexity of {\sc Ordered Linkability} which we restate below.

\medskip
\fbox{
\begin{minipage}{15cm}
\noindent{\sc Ordered Linkability}\\
\noindent{\sl Input}: a graph $G$, $R\subseteq V(G)$, and a graph $H$ where $k=|H|$.\\
\noindent{\sl Question}: is    $R$ $H$-linked in $G$?
\end{minipage}
}
\medskip

We consider the above  problem in the case where the graph $H$ is the 1-regular graph on $2k$ vertices.
In this case,  the task is, given a graph $G$, a subset of vertices $R\subseteq V(G)$, and an integer $k\geq 1$, decide whether for every sequence $s_1,\ldots, s_k,t_1,\ldots,t_k$ of distinct  vertices of $R$, the graph $G$ has pairwise disjoint vertex-disjoint paths between $s_i$ and $t_i$, for $i\in[k]$.

\begin{theorem}\label{thm_link_hard}
{\sc Ordered Linkability} cannot be solved in time $\mathcal{O}_{k}(n^{\mathcal{O}(1)})$ unless ${\sf FPT}={\sf W[1]}$.
\end{theorem}

\begin{proof}
Given a graph $G$, we use $\omega(G)$ to denote the maximum $r$ such that $G$ contains $K_r$ as a subgraph.
Recall that the \textsc{Clique} problem asks, given a graph $G$ and a positive integer $k$, whether $\omega(G)\geq k$. It is well-known that \textsc{Clique} is {\sf W[1]}-hard~\cite{DowneyF13fund} when parameterized by $k$. Furthermore, the optimization version is hard from the parameterized approximation viewpoint. In particular,  by the breakthrough result of Lin~\cite{Lin21cons}, for any positive constant $c\geq 1$, no algorithm running in time $\mathcal{O}_{k}(n^{\mathcal{O}(1)})$ for a computable function $f(k)$ can distinguish between the cases $\omega(G)\geq k$ and $\omega(G)<k/c$, unless  ${\sf FPT}={\sf W[1]}$. We use this result for our reduction.

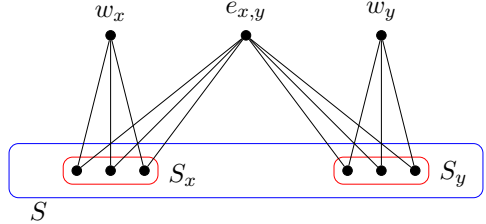
\begin{figure}[ht]
\centering
\scalebox{.9}{
\begin{tikzpicture}
\node[simple, label=$w_x$] (wx) at (0,2) {}; 
\node[simple, label=$e_{x,y}$] (exy) at (2,2) {}; 
\node[simple, label=$w_y$] (wy) at (4,2) {}; 
\draw[red,rounded corners] (-.7,-0.2) rectangle (0.7, 0.2) {};
\node[label=$S_x$, anchor=east] () at (1.2,-0.5) {};
\draw[red,rounded corners] (3.3,-0.2) rectangle (4.7, 0.2) {};
\node[label=$S_y$, anchor=east] () at (5.2,-0.5) {};
\draw[blue,rounded corners] (-1.5,-0.4) rectangle (5.5, 0.4) {};
\node[label=$S$, anchor=west] () at (-1.2,-1) {};

\node[simple] (wx1) at (-.5,0) {}; 
\node[simple] (wx2) at (0,0) {}; 
\node[simple] (wx3) at (0.5,0) {}; 
\node[simple] (wy1) at (3.5,0) {}; 
\node[simple] (wy2) at (4,0) {}; 
\node[simple] (wy3) at (4.5,0) {}; 
\foreach \i in {1,2,3} \draw[-] (wx) -- (wx\i) -- (exy) -- (wy\i) -- (wy);
\end{tikzpicture}
}
\caption{Construction of $G'$.}
\label{fig_gprimecliquereduction}
\end{figure}

Consider an instance $(G,k)$ of \textsc{Clique}. Without loss of generality,
we assume that $k=4r+1$ for an integer $r\geq 2$. We construct the graph $G'$ as follows (see Figure~\ref{fig_gprimecliquereduction}).
\begin{itemize}
\item For every vertex $x\in V(G)$, construct a set $S_x$ of $(k-1)/2$  vertices
(note that $k-1$ is even) and then form a clique from $S=\bigcup_{x\in V(G)}S_x$ by making the vertices pairwise adjacent.
\item For every vertex $x\in V(G)$, construct a vertex $w_x$ and make it adjacent to the vertices of $S_x$; denote $W=\{w_x\colon x\in V(G)\}$.
\item For every edge $\{x,y\}\in E(G)$, construct a vertex $e_{xy}$ and make it adjacent to the vertices of $S_x$ and $S_y$; denote $L=\{e_{xy}\colon \{x,y\}\in E(G)\}$.
\end{itemize}
We set $R=W\cup L$ and $k'=\frac{1}{2}\binom{k}{2}+1$; note that $k'$ is an integer because $k$ is odd and $k-1$ is divisible by 4.

First, we show that if $G$ contains a clique of size $k$ as a subgraph, then there are disjoint $k'$-tuples $(s_1,\ldots,s_{k'})$ and $(t_1,\ldots,t_{k'})$ of vertices of $R$ such that $G'$ has no vertex-disjoint paths between $s_i$ and $t_i$, for $i\in[k']$.

We use the following observation.

\begin{claim}\label{cl_pairs}
Let $H$ be a complete graph with $k=4r+1$ vertices for $r\geq 1$.
Then there is a partition of $E(H)$ into $\ell=(4r+1)r$ pairs $\{e_i,e_i'\}$ for $i\in[\ell]$ such that $e_i$ and $e_i'$ have no common endpoints.  
\end{claim}

\medskip

\noindent\emph{Proof of~\autoref{cl_pairs}.}
The proof is by induction on $r$. If $r=1$, then the partition is shown in Figure~\ref{@misinterpretations} (a). If $r\geq 2$, then we select four arbitrary vertices $v_1,v_2,v_3,v_4$ of $H$ and partition the edges incident to them using the pattern shown in Figure~\ref{@misinterpretations} (b). The remaining edges are partitioned into pairs using the inductive assumption.


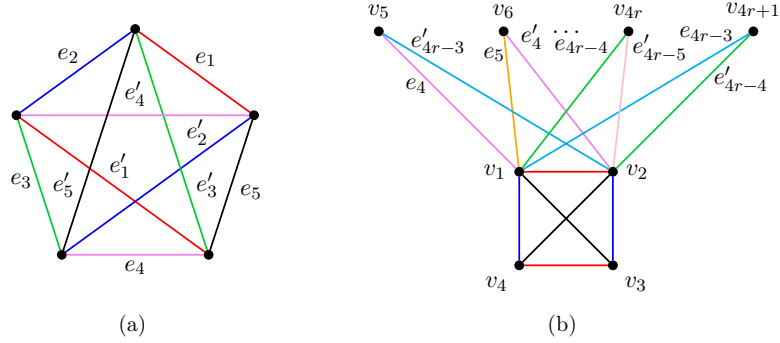
\begin{figure}[ht]
\centering
\scalebox{.83}{
\begin{subfigure}[b]{0.4\textwidth}
\centering
\begin{tikzpicture}
\node[simple] (v1) at (18:2) {};
\node[simple] (v2) at (90:2) {};
\node[simple] (v3) at (162:2) {};
\node[simple] (v4) at (234:2) {};
\node[simple] (v5) at (306:2) {};
\draw[thick,red] (v1) -- (v2) (v3) -- (v5);
\draw[thick,blue] (v2) -- (v3) (v4) -- (v1);
\draw[thick,green!80!blue] (v3) -- (v4)  (v5) -- (v2);
\draw[thick,violet] (v4) -- (v5) (v1) -- (v3);
\draw[thick, black] (v1) -- (v5) (v2) -- (v4);
\node[label=$e_1$] () at (45:1.6) {}; 
\node[label=$e_2$] () at (132:1.6) {}; 
\node[label=$e_3$] () at (204:2) {}; 
\node[label=$e_4$] () at (270:2.2) {}; 
\node[label=$e_5$] () at (-28:2.1) {}; 
\node[label=$e_1'$] () at (250:0.7) {};
\node[label=$e_2'$] () at (-7:1) {};
\node[label=$e_3'$] () at (-40:1.5) {};
\node[label=$e_4'$] () at (0,.5) {};
\node[label=$e_5'$] () at (220:1.5) {};
\end{tikzpicture}
\caption{\ }
\end{subfigure}}
\scalebox{.83}{
\begin{subfigure}[b]{0.4\textwidth}
\centering
\begin{tikzpicture}
\node[simple, label={south west:$v_4$}] (u4) at (0,0) {};
\node[simple, label={south east:$v_3$}] (u3) at (1.5,0) {};
\node[simple, label={east:$v_2$}] (u2) at (1.5,1.5) {};
\node[simple, label={west:$v_1$}] (u1) at (0,1.5) {};

\node[label={center:$\ldots$}] (c) at ($(0.75,0.75)+(90:3)$) {};
\node[simple, label={above:$v_5$}] (u5) at ($(c)+(-3,0)$) {};
\node[simple, label={above:$v_6$}] (u6) at ($(c)+(-1,0)$) {};
\node[simple, label={above:$v_{4r}$}] (u7) at ($(c)+(1,0)$) {};
\node[simple, label={above:$v_{4r+1}$}] (u8) at ($(c)+(3,0)$) {};

\draw[thick, red] (u1) -- (u2) (u3) -- (u4);
\draw[thick,blue] (u2) -- (u3) (u4) -- (u1);
\draw[thick,black] (u1) -- (u3) (u2) -- (u4);
\draw[thick, violet] (u1) -- (u5) (u2) -- (u6);
\draw[thick,orange]  (u1) -- (u6);
\draw[thick,cyan]  (u1) -- (u8) (u2) -- (u5) ;
\draw[thick,green!80!blue] (u2) -- (u8) (u1) -- (u7);
\draw[thick,pink] (u2) -- (u7);
\node[label=$e_4$] () at (-1.65,2.5) {};
\node[label=$e_4'$] () at (0.2,3.2) {};
\node[label=$e_{4r-3}'$] () at (-1.3,3.1) {};
\node[label=$e_5$] () at (-0.4,3) {};
\node[label=$e_{4r-4}'$] () at (3.55,2.5) {};
\node[label=$e_{4r-4}$] () at (1,3.1) {};
\node[label=$e_{4r-3}$] () at (3,3.3) {};
\node[label=$e_{4r-5}'$] () at (2.25,3) {};
\end{tikzpicture}
\caption{\ }
\end{subfigure}
}
\caption{Pairing edges of $H$; the edges incident to $v_3$ and $v_4$ are not shown and are partitioned in pairs in the same pattern as the edges incident to $v_1$ and $v_2$.}
\label{@misinterpretations}
\end{figure}
\hfill$\diamond$

\bigskip

Suppose that $G$ contains a clique $H$ of size $k$ as a subgraph.
By~\autoref{cl_pairs}, the set of edges of $H$ can be partitioned into
$\ell=(4r+1)r=k'-1$ pairs $\{e_i,e_i'\}$ for $i\in[\ell]$ such that
$e_i$ and $e_i'$ have no common endpoints.
Let $e_i=\{x_i,y_i\}$ and $e_i'=\{x_i',y_i'\}$ for $i\in[\ell]$.
We consider the vertices $s_i=e_{x_iy_i}$ and $t_i=e_{x_i'y_i'}$ of $G'$ for $i\in[\ell]$.
We also define $s_{k'}=w_x$ and $t_{k'}=w_y$ for arbitrary distinct $x,y\in V(H)$.
The $k'$-tuples $(s_1,\ldots,s_{k'})$ and $(t_1,\ldots,t_{k'})$ are disjoint
and consist of distinct vertices of $R$.
Because $e_i$ and $e_i'$ have no common
endpoints for every $i\in[\ell]$,
any path between $s_i$ and $t_i$ in $G'$
contains two edges incident to $s_i$ and $t_i$,
respectively.
Because $x$ and $y$ are distinct,
the same holds for any path in $G'$ between $s_{k'}$ and $t_{k'}$.
Suppose that $G'$ has vertex-disjoint paths
between $s_i$ and $t_i$, for $i\in[k']$.
Then the edges of the paths incident
to the vertices $s_1,\ldots,s_{k'}$ and $t_1,\ldots,t_{k'}$ should form a matching.
However, $N_{G'}(\{s_1,\ldots,s_{k'}\}\cup\{t_1,\ldots,t_{k'}\})=\bigcup_{x\in V(H)}S_x$ and 
$|\bigcup_{x\in V(H)}S_x|=k(k-1)/2<2k'$.
Therefore, by Hall's theorem~\cite{Hall45anex},
$G'$ has no matching saturating the vertices of
$\{s_1,\ldots,s_{k'}\}\cup\{t_1,\ldots,t_{k'}\}$,
i.e., $G'$ has no matching $M$ such that the set of the endpoints of the edges in $M$ contains $\{s_1,\ldots,s_{k'}\}\cup\{t_1,\ldots,t_{k'}\}$).
Thus, there are no vertex-disjoint paths between $s_i$ and $t_i$ for $i\in[k']$ in $G'$.

Now we show that if $G$ has no clique with at least $k/5$ vertices, then 
for all pairs of disjoint $k'$-tuples of distinct vertices $(s_1,\ldots,s_{k'})$ and
$(t_1,\ldots,t_{k'})$, the graph $G$ has vertex-disjoint paths between $s_i$ and $t_i$,
for all $i\in[k']$.
Consider arbitrary $(s_1,\ldots,s_{k'})$ and $(t_1,\ldots,t_{k'})$
and set $X=\{s_1,\ldots,s_{k'}\}\cup\{t_1,\ldots,t_{k'}\}$. We prove the following claim.

\begin{claim}\label{@exchangeable}
The graph $G'$ has a matching $M$ saturating every vertex of $X$.
\end{claim}

\noindent\emph{Proof of~\autoref{@exchangeable}.}
We set $Y=X\cap L$. For every $x\in V(G)$, let $S_x'\subset S_x$ be an arbitrary subset of $S_x$ of size $(k-1)/2-1$. We define $S'=\bigcup_{x\in V(G)}S_x'$.  We claim that 
$H=G[Y\cup S']$ has a matching $M'$ saturating every vertex of $Y$. Because $Y$ is an independent set,  we can apply Hall's theorem~\cite{Hall45anex} and observe that it is sufficient to show that for every $Z\subseteq Y$, $|N_H(Z)|\geq |Z|$. Consider a set  $Z\subseteq Y$.
 Let $F=\{\{x,y\}\in E(G)\colon e_{xy}\in Z\}$ and let $U$ be the set of vertices of $G$ incident to the edges of $F$. By Turan's theorem~\cite{Turan41eine}, $|F|\leq\big(1-\frac{5}{k}\big)\frac{|U|^2}{2}$, because $G$ has no clique with at least $k/5$ vertices. Therefore, $\sqrt{\frac{2|F|k}{k-5}}\leq |U|$. 
By the construction of $H$, we have that $|N_H(Z)|=|U|\frac{k-3}{2}$, because $|S_x'|=(k-3)/2$ for every $x\in X$. We obtain that  $|N_H(Z)|\geq \sqrt{\frac{2|F|k}{k-5}}\cdot\frac{k-3}{2}$. 
Because $|Z|=|F|$, it is sufficient to show that  $|Z|\leq \sqrt{\frac{2|Z|k}{k-5}}\cdot\frac{k-3}{2}$. This inequality is equivalent to $|Z|\leq\frac{k(k-3)^2}{2(k-5)}$ and holds because
$\frac{k(k-3)^2}{2(k-5)}\geq \binom{k}{2}+2=2k'\geq |Z|$. We conclude that $H$ has a matching $M'$ saturating every vertex of $Y$.

Note that for every $w_x\in X\cap W$, there is an adjacent vertex $v\in S_x\setminus S_x'$.  By adding $\{w_x,v\}$ to $M'$ for each $w_x\in X\cap W$, we obtain a matching $M$ in $G'$ saturating every vertex of $X$. This concludes the proof of the claim.
\hfill$\diamond$
\medskip

Using~\autoref{@exchangeable}, we construct the paths between $s_i$ and $t_i$ in $G'$ as follows. Let $i\in[k']$. The matching $M$ contains the edges $\{s_i,u\}$ and $\{t_i,v\}$, for some $u,v\in S$. Because $S$ is a clique, $(s_i,u,v,t_i)$ is an $(s_i,t_i)$-path. Since $M$ is a matching, all these paths are vertex-disjoint. 

To conclude the proof, note that the existence of an algorithm solving \textsc{Ordered Linkability} in time $\mathcal{O}_{k}(n^{\mathcal{O}(1)})$ would imply that there is an algorithm distinguishing the cases $\omega(G)\geq k$ and $\omega(G)<k/5$ in time $f(\frac{1}{2}\binom{k}{2}+1)\cdot n^{\mathcal{O}(1)}$, contradicting the result of Lin~\cite{Lin21cons}.
\end{proof}

\subsection{{\sc Monochromatic Path Topological Minor} is  {\sf W[1]}-hard on planar graphs}
\label{@thesprotians}

\paragraph{Topological minor models.}
Given a graph $G$ and a graph $H$, we say that $H$ is a \emph{topological minor}
of $G$ if there is an injection $\varphi:  V(H)\to V(G)$
and a function $\psi$ mapping the edges of $H$ to paths of $G$ such that 
\begin{itemize}
\item For every distinct $e_{1},e_{2}$, $\psi(e_1)$ and $\psi(e_{2})$ are internally vertex disjoint paths of $G$ and 
\item For every $e=\{x,y\}\in E(G)$, $\psi(e)$ is a path joining $\varphi(x)$ and $\varphi(y)$.
\end{itemize} 
Given the above, we say that $H$ is a \emph{topological minor} of $G$ \emph{via the pair} $(\varphi,\psi)$.

%
%

We consider the following problem:

\medskip
\fbox{
\begin{minipage}{15cm}
\noindent{\sc  Monochromatic Path Topological Minor}\\
\noindent{\sl Input}:   multicolored graph $(G,X_{1},\ldots,X_{z})$ a graph $H$ and a coloring\\
\phantom{\noindent{\sl Input}: }function $\lambda: V(H)\to [z]$.\\
\noindent{\sl Question}: does $G$ contain $H$ as a topological  minor via a pair
$(\varphi,\psi)$ where
\medskip

~~~~~$\bullet$   for every $x\in V(H)$, $\varphi(x)\in X_{\lambda(x)}$ (i.e., the image of $x$ carries the color of $x$) and 

~~~~~$\bullet$ every $e\in E(H)$, there is some $i\in [z]$ such that $V(\psi(e))\subseteq V_{i}$ (i.e., the path $\psi(e)$ \\ \phantom{~~~~~$\bullet$ }is monochromatic)
\end{minipage}
}
\medskip

\begin{theorem}
\label{@specialisation}
{\sc  Monochromatic Path Topological Minor}, when parameterized by $k=|H|$,  is {\sf W[1]}-hard on planar graphs, even when $z=4$ and $H$ is the $(k\times k)$-grid.
\end{theorem}

\begin{proof}
We give a parameterized reduction from the following problem:
\medskip

\fbox{
\begin{minipage}{15cm}
\noindent{\sc Grid Tiling}\\
\noindent{\sl Input}: two integers $D,k\in\mathbb{N}$ and a function $M:[k]^2\to 2^{[D]^{2}}$.\\
\noindent{\sl Question}: is there a function $s:[k]^2\to {[D]^{2}}$  such that 
 \medskip

~~~~~$\bullet$ for every $(i,j)\in[k],$ $s(i,j)\in M(i,j)$,

~~~~~$\bullet$ for every $i\in[k],$ all first coordinates of the pairs in $\{s(i,j)\mid j\in[k]\}$ are equal,   \\\phantom{~~~~~$\bullet$} and

~~~~~$\bullet$ for every $j\in[k],$ all  second coordinates of the pairs in $\{s(i,j)\mid i\in[k]\}$ are equal %
\end{minipage}
}\medskip

Intuitively, one may see the input of the {\sc Grid Tiling} problem as the assignment of pairs in $[D]^2$
to the cells of a $k\times k$-matrix an the question is whether it is possible to choose one pair 
from each cell so that, in the occurring $k\times k$-matrix, vertical pairs  agree in the 1st coordinate
and horizontal pairs agree in the 2nd coordinate. It is known (see \cite{CyganFKLMPPS15para})
that {\sc Grid Tiling},when parameterized by $k$, is {\sf W[1]}-hard.
Given an instance $(D,k,M)$ of {\sc Grid Tiling}, we build an instance $(G,B,C,O,T,H,\lambda)$ of {\sc Monochromatic Path Topological Minor} as follows.

\begin{figure}[ht]
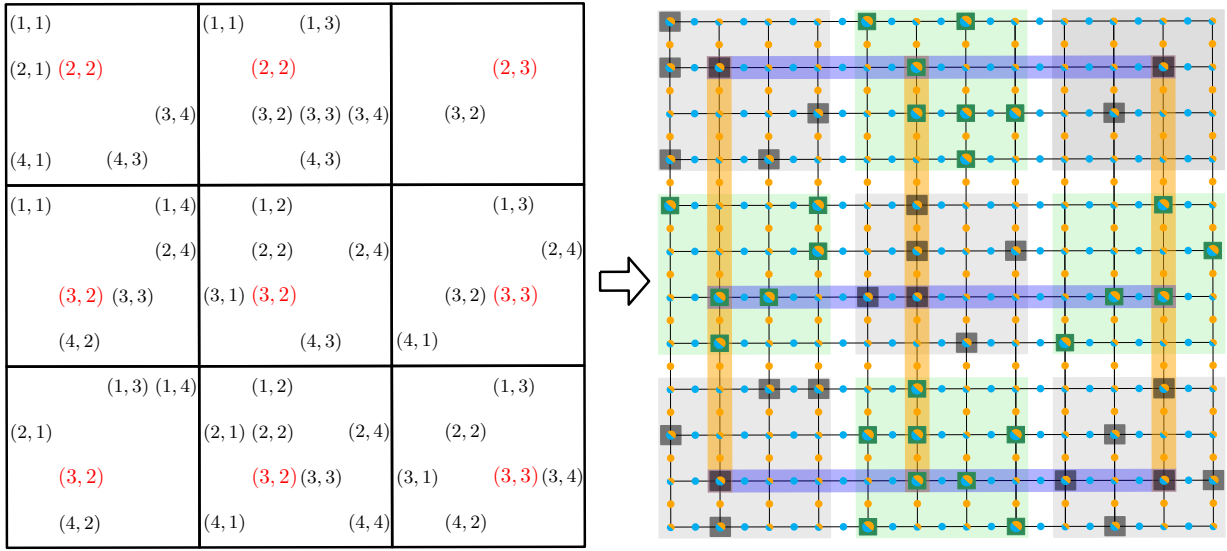

{\hspace{0.1cm}\scalebox{.79}{
}}
\caption{An example of the reduction of \autoref{@specialisation}.
On the left an instance of  {\sc Grid Tiling} is depicted where $D=4$ and $k=3$.
On the right the corresponding instance $(G,B,\green{C},\darkorange{O},\blue{T},H,\lambda)$
of {\sc Monochromatic Path Topological Minor} is depicted, where $H$ is the $(k\times k)$-grid.
The solution for the {\sc Grid Tiling}, depicted by the red pairs, corresponds to a realization of  the $(k\times k)$-grid on the right, depicted by the (long)  \darkorange{orange} and \blue{blue} rectangles.
}
\label{fig_grid_red}
\end{figure}

Consider a $(kD\times kD)$-grid $\Gamma$. Also define $\{\Gamma_{i,j}\mid (i,j)\in [k]^2\}$ to be the (unique) collection of $k^2$ pairwise vertex disjoint $(D\times D)$-grids of $\Gamma$.
We define the black and green color
sets, namely $B$ and $C$ so that for every 
$(i,j)\in[k]^2$ we include, for every pair $(x,y)\in M(i,j)$, 
the intersection vertex of the $x$-th column and the $y$-th row of $\Gamma_{i,j}$ 
 in the set $B$ (if $i+j = 0\pmod 2$) or in the set $C$ (if $i+j = 1\pmod 2$).

The graph $G$ is obtained by subdividing every vertex of $\Gamma$ once.
Clearly, $G$ consists of $kD$ vertical paths and $kD$ horizontal paths. The orange and the turquoise color
sets, namely $O$ anf $T$ are defined so that $O$ contains  
every vertex of a vertical path and $T$ contains every vertex of an horizontal path.
We now consider the $(k\times k)$-grid $H$ and $\lambda$ is a proper two coloring of $H$ in black and green.

It now remains to see that $(D,k,M)$ is a \yes-instance of  {\sc Grid Tiling}
if and only if   the tuple $(G,B,C,O,T,H,\lambda)$  is a \yes-instance of  {\sc Monochromatic Path Topological Minor}.
Just observe that the orange (resp. turquoise) colors force all vertical (resp. horizontal) paths of $H$ to be mapped 
to vertical (resp. horizontal) paths and that the black and green colors force each edge of $H$ to be mapped to a path in some neighboring $\Gamma_{i,j}$'s joining a black vertex and a green vertex of $G$. 
\end{proof}

\end{document}